\documentclass[11pt,twoside]{article}
\usepackage[margin=1in]{geometry}
\usepackage[utf8]{inputenc}
\usepackage{natbib}
\usepackage{amsmath}
\usepackage{mathtools}
\usepackage{amsthm}
\usepackage{amsfonts}
\usepackage{amssymb}
\usepackage{graphicx}
\usepackage[ruled,vlined,linesnumbered,hangingcomment]{algorithm2e}
\usepackage{textcomp}
\usepackage{makecell}
\usepackage[flushleft]{threeparttable}
\usepackage[scr]{rsfso}
\usepackage{fancyhdr}
\usepackage{fontsize}
\usepackage{hyperref}
\usepackage{color}
\usepackage[dvipsnames]{xcolor}
\usepackage{bm}
\usepackage{multicol}
\usepackage{enumerate}
\usepackage[inline]{enumitem}
\usepackage[explicit]{titlesec}
\usepackage{titletoc}
\usepackage{xargs}
\usepackage{xifthen}
\usepackage{xparse}
\usepackage{etoolbox}
\usepackage{colortbl}
\usepackage{pdfpages}
\usepackage{sidecap}

\let\bm\mathbf

\bibpunct{(}{)}{;}{a}{}{,}

\relax

\def\BibTeX{{\rm B\kern-.05em{\sc i\kern-.025em b}\kern-.08em
  T\kern-.1667em\lower.7ex\hbox{E}\kern-.125em}}

\definecolor{LinkColor}{rgb}{0.36,0.40,0.68}
\definecolor{CiteColor}{HTML}{717c8e}
\hypersetup{
    colorlinks=true,
    linkcolor=LinkColor,
    filecolor=magenta,
    urlcolor=cyan,
    citecolor=CiteColor,
    anchorcolor=blue
}

\allowdisplaybreaks

\newboolean{showcomments}
\newboolean{showcolors}
\setboolean{showcomments}{true}
\setboolean{showcolors}{true}

\def\ldelim{\ast}
\def\rdelim{\ast}

\newcommand{\?}[1]{%
  \ifthenelse{\boolean{showcomments}}%
  {\textcolor{blue}{\emph{$\textcolor{blue}{\ldelim}$\,#1\,$\textcolor{blue}{\rdelim}$}}}%
  {}%
}%

\newcommand{\FixThis}[1]{%
  \ifthenelse{\boolean{showcomments}}%
  {\textcolor{gray}{\emph{$\ldelim$ #1 $\rdelim$}}}%
  {}%
}%

\newcommand{\TODO}[1][]{%
  \ifthenelse{\boolean{showcomments}}%
  {\emph{\(\textcolor{red}{\ldelim}\) \textcolor{red}{To-Do\IfNE{#1}{: #1}} \(\textcolor{red}{\rdelim}\)}}%
  {}%
}

\let\emptyset\varnothing

\let\EXISTS\exists
\let\NEXISTS\nexists
\renewcommand{\exists}{\EXISTS\,}
\renewcommand{\nexists}{\NEXISTS\,}
\let\FORALL\forall
\renewcommand{\forall}{\FORALL\,}

\newcommand{\HIDEDRAFT}[1]{}

\DontPrintSemicolon

\definecolor{AlgBodyKWColor}{rgb}{0.20,0.32,0.24}
\definecolor{AlgSetupKWColor}{rgb}{0.20,0.32,0.24}
\definecolor{AlgConditionalColor}{rgb}{0.00,0.41,0.43}

\SetKwInput{Given}{\textcolor{AlgSetupKWColor}{Given}}
\SetKwInput{Param}{\textcolor{AlgSetupKWColor}{Param}}
\SetKwInput{Require}{\textcolor{AlgSetupKWColor}{Require}}
\SetKwInput{Let}{\textcolor{AlgSetupKWColor}{Let}}
\SetKwInput{Fix}{\textcolor{AlgSetupKWColor}{Fix}}
\SetKwInput{Define}{\textcolor{AlgSetupKWColor}{Define}}
\SetKwInput{Design}{\textcolor{AlgSetupKWColor}{Design}}
\SetKwInput{Construct}{\textcolor{AlgSetupKWColor}{Construct}}

\SetKwIF{If}{ElseIf}{Else}%
  {\textcolor{AlgBodyKWColor}{if}}%
  {\textcolor{AlgBodyKWColor}{then:}}%
  {else if}{else}{endif}
\SetKwFor{ForEach}%
  {\textcolor{AlgBodyKWColor}{for each}}%
  {\textcolor{AlgBodyKWColor}{do:}}{endfch}
\SetKwFor{ForAll}%
  {\textcolor{AlgBodyKWColor}{for all}}%
  {\textcolor{AlgBodyKWColor}{do:}}%
  {endfall}

\SetKw{LetStatement}{\textcolor{AlgBodyKWColor}{Let:}}
\SetKw{InitStatement}{\textcolor{AlgBodyKWColor}{Initialize:}}
\SetKw{GivenStatement}{\textcolor{AlgBodyKWColor}{Given:}}
\SetKw{SolveStatement}{Solve:}

\SetSideCommentLeft
\SetNoFillComment
\SetCommentSty{mycommentfnt}
\SetKwComment{Comment}{$\textcolor{black}{\blacktriangleright}$ \ }{}

\makeatletter
\renewcommand\paragraph{\@startsection{paragraph}{8}{\z@}%
  {1ex \@plus1ex \@minus.2ex}%
  {-1em}%
  {\normalfont\normalsize\itshape}}
\makeatother
\makeatletter
\titleformat{\paragraph}[runin]{\normalfont\itshape}{}{}{\theparagraph #1.}
\makeatother

\makeatletter
\renewenvironment{proof}[1]{\par
  \pushQED{\qed}%
  \normalfont \topsep6\p@\@plus6\p@\relax
  \trivlist
  \item\relax
  {\itshape
  \proofname\ifthenelse{\equal{#1}{}}{\@addpunct{.}}{ (#1)\@addpunct{.}}}\hspace\labelsep\ignorespaces
}{%
  \popQED\endtrivlist\@endpefalse \vskip8\p@\@plus6\p@\relax
}
\makeatother

\makeatletter
\newenvironment{subproof}[1]{\par
  \pushQED{\qed}
  \normalfont \topsep6\p@\@plus6\p@\relax
  \trivlist
  \item\relax
  {\itshape
  \proofname\ifthenelse{\equal{#1}{}}{\@addpunct{.}}{ (#1)\@addpunct{.}}}\hspace\labelsep\ignorespaces
}{%
  \popQED\endtrivlist\@endpefalse \vskip8\p@\@plus6\p@\relax
}
\makeatother



\newtheoremstyle{claim}
{3pt}
{3pt}
{\itshape}
{}
{\bfseries}
{.}
{.5em}
{}

\newtheoremstyle{property}
{3pt}
{3pt}
{\itshape}
{}
{}
{}
{.5em}
{}

\newtheoremstyle{LemmaRepeat}
{3pt}
{3pt}
{\itshape}
{}
{\bfseries}
{.}
{.5em}
{\thmname{#1}\thmnote{ #3}}

\newtheorem{theorem}{Theorem}[section]
\newtheorem{corollary}[theorem]{Corollary}
\newtheorem{lemma}[theorem]{Lemma}

\newtheorem{fact}{Fact}[section]
\newtheorem{remark}{Remark}[section]
\newtheorem*{theorem*}{Theorem (restatement)}
\newtheorem*{lemma*}{Lemma (restatement)}
\newtheorem*{corollary*}{Corollary (restatement)}
\newtheorem*{fact*}{Fact (restatement)}
\newtheorem{definition}{Definition}[section]
\theoremstyle{claim}
\newtheorem{claim}[theorem]{Claim}
\theoremstyle{property}

\theoremstyle{LemmaRepeat}



\newlist{itemizetriangle}{itemize}{10}
\setlist[itemizetriangle,1]{label=$\textcolor{black}{\blacktriangleright}$}
\setlist[itemizetriangle,2]{label=\textbullet}
\setlist[itemizetriangle,3]{label=\textendash}
\setlist[itemizetriangle,4]{label=$\textcolor{black}{\triangleright}$}

\newlist{Enumerate}{enumerate}{5}
\setlist[Enumerate,1]{label={(\arabic*)}}
\setlist[Enumerate,2]{label={(\roman*)}}
\setlist[Enumerate,3]{label={(\alph*)}}

\newlist{EnumerateText}{enumerate}{5}
\setlist[EnumerateText,1]{label={(\alph*)}}
\setlist[EnumerateText,2]{label={(\arabic*)}}
\setlist[EnumerateText,3]{label={(\roman*)}}

\newlist{EnumerateThm}{enumerate}{5}
\setlist[EnumerateThm,1]{label={\textup{\thetheorem{} (\arabic*)}}, leftmargin=0.08\linewidth}
\setlist[EnumerateThm,2]{label={(\roman*)}}
\setlist[EnumerateThm,3]{label={(\alph*)}}

\newlist{EnumeratePropRef}{enumerate}{5}
\setlist[EnumeratePropRef,1]{label={\textup{\Prop \thetheorem{} (\arabic*) (HD)}}, leftmargin=0.18\linewidth,first=\itshape}
\setlist[EnumeratePropRef,2]{label={(\roman*)}}
\setlist[EnumeratePropRef,3]{label={(\alph*)}}

\newlist{EnumerateSub}{enumerate}{5}
\setlist[EnumerateSub,1]{label={(\roman*)}}
\setlist[EnumerateSub,2]{label={(\arabic*)}} 
\setlist[EnumerateSub,3]{label={(\alph*)}}

\newlist{EnumerateAlg}{enumerate}{5}
\setlist[EnumerateAlg,1]{label={\bfseries\small\arabic*:},left=0pt,itemindent=0pt}
\setlist[EnumerateAlg,2]{label={(\roman*)}}
\setlist[EnumerateAlg,3]{label={(\alph*)}}

\newlist{EnumerateInline}{enumerate*}{3}
\setlist[EnumerateInline,1]{label={(\roman*)}}
\setlist[EnumerateInline,2]{label={(\alph*)}}
\setlist[EnumerateInline,3]{label={(arabic*)}}

\definecolor{ColorGrayOne}  {rgb}{0.3,0.3,0.3}
\definecolor{ColorGrayTwo}  {rgb}{0.5,0.5,0.5}
\definecolor{ColorGrayThree}{rgb}{0.7,0.7,0.7}
\definecolor{ColorGrayFour} {rgb}{0.8,0.8,0.8}

\definecolor{ColorBlueGrayOne}  {rgb}{0.12,0.15,0.40}
\definecolor{ColorBlueGrayTwo}  {rgb}{0.25,0.30,0.60}
\definecolor{ColorBlueGrayThree}{rgb}{0.70,0.70,0.80}
\definecolor{ColorBlueGrayFour} {rgb}{0.80,0.80,0.90}

\definecolor{ColorGrayBlueOne}  {rgb}{0.30,0.45,0.60}
\definecolor{ColorGrayBlueTwo}  {rgb}{0.15,0.38,0.58}
\definecolor{ColorGrayBlueThree}{rgb}{0.25,0.47,0.38}
\definecolor{ColorGrayBlueFour} {rgb}{0.05,0.31,0.38}

\definecolor{ColorBlueOne}  {rgb}{0.00,0.50,0.65}
\definecolor{ColorBlueTwo}  {rgb}{0.00,0.38,0.58}
\definecolor{ColorBlueThree}{rgb}{0.00,0.47,0.38}
\definecolor{ColorBlueFour} {rgb}{0.00,0.31,0.38}

\definecolor{ColorBrownOne}  {rgb}{0.20,0.10,0.05}
\definecolor{ColorBrownTwo}  {rgb}{0.36,0.10,0.08}
\definecolor{ColorBrownThree}{rgb}{0.60,0.40,0.12}
\definecolor{ColorBrownFour} {rgb}{0.80,0.60,0.50}

\definecolor{ColorPurpleOne}  {rgb}{0.4,0.2,0.4}
\definecolor{ColorPurpleTwo}  {rgb}{0.6,0.3,0.6}
\definecolor{ColorPurpleThree}{rgb}{0.8,0.4,0.8}
\definecolor{ColorPurpleFour} {rgb}{0.9,0.5,0.9}

\definecolor{ColorGreenOne}  {rgb}{0.20,0.30,0.15}
\definecolor{ColorGreenTwo}  {rgb}{0.24,0.40,0.20}
\definecolor{ColorGreenThree}{rgb}{0.37,0.59,0.31}
\definecolor{ColorGreenFour} {rgb}{0.47,0.69,0.41}

\definecolor{ColorBlueGreenOne}  {rgb}{0.00,0.40,0.30}
\definecolor{ColorBlueGreenTwo}  {rgb}{0.10,0.36,0.30}
\definecolor{ColorBlueGreenThree}{rgb}{0.12,0.40,0.28}
\definecolor{ColorBlueGreenFour} {rgb}{0.15,0.45,0.32}

\definecolor{ColorBrown} {rgb}{0.8,0.8,0.8}

\definecolor{KWColor}{rgb}{0.0,0.2,0.4}
\definecolor{TermColor}{rgb}{0.00,0.06,0.56}
\definecolor{EmphColor}{rgb}{0.5,0.3,0.1}
\definecolor{ProbTermColor}{rgb}{0.02,0.30,0.24}

\definecolor{MathInlineColor}{rgb}{0.18,0.38,0.48}
\definecolor{MathInlineBrColor}{rgb}{0.36,0.36,0.20}

\newcommand{\Sub}[1]{\ifthenelse{\isempty{#1}}{}{_{#1}}}
\newcommand{\Sup}[1]{\ifthenelse{\isempty{#1}}{}{^{#1}}}
\newcommand{\IfNE}[2]{\ifthenelse{\isempty{#1}}{}{#2}}
\newcommand{\IfENE}[3]{\ifthenelse{\isempty{#1}}{#2}{#3}}

\newcommand{\OperatorName}[1]{\mathsf{#1}}

\newcommand{\R}{\mathbb{R}}

\newcommand{\Z}{\mathbb{Z}}

\newcommand{\N}{\mathcal{N}}

\newcommand{\dLn}{\longrightarrow~~}

\newcommand{\dCmt}[1][\qquad]{#1\blacktriangleright}

\newcommand{\sgn}[1][]{%
  \OperatorName{sign}%
  \ifthenelse{\isempty{#1}}{}%
  {\ifthenelse{\equal{#1}{-} \OR \equal{#1}{+}}{_{\scriptscriptstyle #1}}{_{#1}}}%
}
\newcommand{\sgnO}{\sgn[0]}

\NewDocumentCommand{\Sgn}{O{} D(){}}{\sgn[#1]\ifthenelse{\isempty{#2}}{}{\left( #2 \right)}}
\NewDocumentCommand{\SgnO}{D(){}}{\sgnO\ifthenelse{\isempty{#1}}{}{\left( #1 \right)}}
\NewDocumentCommand{\Round}{D(){}}{\Round\ifthenelse{\isempty{#1}}{}{\left( #1 \right)}}

\newcommand{\E}{\operatornamewithlimits{\mathbb{E}}}
\NewDocumentCommand{\Ex}{s t_ !O{} D(){}}%
{\mathbb{E}\IfBooleanT{#2}{\Sub{#3}}\ifthenelse{\isempty{#4}}{}{\IfBooleanTF{#1}{[ #4 ]}{\left[ {#4} \right]}}}

\newcommand{\Ent}{H}

\newcommand{\Entropy}[2][]{\Ent\ifthenelse{\isempty{#1}}{}{_{#1}} \ifthenelse{\isempty{#2}}{}{\left( #2 \right)}}

\newcommand{\T}{T}

\NewDocumentCommand{\I}{s m}{{%
 \def\dfmt{\left( #2 \right)}
 \def\tfmt{( #2 )}
 \mathbb{I}
 \IfBooleanTF{#1}{}{\ifthenelse{\isempty{#2}}{}{
    \mathchoice
     {\dfmt} 
     {\tfmt} 
     {\tfmt} 
     {\tfmt} 
}}}}

\RenewDocumentCommand{\Vec}{s +m}{\IfBooleanTF{#1}{#2}{\bm{\mathrm{#2}}}}
\NewDocumentCommand{\Mat}{s +m}{\IfBooleanTF{#1}{#2}{\bm{\mathrm{#2}}}}

\newcommand{\Set}[1]{\mathcal{#1}}

\newcommand{\supp}{\OperatorName{supp}}
\NewDocumentCommand{\Supp}{D(){}}{\supp\ifthenelse{\isempty{#1}}{}{\left( #1 \right)}}

\newcommand{\RV}[1]{#1}

\newcommand{\SparseSubset}[2]{\pi\Sub{#2}\IfNE{#1}{\left[ #1 \right]}}

\newcommand{\SSp}{\text{\hspace{1pt}}}

\newcommand{\MAT}[2]{\mathrm{Mat}\Sub{#1}\IfNE{#2}{( #2 )}}
\NewDocumentCommand{\Rank}{D(){}}{\OperatorName{rank}\ifthenelse{\isempty{#1}}{}{\left( #1 \right)}}
\NewDocumentCommand{\Dim}{D(){}}{\OperatorName{dim}\ifthenelse{\isempty{#1}}{}{\left( #1 \right)}}
\NewDocumentCommand{\Span}{D(){}}{\OperatorName{span}\ifthenelse{\isempty{#1}}{}{\left( #1 \right)}}
\NewDocumentCommand{\Col}{D(){}}{\OperatorName{col}\ifthenelse{\isempty{#1}}{}{\left( #1 \right)}}
\NewDocumentCommand{\Row}{D(){}}{\OperatorName{row}\ifthenelse{\isempty{#1}}{}{\left( #1 \right)}}
\NewDocumentCommand{\Nul}{D(){}}{\OperatorName{nul}\ifthenelse{\isempty{#1}}{}{\left( #1 \right)}}
\NewDocumentCommand{\Ker}{D(){}}{\OperatorName{ker}\ifthenelse{\isempty{#1}}{}{\left( #1 \right)}}
\NewDocumentCommand{\Proj}{s O{} D(){}}%
{\OperatorName{proj}\IfNE{#2}{_{#2}}\ifthenelse{\isempty{#3}}{}{\IfBooleanTF{#1}{(#3)}{\left( #3 \right)}}}
\NewDocumentCommand{\Comp}{s O{} D(){}}%
{\OperatorName{comp}\IfNE{#2}{_{#2}}\ifthenelse{\isempty{#3}}{}{\IfBooleanTF{#1}{(#3)}{\left( #3 \right)}}}

\newcommand{\Diag}{\OperatorName{diag}}

\newcommand{\lnorm}[1]{\ell_{#1}}

\NewDocumentCommand{\Log}{s O{} O{} D(){}}{\log#2#3\IfNE{#4}{\IfBooleanTF{#1}{( #4 )}{\left( #4 \right)}}}
\NewDocumentCommand{\Ln}{s O{} O{} D(){}}{\ln#2#3\IfNE{#4}{\IfBooleanTF{#1}{( #4 )}{\left( #4 \right)}}}
\NewDocumentCommand{\Lg}{s O{} O{} D(){}}{\lg#2#3\IfNE{#4}{\IfBooleanTF{#1}{( #4 )}{\left( #4 \right)}}}

\NewDocumentCommand{\Cos}{s O{} O{} D(){}}{\cos#2#3\IfNE{#4}{\IfBooleanTF{#1}{( #4 )}{\left( #4 \right)}}}
\NewDocumentCommand{\Arccos}{s O{} O{} D(){}}{\arccos#2#3\IfNE{#4}{\IfBooleanTF{#1}{( #4 )}{\left( #4 \right)}}}
\NewDocumentCommand{\Sin}{s O{} O{} D(){}}{\sin#2#3\IfNE{#4}{\IfBooleanTF{#1}{( #4 )}{\left( #4 \right)}}}
\NewDocumentCommand{\Arcsin}{s O{} O{} D(){}}{\arcsin#2#3\IfNE{#4}{\IfBooleanTF{#1}{( #4 )}{\left( #4 \right)}}}
\NewDocumentCommand{\Tan}{s O{} O{} D(){}}{\tan#2#3\IfNE{#4}{\IfBooleanTF{#1}{( #4 )}{\left( #4 \right)}}}
\NewDocumentCommand{\Arctan}{s O{} O{} D(){}}{\arctan#2#3\IfNE{#4}{\IfBooleanTF{#1}{( #4 )}{\left( #4 \right)}}}

\NewDocumentCommand{\pr}{D(){} +m !D||{}}{\operatornamewithlimits{Pr}\ifthenelse{\isempty{#1}}{}{_{#1}}\left[\SSp \ifthenelse{\isempty{#3}}{}{\left.} #2 \ifthenelse{\isempty{#3}}{}{\SSp\right|\SSp #3} \SSp\right]}

\NewDocumentCommand{\Var}{s D(){}}{\OperatorName{\bm{Var}}\ifthenelse{\isempty{#2}}{}{\IfBooleanTF{#1}{(#2)}{\left( #2 \right)}}}

\NewDocumentCommand{\Bern}{D(){}}{\OperatorName{Bern}\ifthenelse{\isempty{#1}}{}{\left( #1 \right)}}
\NewDocumentCommand{\Bin}{O{,} >{\SplitList{#1}} D(){}}{\OperatorName{Bin}\ifthenelse{\isempty{#2}}{}{\left( \printtwonop#2{}{} \right)}}
\NewDocumentCommand{\Geom}{D(){}}{\OperatorName{Geom}\ifthenelse{\isempty{#1}}{}{\left( #1 \right)}}
\NewDocumentCommand{\NBin}{O{,} >{\SplitList{#1}} D(){}}{\OperatorName{NBin}\ifthenelse{\isempty{#2}}{}{\left( \printthreenop#2{}{}{} \right)}}
\NewDocumentCommand{\HGeom}{O{,} >{\SplitList{#1}} D(){}}{\OperatorName{HGeom}\ifthenelse{\isempty{#2}}{}{\left( \printtwonop#2{}{} \right)}}
\NewDocumentCommand{\Pois}{D(){}}{\OperatorName{Pois}\ifthenelse{\isempty{#1}}{}{\left( #1 \right)}}
\NewDocumentCommand{\Unif}{O{,} >{\SplitList{#1}} D(){}}{\OperatorName{Unif}\ifthenelse{\isempty{#2}}{}{\left( \printtwonop#2{}{} \right)}}
\NewDocumentCommand{\DiscreteUnif}{D(){}}{\OperatorName{Unif}\ifthenelse{\isempty{#1}}{}{\left( #1 \right)}}
\NewDocumentCommand{\DiscreteUnifKSet}{m D(){}}%
{\OperatorName{Unif}_{#1}\ifthenelse{\isempty{#2}}{}{\left( #2 \right)}}

\newcommand{\printonenoparen}[1]{\ifthenelse{\isempty{#1}}{}{#1\text{-}}}
\newcommand{\printthree}[3]{\ifthenelse{\isempty{#1}}{}{(#1, #2, #3)\text{-}}}
\newcommand{\printtwo}[2]{\ifthenelse{\isempty{#1}}{}{(#1, #2)\text{-}}}
\newcommand{\printone}[2]{\ifthenelse{\isempty{#1}}{}{#1\text{-}}}
\newcommand{\printtwonop}[2]{#1, #2}
\newcommand{\printthreenop}[3]{#1, #2, #3}
\NewDocumentCommand{\UFF}{O{,} >{\SplitList{#1}} D<>{}}{\printtwo#2{}{}{} \OperatorName{UFF}}
\NewDocumentCommand{\CFF}{O{,} >{\SplitList{#1}} D<>{}}{\printtwo#2{}{}{} \OperatorName{CFF}}
\NewDocumentCommand{\PUFF}{O{,} >{\SplitList{#1}} D<>{}}{\printone#2{}{}{} \OperatorName{PUFF}}
\NewDocumentCommand{\RUFF}{O{,} >{\SplitList{#1}} D<>{}}{\printthree#2{}{}{} \OperatorName{RUFF}}
\NewDocumentCommand{\RUFFs}{O{,} >{\SplitList{#1}} D<>{}}{\printthree#2{}{}{} \OperatorName{RUFF}s}
\newcommand{\defeq}{\triangleq}

\NewDocumentCommand{\Th}{t' O{} O{}}{^{\IfBooleanT{#1}{\prime}#2\mathrm{th}#3}}

\newcommand{\BIGO}{O}
\newcommand{\BIGOMEGA}{\Omega}
\newcommand{\BIGTHETA}{\Theta}
\newcommand{\LITTLEO}{o}
\newcommand{\LITTLEOMEGA}{\omega}

\NewDocumentCommand{\BigO}{t' s D(){}}{\IfBooleanTF{#1}{\tilde{\BIGO}}{\BIGO}\IfBooleanTF{#2}{(#3)}{\ifthenelse{\isempty{#3}}{}{\left( #3 \right)}}}
\NewDocumentCommand{\BigOmega}{t' s D(){}}{\IfBooleanTF{#1}{\tilde{\BIGOMEGA}}{\BIGOMEGA}\IfBooleanTF{#2}{(#3)}{\ifthenelse{\isempty{#3}}{}{\left( #3 \right)}}}
\NewDocumentCommand{\BigTheta}{t' s D(){}}{\IfBooleanTF{#1}{\tilde{\BIGTHETA}}{\BIGTHETA}\IfBooleanTF{#2}{(#3)}{\ifthenelse{\isempty{#3}}{}{\left( #3 \right)}}}
\NewDocumentCommand{\LittleO}{t' s D(){}}{\IfBooleanTF{#1}{\tilde{\LITTLEO}}{\LITTLEO}\IfBooleanTF{#2}{(#3)}{\ifthenelse{\isempty{#3}}{}{\left( #3 \right)}}}
\NewDocumentCommand{\LittleOmega}{t' s D(){}}{\IfBooleanTF{#1}{\tilde{\LITTLEOMEGA}}{\LITTLEOMEGA}\IfBooleanTF{#2}{(#3)}{\ifthenelse{\isempty{#3}}{}{\left( #3 \right)}}}

\newcommand{\cIf}{\text{if}~}

\newcommand{\cOtherwise}{\text{otherwise}}

\newcommand{\+}{\phantom{-}}

\newcommand{\tab}{\ensuremath{~~~~}}
\newcommand{\TAB}{\ensuremath{~~}}
\newcommand{\SPACE}{\ensuremath{~}}
\newcommand{\Tab}[1][1]{%
    \ifthenelse{#1>0}{\tab}{}%
    \ifthenelse{#1>1}{\tab}{}%
    \ifthenelse{#1>2}{\tab}{}%
    \ifthenelse{#1>3}{\tab}{}%
    \ifthenelse{#1>4}{\tab}{}%
    \ifthenelse{#1>5}{\tab}{}%
    \ifthenelse{#1>6}{\tab}{}%
    \ifthenelse{#1>7}{\tab}{}%
    \ifthenelse{#1>8}{\tab}{}%
    \ifthenelse{#1>9}{\tab}{}%
    \ifthenelse{#1>10}{\tab}{}%
    \ifthenelse{#1>11}{\tab}{}%
    \ifthenelse{#1>12}{\tab}{}%
    \ifthenelse{#1>13}{\tab}{}%
    \ifthenelse{#1>14}{\tab}{}%
    \ifthenelse{#1>15}{\tab}{}%
    \ifthenelse{#1>16}{\tab}{}%
    \ifthenelse{#1>17}{\tab}{}%
    \ifthenelse{#1>18}{\tab}{}%
    \ifthenelse{#1>19}{\tab}{}%
    \ifthenelse{#1>20}{\tab}{}%
}

\let\texttick\'
\renewcommand{\'}{\TextOrMath{\texttick}{\;}}

\NewDocumentCommand{\THM}{s !O{ }}{\IfBooleanTF{#1}
{Thm.}{Theorem}#2}
\NewDocumentCommand{\COR}{s !O{ }}{\IfBooleanTF{#1}
{Cor.#1}{\ifthenelse{\equal{#2}{s}}{Corollaries}{Corollary#2}}}
\NewDocumentCommand{\PROP}{s !O{ }}{\IfBooleanTF{#1}
{Prop.}{Proposition}#2}
\NewDocumentCommand{\LEMMA}{s !O{ }}{\IfBooleanTF{#1}{Lemma}{Lemma}#2}
\NewDocumentCommand{\CLAIM}{s !O{ }}{\IfBooleanTF{#1}{Claim}{Claim}#2}
\NewDocumentCommand{\EQN}{!O{~}}{Eq.#1}
\NewDocumentCommand{\EQNS}{!O{~}}{Eqs.#1}
\NewDocumentCommand{\DEF}{s !O{ }}{\IfBooleanTF{#1}
{Def.}{Definition}#2}
\NewDocumentCommand{\ALG}{s !O{ }}{\IfBooleanTF{#1}
{Alg.}{Algorithm}#2}
\NewDocumentCommand{\PROB}{s !O{ }}{\IfBooleanTF{#1}
{Prob.}{Problem}#2}
\NewDocumentCommand{\PROPERTY}{s !O{ }}{Property#2}
\NewDocumentCommand{\LINE}{s !O{ }}{Line#2}
\NewDocumentCommand{\FACT}{s !O{ }}{Fact#2}

\NewDocumentCommand{\SUMMARY}{s !O{ }}{Summary#2}
\NewDocumentCommand{\PG}{s !O{ }}{pp.#2}
\NewDocumentCommand{\LE}{s !O{ }}{Lecture#2}

\NewDocumentCommand{\SECTION}{s !O{ }}{\IfBooleanTF{#1}{Sec.}{Section}#2}

\NewDocumentCommand{\FIG}{s !O{ }}{\IfBooleanTF{#1}
{Fig.}{Figure.}#2}
\NewDocumentCommand{\TABLE}{s !O{ }}{\IfBooleanTF{#1}
{Table}{Table}#2}

\newcommand{\RuleSep}[1][ProbFrameColor]{\textcolor{#1}{\rule{0.33\linewidth}{0.4pt}}}
\newcommand{\CenterRuleSep}[1][]{%
  \IfNE{#1}{\vspace{-\baselineskip}}
  \begin{center}%
  \RuleSep%
  \end{center}%
}

\NewDocumentCommand{\TextItem}{+m !O{}}{(#1)#2}

\newcommand{\Const}[1]{#1}

\newcommand{\lrparens}[2]{\ifthenelse{\isempty{#1}}{\left( #2 \right)}{( #2 )}}

\newcommand{\Net}[2][]{\mathcal{C}_{#2}\IfNE{#1}{^{(#1)}}}
\newcommand{\BallNet}[2][]{\mathcal{D}_{#2}\IfNE{#1}{^{(#1)}}}
\newcommand{\Ball}[2][]{\mathcal{B}_{#2}\IfNE{#1}{^{(#1)}}}
\NewDocumentCommand{\BallSparseSphere}{O{} +m D(){}}{\mathcal{B}_{#2}\IfNE{#1}{^{(#1)}}( #3 ) \cap \Sphere{n} \cap \SparseSubspace{k}{n}}
\newcommand{\Sphere}[2][]{\ifthenelse{\isempty{#1}}{\mathcal{S}^{#2-1}}{\mathcal{S}^{#2}}}

\NewDocumentCommand{\Dist}{s +m +m}{d\IfBooleanTF{#1}{( #2, #3 )}{\left( #2, #3 \right)}}
\NewDocumentCommand{\DistS}{s O{} +m +m}{{%
 \def\dfmt{\big( #3, #4 \big)}
 \def\tfmt{( #3, #4 )}
 d_{\mathcal{S}^{n-1}}#2
 \IfBooleanTF{#1}{\dfmt}{
    \mathchoice
     {\dfmt} 
     {\tfmt} 
     {\tfmt} 
     {\tfmt} 
}}}

\newcommand{\Variable}[1]{#1}
\newcommand{\Value}[1]{#1}
\newcommand{\Function}[1]{#1}
\newcommand{\FunctionVariable}[1]{#1}

\newcommand{\ThresholdOp}{\mathcal{T}\!}
\newcommand{\ThresholdMat}{\Mat{T}\!}
\newcommand{\ThresholdMatEntry}{\Mat*{T}\!}
\NewDocumentCommand{\Threshold}{s +m D(){}}{%
  \ThresholdOp_{#2}%
  \ifthenelse{\isempty{#3}}
  {}
  {\IfBooleanTF{#1}{( #3 )}{\left( #3 \right)}}
}

\NewDocumentCommand{\ThresholdSet}{s +m D(){}}{%
  \ThresholdOp_{#2}%
  \ifthenelse{\isempty{#3}}
  {}
  {\IfBooleanTF{#1}{( #3 )}{\left( #3 \right)}}
}

\NewDocumentCommand{\ThresholdSetMat}{s +m O{}}{%
  \IfBooleanTF{#1}{\ThresholdMatEntry_{\IfNE{#3}{#3;}#2}}{\ThresholdMat_{#2}}%
}

\newcommand{\Text}[1]{\ensuremath{\,\text{#1}\,}}

\newcommand{\ObjectiveFn}{\mathcal{J}}

\newcommand{\Iter}{t}

\renewcommand{\Variable}[1]{#1}

\newcommand{\SparseSubspace}[2]{\Sigma_{#1}^{#2}}
\newcommand{\SparseRealSubspace}[2]{\SparseSubspace{#1}{#2}}
\newcommand{\SparseSphereSubspace}[2]{\mathcal{S}^{#2-1} \cap \SparseSubspace{#1}{#2}}

\NewDocumentCommand{\tNetBIHT}{!O{ }}{\( \frac{\MinErr}{2} \)-net#1}

\newcommand{\PlaceHolder}[1][]{\ensuremath{\textcolor{red}{\triangleright\, \text{todo\IfNE{#1}{: #1}} \,\triangleleft}}}

\NewDocumentCommand{\UB}{s}{%
    \binom{n}{k}^{3}
    \left( \frac{24}{\Epsilon} \right)^{2k}
    \frac{8}{\rho}
}

\newcommand{\Tau}{\tau}
\newcommand{\Epsilon}{\epsilon}
\newcommand{\Rho}{\rho}
\newcommand{\DDelta}{\delta}
\newcommand{\Eta}{\eta}
\newcommand{\Beta}{\beta}
\newcommand{\Alpha}{\alpha}
\newcommand{\Varepsilon}{\varepsilon}

\newcommand{\UnivConstA}{a}
\newcommand{\UnivConstc}{c}
\newcommand{\UnivConstb}{b}
\newcommand{\UnivConstC}{c}
\newcommand{\UnivConstB}{b}

\newcommand{\MeasMat}{\Mat{A}}
\newcommand{\MeasVec}{\Vec{A}}

\newcommand{\hA}[1][]{h_{\MeasMat\IfNE{#1}{;#1}}}
\newcommand{\gA}[1][]{g_{\MeasMat\IfNE{#1}{;#1}}}

\newcommand{\Coords}[1]{#1}

\newcommand{\UBNet}{\binom{n}{k} \left( \frac{6}{\Tau} \right)^{k}}
\newcommand{\UBNetNet}{\binom{n}{k}^{2} \left( \frac{6}{\Tau} \right)^{2k}}

\newcommand{\NUMLSas}{2}
\newcommand{\NUMLSat}{6}
\newcommand{\NUMLSa}{8}

\NewDocumentCommand{\VarepsilonDef}{s O{}}{%
  4\UnivConstc_{1}
  \sqrt{\frac{\Epsilon}{\UnivConstC} \varepsilon( \Iter-1 )}
  +
  4 \UnivConstc_{2} \frac{\Epsilon}{\UnivConstC}
  #2
  \IfBooleanF{#1}{,\quad \Iter \in \Z_{+}}%
}

\newcommand{\VarepsilonAsymptotic}{%
  \left(
    2 \UnivConstc_{1}
    \left( \UnivConstc_{1} + \sqrt{\UnivConstc_{1}^{2} + \UnivConstc_{2}} \right)
    +
    \UnivConstc_{2}
  \right)
  \frac{4 \Epsilon}{\UnivConstC}%
}

\newcommand{\VarepsilonAsymptoticIntermediate}{%
  \frac{32 \Epsilon}{\UnivConstC}%
}

\newcommand{\ErrorRecurrenceDef}[1]{%
  4 \UnivConstc_{1}
  \sqrt{\frac{\Epsilon}{\UnivConstC} \DistS{\Vec{x}}{\Vec{\hat{x}}^{(#1)}}}
  +
  4 \UnivConstc_{2}
  \frac{\Epsilon}{\UnivConstC}%
}

\newcommand{\UnivConstAValue}{16}
\newcommand{\UnivConstBValue}{379.1038}
\newcommand{\UnivConstCValue}{32}

\newcommand{\UnivConstcOneValue}{\sqrt{\frac{3\pi}{\UnivConstB}} \left( 1 + \frac{16\sqrt{2}}{3} \right)}
\newcommand{\UnivConstcTwoValue}{\frac{3}{\UnivConstB} \left( 1 + \frac{4\pi}{3} + \frac{8\sqrt{3\pi}}{3} + 8\sqrt{6\pi} \right)}

\newcommand{\bb}[1]{\mathbb{#1}}
\newcommand{\fl}[1]{\mathbf{#1}}

\newcommand{\s}[1]{\mathsf{#1}}

\newcommand{\remove}[1]{}

\newcommand{\one}{one\xspace}

\newcommand{\onebit}{one-bit\xspace}

\newcommand{\ksparserealvalued}{\( k \)-sparse real-valued\xspace}

\newcommand{\xii}{\xi_{i}}
\newcommand{\subgradient}{subgradient\xspace}

\newcommand{\topkhardthresholding}{top-\( k \) hard thresholding\xspace}
\newcommand{\Topkhardthresholding}{Top-\( k \) hard thresholding\xspace}
\newcommand{\subsethardthresholding}{subset hard thresholding\xspace}
\newcommand{\Subsethardthresholding}{Subset hard thresholding\xspace}

\newcommand{\standardnormal}{Gaussian\xspace}
\usepackage{comment}

\newcommand{\GammaX}{\gamma}
\newcommand{\DDeltaX}{\delta'}
\newcommand{\ConstD}{d'}
\newcommand{\UnivConstD}{d}
\newcommand{\UnivConstDValue}{512}
\newcommand{\UnivConstBX}{\UnivConstB}
\newcommand{\UnivConstAX}{a'}
\newcommand{\UnivConstAXValue}{\EDITX{20}}
\newcommand{\UnivConstAXX}{a''}
\newcommand{\UnivConstAXXValue}{\EDITX{8}}
\newcommand{\kX}{k'}
\newcommand{\JX}{J}

\newcommand{\Luv}{\RV{L}_{\Vec{u},\Vec{v}}}

\newcommand{\Rhatuv}{\Vec{\EDIT{\hat{R}_{\Vec{u},\Vec{v}}}}}
\newcommand{\Rhatuvi}[1][i]{\EDIT{\Vec*{\hat{R}}_{#1;\Vec{u},\Vec{v}}}}
\newcommand{\RLCOND}{\EDIT{\Vec{\hat{R}}_{\Vec{u},\Vec{v}} = \Vec{r}, \Luv = \Variable{\ell}}}
\newcommand{\rSet}{\EDIT{\{ 0,1 \}^{m}}}
\newcommand{\RCOND}{\EDIT{\Vec{\hat{R}}_{\Vec{u},\Vec{v}} = \Vec{r}}}

\renewcommand{\UnivConstcOneValue}{\sqrt{\frac{3\pi}{\UnivConstB \UnivConstD}} \left( 1 + \frac{16\sqrt{2}}{3} \right)}
\renewcommand{\UnivConstcTwoValue}{\frac{90 \sqrt{2}}{\UnivConstBX}}

\renewcommand{\UnivConstcOneValue}{\EDIT{\sqrt{\frac{\pi}{\UnivConstB \UnivConstD}}( \sqrt{3} + 16 )}}

\newcommand{\TauValue}{\frac{\GammaX}{\UnivConstD {\Log( 2e/\GammaX )}}}
\newcommand{\GammaXValue}{\frac{\DDelta}{\UnivConstB \sqrt{\Log( 2e/\GammaX )}}}
\newcommand{\GammaXValueX}{\frac{\UnivConstB_{2} \DDelta}{5 \sqrt{72 \Log( 2e/\GammaX )}}}
\newcommand{\QValue}{\GammaX m}

\newcommand{\CMPLXRAICLogTerm}{%
  \log
  \left(
    \binom{n}{k}^{2} \binom{n}{\kO}
    \left( \frac{12 \UnivConstB \UnivConstD \Log[^{3/2}]( 2e/\GammaX )}{\DDelta} \right)^{\kO}
    \left( \frac{\UnivConstA}{\Rho} \right)
  \right)%
}
\newcommand{\CMPLXRAICLogTermX}{%
  \log
  \left(
    \binom{n}{k}^{2} \binom{n}{\kO}
    \left( \frac{12 \UnivConstB \UnivConstD \Log( 2e/\GammaX )}{\GammaX} \right)^{\kO}
    \left( \frac{\UnivConstA}{\Rho} \right)
  \right)%
}

\newcommand{\CMPLXRAICOne}{%
  \frac{\UnivConstB \UnivConstD}{\DDelta}
  \CMPLXRAICLogTerm
}
\newcommand{\CMPLXRAICOneX}{%
  \frac{\UnivConstB \UnivConstD}{\DDelta}
  \CMPLXRAICLogTermX
}

\newcommand{\CMPLXRAICTwo}{%
  \frac{\UnivConstB \UnivConstD k}{\DDelta} \Log( \frac{1}{\GammaX} ) \sqrt{\Log( \frac{2e}{\GammaX} )}
  +
  \frac{64 \UnivConstB}{\DDelta} \Log( \binom{n}{\kO} \EDITX{\frac{\UnivConstAXX}{\Rho}} ) \sqrt{\Log( \frac{2e}{\GammaX} )}
  +
  \frac{\UnivConstB \kOX}{\DDelta}
  \EDITX{\Log( \frac{en}{\kOX} )}
  +
  \frac{\UnivConstB}{\DDelta}
  \Log( \EDITX{\frac{\UnivConstAX}{\Rho}} )
}

\newcommand{\CMPLXRAICBigO}{%
    \BigO'( \frac{k}{\DDelta} \Log( \frac{n}{k} ) \sqrt{\Log( \frac{1}{\DDelta} )} + \frac{k}{\DDelta} \Log[^{3/2}]( \frac{1}{\DDelta} ) )
}

\newcommand{\xToDo}[1]{\normalfont\textcolor{todocolor}{$\textcolor{todocolor}{\ldelim}$\textbf{To\,do}\IfENE{#1}{\textbf{:}\;Fill this in.}{\textbf{:}\;{#1}}$\textcolor{todocolor}{\rdelim}$}}
\newcommand{\ToDo}[1]{%
  \colorlet{origcolor}{.}\everymath{\color{todocolor}}%
  \ifthenelse{\boolean{showcomments}}%
  {\relax\ifmmode\text{\xToDo{#1}}\else\xToDo{#1}\fi}%
  {}%
  \everymath{\color{MathInlineColor}}%
}%

\newcommand{\Enum}[2][]{\IfENE{#1}{(#2)}{\Label{#1}{(#2)}}\xspace}
\newcommand{\TagEqn}{\stepcounter{equation}\tag{\theequation}}

\MakeRobust{\ref}

\makeatletter
\NewDocumentCommand{\Label}{s m m}{%
  \@bsphack
  \csname phantomsection\endcsname 
  \IfBooleanF{#1}{#3}%
  \def\@currentlabel{#3}{#2}
  \@esphack
}
\makeatother

\makeatletter
\newcommand{\subalign}[1]{%
  \vcenter{%
    \Let@ \restore@math@cr \default@tag
    \baselineskip\fontdimen10 \scriptfont\tw@
    \advance\baselineskip\fontdimen12 \scriptfont\tw@
    \lineskip\thr@@\fontdimen8 \scriptfont\thr@@
    \lineskiplimit\lineskip
    \ialign{\hfil$\m@th\scriptstyle##$&$\m@th\scriptstyle{}##$\hfil\crcr
      #1\crcr
    }%
  }%
}
\makeatother

\NewDocumentCommand{\THEOREM}{s O{~}}{\IfBooleanTF{#1}{Theorem}{Theorem}#2\ignorespaces}
\NewDocumentCommand{\THEOREMS}{s O{~}}{\IfBooleanTF{#1}{Theorems}{Theorems}#2\ignorespaces}
\RenewDocumentCommand{\LEMMA}{s O{~}}{\IfBooleanTF{#1}{Lemma}{Lemma}#2\ignorespaces}
\NewDocumentCommand{\LEMMAS}{s O{~}}{\IfBooleanTF{#1}{Lemmas}{Lemmas}#2\ignorespaces}
\RenewDocumentCommand{\CLAIM}{s O{~}}{\IfBooleanTF{#1}{Claim}{Claim}#2\ignorespaces}
\NewDocumentCommand{\STEP}{s O{~}}{\IfBooleanTF{#1}{Step}{Step}#2\ignorespaces}
\NewDocumentCommand{\STEPS}{s O{~}}{\IfBooleanTF{#1}{Steps}{Steps}#2\ignorespaces}
\NewDocumentCommand{\APPENDIX}{s O{~}}{\IfBooleanTF{#1}{Appendix}{Appendix}#2\ignorespaces}
\NewDocumentCommand{\APPENDICES}{s O{~}}{\IfBooleanTF{#1}{Appendices}{Appendices}#2\ignorespaces}

\NewDocumentCommand{\RHS}{s}{\IfBooleanTF{#1}{RHS}{right-hand-side}\xspace}
\NewDocumentCommand{\LHS}{s}{\IfBooleanTF{#1}{LHS}{left-hand-side}\xspace}
\NewDocumentCommand{\RHSs}{s}{\IfBooleanTF{#1}{RHS}{right-hand-sides}\xspace}
\NewDocumentCommand{\LHSs}{s}{\IfBooleanTF{#1}{LHS}{left-hand-sides}\xspace}

\newcommand{\iid}[1][~]{i.i.d.#1}

\newcommand{\AlignTab}{\phantom{=}~~}

\RenewDocumentCommand{\Vec}{s +m}{\IfBooleanTF{#1}{#2}{\mathbf{#2}}}
\RenewDocumentCommand{\Mat}{s +m}{\IfBooleanTF{#1}{#2}{\mathbf{#2}}}
\NewDocumentCommand{\BVec}{s +m}{\mathbf{1}^{#2}}

\newcommand{\hfrac}[2]{#1/#2}
\newcommand{\ZeroTo}[1]{\{ 0, \dots, #1 \}}
\newcommand{\Mid}{|}

\newcommand{\kO}{\EDITX{k_{0}}}
\newcommand{\kOX}{\EDITX{k'_{0}}}
\newcommand{\KO}{\kO}
\newcommand{\KXO}{\kO}

\definecolor{orig}{HTML}{000000}
\definecolor{edit}{HTML}{000000}
\definecolor{editX}{HTML}{000000}
\newcommand{\ORIG}[1]{}
\newcommand{\EDIT}[1]{#1}
\newcommand{\EDITX}[1]{#1}
\newcommand{\BEGINEDIT}{}
\newcommand{\ENDEDIT}{}
\newcommand{\BEGINEDITX}{}
\newcommand{\ENDEDITX}{}

\newenvironment{EDITb}{%
  \ignorespaces%
	\color{edit}%
	}{\color{black}\noindent%
	\ignorespacesafterend}

\newcommand\blfootnote[1]{%
  \begingroup
  \renewcommand\thefootnote{}\footnote{#1}%
  \addtocounter{footnote}{-1}%
  \endgroup
}

\title{Binary Iterative Hard Thresholding Converges with Optimal Number of Measurements  for 1-Bit Compressed Sensing}

\author{Namiko Matsumoto \\
\and Arya Mazumdar\blfootnote{The authors are with the University of California San Diego. This work is supported in part by NSF awards 2133484 and 2127929. Emails: \texttt{\{nmatsumo,arya\}@ucsd.edu}.}}
\date{}
\begin{document}

\maketitle

\begin{abstract}
     {\em Compressed sensing} has been a very successful high-dimensional signal acquisition and recovery technique that relies on linear operations. However, the actual measurements of signals have to be quantized before storing or processing them. \ORIG{1(One)-bit} \EDIT{1-bit (or one-bit)} compressed sensing is a heavily quantized version of compressed sensing, where each linear measurement of a signal is reduced to just one bit: the sign of the measurement. Once enough of such measurements are collected, the recovery problem in 1-bit compressed sensing aims to find the original signal with as much accuracy as possible. The recovery problem is related to the traditional ``halfspace-learning'' problem in learning theory.

    For recovery of sparse vectors, a popular reconstruction method from \one-bit measurements is the {\em binary iterative hard thresholding (BIHT)} algorithm. The algorithm is a simple projected \subgradient descent method, and is known to converge well empirically, despite the nonconvexity of the problem. The convergence property of BIHT was not  theoretically fully justified 
    (e.g., it is known that a number of measurement greater than $\max\{k^{10}, 24^{48}, k^{3.5}/\epsilon\}$, where $k$ is the sparsity and $\epsilon$ denotes the approximation error, is sufficient, Friedlander et al., 2021). In this paper we show that the BIHT estimates converge to the original signal with only \EDIT{\(   \frac{k}{\epsilon}   \)} measurements (up to logarithmic factors). Note that, this dependence on $k$ and $\epsilon$ is  optimal for any recovery method in 1-bit compressed sensing. With this result, to the best of our knowledge, BIHT is the only practical and efficient (polynomial time) algorithm that requires the optimal number of measurements in all parameters (both $k$ and $\epsilon$). This is also an example of a gradient descent algorithm converging to the correct solution for a nonconvex problem, under suitable structural conditions.
\end{abstract}

\section{Introduction}
One-bit compressed sensing (1bCS) is a basic nonlinear sampling method for high-dimensional sparse signals, introduced first in \cite{DBLP:conf/ciss/BoufounosB08}.
Consider an unknown sparse signal  $\fl{x} \in \bb{R}^n$ with sparsity (number of nonzero coordinates) $\left|\left|\fl{x}\right|\right|_0 \le k$, where $k \ll n$. In the 1bCS framework,
measurements of $\fl{x}$ are obtained with a
 sensing matrix $\fl{A}\in \bb{R}^{m \times n}$ via the observations of signs:
 \begin{align*}
     \fl{b} = \s{sign}(\fl{Ax}).
 \end{align*}
 The sign function (formally defined later) is simply the $\pm$ signs of the coordinates.

Compressed sensing,
the method of obtaining signals by taking few linear projections~\cite{DBLP:journals/tit/Donoho06,candes2006robust} has seen a lot of success in the past two decades. 1bCS is an extremely quantized version of compressed sensing where only \one bit per sample of the signal is observed.
 In terms of nonlinearity, this is  one of the simplest examples of a single-index model~\cite{plan2016generalized}: $y_i = f(\langle \fl{a}_i, \fl{x} \rangle), i =1, \dots, m$, where $f$ is a coordinate-wise nonlinear operation.
As a practical case study and for its aesthetic appeal, 1bCS has been studied with interest in the last few years, for example, in \cite{DBLP:conf/ciss/HauptB11,GNJN13,ABK17,DBLP:journals/tit/PlanV13,DBLP:conf/aistats/Li16}.

Notably, it was shown in \cite{JLBB13} that \ORIG{$m = \tilde{\Theta}(k/\epsilon)$}\EDIT{$m = \Theta(k/\epsilon)$} measurements are necessary and sufficient (up to logarithmic factors) to approximate $\fl{x}$ within an $\epsilon$-ball,
\ORIG{But the reconstruction method used to obtain this measurement complexity is via exhaustive search, which is practically infeasible.}%
\EDIT{but no practically feasible reconstruction algorithm achieving this error rate was proposed.}
A linear programming based solution which runs in polynomial time and requires $O(\frac{k}{\epsilon^5}\log^2\frac{n}{k})$ measurements was provided in \cite{plan2013one}. Note the suboptimal dependence on $\epsilon$. 

An incredibly well-performing algorithm turned out to the {\em binary iterative hard thresholding} (BIHT) algorithm, proposed in the former work~\cite{JLBB13}. BIHT is a simple iterative algorithm that converges to the correct solution quickly in practice. However, until later, the reason of its good performance was somewhat unexplained, barring the fact that it is actually a proximal gradient descent algorithm on a certain loss function (provided in Eq.~\eqref{eqn:costFunction}). In the algorithm, the projection is taken onto a nonconvex set (namely, selecting the ``top-$k$'' coordinates and then normalizing), which usually makes a theoretical analysis unwieldy.
Since the work of \cite{JLBB13} there has been some progress explaining the empirical success of the BIHT algorithm. In particular, it was shown in
\EDIT{\cite[Sec.~5]{jacques2013quantized} that after only the first iteration of the BIHT algorithm, an approximation error $\epsilon$ is achievable with $O(\frac{k}{\epsilon^2})$ measurements, up to logarithmic factors.}
Similar results also appear in
\cite[Sec.~3.5]{plan2017high}.
In all these results,  the dependence on $\epsilon$, which is also referred to as the error-rate, is suboptimal. Furthermore, these works also do not show convergence as the algorithm iterates further. 
Beyond the first iteration, it was shown in \cite{liu2019one} that the iterates of BIHT remain bounded, maintaining the same order of accuracy for the subsequent iterations. This, however, does not imply a reduction in the approximation error after the first iteration. This issue has been  mitigated in \cite{friedlander2021nbiht}, which uses a {\em normalized} version of the BIHT algorithm. While \cite{friedlander2021nbiht} manage to show that the normalized BIHT algorithm can achieve optimal dependence on the error-rate as the number of iterations of BIHT tends to infinity, i.e., $m \sim \frac{1}{\epsilon}$, their result is only valid when $m > \max\{ck^{10} \log^{10}\frac{n}{k}, 24^{48}, \frac{c'}{\epsilon} (k\log\frac{n}{k})^{7/2} \}.$ This clearly is highly sub-optimal in terms of dependence on $k$, and does not explain the empirical performance of the algorithm. This has been left as the main open problem in this area as per \cite{friedlander2021nbiht}.

\subsection{Our Contribution and Techniques}
\label{outline:intro|>contributions}
In this paper, we show that the normalized BIHT algorithm converges with a sample complexity having optimal dependence on both the sparsity $k$ and error $\epsilon$ (see, Theorem~\ref{thm:main:convergence:t-rate} below). As such, we further show the convergence rate with respect to iterations for this algorithm. In particular, we show that the approximation error of BIHT decays as $O(\epsilon^{1-2^{-t}})$ with the number of iteration $t$. This encapsulates the very fast convergence of BIHT to the $\epsilon$-ball of the actual signal. Furthermore, this also shows that after just \one iteration of BIHT, an approximation error of $\sqrt{\epsilon}$ is achievable, with $O(\frac{k}{\epsilon}\log \frac{n}{k})$ measurements, which matches the observations of \cite{jacques2013quantized,plan2017high} regarding the performance of BIHT with just \one iteration. Due to the aforementioned fast rate, the approximation error quickly converges to $\epsilon$ resulting in a polynomial time algorithm for recovery in 1bCS with only $\tilde{O}(\frac{k}{\epsilon})$ measurements, the optimal.

There are several difficulties in analyzing BIHT that were pointed out in the past, for example in \cite{friedlander2021nbiht}. First of all, the loss function is not differentiable, and therefore one has to rely on (sub)gradients, which prohibits an easier analysis of convergence.
Secondly, the algorithm projects onto nonconvex sets, so the improvement of the approximation in each iteration is not immediately apparent.
To tackle these hurdles, the key idea is to use some structural property of the measurement or sampling matrix.
Our result relies on such a property of the sampling matrix $\fl{A}$, called the restricted approximate invertibility condition (RAIC).
A somewhat different invertibility property of a matrix also appears in \cite{friedlander2021nbiht}. However, our definition, which looks more natural, allows for a significantly different analysis that yields the improved sample complexity.
Thereafter, we show that random matrices with i.i.d. \standardnormal entries satisfy the invertibility condition with overwhelmingly large probability.

The invertibility condition that is essential for our proof intuitively states that treating the signed measurements as some ``scaled linear'' measurements should lead to adequate estimates, which is an overarching theme of recovery in generalized linear models.
Further, our condition quantifies the ``goodness'' of these estimates in a way that allows us to show a contraction in the BIHT iterations. This contraction of approximation error comes naturally from our definition.
In contrast, while a similar idea appears in \cite{friedlander2021nbiht}, showing the contraction of approximate error is a highly involved exercise therein.
As another point of interest, \cite[Sec.~4.2]{JLBB13} empirically observed that in normalized BIHT, the step-size of the gradient descent algorithm must be carefully chosen, or else the algorithm will not converge.
Our definition of the invertibility condition gives some intuitive justification on why the algorithm is so sensitive to step-size. Our analysis relies on the step-size being set exactly to $\eta = \sqrt{2\pi}$. More generally, if $\eta$ were to deviate too far from $\sqrt{2\pi}$, the contraction would be lost.

\BEGINEDIT
With this all said, \ORIG{the technical burden of our main result}\EDIT{the crucial technical work for our main result} turns out to be showing that \standardnormal matrices satisfy the invertibility condition (see, Definition~\ref{def:raic:modified} in Section~\ref{outline:intro:overview-techniques|>raic-def-thm}).
We need to show that the condition holds for every pair of sparse unit vectors with bounded probability.
We resort to constructing a cover, an ``epsilon-net,'' of the unit sphere, and then separating the analysis for the invertibility condition into two regimes.
First, in the so-called ``large-distance'' regime, we show that the condition is satisfied for two vectors in the epsilon-net whose distance is above a particular threshold \( \tau > 0 \).
Second, in the ``small-distance'' regime, we show that a similar, though actually stronger, condition holds for every \(  k  \)-sparse unit vector paired with each of its close by ``neighbors,'' or more precisely, each point in the \( \tau \)-ball around it.
This second condition further implies that only a small error is added to the first condition when instead of the net points, vectors close to one or both of them are considered.
Together, the above can be combined to handle all possible cases such that the desired invertibility condition holds for every pair of sparse unit vectors.

In proving the invertibility condition
in each of these two regimes, the primary concern is characterizing and bounding a function
\(
  \hA : \R^{n} \times \R^{n} \to \R^n
\)
of the form
\begin{align*}
  \hA \left( \Vec{u}, \Vec{v} \right)
  =
  \frac{\sqrt{2\pi}}{m}
  \MeasMat^{\T}
  \cdot
  \frac{1}{2}
  \left( \Sgn( \MeasMat \Vec{u} ) - \Sgn( \MeasMat \Vec{v} ) \right)
,\end{align*}
where
\(
  \Vec{u}, \Vec{v} \in \R^{n}
\)
\EDIT{are \(  k  \)-sparse unit vectors.}
(Note that, due to the sparsity induced by the thresholding operation of BIHT, we actually consider the function \(  \hA  \) under a restriction to union of the support of \(  \Vec{u}  \) and \(  \Vec{v}  \), and each subset of coordinates, \(  \Coords{J} \subseteq [n]  \) with \(  | \Coords{J} | \leq k  \). However, for the purposes of this intuitive overview, we will ignore this so as to avoid overloading the discussion with notations and formalities.)
This is achieved by a three-term orthogonal decomposition of \(  \hA  \) and curated concentration inequalities associated with these terms, where the latter form the bulk of the techniques used in this paper.
There are two primary reasons for dividing the analysis for \(  \hA  \) into the ``large-'' and ``small-distances'' regimes.
First, while constructing an epsilon-net over the set of \(  k  \)-sparse unit vectors is a standard and useful approach, the analysis cannot immediately be extended to handle points outside of the net via, e.g., the triangle inequality due to the nonlinearity of the \(  \sgn  \) function.
Therefore, this extension to arbitrary points requires separate analysis which is provided in the ``small-distances'' regime.
Crucially, given that this latter regime considers small neighborhoods of points, the local binary stable embedding of \cite{oymak2015near} can be applied to obtain the uniform result for arbitrary points in these small neighborhoods.
To understand the second reason behind our two-regime approach, first notice
that in the above equation,
\(
  |
    \frac{1}{2}
    ( \Sgn( \MeasMat \Vec{u} ) - \Sgn( \MeasMat \Vec{v} ) )
  |
  =
  \I{\Sgn{}( \MeasMat \Vec{u} ) ) \neq \Sgn{}( \MeasMat \Vec{v} )}
\),
where \( | \cdot | \) takes the absolute value of each of the vector's entries, and
\( \I{} \) denotes the indicator function
(see, \SECTION \ref{outline:preliminaries-notation|>notation} for a more rigorous definition).
A key component of our analysis is characterizing this random vector---or more specifically, the number of nonzero entries in it---as it will facilitate the derivation and use of the three concentration inequalities for \(  \hA  \) and will ultimately lead to the invertibility condition's scaling with the distance between points.
In the ``large-distances'' regime, we can apply standard techniques to bound \(  \| \I{\Sgn{}( \MeasMat \Vec{u} ) ) \neq \Sgn{}( \MeasMat \Vec{v} )} \|_{0}  \) for pairs of points, \(  \Vec{u}, \Vec{v}  \), in the net---in particular, there is a Chernoff bound that provides a sufficient bound on \(  \| \I{\Sgn{}( \MeasMat \Vec{u} ) ) \neq \Sgn{}( \MeasMat \Vec{v} )} \|_{0}  \).
However, when considering points with distances below a certain threshold, such a standard Chernoff bound leads to a suboptimal sample complexity.
Instead, to bound \(  \| \I{\Sgn{}( \MeasMat \Vec{u} ) ) \neq \Sgn{}( \MeasMat \Vec{v} )} \|_{0}  \) for these close-together points, \(  \Vec{u}, \Vec{v}  \), we leverage the local binary stable embedding studied in \cite{oymak2015near}, which is a stronger result and allows the preservation of the optimal sample complexity (up to logarithmic factors).
As an aside, it is additionally worth mentioning that in the ``small-distances'' regime, the uniform result for arbitrary points also in part stems from the observation that, after fixing the measurement matrix, \(  \MeasMat  \), the image of the linear transformation induced by \(  \MeasMat^{\T}  \) over \(  \{ -1, 0, 1 \}^{m} \ni \frac{1}{2} ( \Sgn( \MeasMat \Vec{u} ) - \Sgn( \MeasMat \Vec{v} ) )  \) has a finite cardinality.
This finite cardinality enables a union bound over the image of the said linear transformation, which is needed to obtain a uniform result.
\ENDEDIT


One important aspect of BIHT's convergence is that as the approximation error in the \( \Iter\Th \) iteration improves, it makes possible an even smaller error for the \( (\Iter + 1)\Th \) approximation. 
Analogously to the above discussion, each iteration of BIHT involves fewer and fewer measurements, a phenomenon that can be precisely tracked by the number of measurements whose sign-responses {\em mismatch} between the vector $\Vec{x}$ and its approximation \( \Vec{\hat{x}} \) at the \( \Iter\Th \) iteration, where these \emph{mismatches} are captured by
\( \I{\Sgn{}( \MeasMat \Vec{x} ) ) \neq \Sgn{}( \MeasMat \Vec{y} )} \).
As the number of \emph{mismatches} decreases, so does the variance of \( \hA( \Vec{x}, \Vec{\hat{x}} ) \), leading to higher and higher concentration in the approximations.
A primary difficulty of analysis in 1-bit compressed sensing is the nonlinearity imposed by the sign-responses, which prohibits the use of standard techniques developed for compressed sensing.
However,
it turns out that this same difficulty enables different techniques---e.g., the tracking of \emph{mismatches}---that are possible, and natural, in the setting of 1-bit compressed sensing but may not be useful in the more general setting of compressed sensing.

\subsection{Other Related Works}

A generalization of 1bCS is the noisy version of the problem, where the binary observations  $y_i \in \{+1, -1\}$ are random (noisy): i.e., $y_i=1$ with probability $f(\langle \fl{a}_i, \fl{x} \rangle), i =1, \dots, m$, where $f$ is a potentially nonlinear function, such as the sigmoid function. Recovery guarantees for such models were studied in \cite{DBLP:journals/tit/PlanV13}. In another model, observational noise can appear before the quantization, i.e., $y_i = \Sgn(\langle \fl{a}_i, \fl{x} \rangle+\xii), i=1, \dots, m,$ where $\xii$ is random noise. As observed in \cite{plan2016generalized,friedlander2021nbiht}, the noiseless setting (also considered in this work) is actually more difficult to handle because the randomness of noise allows for a maximum likelihood analysis. Indeed, having some control-over $\xii$s (or just assuming them to be i.i.d. Gaussian), helps estimate the norm of $\Vec{x}$~\cite{knudson2016one}, which is otherwise impossible with just sign measurements, as in our model (this is called introducing {\em dither}, a well-known paradigm in signal processing).
In a related line of work, \onebit measurements are taken by adaptively varying the threshold (in our case the threshold is always 0), which can significantly reduce the error-rate,
for example see~\cite{baraniuk2017exponential} and \cite{saab2018quantization}, the latter being an application of sigma-delta quantization methods.

Yet another line of work in 1bCS literature  takes a more combinatorial avenue and looks at the support recovery problem and constructions of structured measurement matrices. Instances of these works are \cite{GNJN13,ABK17,flodin2019superset,mazumdar2021support}. However, the nature of these works is quite different from ours.

\subsection{Organization}

The rest of the paper is organized as follows. The required notations and definitions to state the main result appear in Section \ref{outline:preliminaries-notation}, where we also formally define the 1-bit compressed sensing problem and the reconstruction method, the normalized binary iterative hard thresholding algorithm (Algorithm~\ref{alg:biht:normalized}).
We provide our main result in Section~\ref{outline:main-thm}, which establishes the convergence rate of BIHT~(Theorem~\ref{thm:main:convergence:t-rate}) and the asymptotic error rate ~(Corollary~\ref{thm:main:convergence:t->infty})  with the optimal measurement complexity.
In Section~\ref{outline:intro:overview-techniques} we also overview the derivation of the result, including our invertibility condition for \standardnormal matrices.
In Section~\ref{outline:pf-main-thm} we provide the main proof of the BIHT convergence algorithm, assuming that a structural property if satisfied by the measurement matrix. Proof of this structural property for \standardnormal matrices is the major technical contribution of this paper (Theorem~\ref{thm:technical:raic:modified}), and it has been delegated to Appendix~\ref{outline:technical:pf}. Proofs of all lemmas and intermediate results can be found in the appendix. We conclude with some future directions in Section~\ref{outline:outlook}.

\section{Preliminaries}
\label{outline:preliminaries-notation}
\subsection{Notations and Definitions}
\label{outline:preliminaries-notation|>notation}
The set of all real-valued, \( k \)-sparse vectors in \( n \) dimension is denoted by
\(
  \SparseRealSubspace{k}{n} \EDITX{\subseteq \R^{n}}
\),
\EDITX{and for a coordinate subset \(  J \subseteq [n]  \), the set of all real-valued, \(  n  \)-dimensional vectors whose support is a (possibly improper) subset of \(  J  \) is written:
\(  \SparseRealSubspace{J}{n} \defeq \{ \Vec{u} \in \R^{n} : \Supp( \Vec{u} ) \subseteq J \} \subseteq \R^{n}  \).}
%
The \( \lnorm{2} \)-sphere in \( \R^{n} \) is written
\(
  \Sphere{n} \subset \R^{n}
\),
such that
\(
  (\SparseSphereSubspace{k}{n})
  \subset \SparseRealSubspace{k}{n}
\)
is the subset real-valued, \( k \)-sparse vectors with unit norm. The Euclidean ball of radius $\tau\ge 0$ and center $\Vec{u}\in \R^{n}$ is defined by $\Ball{\tau}(\Vec{u}) = \{\Vec{x} \in \R^{n}: \|\Vec{u}-\Vec{x}\|_2 \le \tau \}.$
Matrices are denoted in uppercase, boldface text, e.g.,
\(
  \Mat{M} \in \R^{m \times n}
\),
with \( (i,j) \)-entries written
\(
  \Mat*{M}_{i,j}
\).
The \( n \times n \) identity matrix written as
\(
  \Mat{I}_{n \times n}
\).
Vectors are likewise indicated by boldface font, using lowercase and uppercase lettering for
nonrandom and random vectors, respectively, e.g.,
\(
  \Vec{u} \in \R^{n}
\)
and
\(
  \Vec{U} \sim \N(\Vec{0},\Mat{I}_{n \times n})
\),
with entries specified such that, e.g.,
\(
  \Vec{u} = (\Vec*{u}_{1}, \dots, \Vec*{u}_{n})
\).
As customary, \( \N(\Vec{0},\Mat{I}_{n \times n}) \) denotes the i.i.d. \( n \)-variate standard normal
distribution (with the univariate case, \( \N(0,1) \)).
Moreover, random sampling from a distribution \( \mathcal{D} \) is denoted by
\(
  \RV{X} \sim \mathcal{D}
\),
and drawing uniformly at random from a set \( \Set{X} \) is written as
\(
  \RV{X} \sim \Set{X}
\).
For any pair of real-valued vectors
\(
  \Vec{u}, \Vec{v} \in \R^{n}
\),
write
\(
  \DistS{\Vec{u}}{\Vec{v}} \in \R_{\geq 0}
\)
for the distance between their projections onto the \( \lnorm{2} \)-sphere, as well as
\(
  \theta_{\Vec{u},\Vec{v}} \in [0,\pi]
\)
%
for their angular distance
and
\(
  \Vec{\theta}_{\Vec{u},\Vec{v}} \in [-\pi,\pi]
\)
for the angular distance and signed angular distance (for a given convention of positive and
negative directions of rotation), respectively, between them.
Formally,
\begin{gather}
\label{def:distance:l2-sphere}
  \DistS{\Vec{u}}{\Vec{v}}
  =
  \begin{cases}
  \left\|
    \frac{\Vec{u}}{\left\| \Vec{u} \right\|_{2}}
    -
    \frac{\Vec{v}}{\left\| \Vec{v} \right\|_{2}}
  \right\|_{2}
  , & \cIf \Vec{u}, \Vec{v} \neq \Vec{0},
  \\
  0, & \cIf \Vec{u} = \Vec{v} = \Vec{0},
  \\
  1, & \cOtherwise,
  \end{cases}
  \\
  \theta_{\Vec{u},\Vec{v}}
  =
  \Arccos
  (
    \frac
    {\langle \Vec{u}, \Vec{v} \rangle}
    {\left\| \Vec{u} \right\|_{2} \left\| \Vec{v} \right\|_{2}}
  )
.\end{gather}
%
Note that these are related by
\(
  \theta_{\Vec{u},\Vec{v}}
  =
  \Arccos( 1 - \frac{\DistS[^{2}]{\Vec{u}}{\Vec{v}}}{2} )
\),
equivalently,
\(
  \DistS{\Vec{u}}{\Vec{v}}
  =
  \sqrt{2 ( 1 - \cos( \theta_{\Vec{u},\Vec{v}} ) )}
\).


The sign function $\Sgn:\R \to \{+1,-1\}$ is defined in the following way:
$$
\Sgn(x) = \begin{cases}
1, & x \ge 0, \\
-1, & x < 0.
\end{cases}
$$
The function can be extended to vectors, i.e., $\Sgn:\R^n \to \{+1,-1\}^n$ by just applying the it on each coordinate.
Additionally, for a condition \( C \in \{ \mathsf{true}, \mathsf{false} \} \), define the indicator function
\( \I{} : \{ \mathsf{true}, \mathsf{false} \} \to \{ 0,1 \} \) by
\begin{gather}
  \I{}( C ) =
  \begin{cases}
  0, & \text{if \( C = \mathsf{false} \)}, \\
  1, & \text{if \( C = \mathsf{true} \)}.
  \end{cases}
\end{gather}
Again, this notation extends to vectors by applying the function coordinate-wise.
\EDIT{``Big-O,'' ``Big-Omega,'' and ``Big-Theta'' notations are defined as standard: for functions \(  f,g  \), we write \(  f = \BigO( g )  \) if there is a constant \(  C  \) and some \(  x_{0} > 0  \) such that \(  f(x) \leq C g(x)  \) for all \(  x \geq x_{0}  \). We write \(  f = \BigOmega( g )  \) if \(  g = \BigO( f )  \), and \(  f = \BigTheta( g )  \) if \(  f = \BigOmega( g )  \) and \(  f = \BigO( g )  \). Additionally, the notations \(  \BigO'  \), \(  \BigOmega'  \), \(  \BigTheta'  \) are defined analogously but hide logarithmic factors.}
\par
We are going use the following universal constants
\(
  \UnivConstA, \UnivConstB, \UnivConstC, \UnivConstc_{1}, \UnivConstc_{2} > 0
\)
in the statement of our results. Their values are
\begin{gather}
  \UnivConstA = \UnivConstAValue, \quad
  \UnivConstAX = \UnivConstAXValue, \quad
  \UnivConstAXX = \UnivConstAXXValue, \quad
  \UnivConstB \gtrsim \UnivConstBValue, \quad
  \UnivConstC = \UnivConstCValue, \quad
  \UnivConstD = \UnivConstDValue, \nonumber\\
  \label{eqn:univConstants}
  \UnivConstb_{1} =
    \EDIT{\sqrt{\frac{\pi}{\UnivConstB \UnivConstD}}( \sqrt{3} + 16 )}
  , \quad
  \UnivConstb_{2} =
    \frac{30 \sqrt{2}}{\UnivConstBX}
  , \\
  \UnivConstc_{1} = \UnivConstcOneValue, \quad
  \UnivConstc_{2} = \UnivConstcTwoValue.\nonumber
\end{gather}
Additionally, in the BIHT algorithm, the step size
\(
  \Eta > 0
\)
is fixed as
\(
  \Eta = \sqrt{2\pi}
\).

We define two hard thresholding operations:
the \emph{\topkhardthresholding operation} and
the \emph{\subsethardthresholding operation},
defined below in \DEF[s] \ref{def:hard-thresholding:max-weight} and \ref{def:hard-thresholding:set}.
When clear from context, we will omit the distinction simply refer to a \emph{hard thresholding operation}.
%
\EDIT{To write down these definitions, we make use of the following notations.
For a coordinate subset, \(  \Coords{J} \subseteq [n]  \), let \(  \BVec{\Coords{J}} \in \{ 0,1 \}^{n}  \) denote the vector with \(  j\Th  \) entries, \(  ( \BVec{\Coords{J}} )_{j} = \I{j \in \Coords{J}}  \), \(  j \in [n]  \).
Additionally, for \(  \Vec{u} \in \R^{n}  \), let \(  \Diag( \Vec{u} ) \in \R^{n \times n}  \) denote the diagonal matrix with diagonal entries given by \(  \Vec{u}  \).}
\begin{definition}
[\Topkhardthresholding operation]
\label{def:hard-thresholding:max-weight}
\ORIG{For \( k \in \Z_{+} \), \( k \leq n \), the \emph{\topkhardthresholding operation},
\(
  \Threshold{k} : \R^{n} \to \SparseRealSubspace{k}{n}
\),
projects a real-valued vector
\(
  \Vec{u} \in \R^{n}
\)
into the space of \ksparserealvalued vectors by setting all but the \( k \) largest
(in absolute value) entries in \( \Vec{u} \) to \( 0 \) (with ties broken arbitrarily).}%
\EDIT{For \( k \in \Z_{+} \), \( k \leq n \), the \emph{\topkhardthresholding operation},
\(
  \Threshold{k} : \R^{n} \to \SparseRealSubspace{k}{n}
\),
projects a real-valued vector
\(
  \Vec{u} \in \R^{n}
\)
into the space of \ksparserealvalued vectors by \(  \Threshold{k}( \Vec{u} ) = \Diag( \BVec{J_{\Vec{u}}} ) \Vec{u}  \), where \(  J_{\Vec{u}} \subseteq [n]  \), \(  | J_{\Vec{u}} | = k  \), satisfies \(  \| \Diag( \BVec{J_{\Vec{u}}} ) \Vec{u} \|_{1} = \max_{J \subseteq [n] : |J| = k} \| \Diag( \BVec{J} ) \Vec{u} \|_{1}  \).
Note that ties are broken arbitrarily.}
\end{definition}

\begin{definition}
[\Subsethardthresholding operation]
\label{def:hard-thresholding:set}
For a coordinate subset,
\(
  \Coords{J} \subseteq [n]
\),
the \emph{\subsethardthresholding operation} associated with \( \Coords{J} \),
\(
  \ThresholdSet{\Coords{J}} : \R^{n} \to \EDIT{\R^{n}} 
\),
\ORIG{projects a real-valued vector
\(
  \Vec{u} \in \R^{n}
\)
into the space of \ksparserealvalued vectors by
\(
  \ThresholdSet{\Coords{J}}( \Vec{u} )_{j} = u_{j} \cdot \I{j \in \Coords{J}}
\)
for each \( j \in [n] \).}%
\EDIT{is the linear transformation given by \(  \ThresholdSet{J}( \Vec{u} ) = \Diag( \BVec{J} ) \Vec{u}  \).}
%
\ORIG{It is simple to show this operation \( \ThresholdSet{\Coords{J}} \) is a linear transformation
associated with a diagonal \( n \times n \) matrix denoted
\(
  \ThresholdSetMat{\Coords{J}}
  =
  \Diag( \ThresholdSetMat*{\Coords{J}}[1], \dots, \ThresholdSetMat*{\Coords{J}}[n] )
\),
where
\(
  \ThresholdSetMat*{\Coords{J}}[j] = \I{j \in \Coords{J}}
\)
for each \( j \in [n] \).}%
\end{definition}


\subsection{1-Bit Compressed Sensing and the BIHT Algorithm}
\label{outline:intro|>biht-algorithm}
Let $\Vec{x} \in \SparseRealSubspace{k}{n}$. A  measurement matrix is denoted by
\(
  \MeasMat \in \R^{m \times n}
\)
and has rows
\(
  \MeasVec^{(1)}, \dots, \MeasVec^{(m)}
  \sim
  \N(\Vec{0},\Mat{I}_{n \times n})
\)
with i.i.d. entries. The 1-bit measurements of $\Vec{x}$ are performed by:
\begin{equation}
    \Vec{b} = \Sgn({\MeasMat \Vec{x}})
\end{equation}

Throughout this work, the unknown signals,
\(
  \Vec{x} \in \SparseRealSubspace{k}{n}
\),
are assume to have unit norm since information about the norm is lost due to the binarization
of the responses.
(For interested readers, see \cite{knudson2016one} for techniques, e.g., dithering, to reconstruct the signal's norm in 1-bit compressed sensing.)
Given $\MeasMat$ and $\Vec{b}$, the goal of 1-bit compressed sensing is to recover $\Vec{x}$ as accurately as possible. We measure the accuracy of reconstruction by the metric $ \DistS{\cdot}{\cdot}.$

The binary iterative hard thresholding (BIHT) reconstruction algorithm,
proposed by \cite{JLBB13}, comprises two iterative steps:
\begin{EnumerateInline}
\item \label{enum:biht-alg:step:1}
a \EDIT{subgradient} descent step, which finds a non-sparse approximation,
\(
  \Vec{\tilde{x}} \in \R^{n}
\),
followed by
\item \label{enum:biht-alg:step:2}
a projection by
\(
  \Vec{\tilde{x}} \mapsto \Vec{\hat{x}} = \Threshold{k}(\Vec{\tilde{x}})
\)
into the space of \( k \)-sparse, real-valued vectors.
\end{EnumerateInline}
As shown by \cite{JLBB13}, the \EDIT{subgradient} step, \ref{enum:biht-alg:step:1}, aims to
minimize the objective function
\begin{gather}\label{eqn:costFunction}
  \ObjectiveFn(\Sgn(\MeasMat \Vec{x}), \Vec{\hat{x}})
  =
  \left\|
    \left[\, \Sgn(\MeasMat \Vec{x}) \odot \EDIT{(\MeasMat \Vec{\hat{x}})} \,\right]_{-}
  \right\|_{1}
,\end{gather}
where
\(
  \Vec{u} \odot \Vec{v}
  =
  (\Vec*{u}_{1} \Vec*{v}_{1}, \dots, \Vec*{u}_{n} \Vec*{v}_{n})
\)
and
\(
  (\, [\, \Vec{u} \,]_{-} \,)_{j} = \Vec*{u}_{j} \cdot \I{\Vec*{u}_{j} < 0}
\).
\EDIT{Per \cite[\LEMMA 5]{JLBB13}, \(  \ObjectiveFn  \) is convex with respect to \(  \Vec{\hat{x}}  \), and its subgradients include \(  \nabla_{\Vec{\hat{x}}}\, \ObjectiveFn \ni \MeasMat^{\T} \cdot \frac{1}{2} ( \sgn( \MeasMat \Vec{x} ) - \sgn( \MeasMat \Vec{\hat{x}} ) )  \).}
While several variants of the BIHT algorithm have been proposed, \cite{JLBB13}, 
this work focuses on the normalized BIHT algorithm, where the projection step,
\ref{enum:biht-alg:step:2}, is modified to project the approximation onto the \( k \)-sparse,
\( \lnorm{2} \)-unit sphere,
\(
  \SparseSphereSubspace{k}{n}
\).
Algorithm \ref{alg:biht:normalized} provides the version of the BIHT algorithm studied in this
work.

%
\begin{algorithm}
\caption{\label{alg:biht:normalized}Binary iterative hard thresholding (BIHT) algorithm, normalized projections}
Set \( \Eta = \sqrt{2\pi} \)\;
%
%
%
\(
  \Vec{\hat{x}}^{(0)}
  \sim
  \SparseSphereSubspace{k}{n}
\)\;
\For
{\( \Iter = 1, 2, 3, \dots \)}
{
  \(
    \Vec{\tilde{x}}^{(\Iter)}
    \gets
    \Vec{\hat{x}}^{(\Iter-1)}
    +
    \frac{\eta}{m}
    \MeasMat^{\T}
    \cdot
    \frac{1}{2}
    \left(
      \Sgn{} \big( \MeasMat \Vec{x} \big) - \Sgn{} \big( \MeasMat \Vec{\hat{x}}^{(\Iter-1)} \big)
    \right)
  \)\;
  \(
    \Vec{\hat{x}}^{(\Iter)}
    \gets
    \frac
    {\Threshold{k}(\Vec{\tilde{x}}^{(\Iter)})}
    {\left\| \Threshold{k}(\Vec{\tilde{x}}^{(\Iter)}) \right\|_{2}}
  \)\;
}
\end{algorithm}

\section{Main Results and Techniques}
\label{outline:main-thm}

\subsection{BIHT Convergence Theorem}
Our main result is presented below
\EDIT{in Theorem \ref{thm:main:convergence:t-rate}, which characterizes the error decay of BIHT approximations, and Corollary \ref{thm:main:convergence:t->infty}, which bounds the asymptotic error rate as the number of iterations \(  \Iter \to \infty  \).}
Informally, it states that with $m = O(\frac{k}{\epsilon}\log \frac{n}{k\sqrt{\epsilon}})$ \onebit (sign) measurements, it is possible to recover any $k$-sparse unit vector within an $\epsilon$-ball, by means of the normalized BIHT algorithm.
\EDIT{Additionally, Figure \ref{fig:error-decay-plot} corroborates the error decay stated in \EQN \eqref{eqn:main:convergence:t-finite:dist} of Theorem \ref{thm:main:convergence:t-rate}.}

\begin{figure}
%
%
\includegraphics[width=\textwidth]{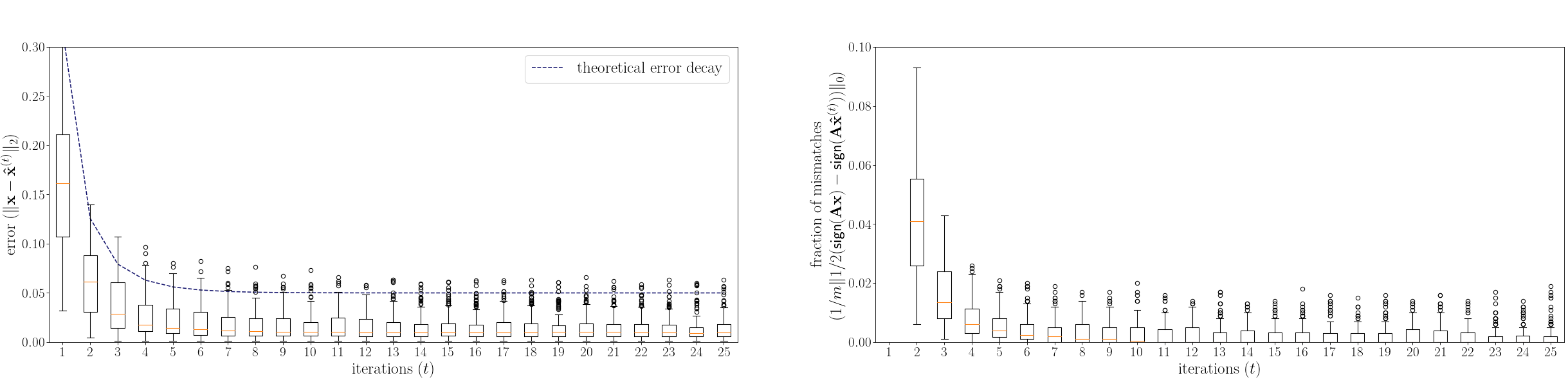}
\caption{\label{fig:error-decay-plot}\BEGINEDITX%
The \LHS shows the error decay of BIHT approximations empirically and theoretically.
The \RHS displays the fraction of measurements which fall onto opposite sides of the hyperplanes associated with the true signal, \(  \Vec{x}  \), and the approximations.
The empirical results were obtained by running \(  100  \) trials of recovering random \(  k  \)-sparse unit vectors via the normalized BIHT algorithm for \(  25  \) iterations.
The parameters were set as: \(  k = 5  \), \(  n = 2000  \), \(  m = 1000  \), \(  \Epsilon = 0.05  \), and \(  \Rho = 0.05  \).%
\ENDEDITX}
%
%
\end{figure}

\begin{theorem}
\label{thm:main:convergence:t-rate}
Let
\(
  \UnivConstA, \UnivConstB, \UnivConstC, \UnivConstD > 0
\)
be universal constants as in Eq.~\eqref{eqn:univConstants}.
Fix
\(
  \Epsilon, \Rho \in (0,1)
\)
and
\(
  k, m, n \in \Z_{+}
\),
where
\begin{EDITb}
\begin{align}
  \nonumber
  m
  &\geq
  \frac{4 \UnivConstB \UnivConstC \UnivConstD k}{\Epsilon}
  \log
  \left(
    \frac{en}{k}
  \right)
  +
  \frac{2 \UnivConstB \UnivConstC \UnivConstD k}{\Epsilon}
  \log
  \left(
    \frac{12 \UnivConstB \UnivConstC \UnivConstD \Log[^{3/2}]( 2e \UnivConstB \UnivConstC/\Epsilon )}{\Epsilon}
  \right)
  +
  \frac{\UnivConstB \UnivConstC \UnivConstD}{\Epsilon}
  \log
  \left(
    \frac{\UnivConstA}{\Rho}
  \right)
  \\ \nonumber
  &\AlignTab+
  \frac{\UnivConstB \UnivConstC \UnivConstD k}{\Epsilon} \Log( \frac{\UnivConstB \UnivConstC}{\Epsilon} ) \sqrt{\Log( \frac{2e \UnivConstB \UnivConstC}{\Epsilon} )}
  +
  \frac{\EDITX{128} \UnivConstB \UnivConstC k}{\Epsilon} \Log( \frac{en}{\EDITX{k}} ) \sqrt{\Log( \frac{2e \UnivConstB \UnivConstC}{\Epsilon} )}
  \\
  &\AlignTab+
  \frac{\EDITX{64} \UnivConstB \UnivConstC}{\Epsilon} \Log( \EDITX{\frac{\UnivConstAXX}{\Rho}} )
  \EDITX{\sqrt{\Log( \frac{2e \UnivConstB \UnivConstC}{\Epsilon} )}}
  +
  \frac{\EDITX{4} \UnivConstB \UnivConstC \EDITX{k}}{\Epsilon}
  \Log( \frac{en}{\EDITX{k}} )
  +
  \frac{\UnivConstB \UnivConstC}{\Epsilon}
  \Log( \EDITX{\frac{\UnivConstAX}{\Rho}} )
\label{eqn:main:convergence:t-finite:m}
.\end{align}
\end{EDITb}
Let the measurement matrix
\(
  \MeasMat \in \R^{m \times n}
\)
have rows with i.i.d. \standardnormal entries.
%
Then, uniformly with probability at least
\(
  1 - \Rho
\),
for every unknown \ksparserealvalued unit vector,
\(
  \Vec{x} \in \SparseSphereSubspace{k}{n}
\),
the normalized BIHT algorithm produces a sequence of approximations,
\(
  \{ \Vec{\hat{x}}^{(\Iter)} \in \SparseSphereSubspace{k}{n} \}_{\Iter \in \Z_{\geq 0}}
\),
which converges to the \( \Epsilon \)-ball around the unknown vector \( \Vec{x} \) at a rate
upper bounded by
\begin{gather}
\label{eqn:main:convergence:t-finite:dist}
  \DistS*{\Vec{x}}{\Vec{\hat{x}}^{(\Iter)}}
  \leq
  2^{2^{-\Iter}}
  \Epsilon^{1 - 2^{-\Iter}}
\end{gather}
for each
\(
  \Iter \in \Z_{\geq 0}
\).
\end{theorem}

\begin{corollary}
\label{thm:main:convergence:t->infty}
Under the conditions stated in Theorem \ref{thm:main:convergence:t-rate},
uniformly with probability at least
\(
  1 - \Rho
\),
for every unknown \ksparserealvalued unit vector,
\(
  \Vec{x} \in \SparseSphereSubspace{k}{n}
\),
the sequence of BIHT approximations,
\(
  \{ \Vec{\hat{x}}^{(\Iter)} \}_{\Iter \in \Z_{\geq 0}}
\),
converges asymptotically to the \( \Epsilon \)-ball around the unknown vector \( \Vec{x} \).
Formally,
\begin{gather}
\label{eqn:main:convergence:t->infty:dist}
  \lim_{\Iter \to \infty} \DistS*{\Vec{x}}{\Vec{\hat{x}}^{(\Iter)}} \leq \Epsilon
.\end{gather}
\end{corollary}

\begin{figure}
%
\centering
\includegraphics[width=0.75\textwidth]{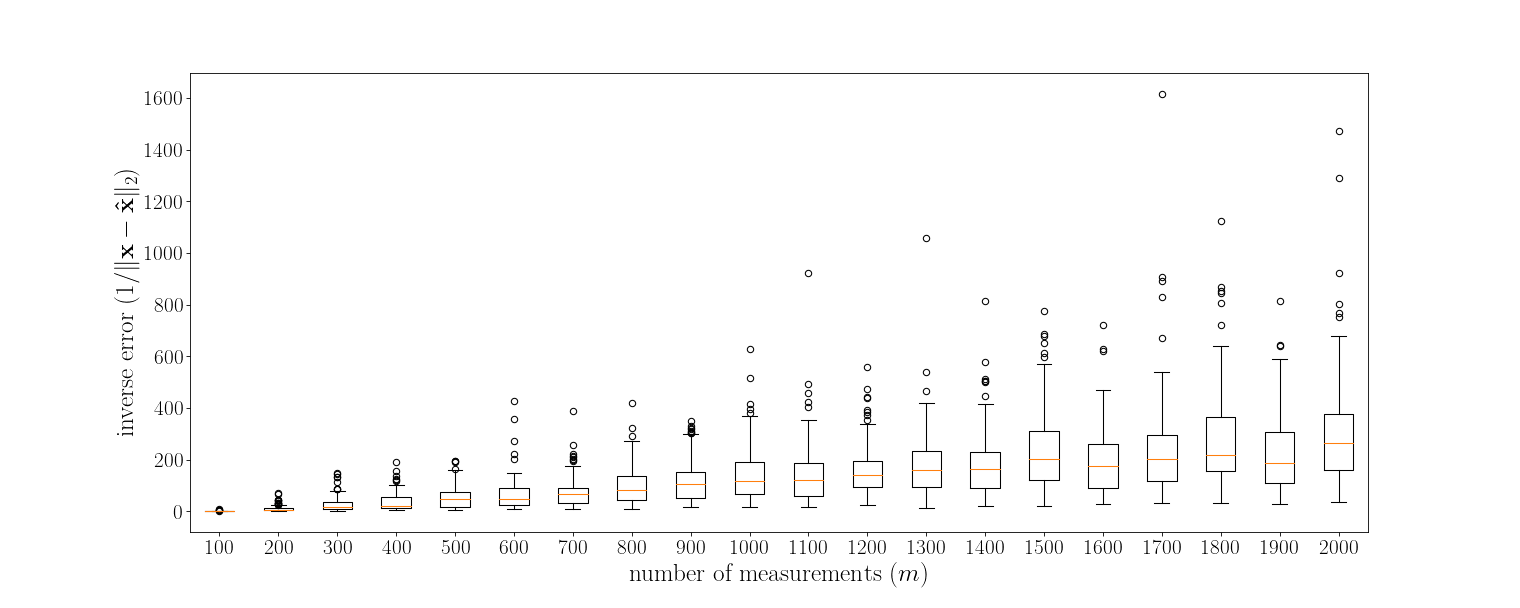}
\caption{\label{fig:m-vs-error-plot}\BEGINEDITX%
This plot shows the (roughly linear) relationship between the number of measurements, \(  m  \), (\(  x  \)-axis) and the inverse error (\(  y  \)-axis), where the error is the \(  \lnorm{2}  \)-distance between the true signal and the approximation obtained after \(  25  \) iterations of the normalized BIHT algorithm.
The sparsity and dimension parameters were set, respectively, as: \(  k = 5  \) and \(  n = 2000  \).%
\ENDEDITX}
%
\end{figure}


\subsection{Technical Overview}
\label{outline:intro:overview-techniques}

The analysis in this work is divided into two components:
\begin{EnumerateInline}[label=(\Roman*)]
\item \label{enum:intro:techniques:1}
the proofs of Theorem \ref{thm:main:convergence:t-rate} and Corollary \ref{thm:main:convergence:t->infty},
which show the universal convergence of the BIHT approximations by using the {\em restricted approximate invertibility condition} (RAIC) for
\standardnormal measurement matries (defined below),
and
\item \label{enum:intro:techniques:2}
the proof of the main technical theorem, Theorem \ref{thm:technical:raic:modified} (also below),
which derives the RAIC for such a measurement matrix.
\end{EnumerateInline}

Informally speaking, we show that the approximation error, $\varepsilon(t)$, of the BIHT algorithm at step $t>0$ satisfies  a recurrence relation of the form $\varepsilon(t) = a_1\sqrt{\epsilon \varepsilon(t-1)} +a_2 \epsilon.$ It is not a difficult exercise to see that we get the desired convergence rate from this recursion, starting from a constant error. The recursion itself is a result of the RAIC property, which tries to capture the fact that the difference between two vectors $\Vec{x}$ and $\Vec{y}$ can be reconstructed by applying $\MeasMat^\T$ on the difference of the corresponding \one-bit measurements. Next we explain the technicalities of these different components of the proof.

\subsubsection{The Restricted Approximate Invertibility Condition}
\label{outline:intro:overview-techniques|>raic-def-thm}

The main technical contribution is an improved sample complexity for the
restricted approximate invertibility condition (RAIC). A different invertibilty condition  was proposed  by \cite{friedlander2021nbiht}. We have included the definition of \cite{friedlander2021nbiht} in Appendix~\ref{outline:technical:raic:friedlander}, for comparison, and to emphasize the major differences.
The definition of RAIC considered in this work is formalized in
Definition \ref{def:raic:modified}, which uses the following notations.
For
\(
  m, n \in \Z_{+}
\),
let
\(
  \MeasMat \in \R^{m \times n}
\)
be a measurement matrix with rows
\(
  \MeasVec^{(i)} \in \R^{n}
\),
\(
  i \in [m]
\).
Then, define the functions
\(
  \hA, \hA[\Coords{J}] : \R^{n} \times \R^{n} \to \R^n
\)
by
\begin{align}
  \hA \left( \Vec{x}, \Vec{y} \right)
  &=
  \frac{\Eta}{m}
  \MeasMat^{\T}
  \cdot
  \frac{1}{2}
  \left( \Sgn( \MeasMat \Vec{x} ) - \Sgn( \MeasMat \Vec{y} ) \right)
\label{eqn:h_Z}
\end{align}
and
\begin{gather}
\label{eqn:h_ZJ}
  \hA[\Coords{J}] \left( \Vec{x}, \Vec{y} \right)
  =
  \ThresholdSet{\supp( \Vec{x} ) \cup \supp( \Vec{y} ) \cup \Coords{J}}
  \left( \hA \left( \Vec{x}, \Vec{y} \right) \right)
\end{gather}
for
\(
  \Vec{x}, \Vec{y} \in \R^{n}
\)
and
\(
  \Coords{J} \subseteq [n]
\),
and where
\(
  \Eta = \sqrt{2\pi}
\).
\EDIT{Recalling from earlier the objective function which BIHT seeks to minimize:
\(  \ObjectiveFn(\Vec{\hat{x}}; \Vec{x}) = \| [\, \sgn(\MeasMat \Vec{x}) \odot \EDIT{(\MeasMat \Vec{\hat{x}})} \,]_{-} \|_{1}  \), note that \(  \nabla_{\Vec{\hat{x}}}\, \ObjectiveFn(\Vec{\hat{x}}; \Vec{x}) \ni \frac{m}{\Eta} \hA( \Vec{x}, \Vec{\hat{x}} )  \).}
%
\begin{definition}
[Restricted approximate invertibility condition (RAIC)]
\label{def:raic:modified}
Fix
\(
  \delta, a_{1}, a_{2} > 0
\)
and
\(
  k, m, n \in \Z_{+}
\)
such that
\(
  0 < k < n
\).
The \( (k, n, \delta, a_{1}, a_{2}) \)-RAIC is satisfied by a measurement matrix
\(
  \MeasMat \in \R^{m \times n}
\)
if
\begin{gather}
\label{eqn:raic:modified:raic:(x,y)}
  \left\| ( \Vec{x} - \Vec{y} ) - \hA[\Coords{J}]( \Vec{x}, \Vec{y} ) \right\|_{2}
  \leq
  a_{1}
  \sqrt{\delta \DistS{\Vec{x}}{\Vec{y}}}
  +
  a_{2} \delta
\end{gather}
uniformly for all
\(
  \Vec{x}, \Vec{y} \in \SparseSphereSubspace{k}{n}
\)
and all
\(
  \Coords{J} \subseteq [n]
\),
\(
  | \Coords{J} | \leq k
\).
\end{definition}
%
\par

Theorem \ref{thm:technical:raic:modified} below is the primary technical result in this analysis
and establishes that \( m \)-many i.i.d. \standardnormal measurements satisfy the
\( (k, n, \DDelta, \UnivConstc_{1}, \UnivConstc_{2} ) \)-RAIC, where the sample complexity
for \( m \) matches the lower bound of \cite[Lemma 1]{JLBB13}.
The proof of the theorem is deferred to Appendix \ref{outline:technical:pf},
while an overview of the proof is given below in Section~\ref{outline:intro:overview-techniques|>raic}.
%
\begin{theorem}
\label{thm:technical:raic:modified}
%
%
%
Let
\(
  \UnivConstA, \UnivConstAX, \UnivConstAXX, \UnivConstB, \UnivConstc_{1}, \UnivConstc_{2}, \UnivConstD > 0
\)
be universal constants as defined in Eq.~\eqref{eqn:univConstants}. 
Fix
\(
  \DDelta, \Rho \in (0,1)
\)
and
\(
  k, m, n \in \Z_{+}
\)
such that
\(
  0 < k \EDITX{\leq} n
\).
\EDITX{Let \(  \kO \defeq \min \{ 2k, n \}  \) and \(  \kOX \defeq \min \{ 4k, n\}  \).}
Define \(  \GammaX \in (0,1)  \) such that
\begin{gather}
\label{eqn:technical:raic:modified:gamma}
  \GammaX = \GammaXValue
,\end{gather}
and let
\begin{align}
  \nonumber
  m
  &= \CMPLXRAICOne
  \\ \label{eqn:technical:raic:modified:m} 
  &\AlignTab
  +  \CMPLXRAICTwo
  \\ \nonumber
  &= \CMPLXRAICBigO
.\end{align}
%
Let
\(
  \MeasMat \in \R^{m \times n}
\)
be a measurement matrix whose rows have i.i.d. \standardnormal entries.
%
Then, \( \MeasMat \) satisfies the
\( (k, n, \DDelta, \UnivConstc_{1}, \UnivConstc_{2}) \)-RAIC
with probability at least
\(
  1 - \Rho
\).
To state this explicitly, uniformly with probability at least
\(
  1 - \Rho
\),
for all
\(
  \Vec{x}, \Vec{y} \in \SparseSphereSubspace{k}{n}
\)
and all
\(
  \Coords{J} \subseteq [n]
\),
\(
  | \Coords{J} | \leq k
\),
\begin{gather}
\label{eqn:technical:raic:modified}
  \left\| ( \Vec{x} - \Vec{y} ) - \hA[\Coords{J}]( \Vec{x}, \Vec{y} ) \right\|_{2}
  \leq
  \UnivConstc_{1}
  \sqrt{\DDelta \DistS{\Vec{x}}{\Vec{y}}}
  +
  \UnivConstc_{2} \DDelta
.\end{gather}
\end{theorem}

\BEGINEDIT
\subsubsection{Comparison of RAIC and Other Properties of Binary Embeddings}
\label{outline:intro:overview-techniques|>comparison}
While not directly comparable, properties of binary embeddings similar to the RAIC appear elsewhere in the literature.
One such property is the sign-product embedding (SPE), studied in \cite{jacques2013quantized} (and previously considered in a more general form in \cite{plan2012robust}), which is a map, \(  \Vec{x} \mapsto \sgn( \Mat{A} \Vec{x} )  \) that upholds:
\begin{gather*}
  \left| \left\langle
    \Vec{x} - \sqrt{\frac{\pi}{2}} \frac{1}{m} \Mat{A}^{\T} \sgn( \Mat{A} \Vec{x} ),
    \Vec{y}
  \right\rangle \right|
  \leq
  \delta
,\end{gather*}
for all \(  k'  \)-sparse \EDITX{unit} vectors, \(  \Vec{x}, \Vec{y} \in \EDITX{\SparseSphereSubspace{k'}{n}}  \), where \(  \Mat{A} \in \R^{m \times n}  \) and \(  \delta > 0  \).
It turns out that the SPE can bound the \LHS of \EQN \eqref{eqn:raic:modified:raic:(x,y)} in the definition of the RAIC (see, Definition~\ref{def:raic:modified}).
Suppose that \(  \Mat{A}  \) satisfies the SPE, where we take \(  k' = 3k  \).
Then, observe:
\begin{align*}
  &
  \left\| ( \Vec{x} - \Vec{y} ) - \hA[\Coords{J}]( \Vec{x}, \Vec{y} ) \right\|_{2}
  \\
  &=
  \left\| ( \Vec{x} - \Vec{y} ) - \ThresholdSet{\supp( \Vec{x} ) \cup \supp( \Vec{y} ) \cup \Coords{J}} \sqrt{\frac{\pi}{2}} \frac{1}{m} \MeasMat^{\T} \left( \Sgn( \MeasMat \Vec{x} ) - \Sgn( \MeasMat \Vec{y} ) \right) \right\|_{2}
  \\
  &=
  \left\| \ThresholdSet{\supp( \Vec{x} ) \cup \supp( \Vec{y} ) \cup \Coords{J}} \left( \Vec{x} - \sqrt{\frac{\pi}{2}} \frac{1}{m} \MeasMat^{\T} \Sgn( \MeasMat \Vec{x} ) \right) - \ThresholdSet{\supp( \Vec{x} ) \cup \supp( \Vec{y} ) \cup \Coords{J}} \left( \Vec{y} - \sqrt{\frac{\pi}{2}} \frac{1}{m} \MeasMat^{\T} \Sgn( \MeasMat \Vec{y} ) \right) \right\|_{2}
  \\
  &\leq
  \left\| \ThresholdSet{\supp( \Vec{x} ) \cup \supp( \Vec{y} ) \cup \Coords{J}} \left( \Vec{x} - \sqrt{\frac{\pi}{2}} \frac{1}{m} \MeasMat^{\T} \Sgn( \MeasMat \Vec{x} ) \right) \right\|_{2}
  +
  \left\| \ThresholdSet{\supp( \Vec{x} ) \cup \supp( \Vec{y} ) \cup \Coords{J}} \left( \Vec{y} - \sqrt{\frac{\pi}{2}} \frac{1}{m} \MeasMat^{\T} \Sgn( \MeasMat \Vec{y} ) \right) \right\|_{2}
  \\
  &\dCmt \text{by the triangle inequality}
  \\
  &=
  \sqrt{\left| \left\langle
    \ThresholdSet{\supp( \Vec{x} ) \cup \supp( \Vec{y} ) \cup \Coords{J}} \left( \Vec{x} - \sqrt{\frac{\pi}{2}} \frac{1}{m} \MeasMat^{\T} \Sgn( \MeasMat \Vec{x} ) \right),
    \ThresholdSet{\supp( \Vec{x} ) \cup \supp( \Vec{y} ) \cup \Coords{J}} \left( \Vec{x} - \sqrt{\frac{\pi}{2}} \frac{1}{m} \MeasMat^{\T} \Sgn( \MeasMat \Vec{x} ) \right)
  \right\rangle \right|}
  \\
  &\AlignTab+
  \sqrt{\left| \left\langle
    \ThresholdSet{\supp( \Vec{x} ) \cup \supp( \Vec{y} ) \cup \Coords{J}} \left( \Vec{y} - \sqrt{\frac{\pi}{2}} \frac{1}{m} \MeasMat^{\T} \Sgn( \MeasMat \Vec{y} ) \right),
    \ThresholdSet{\supp( \Vec{x} ) \cup \supp( \Vec{y} ) \cup \Coords{J}} \left( \Vec{y} - \sqrt{\frac{\pi}{2}} \frac{1}{m} \MeasMat^{\T} \Sgn( \MeasMat \Vec{y} ) \right)
  \right\rangle \right|}
  \\
  &=
  \sqrt{\left| \left\langle
    \Vec{x} - \sqrt{\frac{\pi}{2}} \frac{1}{m} \MeasMat^{\T} \Sgn( \MeasMat \Vec{x} ),
    \ThresholdSet{\supp( \Vec{x} ) \cup \supp( \Vec{y} ) \cup \Coords{J}} \left( \Vec{x} - \sqrt{\frac{\pi}{2}} \frac{1}{m} \MeasMat^{\T} \Sgn( \MeasMat \Vec{x} ) \right)
  \right\rangle \right|}
  \\
  &\AlignTab+
  \sqrt{\left| \left\langle
    \Vec{y} - \sqrt{\frac{\pi}{2}} \frac{1}{m} \MeasMat^{\T} \Sgn( \MeasMat \Vec{y} ),
    \ThresholdSet{\supp( \Vec{x} ) \cup \supp( \Vec{y} ) \cup \Coords{J}} \left( \Vec{y} - \sqrt{\frac{\pi}{2}} \frac{1}{m} \MeasMat^{\T} \Sgn( \MeasMat \Vec{y} ) \right)
  \right\rangle \right|}
  \\
  &\leq
  2 \sqrt{\delta}
  \\
  &\dCmt \text{by the SPE}
.\end{align*}
On the other hand, recall that the RAIC has a bound of the form
\(  \| ( \Vec{x} - \Vec{y} ) - \hA[\Coords{J}]( \Vec{x}, \Vec{y} ) \|_{2} \leq \BigO{}( \sqrt{\delta \DistS{\Vec{x}}{\Vec{y}}} + \delta )  \).
When \(  \Vec{x}  \) and \(  \Vec{y}  \) are far apart---with \(  \DistS{\Vec{x}}{\Vec{y}} = \Theta( 1 )  \)---the SPE provides approximately the same bound as the RAIC, but a comparatively weaker and weaker bound as the distance between \(  \Vec{x}  \) and \(  \Vec{y}  \) decreases.
This is because, unlike the RAIC, the SPE does not scale with the distance between points, and as a result, the SPE is not a sufficient condition to show optimal convergence of BIHT.
A similar phenomenon occurs with the binary stable embedding studied by \cite{JLBB13}, which again does not scale with the distance between points.
\par
A similar notion of the SPE is also considered by \cite{foucart2017flavors}.
Here, a matrix \(  \Mat{A} \in \R^{m \times n}  \) satisfies the SPE if for all \(  \Vec{x} \in \EDITX{\SparseSphereSubspace{k}{n}}  \), the following holds:
\begin{gather*}
  \left\| \Vec{x} - \ThresholdSet{\supp( \Vec{x} )}( \Mat{A}^{\T} \sgn( \Mat{A} \Vec{x} ) ) \right\|_{2}
  =
  \BigO{}( \sqrt{\delta} )
.\end{gather*}
Notice that the RAIC implies that for all \(  \Vec{x} \in \SparseSphereSubspace{k}{n}  \),
\begin{align*}
  \left\| \Vec{x} - \ThresholdSet{\supp( \Vec{x} )}( \left( \sqrt{\frac{\pi}{2}} \frac{1}{m} \Mat{A} \right)^{\T} \Sgn( \left( \sqrt{\frac{\pi}{2}} \frac{1}{m} \Mat{A} \right) \Vec{x} ) ) \right\|_{2}
  &=
  \left\| \Vec{x} - \ThresholdSet{\supp( \Vec{x} )}( \sqrt{\frac{\pi}{2}} \frac{1}{m} \Mat{A}^{\T} \sgn( \Mat{A} \Vec{x} ) ) \right\|_{2}
  \\
  &=
  \frac{1}{2} \left\| ( \Vec{x} - (-\Vec{x}) ) - \hA[\emptyset]( \Vec{x}, -\Vec{x} ) \right\|_{2}
  \\
  &=
  \BigO{}( \sqrt{\delta} + \delta )
,\end{align*}
where the last line applies the RAIC.
Therefore, if a matrix \(  \Mat{A} \in \R^{m \times n}  \) satisfies the RAIC, then the matrix \(  \sqrt{\frac{\pi}{2}} \frac{1}{m} \Mat{A}  \) satisfies the SPE of \cite{foucart2017flavors}.
\par
As a final point of interest, restricted isometry properties (RIP) are well-studied in compressed sensing literature.
The \(  \lnorm{1}/\lnorm{2}  \)-restricted isometry property (\(  \text{RIP}_{\lnorm{1}/\lnorm{2}}  \)) with parameter \(  \lambda > 0  \) (see, e.g., \cite{foucart2017flavors}) for \(  \SparseSphereSubspace{k}{n}  \) is one such property:
a matrix \(  \Mat{A} \in \R^{m \times n}  \) satisfies the \(  \text{RIP}_{\lnorm{1}/\lnorm{2}}  \) with parameter \(  \lambda  \) if for all \(  \Vec{x} \in \SparseSphereSubspace{k}{n}  \),
\begin{gather*}
  \| \Mat{A} \Vec{x} \|_{1} \in [ 1-\lambda, 1+\lambda ]
.\end{gather*}
A simple derivation shows that if \(  \Mat{A}  \) satisfies the RAIC, then \(  \sqrt{\frac{\pi}{2}} \frac{1}{m} \Mat{A}  \) satisfies the \(  \text{RIP}_{\lnorm{1}/\lnorm{2}}  \) with parameter \(  \BigO{}( \sqrt{\delta} + \delta )  \):
\begin{align*}
  \left| 1 - \left\| \sqrt{\frac{\pi}{2}} \frac{1}{m} \Mat{A} \Vec{x} \right\|_{1} \right|
  &=
  \left| 1 - \sum_{j} \left| \left( \sqrt{\frac{\pi}{2}} \frac{1}{m} \Mat{A} \Vec{x} \right)_{j} \right| \right|
  \\
  &=
  \left| 1 - \sum_{j} \left( \sqrt{\frac{\pi}{2}} \frac{1}{m} \Mat{A} \Vec{x} \right)_{j} \sgn( ( \Mat{A} \Vec{x} )_{j} ) \right|
  \\
  &=
  \left| \langle \Vec{x}, \Vec{x} \rangle - \left\langle \sqrt{\frac{\pi}{2}} \frac{1}{m} \Mat{A} \Vec{x}, \sgn( \Mat{A} \Vec{x} ) \right\rangle \right|
  \\
  &=
  \left| \left\langle \Vec{x} - \sqrt{\frac{\pi}{2}} \frac{1}{m} \Mat{A}^{\T} \sgn( \Mat{A} \Vec{x} ), \Vec{x} \right\rangle \right|
  \\
  &\leq
  \left\| \Vec{x} - \sqrt{\frac{\pi}{2}} \frac{1}{m} \Mat{A}^{\T} \sgn( \Mat{A} \Vec{x} ) \right\|_{2}
  \\
  &\dCmt \text{by the Cauchy-Schwarz inequality}
  \\
  &=
  \frac{1}{2} \left\| ( \Vec{x} - (-\Vec{x}) ) - \sqrt{\frac{\pi}{2}} \frac{1}{m} \Mat{A}^{\T}
  \left( \sgn( \Mat{A} \Vec{x} ) - \sgn( \Mat{A} (-\Vec{x}) ) \right) \right\|_{2}
  \\
  &=
  \BigO{}( \sqrt{\delta} + \delta )
  .\\
  &\dCmt \text{by the RAIC}
\end{align*}
\ENDEDIT


\subsubsection{The Uniform Convergence of BIHT Approximations}
\label{outline:intro:overview-techniques|>biht}

Assuming the   desired RAIC property (i.e., the correctness of Theorem \ref{thm:technical:raic:modified}),
the uniform convergence of BIHT approximations 
is shown as
follows.
%
\begin{enumerate}[label=(\alph*)]
\item \label{enum:intro:techniques:biht:1}
The \( 0\Th \) BIHT approximation, which is simply drawn uniformly at random,
\(
  \Vec{\hat{x}}^{(0)} \sim \SparseSphereSubspace{k}{n}
\),
can be seen to have an error of at most \( 2 \) (the diameter of the unit sphere).
Then, the following argument handles each subsequent \( \Iter\Th \) BIHT approximation,
\(
  \Iter \in \Z_{+}
\).
\item \label{enum:intro:techniques:biht:2}
Using standard techniques, the error of any \( \Iter\Th \) BIHT approximation,
\(
  \Iter \in \Z_{+}
\),
can be shown to be (deterministically) upper bounded by
\begin{gather}
\label{eqn:intro:error-upper-bound:alg}
  \DistS*{\Vec{x}}{\Vec{\hat{x}}^{(\Iter)}}
  \leq
  4
  \big\|
    \big( \Vec{x} - \Vec{\hat{x}}^{(\Iter-1)} \big)
    -
    \hA[\supp( \Vec{\hat{x}}^{\Iter} )] \big( \Vec{x}, \Vec{\hat{x}}^{(\Iter-1)} \big)
  \big\|_{2}
.\end{gather}
\item \label{enum:intro:techniques:biht:3}
Subsequently, observing the correspondence between \EQN \eqref{eqn:intro:error-upper-bound:alg}
and the RAIC, Theorem \ref{thm:technical:raic:modified} is applied to further bound the
\( \Iter\Th \) approximation error 
in \eqref{eqn:intro:error-upper-bound:alg} from above by
\begin{align}
\label{eqn:intro:error-upper-bound:raic}
  \DistS*{\Vec{x}}{\Vec{\hat{x}}^{(\Iter)}}
  &\leq
  4
  \left(
    \UnivConstc_{1}
    \sqrt{\frac{\Epsilon}{\UnivConstC} \DistS*{\Vec{x}}{\Vec{\hat{x}}^{(\Iter-1)}}}
    +
    \UnivConstc_{2}
    \frac{\Epsilon}{\UnivConstC}
  \right)\nonumber\\
  &=
  \ErrorRecurrenceDef{\Iter-1}
.\end{align}
\item \label{enum:intro:techniques:biht:4}
Then, the recurrence relation corresponding to the right-hand-side of
\EQN \eqref{eqn:intro:error-upper-bound:raic},
\begin{gather}
  \varepsilon( 0 ) = 2
  ,\\
  \varepsilon( \Iter )
  =
  \VarepsilonDef
,\end{gather}
can be shown to monotonically decrease with \( \Iter \),
asymptotically converging as \( \varepsilon( \Iter ) \sim \Epsilon \), and
pointwise upper bounded by
\(
  \varepsilon( \Iter ) \leq 2^{2^{-\Iter}} \Epsilon^{1-2^{-\Iter}}
\)
for each
\(
  \Iter \in \Z_{\geq 0}
\).
The asymptotic convergence and convergence rate of the BIHT apprximations to the
\( \Epsilon \)-ball around the unknown vector \( \Vec{x} \) directly follow.
This will complete the analysis for the universal convergence of the BIHT algorithm.
\end{enumerate}


\subsubsection{The RAIC for an i.i.d. Gaussian Matrix}
\label{outline:intro:overview-techniques|>raic}

Fixing
\(
  \DDelta, \Rho \in (0,1)
\)
and letting
\(
  \UnivConstc_{1}, \UnivConstc_{2} > 0
\)
be the universal constants specified in Eq.~\eqref{eqn:univConstants}, 
Theorem \ref{thm:technical:raic:modified} establishes that the measurement matrix
\(
  \MeasMat \in \R^{m \times n}
\)
with i.i.d. \standardnormal entries satisfies the
\( (k, n, \DDelta, \UnivConstc_{1}, \UnivConstc_{2}) \)-RAIC
with high probability (at least \( 1 - \Rho \))
when the number of measurements $m$
is at least what is stated in Eq.~\eqref{eqn:technical:raic:modified:m}.
%
The proof of the theorem 
is outlined as follows.
%
\begin{enumerate}[label=(\alph*)]
\item \label{enum:intro:techniques:raic:1}
Writing
\(
  \Tau \defeq \TauValue
\),
suppose
\(
  \{ \Net{\Tau;\Coords{J}} \subseteq \SparseSphereSubspace{k}{n} :
     \Coords{J} \subseteq [n],
     | \Coords{J} | \leq k
  \}
\)
are \( \Tau \)-nets over the subset of vectors in \( \SparseSphereSubspace{k}{n} \) whose
support sets are precisely \( \Coords{J} \).
Then, a \( \Tau \)-net over the entire set of \ksparserealvalued vectors,
\( \SparseSphereSubspace{k}{n} \), is constructed by the union
\(
  \Net{\Tau}
  =
  \bigcup_{\Coords{J} \subseteq [n] : | \Coords{J} | \leq k} \Net{\Tau;\Coords{J}}
\).
%
The goal will be to show that with high probability certain properties hold for (almost) every
ordered pair
\(
  ( \Vec{u}, \Vec{v} ) \in \Net{\Tau} \times \Net{\Tau}
\),
\ORIG{or for every vector
\(
  \Vec{u} \in \Net{\Tau}
\).}%
\EDIT{or for every pair of vectors \(  \Vec{u} \in \SparseSphereSubspace{k}{n}  \) and \(  \Vec{v} \in \Ball{\Tau}( \Vec{u} ) \cap \SparseSphereSubspace{k}{n}  \).}
The desired RAIC will then follow from extending the properties to every pair
\(
  \Vec{x},\Vec{y} \in \SparseSphereSubspace{k}{n}
\).
\item \label{enum:intro:techniques:raic:2}
The first property, corresponding with the ``large distance'' regime
(recall the discussion in \SECTION \ref{outline:intro|>contributions}),
requires that with probability at least
\(
  1 - \Rho_{1}
\),
for every ordered pair,
\(
  ( \Vec{u}, \Vec{v} ) \in \Net{\Tau} \times \Net{\Tau}
\),
in the \( \Tau \)-net with distance at least
\(
  \DistS{\Vec{u}}{\Vec{v}} \geq \Tau
\)
and for every
\(
  \Coords{J} \subseteq [n]
\),
\(
  | \Coords{J} | \leq 2k
\),
\begin{gather}
\label{eqn:intro:techniques:large-scale}
  \left\| ( \Vec{u} - \Vec{v} ) - \hA[\Coords{J}]( \Vec{u}, \Vec{v} ) \right\|_{2}
  \leq
  \UnivConstb_{1}
  \sqrt{\DDelta \DistS{\Vec{u}}{\Vec{v}}}
,\end{gather}
where
\(
  \UnivConstb_{1} > 0
\)
is a small universal constant
(see, \EQN \eqref{eqn:univConstants}).
%
\item \label{enum:intro:techniques:raic:3}
The second property, corresponding with the ``small distance'' regime,
requires that with probability at least
\(
  1 - \Rho_{2}
\),
for each 
\ORIG{\(
  \Vec{u} \in \Net{\Tau}
\),
each
\(
  \Vec{x} \in \BallSparseSphere{\Tau}( \Vec{u} )
\),}%
\EDIT{\(  \Vec{u} \in \SparseSphereSubspace{k}{n}  \) and \(  \Vec{v} \in \Ball{\Tau}( \Vec{u} )  \cap \SparseSphereSubspace{k}{n} \),}
and each
\(
  \Coords{J} \subseteq [n]
\),
\(
  | \Coords{J} | \leq 2k
\),
\begin{gather}
\label{eqn:intro:techniques:small-scale}
  \left\| ( \Vec{x} - \Vec{u} ) - \hA[\Coords{J}]( \Vec{x}, \Vec{u} ) \right\|_{2}
  \leq
  \UnivConstb_{2}
  \DDelta
,\end{gather}
where
\(
  \UnivConstb_{2} > 0
\)
is a small universal constant
(again see, \EQN \eqref{eqn:univConstants}).
%
\item \label{enum:intro:techniques:raic:4}
Requiring
\(
  \Rho_{1} + \Rho_{2} = \Rho
\),
the last step of the proof derives the RAIC claimed in the theorem by using the results from
\STEPS \ref{enum:intro:techniques:raic:2} and \ref{enum:intro:techniques:raic:3}, such that the condition holds with probability at least \( 1 - \Rho \) uniformly in all possible cases.
\end{enumerate}
%
We provide a more thorough overview of
\STEPS \ref{enum:intro:techniques:raic:2} and \ref{enum:intro:techniques:raic:3}
next in \SECTION \ref{outline:intro:overview-techniques|>raic|>large-scale}.


\subsubsection{Large- and Small-Distance Regimes --
Steps \ref{enum:intro:techniques:raic:2} and \ref{enum:intro:techniques:raic:3}}
\label{outline:intro:overview-techniques|>raic|>large-scale}

Before discussing the approach to \STEPS \ref{enum:intro:techniques:raic:2} and
\ref{enum:intro:techniques:raic:3}, let us first motivate the argument.
Let
\(
  \Vec{x}, \Vec{y} \in \SparseSphereSubspace{k}{n}
\).
Notice that the function \( \hA( \Vec{x}, \Vec{y} ) \) can be written as
\begin{align*}
  \hA \left( \Vec{x}, \Vec{y} \right)
  &=
  \frac{\sqrt{2\pi}}{m}
  \MeasMat^{\T}
  \cdot
  \frac{1}{2}
  \big( \Sgn{}( \MeasMat \Vec{x} ) - \Sgn{}( \MeasMat \Vec{y} ) \big)
  \\*
  &=
  \frac{\sqrt{2\pi}}{m}
  \sum_{i=1}^{m}
  \MeasVec^{(i)}
  \cdot
  \frac{1}{2}
  \big(
    \Sgn{}( \langle \MeasVec^{(i)}, \Vec{x} \rangle )
    -
    \Sgn{}( \langle \MeasVec^{(i)}, \Vec{y} \rangle )
  \big)
  \\*
  &=
  \frac{\sqrt{2\pi}}{m}
  \sum_{i=1}^{m}
  \MeasVec^{(i)}
  \cdot
  \Sgn{}( \langle \MeasVec^{(i)}, \Vec{x} \rangle )
  \cdot
  \I{}(\Sgn{}( \langle \MeasVec^{(i)}, \Vec{x} \rangle )
     \neq \Sgn{}( \langle \MeasVec^{(i)}, \Vec{y} \rangle ))
\TagEqn\label{eqn:intro:techniques:hA:rewritten}
.\end{align*}
%
Hence, given the random vector
\begin{gather*}
  \Vec{R}_{\Vec{x},\Vec{y}}
  =
  \frac{1}{2}
  \left( \Sgn( \MeasMat \Vec{x} ) - \Sgn( \MeasMat \Vec{y} ) \right)
,\end{gather*}
which takes values in
\(
  \{ -1, 0, 1 \}^{m}
\),
and defining the random variable
\begin{gather*}
  \RV{L}_{\Vec{x},\Vec{y}}
  = \left\| \Vec{R}_{\Vec{x},\Vec{y}} \right\|_{0}
  = \sum_{i=1}^{m}
    \I{}( \Sgn{}( \langle \MeasVec^{(i)}, \Vec{x} \rangle )
          \neq \Sgn{}( \langle \MeasVec^{(i)}, \Vec{y} \rangle ))
,\end{gather*}
which tracks number of \emph{mismatches}
(again, recall the discussion in \SECTION \ref{outline:intro|>contributions}),
the random vector
\(
  ( \hA \left( \Vec{x}, \Vec{y} \right) \mid \Vec{R}_{\Vec{x},\Vec{y}} )
\)
becomes a function of only \( \RV{L}_{\Vec{x},\Vec{y}} \)-many random vectors, where
\(
  \RV{L}_{\Vec{x},\Vec{y}} \leq m
\).
Such conditioning on \( \Vec{R}_{\Vec{x},\Vec{y}} \) will allow for tighter concentration inequalities related to
(an orthogonal decomposition of)
the random vector
\(
  ( \hA \left( \Vec{x}, \Vec{y} \right) \mid \Vec{R}_{\Vec{x},\Vec{y}} )
\).
Note that these concentration inequalities, stated in
Appendix \ref{outline:technical:pf},
provide the same inequality for any
\( \RV{L}_{\Vec{x},\Vec{y}} = \| \Vec{R}_{\Vec{x},\Vec{y}} \|_{0} \) and
\( \RV{L}_{\Vec{x'},\Vec{y'}} = \| \Vec{R}_{\Vec{x'},\Vec{y'}} \|_{0} \),
whenever
\(
  \RV{L}_{\Vec{x},\Vec{y}} = \RV{L'}_{\Vec{x'},\Vec{y'}}
\),
where
\(
  \Vec{x}, \Vec{y}, \Vec{x'}, \Vec{y'} \in \SparseSphereSubspace{k}{n}
\),
and thus it suffices to have a handle on (an appropriate subset of) the random variables
\(
  \{
    \RV{L}_{\Vec{x},\Vec{y}} :
    \Vec{x}, \Vec{y} \in \SparseSphereSubspace{k}{n}
  \}
\).
\par
With this intuition in mind, we will now lay down the specifics of deriving the results achieved by \STEPS \ref{enum:intro:techniques:raic:2} and \ref{enum:intro:techniques:raic:3} for the ``large-'' and ``small-distance'' regimes.
Each follows from two primary arguments.
%
First, for a given
\(
  \Vec{u}, \Vec{v} \in \Net{\Tau}
\),
the associated random variable \( \RV{L}_{\Vec{u},\Vec{v}} \) is bounded.
Then, conditioning on \( \RV{L}_{\Vec{u},\Vec{v}} \), the desired properties in
\STEPS \ref{enum:intro:techniques:raic:2} and \ref{enum:intro:techniques:raic:3}
follow from the appropriate concentration inequalities related to the decomposition of
\( \hA[\Coords{J}] \left( \Vec{x}, \Vec{y} \right) \) into three orthogonal components.
\par
Specifically, \STEP \ref{enum:intro:techniques:raic:2} is achieved as follows.
\begin{enumerate}[label=(\roman*)]
%
\item \label{enum:intro:techniques:raic:2:1}
Consider any
\(
  ( \Vec{u}, \Vec{v} ) \in \Net{\Tau} \times \Net{\Tau}
\)
such that
\(
  \DistS{\Vec{u}}{\Vec{v}} \geq \Tau
\),
and fix
\(
  \JX \subseteq [n]
\),
\(
  | \JX | \leq 2k
\),
arbitrarily.
\item \label{enum:intro:techniques:raic:2:2}
It can be shown that for a small
\(
  \Variable{s} \in (0,1)
\),
the number, \( \RV{L}_{\Vec{u},\Vec{v}} \), of points among
\(
  \MeasVec^{(i)},  {i \in [m]} 
\),
for which a {mismatch} occurs, i.e.,
\(
  \sgn( \langle \MeasVec^{(i)}, \Vec{u} \rangle )
        \neq \sgn( \langle \MeasVec^{(i)}, \Vec{v} \rangle )
\),
is bounded in the range
\begin{gather}
  \RV{L}_{\Vec{u},\Vec{v}}
  \in
  \left[
    (1-\Variable{s}) \frac{\theta_{\Vec{u},\Vec{v}} m}{\pi},
    (1+\Variable{s}) \frac{\theta_{\Vec{u},\Vec{v}} m}{\pi}
  \right]
\end{gather}
uniformly with high probability for all
\(
  ( \Vec{u}, \Vec{v} ) \in \Net{\Tau} \times \Net{\Tau}
\).
\item \label{enum:intro:techniques:raic:2:3}
Define
\(
  \gA : \R^{n} \times \R^{n} \to \R^{n}
\)
by
\begin{align}
  \gA( \Vec{u}, \Vec{v} )
  =&
  \hA( \Vec{u}, \Vec{v} )
  -
  \left\langle
    \frac{\Vec{u}-\Vec{v}}{\left\| \Vec{u}-\Vec{v} \right\|_{2}},
    \hA( \Vec{u}, \Vec{v} )
  \right\rangle
  \frac{\Vec{u}-\Vec{v}}{\left\| \Vec{u}-\Vec{v} \right\|_{2}}
  -
  \left\langle
    \frac{\Vec{u}+\Vec{v}}{\left\| \Vec{u}+\Vec{v} \right\|_{2}},
    \hA( \Vec{u}, \Vec{v} )
  \right\rangle
  \frac{\Vec{u}+\Vec{v}}{\left\| \Vec{u}+\Vec{v} \right\|_{2}}
\label{eqn:intro:techniques:g_A}
\end{align}
where
\(
  \gA[\JX]( \Vec{u}, \Vec{v} )
  =
  \ThresholdSet{\supp( \Vec{u} ) \cup \supp( \Vec{v} ) \cup \JX}
  ( \gA( \Vec{u}, \Vec{v} ) )
\).
Note that \( \hA \) and \( \hA[\JX] \) can then be orthogonally decomposed into
\begin{align}
  \hA( \Vec{u}, \Vec{v} )
  =&
  \left\langle
    \frac{\Vec{u}-\Vec{v}}{\left\| \Vec{u}-\Vec{v} \right\|_{2}},
    \hA( \Vec{u}, \Vec{v} )
  \right\rangle
  \frac{\Vec{u}-\Vec{v}}{\left\| \Vec{u}-\Vec{v} \right\|_{2}}
  +
  \left\langle
    \frac{\Vec{u}+\Vec{v}}{\left\| \Vec{u}+\Vec{v} \right\|_{2}},
    \hA( \Vec{u}, \Vec{v} )
  \right\rangle
  \frac{\Vec{u}+\Vec{v}}{\left\| \Vec{u}+\Vec{v} \right\|_{2}}
  +
  \gA( \Vec{u}, \Vec{v} )
\label{eqn:intro:techniques:h_A:decomposition}
\end{align}
and
\begin{align}
  \nonumber
  \hA[\JX]( \Vec{u}, \Vec{v} )
  =&
  \ThresholdSet{\supp( \Vec{u} ) \cup \supp( \Vec{v} ) \cup \JX}
  (\hA( \Vec{u}, \Vec{v} ))
  \\
  &=
  \left\langle
    \frac{\Vec{u}-\Vec{v}}{\left\| \Vec{u}-\Vec{v} \right\|_{2}},
    \hA( \Vec{u}, \Vec{v} )
  \right\rangle
  \frac{\Vec{u}-\Vec{v}}{\left\| \Vec{u}-\Vec{v} \right\|_{2}}
  +
  \left\langle
    \frac{\Vec{u}+\Vec{v}}{\left\| \Vec{u}+\Vec{v} \right\|_{2}},
    \hA( \Vec{u}, \Vec{v} )
  \right\rangle
  \frac{\Vec{u}+\Vec{v}}{\left\| \Vec{u}+\Vec{v} \right\|_{2}}
  +
  \gA[\JX]( \Vec{u}, \Vec{v} )
\label{eqn:intro:techniques:h_AJ:decomposition}
.\end{align}
%
Note that \cite{friedlander2021nbiht} similarly uses such a decomposition to show their RAIC,
and this decomposition technique appears earlier in \cite{plan2017high}.
\item \label{enum:intro:techniques:raic:2:4}
Conditioned on
\(
  \RV{L}_{\Vec{u},\Vec{v}}
  \in
  [
    (1-\Variable{s}) \frac{\theta_{\Vec{u},\Vec{v}} m}{\pi},
    (1+\Variable{s}) \frac{\theta_{\Vec{u},\Vec{v}} m}{\pi}
  ]
\),
the desired property in \EQN \eqref{eqn:intro:techniques:large-scale} is derived from
\EQN \eqref{eqn:intro:techniques:h_AJ:decomposition} using a concentration inequality provided by Lemma \ref{lemma:technical:concentration-ineq:(u,v)}
together with standard techniques,
e.g., the triangle inequality.
\item \label{enum:intro:techniques:raic:2:5}
A union bound extends \EQN \eqref{eqn:intro:techniques:large-scale} to hold uniformly over
\(
  \Net{\Tau} \times \Net{\Tau}
\)
and all
\(
  \JX \subseteq [n]
\),
\(
  | \JX | \leq 2k
\),
with high probability, completing \STEP \ref{enum:intro:techniques:raic:2}.
\end{enumerate}
%
\par
Step \ref{enum:intro:techniques:raic:3} takes a similar approach,
\EDIT{but in place of (direct use of) a \(  \Tau  \)-net, the local stability of binary embeddings via Gaussian measurements, established by \cite{oymak2015near}, will lead to a uniform result.
The argument is outlined as follows:}
\begin{enumerate}[label=(\roman*)]
%
\item \label{enum:intro:techniques:raic:3:1}
\ORIG{Let
\(
  \Vec{u} \in \Net{\Tau}
\)
be an arbitrary vector in the \( \Tau \)-net,
and fix any
\(
  \JX \subseteq [n]
\),
\(
  | \JX | \leq 2k
\).}%
\ORIG{Recall that the desired property in \EQN \eqref{eqn:intro:techniques:small-scale}
should hold for all
\(
  \Vec{x} \in \BallSparseSphere{\Tau}( \Vec{u} )
\).}%
\EDIT{Here, we consider pairs of \(  k  \)-sparse points, \(  \Vec{u}, \Vec{x} \in \SparseSphereSubspace{k}{n}  \), where \(  \Vec{x}  \) is contained in a small ball around \(  \Vec{u}  \)---formally, \(  \Vec{u} \in \SparseSphereSubspace{k}{n}  \) and \(  \Vec{x} \in \Ball{\Tau}( \Vec{u} ) \cap \SparseSphereSubspace{k}{n}  \).}
\begin{EDITb}
\item \label{enum:intro:techniques:raic:3:1:2}
Towards obtaining \EQN \eqref{eqn:intro:techniques:small-scale}, the triangle inequality is applied to break up its \LHS:
\begin{gather*}
  \left\| ( \Vec{x} - \Vec{u} ) - \hA[\JX]( \Vec{x}, \Vec{u} ) \right\|_{2}
  \leq
  \left\| \Vec{x} - \Vec{u} \right\|_{2}
  +
  \left\| \hA[\JX]( \Vec{x}, \Vec{u} ) \right\|_{2}
\end{gather*}
where \(  \JX \subseteq [n]  \), \(  | \JX | \leq 2k  \).
\item \label{enum:intro:techniques:raic:3:1:3}
Since \(  \left\| \Vec{x} - \Vec{u} \right\|_{2} \leq \Tau = \BigO{}( \DDelta )  \) by assumption, the main task is uniformly bounding \(  \left\| \hA[\JX]( \Vec{x}, \Vec{u} ) \right\|_{2}  \) with high probability.
\end{EDITb}
\item \label{enum:intro:techniques:raic:3:2}
\ORIG{To ensure this uniform result over \( \BallSparseSphere{\Tau}( \Vec{u} ) \), construct a second net
\(
  \BallNet{\Tau}( \Vec{u} ) \subseteq \BallSparseSphere{\Tau}( \Vec{u} )
\)
such that for each
\(
  \Vec{x} \in \BallSparseSphere{\Tau}( \Vec{u} )
\),
there exits a point
\(
  \Vec{w} \in \BallNet{\Tau}( \Vec{u} )
\)
such that
\(
  \sgn( \MeasMat \Vec{w} ) = \sgn( \MeasMat \Vec{x} )
\).
The next step will upper bound the size of \( \BallNet{\Tau}( \Vec{u} ) \).}%
\EDIT{As in \STEP \ref{enum:intro:techniques:raic:2}, the argument here will use an upper bound on \(  \RV{L}_{\Vec{x},\Vec{u}}  \).
Towards this, let \(  \kX \in \Z_{+}  \), determined later, and let \(  \Set{\hat{W}} \subseteq \R^{n}  \) be a \(  \kX  \)-dimensional subspace of \(  \R^{n}  \), and write \(  \Set{W} \defeq \Set{\hat{W}} \cap \Sphere{n}  \).
Due to \cite[\COR 3.3]{oymak2015near} (see, Lemma \ref{lemma:technical:concentration-ineq:lbe-local-deviations} in \APPENDIX \ref{outline:technical|>concentration-ineq-pfs|>lbe-local-deviations:union}), given \(  m = \BigO{}( \frac{\kX}{\DDelta} \log( \frac{1}{\DDelta} ) )  \) i.i.d. Gaussian vectors, uniformly with high probability, for every pair of vectors, \(  \Vec{u}, \Vec{v} \in \Set{W}  \), which are distance at most \(  \BigO{}( \frac{\DDelta}{\smash[b]{\sqrt{\log( \hfrac{1}{\DDelta} )}}} )  \) apart, the number of the Gaussian vectors lying on opposite sides of \(  \Vec{u}  \) and \(  \Vec{v}  \) is at most \(  \DDelta m  \).}
\item \label{enum:intro:techniques:raic:3:3}
\EDIT{Clearly, the restriction of \(  \R^{n}  \) to a support set of up to \(  \kX  \) coordinates (and subsets thereof) forms a subspace of dimension at most \(  \kX  \).
Hence, \cite[\COR 3.3]{oymak2015near} can be applied repeatedly to each such subspace of \(  \R^{n}  \) induced by a restriction to up to \(  \kX  \) coordinate, and these individual results can be combined with a union bound.
\item \label{enum:intro:techniques:raic:3:4}
Since we are interested in pairs of \(  k  \)-sparse vectors, \(  \Vec{u}, \Vec{v} \in \SparseSphereSubspace{k}{n}  \), here we take \(  \kX = 2k  \).
Moreover, because we only consider vector pairs of the form \(  \Vec{u} \in \SparseSphereSubspace{k}{n}  \), \(  \Vec{x} \in \Ball{\Tau}( \Vec{u} ) \cap \SparseSphereSubspace{k}{n}  \), the locality principle---that \(  \Vec{u}  \) and \(  \Vec{x}  \) are at distance at most \(  \BigO{}( \frac{\DDelta}{\smash[b]{\sqrt{\log( \hfrac{1}{\DDelta} )}}} )  \)---will always be upheld as long as \(  \Tau  \) is defined appropriately.
Ultimately, this leads to a uniform bound on \(  \RV{L}_{\Vec{x},\Vec{u}}  \) for all \(  \Vec{u} \in \SparseSphereSubspace{k}{n}  \) and \(  \Vec{x}
\in \Ball{\Tau}( \Vec{u} ) \cap \SparseSphereSubspace{k}{n}  \): \(  \RV{L}_{\Vec{x},\Vec{u}} \leq \DDelta m  \) with high probability.}
\ORIG{\item 
Let
\(
  \Beta = \arccos( 1 - \frac{\Tau^{2}}{2} )
\)
be the angle associated with the distance \( \Tau \),
and define the random variable
\(
  \RV{M}_{\Beta,\Vec{u}}
  =
  |
    \{
      \MeasVec^{(i)}, i \in [m] :
      \theta_{\Vec{w},\MeasVec^{(i)}}
      \in
      [
        \frac{\pi}{2} - \Beta,
        \frac{\pi}{2} + \Beta
      ]
    \}
  |
\).
Notice that the size of \( \BallNet{\Tau}( \Vec{u} ) \) need not exceed
\(
  2^{\RV{M}_{\Beta,\Vec{u}}}
\).
Moreover, for any
\(
  \Vec{x} \in \BallSparseSphere{\Tau}( \Vec{u} )
\)
with
\(
  \theta_{\Vec{x},\Vec{u}} \in [0,\Beta]
\),
the value taken by the random variable \( \RV{M}_{\Beta,\Vec{u}} \) upper bounds
the number of points
\(
  \MeasVec^{(i)} 
\),
\(
  i \in [m]
\),
on which
\(
  \sgn( \langle \MeasVec^{(i)}, \Vec{x} \rangle )
\)
and
\(
  \sgn( \langle \MeasVec^{(i)}, \Vec{u} \rangle )
\)
mismatch---or more formally,
\(
  \RV{L}_{\Vec{x},\Vec{u}} \leq \RV{M}_{\Beta,\Vec{u}}
\)
for every
\(
  \Vec{x} \in \BallSparseSphere{\Tau}( \Vec{u} )
\).}%
\item \label{enum:intro:techniques:raic:3:5}
\ORIG{Taking any
\(
  \Vec{w} \in \BallNet{\Tau}( \Vec{u} )
\)
and conditioning on
\(
  \RV{L}_{\Vec{x},\Vec{u}}
\),}%
\EDIT{Let \(  \Vec{u} \in \SparseSphereSubspace{k}{n}  \) and \(  \Vec{x}
\in \Ball{\Tau}( \Vec{u} ) \cap \SparseSphereSubspace{k}{n}  \), and fix \(  \JX \subseteq [n]  \), \(  | \JX | \leq 2k  \).}
\EDIT{Conditioning on \(  \RV{L}_{\Vec{x},\Vec{u}} \leq \DDelta m  \),}
the norm of \( \hA[\JX]( \Vec{x}, \Vec{u} ) \) is then bounded using an orthogonal decomposition analogous to that in \STEP \ref{enum:intro:techniques:raic:2}, and again applying
the concentration inequalities in Lemma \ref{lemma:technical:concentration-ineq:(u,v)},
along with standard techniques, to obtain
\(
  \left\| \hA[\JX]( \Vec{x}, \Vec{u} ) \right\|_{2}
  \leq \BigO( \EDIT{\delta} )
\).
\item \label{enum:intro:techniques:raic:3:6}
\EDIT{This bound is then extended to hold uniformly for all \(  \Vec{u} \in \SparseSphereSubspace{k}{n}  \), \(  \Vec{x} \in \Ball{\Tau}( \Vec{u} ) \cap \SparseSphereSubspace{k}{n}  \), and \(  \JX \subseteq [n]  \), \(  | \JX | \leq 2k  \), by the result obtained in \STEP \ref{enum:intro:techniques:raic:3:5} in the case of the first two, and by a union bound bound in the last case.}
\ORIG{\item \label{enum:intro:techniques:raic:3:7}
Step \ref{enum:intro:techniques:raic:3} concludes by arguing that the uniform result from step
\ref{enum:intro:techniques:raic:3:6} suffices to ensure
\EQN \eqref{eqn:intro:techniques:small-scale} holds uniformly for all
\(
  \Vec{u} \in \Net{\Tau}
\),
\(
  \Vec{x} \in \Ball{\Tau}( \Vec{u} )
\),
and 
\(
  \JX \subseteq [n]
\),
\(
  | \JX | \leq 2k
\),
by observing that for each
\(
  \Vec{x} \in \Ball{\Tau}( \Vec{u} )
\),
the construction of the net, \( \BallNet{\Tau}( \Vec{u} ) \), ensures the existence of
\(
  \Vec{w} \in \BallNet{\Tau}( \Vec{u} )
\)
such that
\(
  \left\| \hA[\JX]( \Vec{x}, \Vec{u} ) \right\|_{2}
  =
  \left\| \hA[\JX]( \Vec{w}, \Vec{u} ) \right\|_{2}
  \leq \BigO( \Tau )
\).
The argument additionally applies the triangle inequality:
\(
  \left\| ( \Vec{x} - \Vec{u} ) - \hA[\JX]( \Vec{x}, \Vec{u} ) \right\|_{2}
  \leq
  \left\| \Vec{x} - \Vec{u} \right\|_{2}
  +
  \left\| \hA[\JX]( \Vec{x}, \Vec{u} ) \right\|_{2}
  \leq
  \BigO( \Tau )
\).}%
\end{enumerate}

\subsubsection{Combining the Intermediate Results to Complete the Proof --
Step \ref{enum:intro:techniques:raic:4}}
\label{outline:intro:overview-techniques|>raic|>combine}

\begin{enumerate}[label=(\roman*)]
\item
Fix an arbitrary pair of \( k \)-sparse unit vectors
\( \Vec{x}, \Vec{y} \in \SparseSphereSubspace{k}{n} \),
and let
\( \Vec{u}, \Vec{v} \in \Net{\Tau} \)
be the closest net points, respectively, each pair sharing the same support sets.
Note that it is possible to set \( \Vec{u} = \Vec{x} \) if \( \Vec{x} \in \Net{\Tau} \),
and likewise for \( \Vec{v} \) if \( \Vec{y} \in \Net{\Tau} \).
Let \( \Coords{J} \subseteq [n]  \), \( |\Coords{J}| \leq k  \) be any \( k \)-subset of coordinates.
Moreover, write
\( \Coords{J}_{\Vec{x}} = \Coords{J} \cup \supp( \Vec{x} ) \) and
\( \Coords{J}_{\Vec{y}} = \Coords{J} \cup \supp( \Vec{y} ) \),
each having size no more than \( 2k \).
\item
It is straightforward to show with algebraic manipulation that
\begin{align}
  ( \Vec{x} - \Vec{y} ) - \hA( \Vec{x}, \Vec{y} )
  = ( \Vec{u} - \Vec{v} ) - \hA( \Vec{u}, \Vec{v} )
    + ( \Vec{x} - \Vec{u} ) - \hA( \Vec{x}, \Vec{u} )
    + ( \Vec{v} - \Vec{y} ) - \hA( \Vec{v}, \Vec{y} )
,\end{align}
and similarly that
\begin{align}
  ( \Vec{x} - \Vec{y} ) - \hA[\Coords{J}]( \Vec{x}, \Vec{y} )
  = ( \Vec{u} - \Vec{v} ) - \hA[\Coords{J}]( \Vec{u}, \Vec{v} )
    + ( \Vec{x} - \Vec{u} ) - \hA[\Coords{J}_{\Vec{y}}]( \Vec{x}, \Vec{u} )
    + ( \Vec{v} - \Vec{y} ) - \hA[\Coords{J}_{\Vec{x}}]( \Vec{v}, \Vec{y} )
\label{eqn:d:1}
.\end{align}
\item
The \( \lnorm{2} \)-norm of the left-hand-side of \EQN \eqref{eqn:d:1} can be bounded by splitting
it up into the sum of three terms via the triangle inequality, specifically,
\begin{align}
  \nonumber
  &\| ( \Vec{x} - \Vec{y} ) - \hA[\Coords{J}]( \Vec{x}, \Vec{y} ) \|_{2}
  \\*
  &\quad
  \leq \| ( \Vec{u} - \Vec{v} ) - \hA[\Coords{J}]( \Vec{u}, \Vec{v} ) \|_{2}
    + \| ( \Vec{x} - \Vec{u} ) - \hA[\Coords{J}_{\Vec{y}}]( \Vec{x}, \Vec{u} ) \|_{2}
    + \| ( \Vec{v} - \Vec{y} ) - \hA[\Coords{J}_{\Vec{x}}]( \Vec{v}, \Vec{y} ) \|_{2}
\label{eqn:d:2}
.\end{align}
\item
Now, we consider two cases based on whether \( \DistS{\Vec{u}}{\Vec{v}} \) is above or below the
threshold \( \Tau \) and derive bounds using \EQN \eqref{eqn:d:2}, as well as the results from
\STEPS \ref{enum:intro:techniques:raic:2} and \ref{enum:intro:techniques:raic:3}.
If \( \DistS{\Vec{u}}{\Vec{v}} < \Tau \), then using the result from
\STEP \ref{enum:intro:techniques:raic:3}, we obtain
\begin{align}
  \| ( \Vec{x} - \Vec{y} ) - \hA[\Coords{J}]( \Vec{x}, \Vec{y} ) \|_{2}
  \leq
  3
  \UnivConstb_{2}
  \DDelta
\label{eqn:d:3}
.\end{align}
Otherwise, when \( \DistS{\Vec{u}}{\Vec{v}} \geq \Tau \), using the results from both
\STEPS \ref{enum:intro:techniques:raic:2} and \ref{enum:intro:techniques:raic:3}
we obtain
\begin{align}
  \| ( \Vec{x} - \Vec{y} ) - \hA[\Coords{J}]( \Vec{x}, \Vec{y} ) \|_{2}
  \leq
  \UnivConstb_{1}
  \sqrt{\DDelta \DistS{\Vec{u}}{\Vec{v}}}
  +
  2
  \UnivConstb_{2}
  \DDelta
\label{eqn:d:4}
.\end{align}
Moreover, \EQN \eqref{eqn:d:3} and \eqref{eqn:d:4} are both trivially upper bounded by
\begin{align}
  \| ( \Vec{x} - \Vec{y} ) - \hA[\Coords{J}]( \Vec{x}, \Vec{y} ) \|_{2}
  \leq
  \UnivConstb_{1}
  \sqrt{\DDelta \DistS{\Vec{u}}{\Vec{v}}}
  +
  3
  \UnivConstb_{2}
  \DDelta
\label{eqn:d:5}
.\end{align}
\item
Then, using the universal constants defined in \EQN \eqref{eqn:univConstants},
the RAIC claimed in Theorem \ref{thm:technical:raic:modified} follows.
\end{enumerate}


\section{Proof of the Main Result---BIHT Convergence}
\label{outline:pf-main-thm}

\subsection{Intermediate Results}
\label{outline:pf-main-thm|>intermediate-lemmas}

Before proving the main theorems, Theorem \ref{thm:main:convergence:t-rate} and
\ref{thm:main:convergence:t->infty}, three intermediate results, in
Lemmas \ref{lemma:biht:error-upper-bound:alg}-\ref{lemma:biht:error:explicit},
are presented to facilitate the analysis for the convergence of BIHT approximations.
The proofs for these intermediate results are in
Section \ref{outline:biht:pf-main-thm|>intermediate-lemmas-pf}.
%
\begin{lemma}
\label{lemma:biht:error-upper-bound:alg}
Consider any
\(
  \Vec{x} \in \SparseSphereSubspace{k}{n}
\)
and any
\(
  \Iter \in \Z_{+}
\).
The error of the \( \Iter\Th \) approximation produced by the BIHT algorithm satisfies
\begin{align}
  \DistS*{\Vec{x}}{\Vec{\hat{x}}^{(\Iter)}}
  &\leq
  4
  \left\|
    \big( \Vec{x} - \Vec{\hat{x}}^{(\Iter-1)} \big)
    -
    \hA[\supp( \Vec{\hat{x}}^{(\Iter)} )]\big( \Vec{x}, \Vec{\hat{x}}^{(\Iter-1)} \big)
  \right\|_{2}
\label{eqn:biht:error-upper-bound:alg}
.\end{align}
\end{lemma}
%
Note that Lemma \ref{lemma:biht:error-upper-bound:alg} is a deterministic result, arising from the
equation by which the BIHT algorithm computes its \( \Iter\Th \) approximations,
\(
  \Iter \in \Z_{+}
\).
Hence, it holds for all
\(
  \Vec{x} \in \SparseSphereSubspace{k}{n}
\)
and all iterations
\(
  \Iter \in \Z_{+}
\).
%
\begin{lemma}
\label{lemma:biht:error:recurrence}
Let
\(
  \Varepsilon : \Z_{\geq 0} \to \R
\)
be a function given by the recurrence relation
\begin{gather}
\label{eqn:biht:error:recurrence:def}
  \Varepsilon( 0 ) = 2
  ,\\
  \Varepsilon( \Iter )
  =
  \VarepsilonDef
.\end{gather}
%
The function \( \Varepsilon \) decreases monotonically with \( \Iter \) and asymptotically tends
to a value not exceeding \( \Epsilon \)---formally,
\begin{gather}
\label{eqn:biht:error:recurrence:t->infty}
  \lim_{\Iter \to \infty} \Varepsilon( \Iter )
  =
  \VarepsilonAsymptotic
  <
  \Epsilon
.\end{gather}
\end{lemma}
%
\begin{lemma}
\label{lemma:biht:error:explicit}
Let
\(
  \Varepsilon : \Z_{\geq 0} \to \R
\)
be the function as defined in Lemma \ref{lemma:biht:error:recurrence}.
Then, the sequence
\(
  \{ \Varepsilon( \Iter ) \}_{\Iter \in \Z_{\geq 0}}
\)
is bound from above by the sequence
\(
  \{ 2^{2^{-\Iter}} \Epsilon^{1-2^{-\Iter}} \}_{\Iter \in \Z_{\geq 0}}
\).
\end{lemma}


\subsection{Proofs of Theorems \ref{thm:main:convergence:t->infty} and
\ref{thm:main:convergence:t-rate}}
\label{outline:pf-main-thm|>pf}

The main theorems for the analysis of the BIHT algorithm are restated below for convenience and
will subsequently be proved in tandem.
%
\begin{theorem*}[Theorem \ref{thm:main:convergence:t-rate}]
Let
\(
  \UnivConstA, \UnivConstB, \UnivConstC > 0
\)
be universal constants as in Eq.~\eqref{eqn:univConstants}.
Fix
\(
  \Epsilon, \Rho \in (0,1)
\)
and
\(
  k, m, n \in \Z_{+}
\)
where
\begin{EDITb}
\begin{align*}
  \nonumber
  m
  &\geq
  \frac{4 \UnivConstB \UnivConstC \UnivConstD k}{\Epsilon}
  \log
  \left(
    \frac{en}{k}
  \right)
  +
  \frac{2 \UnivConstB \UnivConstC \UnivConstD k}{\Epsilon}
  \log
  \left(
    \frac{12 \UnivConstB \UnivConstC \UnivConstD \Log[^{3/2}]( 2e \UnivConstB \UnivConstC/\Epsilon )}{\Epsilon}
  \right)
  +
  \frac{\UnivConstB \UnivConstC \UnivConstD}{\Epsilon}
  \log
  \left(
    \frac{\UnivConstA}{\Rho}
  \right)
  \\ \nonumber
  &\AlignTab+
  \frac{\UnivConstB \UnivConstC \UnivConstD k}{\Epsilon} \Log( \frac{\UnivConstB \UnivConstC}{\Epsilon} ) \sqrt{\Log( \frac{2e \UnivConstB \UnivConstC}{\Epsilon} )}
  +
  \frac{\EDITX{128} \UnivConstB \UnivConstC k}{\Epsilon} \Log( \frac{en}{\EDITX{k}} ) \sqrt{\Log( \frac{2e \UnivConstB \UnivConstC}{\Epsilon} )}
  \\
  &\AlignTab+
  \frac{\EDITX{64} \UnivConstB \UnivConstC}{\Epsilon} \Log( \EDITX{\frac{\UnivConstAXX}{\Rho}} )
  \EDITX{\sqrt{\Log( \frac{2e \UnivConstB \UnivConstC}{\Epsilon} )}}
  +
  \frac{\EDITX{4} \UnivConstB \UnivConstC \EDITX{k}}{\Epsilon}
  \Log( \frac{en}{\EDITX{k}} )
  +
  \frac{\UnivConstB \UnivConstC}{\Epsilon}
  \Log( \EDITX{\frac{\UnivConstAX}{\Rho}} )
.\end{align*}
\end{EDITb}
Let the measurement matrix
\(
  \MeasMat \in \R^{m \times n}
\)
have rows with i.i.d. \standardnormal entries.
%
Then, uniformly with probability at least
\(
  1 - \Rho
\),
for every unknown \ksparserealvalued unit vector,
\(
  \Vec{x} \in \SparseSphereSubspace{k}{n}
\),
the normalized BIHT algorithm produces a sequence of approximations,
\(
  \{ \Vec{\hat{x}}^{(\Iter)} \in \SparseSphereSubspace{k}{n} \}_{\Iter \in \Z_{\geq 0}}
\),
which converges to the \( \Epsilon \)-ball around the unknown vector \( \Vec{x} \) at a rate
upper bounded by
\begin{gather*}
  \DistS*{\Vec{x}}{\Vec{\hat{x}}^{(\Iter)}}
  \leq
  2^{2^{-\Iter}}
  \Epsilon^{1 - 2^{-\Iter}}
\end{gather*}
for each
\(
  \Iter \in \Z_{\geq 0}
\).
\end{theorem*}
\begin{corollary*}
[Corollary \ref{thm:main:convergence:t->infty}]
Under the conditions stated in Theorem \ref{thm:main:convergence:t-rate},
uniformly with probability at least
\(
  1 - \Rho
\),
for every unknown \ksparserealvalued unit vector,
\(
  \Vec{x} \in \SparseSphereSubspace{k}{n}
\),
the sequence of BIHT approximations,
\(
  \{ \Vec{\hat{x}}^{(\Iter)} \}_{\Iter \in \Z_{\geq 0}}
\),
converges asymptotically to the \( \Epsilon \)-ball around the unknown vector \( \Vec{x} \).
Formally,
\begin{gather*}
  \lim_{\Iter \to \infty} \DistS*{\Vec{x}}{\Vec{\hat{x}}^{(\Iter)}} \leq \Epsilon
.\end{gather*}
\end{corollary*}
%
\begin{proof}
{Theorem \ref{thm:main:convergence:t-rate} and Corollary \ref{thm:main:convergence:t->infty}}
\label{pf:thm:main:convergence}
The convergence of BIHT approximations for an arbitrary unknown, \( k \)-sparse unit vector,
\(
  \Vec{x} \in \SparseSphereSubspace{k}{n}
\),
will follow from the main technical theorem,
Theorem \ref{thm:technical:raic:modified}, and the intermediate lemmas,
Lemmas \ref{lemma:biht:error-upper-bound:alg}-\ref{lemma:biht:error:explicit}.
Recalling that Theorem \ref{thm:technical:raic:modified} and
Lemma \ref{lemma:biht:error-upper-bound:alg} hold uniformly over \( \SparseSphereSubspace{k}{n} \)
(respectively, with bounded probability and deterministically),
the argument then implies uniform convergence for all unknown \( k \)-sparse vectors,
\(
  \Vec{x} \in \SparseSphereSubspace{k}{n}
\).
\par
Consider any unknown, \( k \)-sparse unit vector
\(
  \Vec{x} \in \SparseSphereSubspace{k}{n}
\)
with an associated sequence of BIHT approximations,
\(
  \{ \Vec{\hat{x}}^{(\Iter)} \in \SparseSphereSubspace{k}{n} \}_{\Iter \in \Z_{\geq 0}}
\).
For each
\(
  \Iter \in \Z_{+}
\),
Lemma \ref{lemma:biht:error-upper-bound:alg} bounds the error of the \( \Iter\Th \) approximation
from above by
\begin{align}
  \DistS*{\Vec{x}}{\Vec{\hat{x}}^{(\Iter)}}
  &\leq
  4
  \left\|
    \big( \Vec{x} - \Vec{\hat{x}}^{(\Iter-1)} \big)
    -
    \hA[\supp( \Vec{\hat{x}}^{(\Iter)} )] \big( \Vec{x}, \Vec{\hat{x}}^{(\Iter-1)} \big)
  \right\|_{2}
\label{pf:thm:main:convergence:eqn:1}
\end{align}
which is further bounded by Theorem \ref{thm:technical:raic:modified}
(by setting
\(
  \DDelta
  = \frac{\Epsilon}{\UnivConstC}
  = \frac{\Epsilon}{32}
\)
in the theorem) as
\begin{subequations}
\label{pf:thm:main:convergence:eqn:2:main}
\begin{align}
\label{pf:thm:main:convergence:eqn:2}
  \DistS*{\Vec{x}}{\Vec{\hat{x}}^{(\Iter)}}
  & \leq
  4
  \big\|
    \big( \Vec{x} - \Vec{\hat{x}}^{(\Iter-1)} \big)
    -
    \hA[\supp( \Vec{\hat{x}}^{(\Iter)} )] \big( \Vec{x}, \Vec{\hat{x}}^{(\Iter-1)} \big)
  \big\|_{2}
  \\*
  &\leq
  4
  \left(
    \UnivConstc_{1}
    \sqrt{\frac{\Epsilon}{\UnivConstC} \DistS*{\Vec{x}}{\Vec{\hat{x}}^{(\Iter-1)}}}
    +
    \UnivConstc_{2}
    \frac{\Epsilon}{\UnivConstC}
  \right)
  \\*
  &=
  \ErrorRecurrenceDef{\Iter-1}
\label{pf:thm:main:convergence:eqn:2:end}
\end{align}
\end{subequations}
where in the case of
\(
  \Iter = 1
\),
\eqref{pf:thm:main:convergence:eqn:2:end},
\begin{align}
  \nonumber
  \DistS*{\Vec{x}}{\Vec{\hat{x}}^{(1)}}
  &\leq
  \ErrorRecurrenceDef{0}
  \\* 
  &\leq
  4 \UnivConstc_{1}
  \sqrt{\frac{\Epsilon}{\UnivConstC} \DistS{\Vec{x}}{-\Vec{x}}}
  +
  4 \UnivConstc_{2}
  \frac{\Epsilon}{\UnivConstC}
  =
  \UnivConstc_{1}
  \sqrt{\Epsilon}
  +
  \frac{\UnivConstc_{2}}{8}
  \Epsilon
\label{pf:thm:main:convergence:eqn:3}
.\end{align}
%
Recall that Lemma \ref{lemma:biht:error:recurrence} defines a function
\(
  \Varepsilon : \Z_{\geq 0} \to \R
\)
by the recurrence relation
\begin{gather}
\label{pf:thm:main:convergence:eqn:4}
  \Varepsilon( 0 ) = 2
  ,\\*
  \Varepsilon( \Iter )
  =
  \VarepsilonDef
,\end{gather}
whose form is similar to \eqref{pf:thm:main:convergence:eqn:2:end}.
It can be argued inductively that for every
\(
  \Iter \in \Z_{\geq 0}
\),
the function \( \Varepsilon( \Iter ) \) upper bounds the error of the \( \Iter\Th \) BIHT
approximation,
\(
  \DistS*{\Vec{x}}{\Vec{\hat{x}}^{(\Iter)}}
\),
as discussed next.
The base case,
\(
  \Iter = 0
\),
is trivial since
\begin{gather}
\label{pf:thm:main:convergence:eqn:5}
  \DistS*{\Vec{x}}{\Vec{\hat{x}}^{(0)}}
  \leq \DistS*{\Vec{x}}{-\Vec{x}}
  = 2
  = \Varepsilon( 0 )
.\end{gather}
%
Meanwhile, arbitrarily fixing
\(
  \Iter \in \Z_{+}
\),
suppose that for each
\(
  \Iter' \in [\Iter-1]
\),
the error is upper bounded by
\begin{gather}
\label{pf:thm:main:convergence:eqn:6}
  \DistS*{\Vec{x}}{\Vec{\hat{x}}^{(\Iter')}} \leq \Varepsilon( \Iter' )
.\end{gather}
%
Then, applying \EQN \eqref{pf:thm:main:convergence:eqn:2:main},
the \( \Iter\Th \) approximation satisfies
\begin{align}
  \nonumber
  \DistS*{\Vec{x}}{\Vec{\hat{x}}^{(\Iter)}}
  &\leq
  \ErrorRecurrenceDef{\Iter-1}
  \\* 
  &\leq
  \VarepsilonDef*
  = \Varepsilon( \Iter )
\label{pf:thm:main:convergence:eqn:7}
\end{align}
as desired.
By induction, it follows that the sequence of BIHT approximations for the unknown vector
\( \Vec{x} \) satisfies
\begin{gather}
\label{pf:thm:main:convergence:eqn:8}
  \DistS*{\Vec{x}}{\Vec{\hat{x}}^{(\Iter)}} \leq \Varepsilon( \Iter )
  ,\quad
  \forall \Iter \in \Z_{\geq 0}
.\end{gather}
%
Then, Lemmas \ref{lemma:biht:error:recurrence} and \ref{lemma:biht:error:explicit}
immediately imply the desired results since asymptotically (Lemma \ref{lemma:biht:error:recurrence}),
\begin{align}
  \nonumber
  \lim_{\Iter \to \infty} \DistS*{\Vec{x}}{\Vec{\hat{x}}^{(\Iter)}}
  &\leq
  \lim_{\Iter \to \infty} \Varepsilon( \Iter )
  \\* 
  &=
  \VarepsilonAsymptotic
  <
  \Epsilon
\label{pf:thm:main:convergence:eqn:9}
\end{align}
whereas pointwise (Lemma \ref{lemma:biht:error:explicit}),
\begin{gather}
\label{pf:thm:main:convergence:eqn:10}
  \DistS*{\Vec{x}}{\Vec{\hat{x}}^{(\Iter)}}
  \leq
  \Varepsilon( \Iter )
  \leq
  2^{2^{-\Iter}}
  \Epsilon^{1-2^{-\Iter}}
.\end{gather}
%
This completes the first step of the proof.
Next, the proof concludes by extending the argument to the uniform results claimed in the theorems.
\par
As briefly mentioned at the beginning of the proof, in the argument laid out above,
Lemma \ref{lemma:biht:error-upper-bound:alg} and
Theorem \ref{thm:technical:raic:modified} hold uniformly for every
\(
  \Vec{x} \in \SparseSphereSubspace{k}{n}
\),
where Lemma \ref{lemma:biht:error-upper-bound:alg} is deterministic while
Theorem \ref{thm:technical:raic:modified} ensures the bound with probability at least
\(
  1 - \Rho
\).
Thus, for every
\(
  \Vec{x} \in \SparseSphereSubspace{k}{n}
\),
the \( \Iter\Th \) BIHT approximation has error upper bounded by
\begin{align}
\label{pf:thm:main:convergence:eqn:11}
  \DistS*{\Vec{x}}{\Vec{\hat{x}}^{(\Iter)}}
  \leq
  \ErrorRecurrenceDef{\Iter-1}
\end{align}
uniformly with probability at least
\(
  1 - \Rho
\).
Furthermore, because Lemmas \ref{lemma:biht:error:recurrence} and
\ref{lemma:biht:error:explicit} are deterministic, the rate of decay and asymptotic behavior
stated in the theorems also hold uniformly---specifically, for all
\(
  \Vec{x} \in \SparseSphereSubspace{k}{n}
\),
\begin{align}
  \nonumber
  \lim_{\Iter \to \infty} \ & \DistS*{\Vec{x}}{\Vec{\hat{x}}^{(\Iter)}}
  \leq
  \lim_{\Iter \to \infty} \Varepsilon( \Iter )
  \\*
  &=
  \VarepsilonAsymptotic
  <
  \Epsilon
\label{pf:thm:main:convergence:eqn:12}
\end{align}
and
\begin{align}
  \DistS*{\Vec{x}}{\Vec{\hat{x}}^{(\Iter)}}
  \leq
  \Varepsilon( \Iter )
  \leq
  2^{2^{-\Iter}}
  \Epsilon^{1-2^{-\Iter}}
  ,\quad \forall \Iter \in \Z_{\geq 0}
\label{pf:thm:main:convergence:eqn:12:b}
\end{align}
with probability at least
\(
  1 - \Rho
\).
\end{proof}

\remove{
\subsection{Proof of the Intermediate Lemmas
(Lemmas \ref{lemma:biht:error-upper-bound:alg}-\ref{lemma:biht:error:explicit})}
\label{outline:biht:pf-main-thm|>intermediate-lemmas-pf}

\subsubsection{Proof of Lemma \ref{lemma:biht:error-upper-bound:alg}}
\label{outline:biht:pf-main-thm|>intermediate-lemmas-pf|>error-alg}

\begin{proof}
{Lemma \ref{lemma:biht:error-upper-bound:alg}}
\label{pf:lemma:biht:error-upper-bound:alg}
Let
\(
  \Vec{x} \in \SparseSphereSubspace{k}{n}
\)
be an arbitrary unknown, \( k \)-sparse vector of unit norm,
and consider any \( \Iter\Th \) BIHT approximation,
\(
  \Vec{\hat{x}}^{(\Iter)} \in \SparseSphereSubspace{k}{n}
\),
\(
  \Iter \in \Z_{+}
\).
Recall that the BIHT algorithm computes its \( \Iter\Th \) approximation by
\begin{gather}
\label{pf:lemma:biht:error-upper-bound:alg:eqn:1}
  \Vec{\tilde{x}}^{(\Iter)}
  =
  \Vec{\hat{x}}^{(\Iter-1)}
  +
  \frac{\Eta}{m}
  \MeasMat^{\T}
  \cdot
  \frac{1}{2}
  \left( \Sgn( \MeasMat \Vec{x} ) - \Sgn( \MeasMat \Vec{\hat{x}}^{(\Iter-1)} ) \right)
  \\
  \Vec{\hat{x}}^{(\Iter)}
  =
  \frac
  {\Threshold{k}( \Vec{\tilde{x}}^{(\Iter)} )}
  {\left\| \Threshold{k}( \Vec{\tilde{x}}^{(\Iter)} ) \right\|_{2}}
\end{gather}
and notice that
\begin{gather}
\label{pf:lemma:biht:error-upper-bound:alg:eqn:2}
  \Vec{\tilde{x}}^{(\Iter)}
  =
  \Vec{\hat{x}}^{(\Iter-1)}
  +
  \hA( \Vec{x}, \Vec{\hat{x}}^{(\Iter-1)} )
  \\
  \ThresholdSet{\supp( \Vec{x} ) \cup \supp( \Vec{\hat{x}}^{(\Iter-1)} )
                \cup \supp( \Vec{\hat{x}}^{(\Iter)} )}
  ( \Vec{\tilde{x}}^{(\Iter)} )
 \nonumber \\ =
  \Vec{\hat{x}}^{(\Iter-1)}
  +
  \hA[\supp( \Vec{\hat{x}}^{(\Iter)} )]( \Vec{x}, \Vec{\hat{x}}^{(\Iter-1)} )
.\end{gather}
%
Applying the triangle inequality, the error of the \( \Iter\Th \) BIHT approximation,
\( \Vec{\hat{x}}^{(\Iter)} \), can be bounded from above.
\begin{subequations}
\begin{align}
\label{pf:lemma:biht:error-upper-bound:alg:eqn:3:1}
  &
  \DistS{\Vec{x}}{\Vec{\hat{x}}^{(\Iter)}}
  \\
  &=
  \left\| \Vec{x} - \Vec{\hat{x}}^{(\Iter)} \right\|_{2}
  \\
  &=
  \left\|
    \left(
      \Vec{x}
      -
      \ThresholdSet{\supp( \Vec{x} ) \cup \supp( \Vec{\hat{x}}^{(\Iter)} )}
      ( \Vec{\tilde{x}}^{(\Iter)} )
    \right) \nonumber\\
    &+
    \left(
      \ThresholdSet{\supp( \Vec{x} ) \cup \supp( \Vec{\hat{x}}^{(\Iter)} )}
      ( \Vec{\tilde{x}}^{(\Iter)} )
      -
      \ThresholdSet{\supp( \Vec{\hat{x}}^{(\Iter)} )}
      ( \Vec{\tilde{x}}^{(\Iter)} )
    \right) \nonumber\\
    &+
    \left(
      \ThresholdSet{\supp( \Vec{\hat{x}}^{(\Iter)} )}
      ( \Vec{\tilde{x}}^{(\Iter)} )
      -
      \Vec{\hat{x}}^{(\Iter)}
    \right)
  \right\|_{2}
  \\ \nonumber
  &\leq
  \left\|
    \Vec{x}
    -
    \ThresholdSet{\supp( \Vec{x} ) \cup \supp( \Vec{\hat{x}}^{(\Iter)} )}
    ( \Vec{\tilde{x}}^{(\Iter)} )
  \right\|_{2} \\
  &+
  \left\|
    \ThresholdSet{\supp( \Vec{x} ) \cup \supp( \Vec{\hat{x}}^{(\Iter)} )}
    ( \Vec{\tilde{x}}^{(\Iter)} )
    -
    \ThresholdSet{\supp( \Vec{\hat{x}}^{(\Iter)} )}
    ( \Vec{\tilde{x}}^{(\Iter)} )
  \right\|_{2} \nonumber\\
  &+
  \left\|
    \ThresholdSet{\supp( \Vec{\hat{x}}^{(\Iter)} )}
    ( \Vec{\tilde{x}}^{(\Iter)} )
    -
    \Vec{\hat{x}}^{(\Iter)}
  \right\|_{2}
  \\
  &\dCmt \Text{by the triangle inequality}
  \\ \label{pf:lemma:biht:error-upper-bound:alg:eqn:3:2}
  &=
  \left\|
    \Vec{x}
    -
    \ThresholdSet{\supp( \Vec{x} ) \cup \supp( \Vec{\hat{x}}^{(\Iter)} )}
    ( \Vec{\tilde{x}}^{(\Iter)} )
  \right\|_{2} \nonumber\\
  &+
  \left\|
    \ThresholdSet{\supp( \Vec{x} ) \cup \supp( \Vec{\hat{x}}^{(\Iter)} )}
    ( \Vec{\tilde{x}}^{(\Iter)} )
    -
    \ThresholdSet{\supp( \Vec{\hat{x}}^{(\Iter)} )}
    ( \Vec{\tilde{x}}^{(\Iter)} )
  \right\|_{2}
  \\ \nonumber
  &\Tab[4]
  +
  \left\|
    \ThresholdSet{\supp( \Vec{\hat{x}}^{(\Iter)} )}( \Vec{\tilde{x}}^{(\Iter)} )
    -
    \frac
    {\ThresholdSet{\supp( \Vec{\hat{x}}^{(\Iter)} )}( \Vec{\tilde{x}}^{(\Iter)} )}
    {\left\| \ThresholdSet{\supp(\Vec{\hat{x}}^{(\Iter)})}(\Vec{\tilde{x}}^{(\Iter)}) \right\|_{2}}
  \right\|_{2}
\end{align}
\end{subequations}
%
The rightmost term in the last line can be upper bounded as follows.
\begin{subequations}
\begin{align}
\label{pf:lemma:biht:error-upper-bound:alg:eqn:4:1}
  &
  \left\|
    \ThresholdSet{\supp( \Vec{\hat{x}}^{(\Iter)} )}( \Vec{\tilde{x}}^{(\Iter)} )
    -
    \frac
    {\ThresholdSet{\supp( \Vec{\hat{x}}^{(\Iter)} )}( \Vec{\tilde{x}}^{(\Iter)} )}
    {\left\| \ThresholdSet{\supp(\Vec{\hat{x}}^{(\Iter)})}(\Vec{\tilde{x}}^{(\Iter)}) \right\|_{2}}
  \right\|_{2}
  \\
  &=
  \left|
    \left\|
      \ThresholdSet{\supp( \Vec{\hat{x}}^{(\Iter)} )}( \Vec{\tilde{x}}^{(\Iter)} )
    \right\|_{2}
    -
    1
  \right|
  \left\|
    \frac
    {\ThresholdSet{\supp( \Vec{\hat{x}}^{(\Iter)} )}( \Vec{\tilde{x}}^{(\Iter)} )}
    {\left\| \ThresholdSet{\supp(\Vec{\hat{x}}^{(\Iter)})}(\Vec{\tilde{x}}^{(\Iter)}) \right\|_{2}}
  \right\|_{2}
  \\
  &=
  \left|
    \left\|
      \ThresholdSet{\supp( \Vec{\hat{x}}^{(\Iter)} )}( \Vec{\tilde{x}}^{(\Iter)} )
    \right\|_{2}
    -
    1
  \right|
  \\
  &=
  \left|
    \left\|
      \ThresholdSet{\supp( \Vec{\hat{x}}^{(\Iter)} )}( \Vec{\tilde{x}}^{(\Iter)} )
    \right\|_{2}
    -
    \left\| \Vec{x} \right\|_{2}
  \right|
  \\
  &\leq
  \left\|
    \ThresholdSet{\supp( \Vec{\hat{x}}^{(\Iter)} )}( \Vec{\tilde{x}}^{(\Iter)} )
    -
    \Vec{x}
  \right\|_{2}
  \dCmt \Text{by the triangle inequality}
  \\
  &=
  \left\|
    \left(
      \ThresholdSet{\supp( \Vec{\hat{x}}^{(\Iter)} )}( \Vec{\tilde{x}}^{(\Iter)} )
      -
      \ThresholdSet{\supp( \Vec{x} ) \cup \supp( \Vec{\hat{x}}^{(\Iter)} )}
      ( \Vec{\tilde{x}}^{(\Iter)} )
    \right)\nonumber\\
    &+
    \left(
      \ThresholdSet{\supp( \Vec{x} ) \cup \supp( \Vec{\hat{x}}^{(\Iter)} )}
      ( \Vec{\tilde{x}}^{(\Iter)} )
      -
      \Vec{x}
    \right)
  \right\|_{2}
  \\
  &\leq
  \left\|
    \ThresholdSet{\supp( \Vec{\hat{x}}^{(\Iter)} )}( \Vec{\tilde{x}}^{(\Iter)} )
    -
    \ThresholdSet{\supp( \Vec{x} ) \cup \supp( \Vec{\hat{x}}^{(\Iter)} )}
    ( \Vec{\tilde{x}}^{(\Iter)} )
  \right\|_{2} \nonumber\\
  &+
  \left\|
    \ThresholdSet{\supp( \Vec{x} ) \cup \supp( \Vec{\hat{x}}^{(\Iter)} )}
    ( \Vec{\tilde{x}}^{(\Iter)} )
    -
    \Vec{x}
  \right\|_{2}
  \\ \nonumber
  &\Tab
  \dCmt \Text{by the triangle inequality}
  \\
  &=
  \left\|
    \Vec{x}
    -
    \ThresholdSet{\supp( \Vec{x} ) \cup \supp( \Vec{\hat{x}}^{(\Iter)} )}
    ( \Vec{\tilde{x}}^{(\Iter)} )
  \right\|_{2}\nonumber\\
  &+
  \left\|
    \ThresholdSet{\supp( \Vec{x} ) \cup \supp( \Vec{\hat{x}}^{(\Iter)} )}
    ( \Vec{\tilde{x}}^{(\Iter)} )
    -
    \ThresholdSet{\supp( \Vec{\hat{x}}^{(\Iter)} )}( \Vec{\tilde{x}}^{(\Iter)} )
  \right\|_{2}
\label{pf:lemma:biht:error-upper-bound:alg:eqn:4:end}
\end{align}
\end{subequations}
%
Combining \eqref{pf:lemma:biht:error-upper-bound:alg:eqn:3:2} and
\eqref{pf:lemma:biht:error-upper-bound:alg:eqn:4:end} yields
\begin{gather}
\label{pf:lemma:biht:error-upper-bound:alg:eqn:5}
  \DistS{\Vec{x}}{\Vec{\hat{x}}^{(\Iter)}}
  =
  2
  \left\|
    \Vec{x}
    -
    \ThresholdSet{\supp( \Vec{x} ) \cup \supp( \Vec{\hat{x}}^{(\Iter)} )}
    ( \Vec{\tilde{x}}^{(\Iter)} )
  \right\|_{2} \nonumber\\
  +
  2
  \left\|
    \ThresholdSet{\supp( \Vec{x} ) \cup \supp( \Vec{\hat{x}}^{(\Iter)} )}
    ( \Vec{\tilde{x}}^{(\Iter)} )
    -
    \ThresholdSet{\supp( \Vec{\hat{x}}^{(\Iter)} )}
    ( \Vec{\tilde{x}}^{(\Iter)} )
  \right\|_{2}
.\end{gather}
%
Taking a closer look at the last term in \eqref{pf:lemma:biht:error-upper-bound:alg:eqn:5},
\begin{gather}
\label{pf:lemma:biht:error-upper-bound:alg:eqn:6}
  \left\|
    \ThresholdSet{\supp( \Vec{x} ) \cup \supp( \Vec{\hat{x}}^{(\Iter)} )}
    ( \Vec{\tilde{x}}^{(\Iter)} )
    -
    \ThresholdSet{\supp( \Vec{\hat{x}}^{(\Iter)} )}
    ( \Vec{\tilde{x}}^{(\Iter)} )
  \right\|_{2} \nonumber\\
  =
  \left\|
    \ThresholdSet{\supp( \Vec{x} ) \setminus \supp( \Vec{\hat{x}}^{(\Iter)} )}
    ( \Vec{\tilde{x}}^{(\Iter)} )
  \right\|_{2}
  \leq
  \left\|
    \ThresholdSet{\supp( \Vec{\hat{x}}^{(\Iter)} ) \setminus \supp( \Vec{x} )}
    ( \Vec{\tilde{x}}^{(\Iter)} )
  \right\|_{2}
\end{gather}
where the rightmost inequality follows from the definition of the thresholding operation
\( \Threshold{k} \), which ensures that for each
\(
  j \in \supp( \Vec{x} ) \setminus \supp( \Vec{\hat{x}}^{(\Iter)} )
\),
the \( j\Th \) entry of \( \Vec{\tilde{x}}^{(\Iter)} \) satisfies
\(
  | \Vec*{\tilde{x}}_{j}^{(\Iter)} |
  \leq
  \min_{j' \in \supp( \Vec{\hat{x}}^{(\Iter)} ) \setminus \supp( \Vec{x} )}
  | \Vec*{\tilde{x}}_{j'}^{(\Iter)} |
\).
Then, observe
\begin{subequations}
\begin{align}
\label{pf:lemma:biht:error-upper-bound:alg:eqn:7}
 & \left\|
    \Vec{x}
    -
    \ThresholdSet{\supp( \Vec{x} ) \cup \supp( \Vec{\hat{x}}^{(\Iter)} )}
    ( \Vec{\tilde{x}}^{(\Iter)} )
  \right\|_{2}^{2}
  =
  \sum_{j \in \supp( \Vec{x} ) \cup \supp( \Vec{\hat{x}}^{(\Iter)} )}
  \left( \Vec*{x}_{j} - \Vec*{\tilde{x}}_{j}^{(\Iter)} \right)^{2}
  \\
  &=
  \sum_{j \in \supp( \Vec{\hat{x}}^{(\Iter)} ) \setminus \supp( \Vec{x} )}
  \left( \Vec*{x}_{j} - \Vec*{\tilde{x}}_{j}^{(\Iter)} \right)^{2}
  +
  \sum_{j \in \supp( \Vec{x} )}
  \left( \Vec*{x}_{j} - \Vec*{\tilde{x}}_{j}^{(\Iter)} \right)^{2}
  \\
  &=
  \sum_{j \in \supp( \Vec{\hat{x}}^{(\Iter)} ) \setminus \supp( \Vec{x} )}
  \left( 0 - \Vec*{\tilde{x}}_{j}^{(\Iter)} \right)^{2}
  +
  \sum_{j \in \supp( \Vec{x} )}
  \left( \Vec*{x}_{j} - \Vec*{\tilde{x}}_{j}^{(\Iter)} \right)^{2}
  \\
  &=
  \sum_{j \in \supp( \Vec{\hat{x}}^{(\Iter)} ) \setminus \supp( \Vec{x} )}
  \left( \Vec*{\tilde{x}}_{j}^{(\Iter)} \right)^{2}
  +
  \sum_{j \in \supp( \Vec{x} )}
  \left( \Vec*{x}_{j} - \Vec*{\tilde{x}}_{j}^{(\Iter)} \right)^{2}
  \\
  &=
  \left\|
    \ThresholdSet{\supp( \Vec{\hat{x}}^{(\Iter)} ) \setminus \supp( \Vec{x} )}
    ( \Vec{\tilde{x}}^{(\Iter)} )
  \right\|_{2}^{2}
  +
  \left\|
    \Vec{x}
    -
    \ThresholdSet{\supp( \Vec{x} )}
    ( \Vec{\tilde{x}}^{(\Iter)} )
  \right\|_{2}^{2}
\end{align}
\end{subequations}
%
It follows that
\begin{subequations}
\begin{align}
\label{pf:lemma:biht:error-upper-bound:alg:eqn:8:1}
  &
  \left\|
    \ThresholdSet{\supp( \Vec{\hat{x}}^{(\Iter)} ) \setminus \supp( \Vec{x} )}
    ( \Vec{\tilde{x}}^{(\Iter)} )
  \right\|_{2}^{2}
  +
  \left\|
    \Vec{x}
    -
    \ThresholdSet{\supp( \Vec{x} )}
    ( \Vec{\tilde{x}}^{(\Iter)} )
  \right\|_{2}^{2} \nonumber\\
  &\qquad=
  \left\|
    \Vec{x}
    -
    \ThresholdSet{\supp( \Vec{x} ) \cup \supp( \Vec{\hat{x}}^{(\Iter)} )}
    ( \Vec{\tilde{x}}^{(\Iter)} )
  \right\|_{2}^{2}
  \\
  &\dLn
  \left\|
    \ThresholdSet{\supp( \Vec{\hat{x}}^{(\Iter)} ) \setminus \supp( \Vec{x} )}
    ( \Vec{\tilde{x}}^{(\Iter)} )
  \right\|_{2}^{2} \nonumber\\
 &\qquad =
  \left\|
    \Vec{x}
    -
    \ThresholdSet{\supp( \Vec{x} ) \cup \supp( \Vec{\hat{x}}^{(\Iter)} )}
    ( \Vec{\tilde{x}}^{(\Iter)} )
  \right\|_{2}^{2}
  -
  \left\|
    \Vec{x}
    -
    \ThresholdSet{\supp( \Vec{x} )}
    ( \Vec{\tilde{x}}^{(\Iter)} )
  \right\|_{2}^{2}
  \\
  &\dLn
  \left\|
    \ThresholdSet{\supp( \Vec{\hat{x}}^{(\Iter)} ) \setminus \supp( \Vec{x} )}
    ( \Vec{\tilde{x}}^{(\Iter)} )
  \right\|_{2}^{2}
  \leq
  \left\|
    \Vec{x}
    -
    \ThresholdSet{\supp( \Vec{x} ) \cup \supp( \Vec{\hat{x}}^{(\Iter)} )}
    ( \Vec{\tilde{x}}^{(\Iter)} )
  \right\|_{2}^{2}
  \\
  &\dLn
  \left\|
    \ThresholdSet{\supp( \Vec{\hat{x}}^{(\Iter)} ) \setminus \supp( \Vec{x} )}
    ( \Vec{\tilde{x}}^{(\Iter)} )
  \right\|_{2}
  \leq
  \left\|
    \Vec{x}
    -
    \ThresholdSet{\supp( \Vec{x} ) \cup \supp( \Vec{\hat{x}}^{(\Iter)} )}
    ( \Vec{\tilde{x}}^{(\Iter)} )
  \right\|_{2}
\label{pf:lemma:biht:error-upper-bound:alg:eqn:8:end}
\end{align}
\end{subequations}
%
Likewise,
\begin{subequations}
\begin{align}
\label{pf:lemma:biht:error-upper-bound:alg:eqn:10:1}
  &
  \left\|
    \Vec{x}
    -
    \ThresholdSet{\supp( \Vec{x} ) \cup \supp( \Vec{\hat{x}}^{(\Iter-1)} )
                  \cup \supp( \Vec{\hat{x}}^{(\Iter)} )}
    ( \Vec{\tilde{x}}^{(\Iter)} )
  \right\|_{2}^{2}
  \\
  &=
  \sum_{j \in \supp( \Vec{x} ) \cup \supp( \Vec{\hat{x}}^{(\Iter-1)} )
              \cup \supp( \Vec{\hat{x}}^{(\Iter)} )}
  \left( \Vec*{x}_{j} - \Vec*{\tilde{x}}_{j}^{(\Iter)} \right)^{2}
  \\
  &=
  \sum_{j \in \supp( \Vec{x} ) \cup \supp( \Vec{\hat{x}}^{(\Iter)} )}
  \left( \Vec*{x}_{j} - \Vec*{\tilde{x}}_{j}^{(\Iter)} \right)^{2}
  +
  \sum_{j \in \supp( \Vec{\hat{x}}^{(\Iter-1)} )
              \setminus ( \supp( \Vec{x} ) \cup \supp( \Vec{\hat{x}}^{(\Iter)} ) )}
  \left( \Vec*{x}_{j} - \Vec*{\tilde{x}}_{j}^{(\Iter)} \right)^{2}
  \\
  &=
  \left\|
    \ThresholdSet{\supp( \Vec{x} ) \cup \supp( \Vec{\hat{x}}^{(\Iter)} )}
    ( \Vec{x} - \Vec{\tilde{x}}^{(\Iter)} )
  \right\|_{2}^{2}
  +
  \left\|
    \ThresholdSet{\supp( \Vec{\hat{x}}^{(\Iter-1)} )
                  \setminus ( \supp( \Vec{x} ) \cup \supp( \Vec{\hat{x}}^{(\Iter)} ) )}
    ( \Vec{x} - \Vec{\tilde{x}}^{(\Iter)} )
  \right\|_{2}^{2}
  \\
  &\geq
  \left\|
    \ThresholdSet{\supp( \Vec{x} ) \cup \supp( \Vec{\hat{x}}^{(\Iter)} )}
    ( \Vec{x} - \Vec{\tilde{x}}^{(\Iter)} )
  \right\|_{2}^{2}
  \\
  &=
  \left\|
    \Vec{x}
    -
    \ThresholdSet{\supp( \Vec{x} ) \cup \supp( \Vec{\hat{x}}^{(\Iter)} )}
    ( \Vec{\tilde{x}}^{(\Iter)} )
  \right\|_{2}^{2}
  \\
  &\dLn
  \left\|
    \Vec{x}
    -
    \ThresholdSet{\supp( \Vec{x} ) \cup \supp( \Vec{\hat{x}}^{(\Iter)} )}
    ( \Vec{\tilde{x}}^{(\Iter)} )
  \right\|_{2}
  \leq
  \left\|
    \Vec{x}
    -
    \ThresholdSet{\supp( \Vec{x} ) \cup \supp( \Vec{\hat{x}}^{(\Iter-1)} )
                  \cup \supp( \Vec{\hat{x}}^{(\Iter)} )}
    ( \Vec{\tilde{x}}^{(\Iter)} )
  \right\|_{2}
\label{pf:lemma:biht:error-upper-bound:alg:eqn:10:end}
\end{align}
\end{subequations}
%
Continuing from \eqref{pf:lemma:biht:error-upper-bound:alg:eqn:5},
\begin{subequations}
\begin{align}
\label{pf:lemma:biht:error-upper-bound:alg:eqn:9}
  &
  \DistS{\Vec{x}}{\Vec{\hat{x}}^{(\Iter)}}
  \\
  &=
  2
  \left\|
    \Vec{x}
    -
    \ThresholdSet{\supp( \Vec{x} ) \cup \supp( \Vec{\hat{x}}^{(\Iter)} )}
    ( \Vec{\tilde{x}}^{(\Iter)} )
  \right\|_{2}
  +
  2
  \left\|
    \ThresholdSet{\supp( \Vec{x} ) \cup \supp( \Vec{\hat{x}}^{(\Iter)} )}
    ( \Vec{\tilde{x}}^{(\Iter)} )
    -
    \ThresholdSet{\supp( \Vec{\hat{x}}^{(\Iter)} )}
    ( \Vec{\tilde{x}}^{(\Iter)} )
  \right\|_{2}
  \\
  &=
  2
  \left\|
    \Vec{x}
    -
    \ThresholdSet{\supp( \Vec{x} ) \cup \supp( \Vec{\hat{x}}^{(\Iter)} )}
    ( \Vec{\tilde{x}}^{(\Iter)} )
  \right\|_{2}
  +
  2
  \left\|
    \ThresholdSet{\supp( \Vec{\hat{x}}^{(\Iter)} ) \setminus \supp( \Vec{x} )}
    ( \Vec{\tilde{x}}^{(\Iter)} )
  \right\|_{2}
  \dCmt \Text{by \EQN \eqref{pf:lemma:biht:error-upper-bound:alg:eqn:6}}
  \\
  &\leq
  4
  \left\|
    \Vec{x}
    -
    \ThresholdSet{\supp( \Vec{x} ) \cup \supp( \Vec{\hat{x}}^{(\Iter)} )}
    ( \Vec{\tilde{x}}^{(\Iter)} )
  \right\|_{2}
  \dCmt \Text{by \EQN \eqref{pf:lemma:biht:error-upper-bound:alg:eqn:8:end}}
  \\
  &\leq
  4
  \left\|
    \Vec{x}
    -
    \ThresholdSet{\supp( \Vec{x} ) \cup \supp( \Vec{\hat{x}}^{(\Iter-1)} )
                  \cup \supp( \Vec{\hat{x}}^{(\Iter)} )}
    ( \Vec{\tilde{x}}^{(\Iter)} )
  \right\|_{2}
  \dCmt \Text{by \EQN \eqref{pf:lemma:biht:error-upper-bound:alg:eqn:10:end}}
  \\
  &=
  4
  \left\|
    \Vec{x}
    -
    \ThresholdSet{\supp( \Vec{x} ) \cup \supp( \Vec{\hat{x}}^{(\Iter-1)} )
                  \cup \supp( \Vec{\hat{x}}^{(\Iter)} )}
    ( \Vec{\hat{x}}^{(\Iter-1)} + \hA( \Vec{x}, \Vec{\hat{x}}^{(\Iter-1)} ) )
  \right\|_{2}
  \\
  &=
  4
  \left\|
    \Vec{x}
    -
    \Vec{\hat{x}}^{(\Iter-1)}
    -
    \ThresholdSet{\supp( \Vec{x} ) \cup \supp( \Vec{\hat{x}}^{(\Iter-1)} )
                  \cup \supp( \Vec{\hat{x}}^{(\Iter)} )}
    ( \hA( \Vec{x}, \Vec{\hat{x}}^{(\Iter-1)} ) )
  \right\|_{2}
  \\
  &=
  4
  \left\|
    \left( \Vec{x} - \Vec{\hat{x}}^{(\Iter-1)} \right)
    -
    \hA[\supp( \Vec{\hat{x}}^{(\Iter)} )]( \Vec{x}, \Vec{\hat{x}}^{(\Iter-1)} )
  \right\|_{2}
\end{align}
\end{subequations}
as desired.
\end{proof}

\subsubsection{Proof of Lemmas \ref{lemma:biht:error:recurrence} and
\ref{lemma:biht:error:explicit}}
\label{outline:biht:pf-main-thm|>intermediate-lemmas-pf|>error}

Lemmas \ref{lemma:biht:error:recurrence} and \ref{lemma:biht:error:explicit}, will
be verified in tandem.
Fact \ref{fact:misc:error-decay-recurrence}, stated below and proved in
Section \ref{outline:misc:error-decay-recurrence}, will facilitate the proof.
%
\begin{fact}
\label{fact:misc:error-decay-recurrence}
Let
\(
  \Variable{u}, \Variable{v}, \Variable{w}, \Variable{w}_{0} \in \R_{+}
\)
such that
\(
  \Variable{u} = \frac{1}{2} \left( 1 + \sqrt{1 + 4 \Variable{w}} \right)
\),
and
\(
  1 \leq \Variable{u} \leq \frac{2}{\sqrt{\Variable{v}}}
\).
Define the functions
\(
  \Function{f}_{1}, \Function{f}_{2} : \Z_{\geq 0} \to \R
\)
by
\begin{gather}
\label{eqn:fact:misc:error-decay-recurrence:g-h-def}
  \Function{f}_{1}(0) = 2
  \\
  \Function{f}_{1}(t) = \Variable{v} \Variable{w} + \sqrt{\Variable{v} g(t-1)}
  ,\quad t \in \Z_{+}
  \\
  \Function{f}_{2}(t) = 2^{2^{-t}} (\Variable{u}^{2} \Variable{v})^{1 - 2^{-t}}
  ,\quad t \in \Z_{\geq 0}
.\end{gather}
%
Then, \( \Function{f}_{1} \) and \( \Function{f}_{2} \) are strictly monotonically decreasing and
asymptotically converges to
\(
  \Variable{u}^{2} \Variable{v}
\).
Moreover, \( \Function{f}_{2} \) pointwise upper bounds \( \Function{f}_{1} \).
Formally,
\begin{gather}
\label{eqn:fact:misc:error-decay-recurrence:g<=h}
  \Function{f}_{1}(t) \leq \Function{f}_{2}(t)
  ,\quad \forall t \in \Z_{\geq 0}
\\
\label{eqn:fact:misc:error-decay-recurrence:t->infty}
  \lim_{t \to \infty} \Function{f}_{2}(t)
  =
  \lim_{t \to \infty} \Function{f}_{1}(t)
  =
  \Variable{u}^{2} \Variable{v}
.\end{gather}
\end{fact}
%
s
\begin{lemma*}
[Lemma \ref{lemma:biht:error:recurrence}]
Let
\(
  \Varepsilon : \Z_{\geq 0} \to \R
\)
be a function given by the recurrence relation
\begin{gather*}
  \Varepsilon( 0 ) = 2
  \\
  \Varepsilon( \Iter )
  =
  \VarepsilonDef
\end{gather*}
%
The function \( \varepsilon \) decreases monotonically with \( \Iter \) and asymptotically tends
to a value not exceeding \( \Epsilon \), formally,
\begin{gather*}
  \lim_{\Iter \to \infty} \Varepsilon( \Iter )
  =
  \VarepsilonAsymptotic
  <
  \Epsilon
\end{gather*}
\end{lemma*}
\begin{lemma*}
[Lemma \ref{lemma:biht:error:explicit}]
Let
\(
  \varepsilon : \Z_{\geq 0} \to \R
\)
be the function as defined in Lemma \ref{lemma:biht:error:recurrence}.
Then, the sequence
\(
  \{ \varepsilon( \Iter ) \}_{\Iter \in \Z_{\geq 0}}
\)
is bound from above by the sequence
\(
  \{ 2^{2^{-\Iter}} \Epsilon^{1-2^{-\Iter}} \}_{\Iter \in \Z_{\geq 0}}
\).
\end{lemma*}
%
\begin{proof}
{Lemmas \ref{lemma:biht:error:recurrence} and \ref{lemma:biht:error:explicit}}
\label{pf:lemma:biht:error:recurrence-seq}
The lemmas are corollaries to Fact \ref{fact:misc:error-decay-recurrence}.
All that is necessary is writing \( \Varepsilon \) in the form of \( \Function{f}_{1} \) in
Fact \ref{fact:misc:error-decay-recurrence} and verifying that it satisfies the conditions of the
fact.
For
\(
  \Iter = 0
\),
\(
  \Varepsilon(0) = 2 = \Function{f}_{1}(0)
\).
Otherwise, for
\(
  \Iter > 0
\),
observe
\begin{subequations}
\begin{align}
  \Varepsilon( \Iter )
  &=
  4 \UnivConstc_{1}
  \sqrt{\frac{\Epsilon}{\UnivConstC} \Varepsilon( \Iter-1 )}
  +
  4 \UnivConstc_{2} \frac{\Epsilon}{\UnivConstC}
  =
  \left( \frac{16 \UnivConstc_{1}^{2} \Epsilon}{\UnivConstC} \right)
  \left( \frac{16 \UnivConstc_{1}^{2} \Epsilon}{\UnivConstC} \right)^{-1}
  4 \UnivConstc_{2} \frac{\Epsilon}{\UnivConstC}
  +
  \sqrt{\left( \frac{16 \UnivConstc_{1}^{2} \Epsilon}{\UnivConstC} \right) \Varepsilon( \Iter-1 )}
  \\
  &=
  \left( \frac{16 \UnivConstc_{1}^{2} \Epsilon}{\UnivConstC} \right)
  \left( \frac{\UnivConstc_{2}}{4 \UnivConstc_{1}^{2}} \right)
  +
  \sqrt{\left( \frac{16 \UnivConstc_{1}^{2} \Epsilon}{\UnivConstC} \right) \Varepsilon( \Iter-1 )}
  \\
  &=
  \Variable{v} \Variable{w} + \sqrt{\Variable{v} \Varepsilon( \Iter-1 )}
\end{align}
\end{subequations}
where
\(
  \Variable{v} = \frac{16 \UnivConstc_{1}^{2} \Epsilon}{\UnivConstC}
\),
\(
  \Variable{w} = \frac{\UnivConstc_{2}}{4 \UnivConstc_{1}^{2}}
\),
and
\(
  \Variable{u}
  = \frac{1}{2} ( 1 + \sqrt{1 + 4 \cdot \frac{\UnivConstc_{2}}{4 \UnivConstc_{1}^{2}}} )
  = \frac{1}{2} ( 1 + \sqrt{1 + \frac{\UnivConstc_{2}}{\UnivConstc_{1}^{2}}} )
  = \frac{1}{2 \UnivConstc_{1}} ( \UnivConstc_{1} + \sqrt{\UnivConstc_{1}^{2} + \UnivConstc_{2}} )
\).
Recall that the universal constants are fixed as
\(
  \UnivConstc_{1} = \UnivConstcOneValue,
  \UnivConstc_{2} = \UnivConstcTwoValue,
  \UnivConstC = \UnivConstCValue
\).
By numerical calculations, it can be shown that
\(
  \Variable{u} \sqrt{\Variable{v}}
  < \sqrt{2}
\)
whenever
\(
  \UnivConstB \gtrsim \UnivConstBValue
\),
and hence
\(
  \Variable{u}
  < \sqrt{\frac{2}{\Variable{v}}}
\),
as required by Fact \ref{fact:misc:error-decay-recurrence}.
It then follows that \( \Varepsilon \) monotonically decreases with
\(
  \Iter \in \Z_{\geq 0}
\)
and
\begin{gather}
  \lim_{\Iter \to \infty} \Varepsilon(\Iter)
  = \Variable{u}^{2} \Variable{v}
  = \VarepsilonAsymptotic
  < \VarepsilonAsymptoticIntermediate
  = \Epsilon
,\end{gather}
where the last inequality follows from a numerical calculation.
Moreover, Fact \ref{fact:misc:error-decay-recurrence} further implies
\begin{gather}
  \Varepsilon(\Iter)
  \leq
  2^{2^{-\Iter}}
  (\Variable{u}^{2} \Variable{v})^{1-2^{-\Iter}}
  <
  2^{2^{-\Iter}}
  \Epsilon^{1-2^{-\Iter}}
.\end{gather}
\end{proof}

}


\subsection{Proof of the Intermediate Lemmas
(Lemmas \ref{lemma:biht:error-upper-bound:alg}-\ref{lemma:biht:error:explicit})}
\label{outline:biht:pf-main-thm|>intermediate-lemmas-pf}

\subsubsection{Proof of Lemma \ref{lemma:biht:error-upper-bound:alg}}
\label{outline:biht:pf-main-thm|>intermediate-lemmas-pf|>error-alg}

\begin{proof}
{Lemma \ref{lemma:biht:error-upper-bound:alg}}
\label{pf:lemma:biht:error-upper-bound:alg}
Let
\(
  \Vec{x} \in \SparseSphereSubspace{k}{n}
\)
be an arbitrary unknown, \( k \)-spare vector of unit norm,
and consider any \( \Iter\Th \) BIHT approximation,
\(
  \Vec{\hat{x}}^{(\Iter)} \in \SparseSphereSubspace{k}{n}
\),
\(
  \Iter \in \Z_{+}
\).
Recall that the BIHT algorithm computes its \( \Iter\Th \) approximation by
\begin{gather}
\label{pf:lemma:biht:error-upper-bound:alg:eqn:1}
  \Vec{\tilde{x}}^{(\Iter)}
  =
  \Vec{\hat{x}}^{(\Iter-1)}
  +
  \frac{\Eta}{m}
  \MeasMat^{\T}
  \cdot
  \frac{1}{2}
  \left( \Sgn( \MeasMat \Vec{x} ) - \Sgn( \MeasMat \Vec{\hat{x}}^{(\Iter-1)} ) \right)
  \\
  \Vec{\hat{x}}^{(\Iter)}
  =
  \frac
  {\Threshold{k}( \Vec{\tilde{x}}^{(\Iter)} )}
  {\left\| \Threshold{k}( \Vec{\tilde{x}}^{(\Iter)} ) \right\|_{2}}
\end{gather}
and notice that
\begin{gather}
\label{pf:lemma:biht:error-upper-bound:alg:eqn:2}
  \Vec{\tilde{x}}^{(\Iter)}
  =
  \Vec{\hat{x}}^{(\Iter-1)}
  +
  \hA( \Vec{x}, \Vec{\hat{x}}^{(\Iter-1)} )
  \\
  \ThresholdSet{\supp( \Vec{x} ) \cup \supp( \Vec{\hat{x}}^{(\Iter-1)} )
                \cup \supp( \Vec{\hat{x}}^{(\Iter)} )}
  ( \Vec{\tilde{x}}^{(\Iter)} )
  =
  \Vec{\hat{x}}^{(\Iter-1)}
  +
  \hA[\supp( \Vec{\hat{x}}^{(\Iter)} )]( \Vec{x}, \Vec{\hat{x}}^{(\Iter-1)} )
.\end{gather}
%
Applying the triangle inequality, the error of the \( \Iter\Th \) BIHT approximation,
\( \Vec{\hat{x}}^{(\Iter)} \), can be bounded from above.
\begin{subequations}
\begin{align}
\label{pf:lemma:biht:error-upper-bound:alg:eqn:3:1}
  &
  \DistS{\Vec{x}}{\Vec{\hat{x}}^{(\Iter)}}
  \\
  &=
  \left\| \Vec{x} - \Vec{\hat{x}}^{(\Iter)} \right\|_{2}
  \\
  &=
  \left\|
    \left(
      \Vec{x}
      -
      \ThresholdSet{\supp( \Vec{x} ) \cup \supp( \Vec{\hat{x}}^{(\Iter)} )}
      ( \Vec{\tilde{x}}^{(\Iter)} )
    \right)
    +
    \left(
      \ThresholdSet{\supp( \Vec{x} ) \cup \supp( \Vec{\hat{x}}^{(\Iter)} )}
      ( \Vec{\tilde{x}}^{(\Iter)} )
      -
      \ThresholdSet{\supp( \Vec{\hat{x}}^{(\Iter)} )}
      ( \Vec{\tilde{x}}^{(\Iter)} )
    \right)
  \right.
   \\ & \qquad +
   \left.
    \left(
      \ThresholdSet{\supp( \Vec{\hat{x}}^{(\Iter)} )}
      ( \Vec{\tilde{x}}^{(\Iter)} )
      -
      \Vec{\hat{x}}^{(\Iter)}
    \right)
  \right\|_{2}\nonumber
  \\
  &\leq
  \left\|
    \Vec{x}
    -
    \ThresholdSet{\supp( \Vec{x} ) \cup \supp( \Vec{\hat{x}}^{(\Iter)} )}
    ( \Vec{\tilde{x}}^{(\Iter)} )
  \right\|_{2}
  +
  \left\|
    \ThresholdSet{\supp( \Vec{x} ) \cup \supp( \Vec{\hat{x}}^{(\Iter)} )}
    ( \Vec{\tilde{x}}^{(\Iter)} )
    -
    \ThresholdSet{\supp( \Vec{\hat{x}}^{(\Iter)} )}
    ( \Vec{\tilde{x}}^{(\Iter)} )
  \right\|_{2} \nonumber
 \\ & \qquad +
  \left\|
    \ThresholdSet{\supp( \Vec{\hat{x}}^{(\Iter)} )}
    ( \Vec{\tilde{x}}^{(\Iter)} )
    -
    \Vec{\hat{x}}^{(\Iter)}
  \right\|_{2}
  \\
  &\dCmt \Text{by the triangle inequality}
  \\ \label{pf:lemma:biht:error-upper-bound:alg:eqn:3:2}
  &=
  \left\|
    \Vec{x}
    -
    \ThresholdSet{\supp( \Vec{x} ) \cup \supp( \Vec{\hat{x}}^{(\Iter)} )}
    ( \Vec{\tilde{x}}^{(\Iter)} )
  \right\|_{2}
  +
  \left\|
    \ThresholdSet{\supp( \Vec{x} ) \cup \supp( \Vec{\hat{x}}^{(\Iter)} )}
    ( \Vec{\tilde{x}}^{(\Iter)} )
    -
    \ThresholdSet{\supp( \Vec{\hat{x}}^{(\Iter)} )}
    ( \Vec{\tilde{x}}^{(\Iter)} )
  \right\|_{2}
  \\ \nonumber
  &\Tab[4]
  +
  \left\|
    \ThresholdSet{\supp( \Vec{\hat{x}}^{(\Iter)} )}( \Vec{\tilde{x}}^{(\Iter)} )
    -
    \frac
    {\ThresholdSet{\supp( \Vec{\hat{x}}^{(\Iter)} )}( \Vec{\tilde{x}}^{(\Iter)} )}
    {\left\| \ThresholdSet{\supp(\Vec{\hat{x}}^{(\Iter)})}(\Vec{\tilde{x}}^{(\Iter)}) \right\|_{2}}
  \right\|_{2}
\end{align}
\end{subequations}
%
The rightmost term in the last line can be upper bounded as follows.
\begin{subequations}
\begin{align}
\label{pf:lemma:biht:error-upper-bound:alg:eqn:4:1}
  &
  \left\|
    \ThresholdSet{\supp( \Vec{\hat{x}}^{(\Iter)} )}( \Vec{\tilde{x}}^{(\Iter)} )
    -
    \frac
    {\ThresholdSet{\supp( \Vec{\hat{x}}^{(\Iter)} )}( \Vec{\tilde{x}}^{(\Iter)} )}
    {\left\| \ThresholdSet{\supp(\Vec{\hat{x}}^{(\Iter)})}(\Vec{\tilde{x}}^{(\Iter)}) \right\|_{2}}
  \right\|_{2}
  \\
  &\leq
  \left\|
    \ThresholdSet{\supp( \Vec{\hat{x}}^{(\Iter)} )}( \Vec{\tilde{x}}^{(\Iter)} )
    -
    \Vec{x}
  \right\|_{2}
  \dCmt \Text{since} \arg\min_{u \in \Sphere{n}}\|v-u\|_2 = \frac{v}{\|v\|_2}
  \\
  &=
  \left\|
    \left(
      \ThresholdSet{\supp( \Vec{\hat{x}}^{(\Iter)} )}( \Vec{\tilde{x}}^{(\Iter)} )
      -
      \ThresholdSet{\supp( \Vec{x} ) \cup \supp( \Vec{\hat{x}}^{(\Iter)} )}
      ( \Vec{\tilde{x}}^{(\Iter)} )
    \right)
    +
    \left(
      \ThresholdSet{\supp( \Vec{x} ) \cup \supp( \Vec{\hat{x}}^{(\Iter)} )}
      ( \Vec{\tilde{x}}^{(\Iter)} )
      -
      \Vec{x}
    \right)
  \right\|_{2}
  \\
  &\leq
  \left\|
    \Vec{x}
    -
    \ThresholdSet{\supp( \Vec{x} ) \cup \supp( \Vec{\hat{x}}^{(\Iter)} )}
    ( \Vec{\tilde{x}}^{(\Iter)} )
  \right\|_{2}
  +
  \left\|
    \ThresholdSet{\supp( \Vec{x} ) \cup \supp( \Vec{\hat{x}}^{(\Iter)} )}
    ( \Vec{\tilde{x}}^{(\Iter)} )
    -
    \ThresholdSet{\supp( \Vec{\hat{x}}^{(\Iter)} )}( \Vec{\tilde{x}}^{(\Iter)} )
  \right\|_{2}
  \label{pf:lemma:biht:error-upper-bound:alg:eqn:4:end}
  \\ \nonumber
  &\Tab
  \dCmt \Text{by the triangle inequality.}
\end{align}
\end{subequations}
%
Combing \eqref{pf:lemma:biht:error-upper-bound:alg:eqn:3:2} and
\eqref{pf:lemma:biht:error-upper-bound:alg:eqn:4:end} yields
\begin{gather}
\label{pf:lemma:biht:error-upper-bound:alg:eqn:5}
  \DistS{\Vec{x}}{\Vec{\hat{x}}^{(\Iter)}}
  \leq
  2
  \left\|
    \Vec{x}
    -
    \ThresholdSet{\supp( \Vec{x} ) \cup \supp( \Vec{\hat{x}}^{(\Iter)} )}
    ( \Vec{\tilde{x}}^{(\Iter)} )
  \right\|_{2}
  +
  2
  \left\|
    \ThresholdSet{\supp( \Vec{x} ) \cup \supp( \Vec{\hat{x}}^{(\Iter)} )}
    ( \Vec{\tilde{x}}^{(\Iter)} )
    -
    \ThresholdSet{\supp( \Vec{\hat{x}}^{(\Iter)} )}
    ( \Vec{\tilde{x}}^{(\Iter)} )
  \right\|_{2}
.\end{gather}
%
\EDIT{Recall that \(  | \supp( \Vec{x} ) \cup \supp( \Vec{\hat{x}}^{(\Iter)} ) | \leq 2k  \).}
Taking a closer look at the last term in \eqref{pf:lemma:biht:error-upper-bound:alg:eqn:5},
\begin{gather}
\label{pf:lemma:biht:error-upper-bound:alg:eqn:6}
  \left\|
    \ThresholdSet{\supp( \Vec{x} ) \cup \supp( \Vec{\hat{x}}^{(\Iter)} )}
    ( \Vec{\tilde{x}}^{(\Iter)} )
    -
    \ThresholdSet{\supp( \Vec{\hat{x}}^{(\Iter)} )}
    ( \Vec{\tilde{x}}^{(\Iter)} )
  \right\|_{2}
  =
  \left\|
    \ThresholdSet{\supp( \Vec{x} ) \setminus \supp( \Vec{\hat{x}}^{(\Iter)} )}
    ( \Vec{\tilde{x}}^{(\Iter)} )
  \right\|_{2}
  \leq
  \left\|
    \ThresholdSet{\supp( \Vec{\hat{x}}^{(\Iter)} ) \setminus \supp( \Vec{x} )}
    ( \Vec{\tilde{x}}^{(\Iter)} )
  \right\|_{2}
\end{gather}
where the rightmost inequality follows from the definition of the thresholding operation
\( \Threshold{k} \), which ensures that for each
\(
  j \in \supp( \Vec{x} ) \setminus \supp( \Vec{\hat{x}}^{(\Iter)} )
\),
the \( j\Th \) entry of \( \Vec{\tilde{x}}^{(\Iter)} \) satisfies
\(
  | \Vec*{\tilde{x}}_{j}^{(\Iter)} |
  \leq
  \min_{j' \in \supp( \Vec{\hat{x}}^{(\Iter)} ) \setminus \supp( \Vec{x} )}
  | \Vec*{\tilde{x}}_{j'}^{(\Iter)} |
\).
Then, observe
\begin{subequations}
\begin{align}
\label{pf:lemma:biht:error-upper-bound:alg:eqn:7}
  \left\|
    \Vec{x}
    -
    \ThresholdSet{\supp( \Vec{x} ) \cup \supp( \Vec{\hat{x}}^{(\Iter)} )}
    ( \Vec{\tilde{x}}^{(\Iter)} )
  \right\|_{2}^{2}
  &=
  \sum_{j \in \supp( \Vec{x} ) \cup \supp( \Vec{\hat{x}}^{(\Iter)} )}
  \left( \Vec*{x}_{j} - \Vec*{\tilde{x}}_{j}^{(\Iter)} \right)^{2}
  \\
  &=
  \sum_{j \in \supp( \Vec{\hat{x}}^{(\Iter)} ) \setminus \supp( \Vec{x} )}
  \left( \Vec*{x}_{j} - \Vec*{\tilde{x}}_{j}^{(\Iter)} \right)^{2}
  +
  \sum_{j \in \supp( \Vec{x} )}
  \left( \Vec*{x}_{j} - \Vec*{\tilde{x}}_{j}^{(\Iter)} \right)^{2}
  \\
  &=
  \sum_{j \in \supp( \Vec{\hat{x}}^{(\Iter)} ) \setminus \supp( \Vec{x} )}
  \left( \Vec*{\tilde{x}}_{j}^{(\Iter)} \right)^{2}
  +
  \sum_{j \in \supp( \Vec{x} )}
  \left( \Vec*{x}_{j} - \Vec*{\tilde{x}}_{j}^{(\Iter)} \right)^{2}
  \\
  &=
  \left\|
    \ThresholdSet{\supp( \Vec{\hat{x}}^{(\Iter)} ) \setminus \supp( \Vec{x} )}
    ( \Vec{\tilde{x}}^{(\Iter)} )
  \right\|_{2}^{2}
  +
  \left\|
    \Vec{x}
    -
    \ThresholdSet{\supp( \Vec{x} )}
    ( \Vec{\tilde{x}}^{(\Iter)} )
  \right\|_{2}^{2}
\end{align}
\end{subequations}
%
It follows that
\begin{subequations}
\begin{align}
\label{pf:lemma:biht:error-upper-bound:alg:eqn:8:1}
  &
  \left\|
    \ThresholdSet{\supp( \Vec{\hat{x}}^{(\Iter)} ) \setminus \supp( \Vec{x} )}
    ( \Vec{\tilde{x}}^{(\Iter)} )
  \right\|_{2}^{2}
  +
  \left\|
    \Vec{x}
    -
    \ThresholdSet{\supp( \Vec{x} )}
    ( \Vec{\tilde{x}}^{(\Iter)} )
  \right\|_{2}^{2}
  =
  \left\|
    \Vec{x}
    -
    \ThresholdSet{\supp( \Vec{x} ) \cup \supp( \Vec{\hat{x}}^{(\Iter)} )}
    ( \Vec{\tilde{x}}^{(\Iter)} )
  \right\|_{2}^{2}
  \\
  &\dLn
  \left\|
    \ThresholdSet{\supp( \Vec{\hat{x}}^{(\Iter)} ) \setminus \supp( \Vec{x} )}
    ( \Vec{\tilde{x}}^{(\Iter)} )
  \right\|_{2}^{2}
  =
  \left\|
    \Vec{x}
    -
    \ThresholdSet{\supp( \Vec{x} ) \cup \supp( \Vec{\hat{x}}^{(\Iter)} )}
    ( \Vec{\tilde{x}}^{(\Iter)} )
  \right\|_{2}^{2}
  -
  \left\|
    \Vec{x}
    -
    \ThresholdSet{\supp( \Vec{x} )}
    ( \Vec{\tilde{x}}^{(\Iter)} )
  \right\|_{2}^{2}
  \\
  &\dLn
  \left\|
    \ThresholdSet{\supp( \Vec{\hat{x}}^{(\Iter)} ) \setminus \supp( \Vec{x} )}
    ( \Vec{\tilde{x}}^{(\Iter)} )
  \right\|_{2}^{2}
  \leq
  \left\|
    \Vec{x}
    -
    \ThresholdSet{\supp( \Vec{x} ) \cup \supp( \Vec{\hat{x}}^{(\Iter)} )}
    ( \Vec{\tilde{x}}^{(\Iter)} )
  \right\|_{2}^{2}
\label{pf:lemma:biht:error-upper-bound:alg:eqn:8:end}
\end{align}
\end{subequations}
%
Likewise,
\begin{subequations}
\begin{align}
\label{pf:lemma:biht:error-upper-bound:alg:eqn:10:1}
  &
  \left\|
    \Vec{x}
    -
    \ThresholdSet{\supp( \Vec{x} ) \cup \supp( \Vec{\hat{x}}^{(\Iter-1)} )
                  \cup \supp( \Vec{\hat{x}}^{(\Iter)} )}
    ( \Vec{\tilde{x}}^{(\Iter)} )
  \right\|_{2}^{2}
  \\
  &=
  \sum_{j \in \supp( \Vec{x} ) \cup \supp( \Vec{\hat{x}}^{(\Iter-1)} )
              \cup \supp( \Vec{\hat{x}}^{(\Iter)} )}
  \left( \Vec*{x}_{j} - \Vec*{\tilde{x}}_{j}^{(\Iter)} \right)^{2}
  \\
  &=
  \sum_{j \in \supp( \Vec{x} ) \cup \supp( \Vec{\hat{x}}^{(\Iter)} )}
  \left( \Vec*{x}_{j} - \Vec*{\tilde{x}}_{j}^{(\Iter)} \right)^{2}
  +
  \sum_{j \in \supp( \Vec{\hat{x}}^{(\Iter-1)} )
              \setminus ( \supp( \Vec{x} ) \cup \supp( \Vec{\hat{x}}^{(\Iter)} ) )}
  \left( \Vec*{x}_{j} - \Vec*{\tilde{x}}_{j}^{(\Iter)} \right)^{2}
  \\
  &=
  \left\|
    \ThresholdSet{\supp( \Vec{x} ) \cup \supp( \Vec{\hat{x}}^{(\Iter)} )}
    ( \Vec{x} - \Vec{\tilde{x}}^{(\Iter)} )
  \right\|_{2}^{2}
  +
  \left\|
    \ThresholdSet{\supp( \Vec{\hat{x}}^{(\Iter-1)} )
                  \setminus ( \supp( \Vec{x} ) \cup \supp( \Vec{\hat{x}}^{(\Iter)} ) )}
    ( \Vec{x} - \Vec{\tilde{x}}^{(\Iter)} )
  \right\|_{2}^{2}
  \\
  &\geq
  \left\|
    \ThresholdSet{\supp( \Vec{x} ) \cup \supp( \Vec{\hat{x}}^{(\Iter)} )}
    ( \Vec{x} - \Vec{\tilde{x}}^{(\Iter)} )
  \right\|_{2}^{2}
  \\
  &=
  \left\|
    \Vec{x}
    -
    \ThresholdSet{\supp( \Vec{x} ) \cup \supp( \Vec{\hat{x}}^{(\Iter)} )}
    ( \Vec{\tilde{x}}^{(\Iter)} )
  \right\|_{2}^{2}
  \\
  &\dLn
  \left\|
    \Vec{x}
    -
    \ThresholdSet{\supp( \Vec{x} ) \cup \supp( \Vec{\hat{x}}^{(\Iter)} )}
    ( \Vec{\tilde{x}}^{(\Iter)} )
  \right\|_{2}
  \leq
  \left\|
    \Vec{x}
    -
    \ThresholdSet{\supp( \Vec{x} ) \cup \supp( \Vec{\hat{x}}^{(\Iter-1)} )
                  \cup \supp( \Vec{\hat{x}}^{(\Iter)} )}
    ( \Vec{\tilde{x}}^{(\Iter)} )
  \right\|_{2}
\label{pf:lemma:biht:error-upper-bound:alg:eqn:10:end}
\end{align}
\end{subequations}
%
Continuing from \eqref{pf:lemma:biht:error-upper-bound:alg:eqn:5},
\begin{subequations}
\begin{align}
\label{pf:lemma:biht:error-upper-bound:alg:eqn:9}
  &
  \DistS{\Vec{x}}{\Vec{\hat{x}}^{(\Iter)}}
  \\
  &\leq
  2
  \left\|
    \Vec{x}
    -
    \ThresholdSet{\supp( \Vec{x} ) \cup \supp( \Vec{\hat{x}}^{(\Iter)} )}
    ( \Vec{\tilde{x}}^{(\Iter)} )
  \right\|_{2}
  +
  2
  \left\|
    \ThresholdSet{\supp( \Vec{x} ) \cup \supp( \Vec{\hat{x}}^{(\Iter)} )}
    ( \Vec{\tilde{x}}^{(\Iter)} )
    -
    \ThresholdSet{\supp( \Vec{\hat{x}}^{(\Iter)} )}
    ( \Vec{\tilde{x}}^{(\Iter)} )
  \right\|_{2}
  \\
  &=
  2
  \left\|
    \Vec{x}
    -
    \ThresholdSet{\supp( \Vec{x} ) \cup \supp( \Vec{\hat{x}}^{(\Iter)} )}
    ( \Vec{\tilde{x}}^{(\Iter)} )
  \right\|_{2}
  +
  2
  \left\|
    \ThresholdSet{\supp( \Vec{\hat{x}}^{(\Iter)} ) \setminus \supp( \Vec{x} )}
    ( \Vec{\tilde{x}}^{(\Iter)} )
  \right\|_{2}
  \dCmt \Text{by \EQN \eqref{pf:lemma:biht:error-upper-bound:alg:eqn:6}}
  \\
  &\leq
  4
  \left\|
    \Vec{x}
    -
    \ThresholdSet{\supp( \Vec{x} ) \cup \supp( \Vec{\hat{x}}^{(\Iter)} )}
    ( \Vec{\tilde{x}}^{(\Iter)} )
  \right\|_{2}
  \dCmt \Text{by \EQN \eqref{pf:lemma:biht:error-upper-bound:alg:eqn:8:end}}
  \\
  &\leq
  4
  \left\|
    \Vec{x}
    -
    \ThresholdSet{\supp( \Vec{x} ) \cup \supp( \Vec{\hat{x}}^{(\Iter-1)} )
                  \cup \supp( \Vec{\hat{x}}^{(\Iter)} )}
    ( \Vec{\tilde{x}}^{(\Iter)} )
  \right\|_{2}
  \dCmt \Text{by \EQN \eqref{pf:lemma:biht:error-upper-bound:alg:eqn:10:end}}
  \\
  &=
  4
  \left\|
    \Vec{x}
    -
    \ThresholdSet{\supp( \Vec{x} ) \cup \supp( \Vec{\hat{x}}^{(\Iter-1)} )
                  \cup \supp( \Vec{\hat{x}}^{(\Iter)} )}
    ( \Vec{\hat{x}}^{(\Iter-1)} + \hA( \Vec{x}, \Vec{\hat{x}}^{(\Iter-1)} ) )
  \right\|_{2}
  \\
  &=
  4
  \left\|
    \Vec{x}
    -
    \Vec{\hat{x}}^{(\Iter-1)}
    -
    \ThresholdSet{\supp( \Vec{x} ) \cup \supp( \Vec{\hat{x}}^{(\Iter-1)} )
                  \cup \supp( \Vec{\hat{x}}^{(\Iter)} )}
    ( \hA( \Vec{x}, \Vec{\hat{x}}^{(\Iter-1)} ) )
  \right\|_{2}
  \\
  &=
  4
  \left\|
    \left( \Vec{x} - \Vec{\hat{x}}^{(\Iter-1)} \right)
    -
    \hA[\supp( \Vec{\hat{x}}^{(\Iter)} )]( \Vec{x}, \Vec{\hat{x}}^{(\Iter-1)} )
  \right\|_{2}
\end{align}
\end{subequations}
as desired.
\end{proof}

\subsubsection{Proof of Lemmas \ref{lemma:biht:error:recurrence} and
\ref{lemma:biht:error:explicit}}
\label{outline:biht:pf-main-thm|>intermediate-lemmas-pf|>error}

Lemmas \ref{lemma:biht:error:recurrence} and \ref{lemma:biht:error:explicit}, will
be verified in tandem.
Fact \ref{fact:misc:error-decay-recurrence}, stated below and proved in
Section \ref{outline:misc:error-decay-recurrence}, will facilitate the proof.
%
\begin{fact}
\label{fact:misc:error-decay-recurrence}
Let
\(
  \Variable{u}, \Variable{v}, \Variable{w}, \Variable{w}_{0} \in \R_{+}
\)
such that
\(
  \Variable{u} = \frac{1}{2} \left( 1 + \sqrt{1 + 4 \Variable{w}} \right)
\),
and
\(
  1 \leq \Variable{u} \leq \frac{2}{\sqrt{\Variable{v}}}
\).
Define the functions
\(
  \Function{f}_{1}, \Function{f}_{2} : \Z_{\geq 0} \to \R
\)
by
\begin{gather}
\label{eqn:fact:misc:error-decay-recurrence:g-h-def}
  \Function{f}_{1}(0) = 2
  \\
  \Function{f}_{1}(t) = \Variable{v} \Variable{w} + \sqrt{\Variable{v} \Function{f}_{1}(t-1)}
  ,\quad t \in \Z_{+}
  \\
  \Function{f}_{2}(t) = 2^{2^{-t}} (\Variable{u}^{2} \Variable{v})^{1 - 2^{-t}}
  ,\quad t \in \Z_{\geq 0}
.\end{gather}
%
Then, \( \Function{f}_{1} \) and \( \Function{f}_{2} \) are strictly monotonically decreasing and
asymptotically converges to
\(
  \Variable{u}^{2} \Variable{v}
\).
Moreover, \( \Function{f}_{2} \) pointwise upper bounds \( \Function{f}_{1} \).
Formally,
\begin{gather}
\label{eqn:fact:misc:error-decay-recurrence:g<=h}
  \Function{f}_{1}(t) \leq \Function{f}_{2}(t)
  ,\quad \forall t \in \Z_{\geq 0}
\\
\label{eqn:fact:misc:error-decay-recurrence:t->infty}
  \lim_{t \to \infty} \Function{f}_{2}(t)
  =
  \lim_{t \to \infty} \Function{f}_{1}(t)
  =
  \Variable{u}^{2} \Variable{v}
.\end{gather}
\end{fact}
%
s
\begin{lemma*}
[Lemma \ref{lemma:biht:error:recurrence}]
Let
\(
  \Varepsilon : \Z_{\geq 0} \to \R
\)
be a function given by the recurrence relation
\begin{gather*}
  \Varepsilon( 0 ) = 2
  \\
  \Varepsilon( \Iter )
  =
  \VarepsilonDef
\end{gather*}
%
The function \( \varepsilon \) decreases monotonically with \( \Iter \) and asymptotically tends
to a value not exceeding \( \Epsilon \), formally,
\begin{gather*}
  \lim_{\Iter \to \infty} \Varepsilon( \Iter )
  =
  \VarepsilonAsymptotic
  <
  \Epsilon
\end{gather*}
\end{lemma*}
\begin{lemma*}
[Lemma \ref{lemma:biht:error:explicit}]
Let
\(
  \varepsilon : \Z_{\geq 0} \to \R
\)
be the function as defined in Lemma \ref{lemma:biht:error:recurrence}.
Then, the sequence
\(
  \{ \varepsilon( \Iter ) \}_{\Iter \in \Z_{\geq 0}}
\)
is bound from above by the sequence
\(
  \{ 2^{2^{-\Iter}} \Epsilon^{1-2^{-\Iter}} \}_{\Iter \in \Z_{\geq 0}}
\).
\end{lemma*}
%
\begin{proof}
{Lemmas \ref{lemma:biht:error:recurrence} and \ref{lemma:biht:error:explicit}}
\label{pf:lemma:biht:error:recurrence-seq}
The lemmas are corollaries to Fact \ref{fact:misc:error-decay-recurrence}.
All that is necessary is writing \( \Varepsilon \) in the form of \( \Function{f}_{1} \) in
Fact \ref{fact:misc:error-decay-recurrence} and verifying that it satisfies the conditions of the
fact.
For
\(
  \Iter = 0
\),
\(
  \Varepsilon(0) = 2 = \Function{f}_{1}(0)
\).
Otherwise, for
\(
  \Iter > 0
\),
observe
\begin{subequations}
\begin{align}
  \Varepsilon( \Iter )
  &=
  4 \UnivConstc_{1}
  \sqrt{\frac{\Epsilon}{\UnivConstC} \Varepsilon( \Iter-1 )}
  +
  4 \UnivConstc_{2} \frac{\Epsilon}{\UnivConstC}
  =
  \left( \frac{16 \UnivConstc_{1}^{2} \Epsilon}{\UnivConstC} \right)
  \left( \frac{16 \UnivConstc_{1}^{2} \Epsilon}{\UnivConstC} \right)^{-1}
  4 \UnivConstc_{2} \frac{\Epsilon}{\UnivConstC}
  +
  \sqrt{\left( \frac{16 \UnivConstc_{1}^{2} \Epsilon}{\UnivConstC} \right) \Varepsilon( \Iter-1 )}
  \\
  &=
  \left( \frac{16 \UnivConstc_{1}^{2} \Epsilon}{\UnivConstC} \right)
  \left( \frac{\UnivConstc_{2}}{4 \UnivConstc_{1}^{2}} \right)
  +
  \sqrt{\left( \frac{16 \UnivConstc_{1}^{2} \Epsilon}{\UnivConstC} \right) \Varepsilon( \Iter-1 )}
  \\
  &=
  \Variable{v} \Variable{w} + \sqrt{\Variable{v} \Varepsilon( \Iter-1 )}
\end{align}
\end{subequations}
where
\(
  \Variable{v} = \frac{16 \UnivConstc_{1}^{2} \Epsilon}{\UnivConstC}
\),
\(
  \Variable{w} = \frac{\UnivConstc_{2}}{4 \UnivConstc_{1}^{2}}
\),
and
\(
  \Variable{u}
  = \frac{1}{2} ( 1 + \sqrt{1 + 4 \cdot \frac{\UnivConstc_{2}}{4 \UnivConstc_{1}^{2}}} )
  = \frac{1}{2} ( 1 + \sqrt{1 + \frac{\UnivConstc_{2}}{\UnivConstc_{1}^{2}}} )
  = \frac{1}{2 \UnivConstc_{1}} ( \UnivConstc_{1} + \sqrt{\UnivConstc_{1}^{2} + \UnivConstc_{2}} )
\).
Recall that the universal constants are fixed as
\(
  \UnivConstc_{1} = \UnivConstcOneValue,
  \UnivConstc_{2} = \UnivConstcTwoValue,
  \UnivConstC = \UnivConstCValue
\).
By numerical calculations, it can be shown that
\(
  \Variable{u} \sqrt{\Variable{v}}
  < \sqrt{2}
\)
whenever
\(
  \UnivConstB \gtrsim \UnivConstBValue
\),
and hence
\(
  \Variable{u}
  < \sqrt{\frac{2}{\Variable{v}}}
\),
as required by Fact \ref{fact:misc:error-decay-recurrence}.
It then follows that \( \Varepsilon \) monotonically decreases with
\(
  \Iter \in \Z_{\geq 0}
\)
and
\begin{gather}
  \lim_{\Iter \to \infty} \Varepsilon(\Iter)
  = \Variable{u}^{2} \Variable{v}
  = \VarepsilonAsymptotic
  < \VarepsilonAsymptoticIntermediate
  = \Epsilon
,\end{gather}
where the last inequality follows from a numerical calculation.
Moreover, Fact \ref{fact:misc:error-decay-recurrence} further implies
\begin{gather}
  \Varepsilon(\Iter)
  \leq
  2^{2^{-\Iter}}
  (\Variable{u}^{2} \Variable{v})^{1-2^{-\Iter}}
  <
  2^{2^{-\Iter}}
  \Epsilon^{1-2^{-\Iter}}
.\end{gather}
\end{proof}

\section{Outlook}
\label{outline:outlook}
In this paper, we have shown that the binary iterative hard thresholding algorithm, an iterative (proximal) \subgradient descent algorithm for a nonconvex optimization problem, converges under certain structural assumptions, with the optimal number of measurements. It is worth exploring how general this result can be:
what other nonlinear measurements can be handled this way, and what type of measurement noise can be tolerated by such iterative algorithms?
This direction is hopeful because the noiseless sign measurements are often thought to be the hardest to analyze.
As another point of interest, our result is deterministic given a measurement matrix with a certain property. Incidentally, Gaussian measurements satisfy this property with high probability. However, the spherical symmetry of these measurements is crucial in the proof laid out in this work, and it is not clear whether other non-Gaussian (even sub-Gaussian) measurement matrices can have this property, or whether derandomized, explicit construction of measurement matrices is possible.

\paragraph{Acknowledgements.} We would like to thank the anonymous reviewers who helped correct some errors in the initially submitted version, as well as gave suggestions to improve readability significantly.

\bibliography{refs}

@inproceedings{foucart2017flavors,
  title={Flavors of compressive sensing},
  author={Foucart, Simon},
  booktitle={Approximation Theory XV: San Antonio 2016 15},
  pages={61--104},
  year={2017},
  organization={Springer},
  publisher={Springer}
}

@article{plan2012robust,
  title={Robust 1-bit compressed sensing and sparse logistic regression: A convex programming approach},
  author={Plan, Yaniv and Vershynin, Roman},
  journal={IEEE Transactions on Information Theory},
  volume={59},
  number={1},
  pages={482--494},
  year={2012},
  publisher={IEEE}
}

@article{liu2019one,
  title={One-bit compressive sensing with projected subgradient method under sparsity constraints},
  author={Liu, Dekai and Li, Song and Shen, Yi},
  journal={IEEE Transactions on Information Theory},
  volume={65},
  number={10},
  pages={6650--6663},
  year={2019},
  publisher={IEEE}
}

@inproceedings{mazumdar2021support,
  author    = {Arya Mazumdar and
               Soumyabrata Pal},
  editor    = {Mark Braverman},
  title     = {Support Recovery in Universal One-Bit Compressed Sensing},
  booktitle = {13th Innovations in Theoretical Computer Science Conference, {ITCS}
               2022, January 31 - February 3, 2022, Berkeley, CA, {USA}},
  series    = {LIPIcs},
  volume    = {215},
  pages     = {106:1--106:20},
  publisher = {Schloss Dagstuhl - Leibniz-Zentrum f{\"{u}}r Informatik},
  year      = {2022}
}

@article{friedlander2021nbiht,
  title={NBIHT: An Efficient Algorithm for 1-bit Compressed Sensing with Optimal Error Decay Rate},
  author={Friedlander, Michael P and Jeong, Halyun and Plan, Yaniv and Y{\i}lmaz, {\"O}zg{\"u}r},
  journal={IEEE Transactions on Information Theory},
  volume={68},
  number={2},
  pages={1157--1177},
  year={2021},
  publisher={IEEE}
}

@article{jacques2013quantized,
  title={Quantized iterative hard thresholding: Bridging 1-bit and high-resolution quantized compressed sensing},
  author={Jacques, Laurent and Degraux, K{\'e}vin and De Vleeschouwer, Christophe},
  journal={arXiv preprint arXiv:1305.1786},
  volume       = {abs/1305.1786},
  year         = {2013},
  numpages     = {8},
  url          = {http://arxiv.org/abs/1305.1786},
  eprinttype    = {arXiv},
  eprint       = {1305.1786}
}

@article{oymak2015near,
  title={Near-optimal bounds for binary embeddings of arbitrary sets},
  author={Oymak, Samet and Recht, Ben},
  journal={arXiv preprint arXiv:1512.04433},
  volume       = {abs/1512.04433},
  year         = {2015},
  numpages     = {20},
  url          = {http://arxiv.org/abs/1512.04433},
  eprinttype    = {arXiv},
  eprint       = {1512.04433}
}

@article{plan2017high,
  title={High-dimensional estimation with geometric constraints},
  author={Plan, Yaniv and Vershynin, Roman and Yudovina, Elena},
  journal={Information and Inference: A Journal of the IMA},
  volume={6},
  number={1},
  pages={1--40},
  year={2017},
  publisher={Oxford University Press}
}

@book{vershynin2018high,
  title={High-dimensional probability: An introduction with applications in data science},
  author={Vershynin, Roman},
  volume={47},
  year={2018},
  publisher={Cambridge university press}
}

@inproceedings{flodin2019superset,
  title={Superset Technique for Approximate Recovery in One-Bit Compressed Sensing},
  author={Flodin, Larkin and Gandikota, Venkata and Mazumdar, Arya},
  editor       = {Hanna M. Wallach and
                  Hugo Larochelle and
                  Alina Beygelzimer and
                  Florence d'Alch{\'{e}}{-}Buc and
                  Emily B. Fox and
                  Roman Garnett},
  booktitle    = {Advances in Neural Information Processing Systems 32: Annual Conference
                  on Neural Information Processing Systems 2019, NeurIPS 2019, December
                  8-14, 2019, Vancouver, BC, Canada},
  pages        = {10387--10396},
  year         = {2019},
  publisher = {Curran Associates, Inc.},
  url          = {https://proceedings.neurips.cc/paper/2019/hash/c900ced7451da79502d29aa37ebb7b60-Abstract.html},
  bibsource    = {dblp computer science bibliography, https://dblp.org}
}

@article{Geom,
  author    = {Weisstein, Eric W.},
  title     = {Geometric Distribution},
  journal   = {From MathWorld--A Wolfram Web Resource},
  url       = { http://mathworld.wolfram.com/GeometricDistribution.html},
  year = "2019"
 }

@article{candes2006robust,
  title={Robust uncertainty principles: exact signal reconstruction from highly incomplete frequency information},
  author={Cand{\`e}s, Emmanuel J and Romberg, Justin and Tao, Terence},
  journal={IEEE Transactions on Information Theory},
  volume={52},
  number={2},
  pages={489--509},
  year={2006},
  publisher={IEEE}
}

@article{knudson2016one,
  title={One-bit compressive sensing with norm estimation},
  author={Knudson, Karin and Saab, Rayan and Ward, Rachel},
  journal={IEEE Transactions on Information Theory},
  volume={62},
  number={5},
  pages={2748--2758},
  year={2016},
  publisher={IEEE}
}

@article{plan2013one,
  title={One-Bit Compressed Sensing by Linear Programming},
  author={Plan, Yaniv and Vershynin, Roman},
  journal={Communications on Pure and Applied Mathematics},
  volume={66},
  number={8},
  pages={1275--1297},
  year={2013},
  publisher={Wiley Online Library}
}

@article{JLBB13,
  title={Robust 1-bit compressive sensing via binary stable embeddings of sparse vectors},
  author={Jacques, Laurent and Laska, Jason N and Boufounos, Petros T and Baraniuk, Richard G},
  journal={IEEE Transactions on Information Theory},
  volume={59},
  number={4},
  pages={2082--2102},
  year={2013},
  publisher={IEEE}
}

@inproceedings{GNJN13,
  title={One-bit compressed sensing: Provable support and vector recovery},
  author={Gopi, Sivakant and Netrapalli, Praneeth and Jain, Prateek and Nori, Aditya},
  booktitle    = {Proceedings of the 30th International Conference on Machine Learning,
                  {ICML} 2013, Atlanta, GA, USA, 16-21 June 2013},
  series       = {{JMLR} Workshop and Conference Proceedings},
  volume       = {28},
  pages        = {154--162},
  publisher    = {JMLR.org},
  year         = {2013},
  url          = {http://proceedings.mlr.press/v28/gopi13.html},
  bibsource    = {dblp computer science bibliography, https://dblp.org}
}

@inproceedings{ABK17,
  title={Improved bounds for universal one-bit compressive sensing},
  author={Acharya, Jayadev and Bhattacharyya, Arnab and Kamath, Pritish},
  booktitle    = {2017 {IEEE} International Symposium on Information Theory, {ISIT}
                  2017, Aachen, Germany, June 25-30, 2017},
  pages        = {2353--2357},
  publisher    = {{IEEE}},
  year         = {2017},
  url          = {https://doi.org/10.1109/ISIT.2017.8006950},
  doi          = {10.1109/ISIT.2017.8006950},
  bibsource    = {dblp computer science bibliography, https://dblp.org}
}

@inproceedings{charikar2002similarity,
  author       = {Moses Charikar},
  editor       = {John H. Reif},
  title        = {Similarity estimation techniques from rounding algorithms},
  booktitle    = {Proceedings on 34th Annual {ACM} Symposium on Theory of Computing,
                  May 19-21, 2002, Montr{\'{e}}al, Qu{\'{e}}bec, Canada},
  pages        = {380--388},
  publisher    = {{ACM}},
  year         = {2002},
  url          = {https://doi.org/10.1145/509907.509965},
  doi          = {10.1145/509907.509965},
  bibsource    = {dblp computer science bibliography, https://dblp.org}
}

@article{plan2016generalized,
  title={The generalized lasso with non-linear observations},
  author={Plan, Yaniv and Vershynin, Roman},
  journal={IEEE Transactions on information theory},
  volume={62},
  number={3},
  pages={1528--1537},
  year={2016},
  publisher={IEEE}
}

@article{DBLP:journals/tit/Donoho06,
  author    = {David L. Donoho},
  title     = {Compressed sensing},
  journal   = {{IEEE} Trans. Information Theory},
  volume    = {52},
  number    = {4},
  pages     = {1289--1306},
  year      = {2006}
}

@inproceedings{DBLP:conf/ciss/BoufounosB08,
  author    = {Petros Boufounos and
               Richard G. Baraniuk},
  title     = {1-Bit compressive sensing},
  booktitle = {42nd Annual Conference on Information Sciences and Systems, {CISS}
               2008, Princeton, NJ, USA, 19-21 March 2008},
  pages     = {16--21},
  year      = {2008},
  crossref  = {DBLP:conf/ciss/2008},
  url       = {https://doi.org/10.1109/CISS.2008.4558487},
  doi       = {10.1109/CISS.2008.4558487},
  bibsource = {dblp computer science bibliography, https://dblp.org}
}

@proceedings{DBLP:conf/ciss/2008,
  title     = {42nd Annual Conference on Information Sciences and Systems, {CISS}
               2008, Princeton, NJ, USA, 19-21 March 2008},
  publisher = {{IEEE}},
  year      = {2008},
  url       = {http://ieeexplore.ieee.org/xpl/mostRecentIssue.jsp?punumber=4555640},
  isbn      = {978-1-4244-2246-3},
  bibsource = {dblp computer science bibliography, https://dblp.org}
}

@inproceedings{DBLP:conf/ciss/HauptB11,
  author    = {Jarvis D. Haupt and
               Richard G. Baraniuk},
  title     = {Robust support recovery using sparse compressive sensing matrices},
  booktitle = {45st Annual Conference on Information Sciences and Systems, {CISS}
               2011, The John Hopkins University, Baltimore, MD, USA, 23-25 March
               2011},
  pages     = {1--6},
  year      = {2011},
  crossref  = {DBLP:conf/ciss/2011},
  url       = {https://doi.org/10.1109/CISS.2011.5766202},
  doi       = {10.1109/CISS.2011.5766202},
  bibsource = {dblp computer science bibliography, https://dblp.org}
}

@proceedings{DBLP:conf/ciss/2011,
  title     = {45st Annual Conference on Information Sciences and Systems, {CISS}
               2011, The John Hopkins University, Baltimore, MD, USA, 23-25 March
               2011},
  publisher = {{IEEE}},
  year      = {2011},
  url       = {http://ieeexplore.ieee.org/xpl/mostRecentIssue.jsp?punumber=5764460},
  isbn      = {978-1-4244-9846-8},
  bibsource = {dblp computer science bibliography, https://dblp.org}
}

@article{saab2018quantization,
  title={Quantization of compressive samples with stable and robust recovery},
  author={Saab, Rayan and Wang, Rongrong and Y{\i}lmaz, {\"O}zg{\"u}r},
  journal={Applied and Computational Harmonic Analysis},
  volume={44},
  number={1},
  pages={123--143},
  year={2018},
  publisher={Elsevier}
}

@article{baraniuk2017exponential,
  title={Exponential decay of reconstruction error from binary measurements of sparse signals},
  author={Baraniuk, Richard G and Foucart, Simon and Needell, Deanna and Plan, Yaniv and Wootters, Mary},
  journal={IEEE Transactions on Information Theory},
  volume={63},
  number={6},
  pages={3368--3385},
  year={2017},
  publisher={IEEE}
}

@article{DBLP:journals/tit/PlanV13,
  author    = {Yaniv Plan and
               Roman Vershynin},
  title     = {Robust 1-bit Compressed Sensing and Sparse Logistic Regression: {A}
               Convex Programming Approach},
  journal   = {{IEEE} Trans. Information Theory},
  volume    = {59},
  number    = {1},
  pages     = {482--494},
  year      = {2013}
}

@inproceedings{DBLP:conf/aistats/Li16,
  author    = {Ping Li},
  title     = {One Scan 1-Bit Compressed Sensing},
  booktitle = {Proceedings of the 19th International Conference on Artificial Intelligence
               and Statistics, {AISTATS} 2016, Cadiz, Spain, May 9-11, 2016},
  pages     = {1515--1523},
  year      = {2016},
  crossref  = {DBLP:conf/aistats/2016},
  url       = {http://jmlr.org/proceedings/papers/v51/li16g.html},
  bibsource = {dblp computer science bibliography, https://dblp.org}
}

@proceedings{DBLP:conf/aistats/2016,
  editor    = {Arthur Gretton and
               Christian C. Robert},
  title     = {Proceedings of the 19th International Conference on Artificial Intelligence
               and Statistics, {AISTATS} 2016, Cadiz, Spain, May 9-11, 2016},
  series    = {{JMLR} Workshop and Conference Proceedings},
  volume    = {51},
  publisher = {JMLR.org},
  year      = {2016},
  url       = {http://jmlr.org/proceedings/papers/v51/},
  bibsource = {dblp computer science bibliography, https://dblp.org}
}

\begin{appendix}
\label{outline:appendix}

\section{Proof of Theorem \ref{thm:technical:raic:modified}}
\label{outline:technical:pf}

This section proves the main technical theorem, Theorem \ref{thm:technical:raic:modified},
which is restated for convenience.

\begin{theorem*}
Let
\(
  \UnivConstA, \UnivConstAX, \UnivConstAXX, \UnivConstB, \UnivConstc_{1}, \UnivConstc_{2}, \UnivConstD > 0
\)
be universal constants as defined in Eq.~\eqref{eqn:univConstants}. 
Fix
\(
  \DDelta, \Rho \in (0,1)
\)
and
\(
  k, m, n \in \Z_{+}
\)
such that
\(
  0 < k \EDITX{\leq} n
\).
\EDITX{Let \(  \kO \defeq \min \{ 2k, n \}  \) and \(  \kOX \defeq \min \{ 4k, n\}  \).}
Define \(  \GammaX \in (0,1)  \) such that
\begin{gather*}
  \GammaX = \GammaXValue
,\end{gather*}
and let
\begin{align*}
  m
  &= \CMPLXRAICOne
  \\ &\AlignTab
  +  \CMPLXRAICTwo
  \\
  &= \CMPLXRAICBigO
.\end{align*}
%
Let
\(
  \MeasMat \in \R^{m \times n}
\)
be a measurement matrix whose rows
\(
  \MeasVec^{(i)} \sim \N( \Vec{0}, \Mat{I}_{n \times n} )
\),
\(
  i \in [m]
\),
have i.i.d. standard normal entries.
Then, the measurement matrix \( \MeasMat \) satisfies the
\( (k, n, \DDelta, \UnivConstc_{1}, \UnivConstc_{2}) \)-RAIC.
Formally, uniformly with probability at least
\(
  1 - \Rho
\),
for all
\(
  \Vec{x}, \Vec{y} \in \SparseSphereSubspace{k}{n}
\)
and all
\(
  \Coords{J} \subseteq [n]
\),
\(
  | \Coords{J} | \leq k
\),
\begin{gather*}
  \left\| ( \Vec{x} - \Vec{y} ) - \hA[\Coords{J}]( \Vec{x}, \Vec{y} ) \right\|_{2}
  \leq
  \UnivConstc_{1}
  \sqrt{\DDelta \DistS{\Vec{x}}{\Vec{y}}}
  +
  \UnivConstc_{2} \DDelta
.\end{gather*}
\end{theorem*}

The proof of the theorem will consider two regimes---the first, in
Section \ref{outline:technical:pf|>large-scale}, looks at points which are at least
distance \( \TauValue \) apart, while the second, in
Section \ref{outline:technical:pf|>small-scale}, handles points which are very close
(less than distance \( \TauValue \)).
Section \ref{outline:technical:pf|>combine} then combines the two regimes to establish the theorem.
\par
Before beginning the proof, let us introduce some notation and intermediate results.
\EDITX{As in \THEOREM \ref{thm:technical:raic:modified}, the notations of  
\(  \kO \defeq \min \{ 2k, n \}  \) and \(  \kOX \defeq \min \{ 4k, n \}  \)
will appear throughout these appendices.}
In addition, recall the definition of
\(
  \hA : \R^{n} \times \R^{n} \to \R
\),
\begin{gather}
\label{eqn:h_Z:restatement}
  \hA \left( \Vec{x}, \Vec{y} \right)
  =
  \sqrt{2\pi}
  \frac{1}{m}
  \MeasMat^{\T}
  \cdot
  \frac{1}{2}
  \left( \Sgn( \MeasMat \Vec{x} ) - \Sgn( \MeasMat \Vec{y} ) \right)
\\
\label{eqn:h_ZJ:restatement}
  \hA[\Coords{J}] \left( \Vec{x}, \Vec{y} \right)
  =
  \ThresholdSet{\supp( \Vec{x} ) \cup \supp( \Vec{y} ) \cup \Coords{J}}
  \left( \hA \left( \Vec{x}, \Vec{y} \right) \right)
\end{gather}
and further define
\begin{gather}
\label{eqn:g_Z}
  \gA \left( \Vec{x}, \Vec{y} \right)
  =
  \hA \left( \Vec{x}, \Vec{y} \right)
  -
  \left\langle
    \frac{
      \frac{\Vec{x}}{\left\| \Vec{x} \right\|_{2}}
      -
      \frac{\Vec{y}}{\left\| \Vec{y} \right\|_{2}}
    }{
      \left\|
      \frac{\Vec{x}}{\left\| \Vec{x} \right\|_{2}}
      -
      \frac{\Vec{y}}{\left\| \Vec{y} \right\|_{2}}
      \right\|_{2}
    },
    \hA \left( \Vec{x}, \Vec{y} \right)
  \right\rangle
  \frac{
    \frac{\Vec{x}}{\left\| \Vec{x} \right\|_{2}}
    -
    \frac{\Vec{y}}{\left\| \Vec{y} \right\|_{2}}
  }{
    \left\|
    \frac{\Vec{x}}{\left\| \Vec{x} \right\|_{2}}
    -
    \frac{\Vec{y}}{\left\| \Vec{y} \right\|_{2}}
    \right\|_{2}
  }
  \\ \nonumber
  -
  \left\langle
    \frac{
      \frac{\Vec{x}}{\left\| \Vec{x} \right\|_{2}}
      +
      \frac{\Vec{y}}{\left\| \Vec{y} \right\|_{2}}
    }{
      \left\|
      \frac{\Vec{x}}{\left\| \Vec{x} \right\|_{2}}
      +
      \frac{\Vec{y}}{\left\| \Vec{y} \right\|_{2}}
      \right\|_{2}
    },
    \hA \left( \Vec{x}, \Vec{y} \right)
  \right\rangle
  \frac{
    \frac{\Vec{x}}{\left\| \Vec{x} \right\|_{2}}
    +
    \frac{\Vec{y}}{\left\| \Vec{y} \right\|_{2}}
  }{
    \left\|
    \frac{\Vec{x}}{\left\| \Vec{x} \right\|_{2}}
    +
    \frac{\Vec{y}}{\left\| \Vec{y} \right\|_{2}}
    \right\|_{2}
  }
\\
\label{eqn:g_ZJ}
  \gA[\Coords{J}] \left( \Vec{x}, \Vec{y} \right)
  =
  \ThresholdSet{\supp( \Vec{x} ) \cup \supp( \Vec{y} ) \cup \Coords{J}}
  \left( \gA \left( \Vec{x}, \Vec{y} \right) \right)
\end{gather}
for
\(
  \Vec{x}, \Vec{y} \in \R^{n}
\)
and
\(
  \Coords{J} \subseteq [n]
\).
\EDIT{The following three lemmas, whose proofs are deferred to \APPENDIX \ref{outline:normal|>concentration-ineq-pfs}, are instrumental in deriving the RAIC.}
The first of the \EDIT{these} lemmas provides concentration inequalities related to these functions
\( \hA \) and \( \gA \).
\EDIT{The second lemma characterizes the number of measurements which lie in an angularly defined, \(  2  \)-dimensional subspace of \( \R^{n} \).}
\EDIT{The third lemma is a corollary to \cite[\COR 3.3]{oymak2015near}, which is stated in the proof of \LEMMA \ref{lemma:technical:concentration-ineq:lbe-local-deviations:union} found in \APPENDIX \ref{outline:normal|>concentration-ineq-pfs}. It is related to the stability of binary embeddings of nearby points via Gaussian measurements.}
%
%
It should be noted that \cite{friedlander2021nbiht} also used results from \cite{oymak2015near} in their analysis of BIHT.

\begin{lemma}
\label{lemma:technical:concentration-ineq:(u,v)}
Fix
\(
  \Variable{\ell}, \Variable{t} > 0
\),
\EDIT{\(  \Eta \defeq \sqrt{2\pi}  \),}
\EDITX{\(  \kO \defeq \min \{ 2k, n \}  \),}
\(
  \Vec{r} \in \{-1,0,1\}^{m}
\),
and
\(
  \Coords{J} \subseteq [n]
\),
such that
\(
  \left\| \Vec{r} \right\|_{0} = \Variable{\ell} > 0
\)
and
\(
  | \Coords{J} | \leq \KO
\).
Let
\(
  ( \Vec{u}, \Vec{v} ) \in \SparseSphereSubspace{k}{n} \times \SparseSphereSubspace{k}{n}
\)
be an ordered pair of real-valued unit vectors,
and define the random variables
\(
  \Vec{R}_{\Vec{u},\Vec{v}}
  =
  ( \Vec*{R}_{1;\Vec{u},\Vec{v}}, \dots, \Vec*{R}_{m;\Vec{u},\Vec{v}} )
  =
  \frac{1}{2}
  ( \sgn( \MeasMat \Vec{u} ) - \sgn( \MeasMat \Vec{v} ) )
\)
and
\(
  \RV{L}_{\Vec{u},\Vec{v}} = \left\| \Vec{R}_{\Vec{u},\Vec{v}} \right\|_{0}
\),
and suppose
\(
  \Vec{R}_{\Vec{u},\Vec{v}} = \Vec{r}
\)
and
\(
  \RV{L}_{\Vec{u},\Vec{v}} = \Variable{\ell}
\).
Then, conditioned on
\ORIG{\(
  \Vec{R}_{\Vec{u},\Vec{v}} = \Vec{r}
\)
and}%
\(
  \RV{L}_{\Vec{u},\Vec{v}} = \Variable{\ell}
\),
the following concentration inequalities hold.
\begin{gather}
\label{eqn:technical:concentration-ineq:(u,v):u-v}
  \Pr
  \left(
    \left|
      \left\langle
        \frac{\Vec{u} - \Vec{v}}{\left\| \Vec{u} - \Vec{v} \right\|_{2}},
        \frac{1}{\Eta}
        \hA[\Coords{J}]( \Vec{u}, \Vec{v} )
      \right\rangle
      -
      \sqrt{\frac{\pi}{2}}
      \frac{\Variable{\ell}}{m}
      \frac{\DistS{\Vec{u}}{\Vec{v}}}{\theta_{\Vec{u},\Vec{v}}}
    \right|
    \geq
    \frac{\Variable{\ell} \Variable{t}}{m}
    \middle|
    \RV{L}_{\Vec{u},\Vec{v}} = \Variable{\ell}
  \right)
  \leq
  2 e^{-\frac{1}{2} \Variable{\ell} \Variable{t}^{2}}
\\
\label{eqn:technical:concentration-ineq:(u,v):u+v}
  \Pr
  \left(
    \left|
      \left\langle
        \frac{\Vec{u} + \Vec{v}}{\left\| \Vec{u} + \Vec{v} \right\|_{2}},
        \frac{1}{\Eta}
        \hA[\Coords{J}]( \Vec{u}, \Vec{v} )
      \right\rangle
    \right|
    \geq
    \frac{\Variable{\ell} \Variable{t}}{m}
    \middle|
    \RV{L}_{\Vec{u},\Vec{v}} = \Variable{\ell}
  \right)
  \leq
  2 e^{-\frac{1}{2} \Variable{\ell} \Variable{t}^{2}}
\\
\label{eqn:technical:concentration-ineq:(u,v):g_Z}
  \Pr
  \left(
    \left\|
      \frac{1}{\Eta}
      \gA[\Coords{J}]( \Vec{u}, \Vec{v} )
    \right\|_{2}
    \geq
    \frac{2\sqrt{\KO \Variable{\ell}}}{m}
    +
    \frac{\Variable{\ell} \Variable{t}}{m}
    \middle|
    \RV{L}_{\Vec{u},\Vec{v}} = \Variable{\ell}
  \right)
  \leq
  2 e^{-\frac{1}{8} \Variable{\ell} \Variable{t}^{2}}
\end{gather}
\end{lemma}
%
\begin{lemma}
\label{lemma:technical:concentration-ineq:counting}
Fix
\(
  \Variable{t} \in (0,1)
\).
%
Let
\(
  \Vec{u}, \Vec{v} \in \R^{n}
\),
and define the random variable
\(  \RV{L}_{\Vec{u},\Vec{v}} = \left\|   \frac{1}{2} ( \sgn( \MeasMat \Vec{u} ) - \sgn( \MeasMat \Vec{v} ) ) \right\|_{0}  \),
as in \LEMMA \ref{lemma:technical:concentration-ineq:(u,v)}.
Then,
\begin{gather}
  \mu_{\RV{L}_{\Vec{u},\Vec{v}}}
  =
  \E \left[ \RV{L}_{\Vec{u},\Vec{v}} \right]
  =
  \frac{\theta_{\Vec{u},\Vec{v}} m}{\pi}
\end{gather}
and
\begin{gather}
  \Pr
  \left(
    \RV{L}_{\Vec{u},\Vec{v}}
    \notin
    \left[
      (1 - \Variable{t}) \mu_{\RV{L}_{\Vec{u},\Vec{v}}},
      (1 + \Variable{t}) \mu_{\RV{L}_{\Vec{u},\Vec{v}}}
    \right]
  \right)
  \leq
  2e^{-\frac{1}{3} \mu_{\RV{L}_{\Vec{u},\Vec{v}}} \Variable{t}^{2}}
.\end{gather}
\end{lemma}
%
\begin{lemma}[{Corollary to \cite[\COR 3.3]{oymak2015near}}]
\label{lemma:technical:concentration-ineq:lbe-local-deviations:union}
Fix \(  \UnivConstD > 0  \) as the universal constant specified in \EQN \eqref{eqn:univConstants}, and let \(  \DDeltaX \in (0,1)  \) \EDITX{and \(  \kO \defeq \min \{ \KO, n \}  \)}.
Let \(  \MeasMat \in \R^{m \times n}  \) be a standard Gaussian matrix with i.i.d. entries.
If
\(  m \geq \frac{\UnivConstD k}{\DDeltaX} \Log( \frac{1}{\DDeltaX} )  \),
then with probability at least
\(  1 - 2 \binom{n}{\kO} e^{-\frac{1}{64} \DDeltaX m}  \),
 uniformly for all \(  \Vec{u}, \Vec{v} \in \SparseSphereSubspace{k}{n}  \) such that
\(  \| \Vec{u} - \Vec{v} \|_{2} \leq \frac{\DDeltaX}{\UnivConstD \sqrt{\Log( 1/\DDeltaX )}}  \),
the random variable \(  \RV{L}_{\Vec{u},\Vec{v}} = \left\|   \frac{1}{2} ( \sgn( \MeasMat \Vec{u} ) - \sgn( \MeasMat \Vec{v} ) ) \right\|_{0}  \), defined as in \LEMMA \ref{lemma:technical:concentration-ineq:(u,v)}, satisfies \(  \RV{L}_{\Vec{u},\Vec{v}} \leq \DDeltaX m  \).
\end{lemma}

Lastly, for the purposes of the proof, a \( \Tau \)-net
\(
  \Net{\Tau} \subset \SparseSphereSubspace{k}{n}
\)
over the set of \( k \)-sparse, real-valued unit vectors is designed as follows, where
\(
  \Tau \defeq \TauValue
\)
is defined to lighten the notation.
For each
\(
  \Coords{J} \subseteq [n]
\),
\(
  | \Coords{J} | \leq k
\),
let
\(
  \Net{\Tau;\Coords{J}} \subset \SparseSphereSubspace{k}{n}
\)
be a \( \Tau \)-net over the set
\(
  \{ \Vec{x} \in \SparseSphereSubspace{k}{n} : \supp( \Vec{x} ) = \Coords{J} \}
\).
Then, construct the \( \Tau \)-net
\(
  \Net{\Tau} \subset \SparseSphereSubspace{k}{n}
\)
as their union,
\(
  \Net{\Tau}
  =
  \bigcup_{\Coords{J} \subseteq [n] : | \Coords{J} | \leq k}
  \Net{\Tau;\Coords{J}}
\).
Note that
\(
  | \Net{\Tau} |
  \leq \binom{n}{k} \left( \frac{3}{\Tau} \right)^{k} 2^{k}
  = \UBNet
\)
and
\(
  | \Net{\Tau} \times \Net{\Tau} |
  \leq \binom{n}{k}^{2} \left( \frac{3}{\Tau} \right)^{\KO} 2^{\KO}
  = \UBNetNet
\).
This construction is consistent throughout
Sections \ref{outline:technical:pf|>large-scale}-\ref{outline:technical:pf|>combine}.


\subsection{``Large distances'' regime}
\label{outline:technical:pf|>large-scale}

The first regime considers the RAIC for ordered pairs of points in the \( \Tau \)-net which are
at least distance \( \Tau \) from each other.
Lemma \ref{lemma:technical:raic:large-scale} formalizes a uniform result in this regime.

\begin{lemma}
\label{lemma:technical:raic:large-scale}
Let
\(
  \UnivConstb_{1} > 0
\)
be a universal constant.
\EDIT{Define \(  \GammaX, m > 0  \) as in \THEOREM \ref{thm:technical:raic:modified}.}
Fix
\(
  \DDelta, \Rho_{1} \in (0,1)
\),
where
\( \Rho_{1} \defeq \frac{\Rho}{2} \),
and let
\(
  \Tau \defeq \TauValue
\).
%
%
Uniformly with probability at least
\(
  1 - \Rho_{1}
\),
\begin{gather}
\label{eqn:technical:raic:large-scale}
  \left\| ( \Vec{u} - \Vec{v} ) - \hA[\Coords{J}]( \Vec{u}, \Vec{v} ) \right\|_{2}
  \leq
  \UnivConstb_{1}
  \sqrt{\DDelta \DistS{\Vec{u}}{\Vec{v}}}
\end{gather}
for all
\(
  ( \Vec{u}, \Vec{v} )
  \in
  \Net{\Tau} \times \Net{\Tau}
\)
satisfying
\(
  \DistS{\Vec{u}}{\Vec{v}} \geq \Tau
\),
and
\(
  \Coords{J} \subseteq [n]
\),
\(
  | \Coords{J} | \leq \KO
\).
\end{lemma}

\BEGINEDIT
Before proving \LEMMA \ref{lemma:technical:raic:large-scale}, we introduce the following fact which bounds the ratio of \(  \theta_{\Vec{u},\Vec{v}}/\DistS{\Vec{u}}{\Vec{v}}  \).

\BEGINEDITX
\begin{fact}
\label{fact:technical:theta/d}
For \(  \Vec{u}, \Vec{v} \in \Sphere{n}  \),
\begin{gather}
  \DistS{\Vec{u}}{\Vec{v}}
  \leq
  \theta_{\Vec{u},\Vec{v}}
  \leq
  \frac{\pi}{2} \DistS{\Vec{u}}{\Vec{v}}
.\end{gather}
\end{fact}

\begin{proof}
{\FACT \ref{fact:technical:theta/d}}
Since, \(  \sin(x) \leq x  \)
for \(  x \geq 0  \), it follows that \(  \DistS{\Vec{u}}{\Vec{v}} = \| \Vec{u} - \Vec{v} \|_{2} = \sqrt{2 ( 1 - \cos( \theta_{\Vec{u},\Vec{v}} ) )}  = 2 \sin( \frac{\theta_{\Vec{u},\Vec{v}}}{2} )  \leq    \theta_{\Vec{u},\Vec{v}}\).
Additionally, using basic calculus, it can be shown that on the interval \(  x \in [0,\pi/2]  \),
\(  \frac{\sin(x)}{x}  \)
decreases, implying (since \(  \theta_{\Vec{u},\Vec{v}} \in [0,\pi]\)),

\begin{align*}
  \frac{\DistS{\Vec{u}}{\Vec{v}}}{\theta_{\Vec{u},\Vec{v}}}
  =
  \frac{2 \sin( \frac{\theta_{\Vec{u},\Vec{v}}}{2} )}{\theta_{\Vec{u},\Vec{v}}}
  =
  \frac{\sin( \frac{\theta_{\Vec{u},\Vec{v}}}{2} )}{\frac{\theta_{\Vec{u},\Vec{v}}}{2}}
  \geq \frac{\sin(\frac{\pi}{2})}{\frac{\pi}{2}} =
  \frac{2}{\pi}
.\end{align*}
%
\end{proof}
\ENDEDITX

With this, we are ready to prove Lemma \ref{lemma:technical:raic:large-scale}.
\ENDEDIT

\begin{proof}
{Lemma \ref{lemma:technical:raic:large-scale}}
\label{pf:lemma:technical:raic:large-scale}
Let
\(
  ( \Vec{u}, \Vec{v} )
  \in
  \Net{\Tau} \times \Net{\Tau}
\)
be an arbitrary ordered pair of points in the \( \Tau \)-net whose distance is at least
\(
  \DistS{\Vec{u}}{\Vec{v}} \geq \Tau
\).
Similar to the approach by \cite{friedlander2021nbiht} and seen in \cite{plan2016generalized},
the function \( \hA[\Coords{J}] \) can be orthogonally decomposed as
\begin{gather}
\label{pf:lemma:technical:raic:large-scale:eqn:1}
  \hA[\Coords{J}]( \Vec{u}, \Vec{v} )
  =
  \left\langle
  \frac{\Vec{u} - \Vec{v}}{\left\| \Vec{u} - \Vec{v} \right\|_{2}},
  \hA[\Coords{J}]( \Vec{u}, \Vec{v} )
  \right\rangle
  \frac{\Vec{u} - \Vec{v}}{\left\| \Vec{u} - \Vec{v} \right\|_{2}}
  +
  \left\langle
  \frac{\Vec{u} + \Vec{v}}{\left\| \Vec{u} + \Vec{v} \right\|_{2}},
  \hA[\Coords{J}]( \Vec{u}, \Vec{v} )
  \right\rangle
  \frac{\Vec{u} + \Vec{v}}{\left\| \Vec{u} + \Vec{v} \right\|_{2}}
  +
  \gA[\Coords{J}]( \Vec{u}, \Vec{v} )
\end{gather}
%
Combining \eqref{pf:lemma:technical:raic:large-scale:eqn:1} with the triangle inequality yields
\begin{subequations}
\begin{align}
\label{pf:lemma:technical:raic:large-scale:eqn:2}
  &
  \left\| ( \Vec{u} - \Vec{v} ) - \hA[\Coords{J}]( \Vec{u}, \Vec{v} ) \right\|_{2}
  \\
  &=
  \left\|
    ( \Vec{u} - \Vec{v} )
    -
    \left(
      \left\langle
      \frac{\Vec{u} - \Vec{v}}{\left\| \Vec{u} - \Vec{v} \right\|_{2}},
      \hA[\Coords{J}]( \Vec{u}, \Vec{v} )
      \right\rangle
      \frac{\Vec{u} - \Vec{v}}{\left\| \Vec{u} - \Vec{v} \right\|_{2}}
      +
      \left\langle
      \frac{\Vec{u} + \Vec{v}}{\left\| \Vec{u} + \Vec{v} \right\|_{2}},
      \hA[\Coords{J}]( \Vec{u}, \Vec{v} )
      \right\rangle
      \frac{\Vec{u} + \Vec{v}}{\left\| \Vec{u} + \Vec{v} \right\|_{2}}
      +
      \gA[\Coords{J}]( \Vec{u}, \Vec{v} )
    \right)
  \right\|_{2}
  \\
  &\leq
  \left\|
    ( \Vec{u} - \Vec{v} )
    -
    \left\langle
    \frac{\Vec{u} - \Vec{v}}{\left\| \Vec{u} - \Vec{v} \right\|_{2}},
    \hA[\Coords{J}]( \Vec{u}, \Vec{v} )
    \right\rangle
    \frac{\Vec{u} - \Vec{v}}{\left\| \Vec{u} - \Vec{v} \right\|_{2}}
  \right\|_{2}
  +
  \left\|
    \left\langle
    \frac{\Vec{u} + \Vec{v}}{\left\| \Vec{u} + \Vec{v} \right\|_{2}},
    \hA[\Coords{J}]( \Vec{u}, \Vec{v} )
    \right\rangle
    \frac{\Vec{u} + \Vec{v}}{\left\| \Vec{u} + \Vec{v} \right\|_{2}}
  \right\|_{2}
  +
  \left\|
    \gA[\Coords{J}]( \Vec{u}, \Vec{v} )
  \right\|_{2}
  \\ \nonumber
  &\dCmt \Text{by the triangle inequality}
  \\
  &=
  \left|
    \left\| \Vec{u} - \Vec{v} \right\|_{2}
    -
    \left\langle
    \frac{\Vec{u} - \Vec{v}}{\left\| \Vec{u} - \Vec{v} \right\|_{2}},
    \hA[\Coords{J}]( \Vec{u}, \Vec{v} )
    \right\rangle
  \right|
  \left\|
    \frac{\Vec{u} - \Vec{v}}{\left\| \Vec{u} - \Vec{v} \right\|_{2}}
  \right\|_{2}
  +
  \left|
    \left\langle
    \frac{\Vec{u} + \Vec{v}}{\left\| \Vec{u} + \Vec{v} \right\|_{2}},
    \hA[\Coords{J}]( \Vec{u}, \Vec{v} )
    \right\rangle
  \right|
  \left\|
    \frac{\Vec{u} + \Vec{v}}{\left\| \Vec{u} + \Vec{v} \right\|_{2}}
  \right\|_{2}
  +
  \left\|
    \gA[\Coords{J}]( \Vec{u}, \Vec{v} )
  \right\|_{2}
  \\
  &=
  \left|
    \left\| \Vec{u} - \Vec{v} \right\|_{2}
    -
    \left\langle
    \frac{\Vec{u} - \Vec{v}}{\left\| \Vec{u} - \Vec{v} \right\|_{2}},
    \hA[\Coords{J}]( \Vec{u}, \Vec{v} )
    \right\rangle
  \right|
  +
  \left|
    \left\langle
    \frac{\Vec{u} + \Vec{v}}{\left\| \Vec{u} + \Vec{v} \right\|_{2}},
    \hA[\Coords{J}]( \Vec{u}, \Vec{v} )
    \right\rangle
  \right|
  +
  \left\|
    \gA[\Coords{J}]( \Vec{u}, \Vec{v} )
  \right\|_{2}
\end{align}
\end{subequations}
%
Lemma \ref{lemma:technical:concentration-ineq:(u,v)} provides the following concentration
inequalities.
\begin{gather}
\label{pf:lemma:technical:raic:large-scale:eqn:3:1}
  \Pr
  \left(
    \left|
      \left\langle
        \frac{\Vec{u} - \Vec{v}}{\left\| \Vec{u} - \Vec{v} \right\|_{2}},
        \frac{1}{\Eta}
        \hA[\Coords{J}]( \Vec{u}, \Vec{v} )
      \right\rangle
      -
      \sqrt{\frac{\pi}{2}}
      \frac{\Variable{\ell}_{\Vec{u},\Vec{v}}}{m}
      \frac{\DistS{\Vec{u}}{\Vec{v}}}{\theta_{\Vec{u},\Vec{v}}}
    \right|
    >
    \frac{\Variable{\ell}_{\Vec{u},\Vec{v}} \Variable{t}_{\Vec{u},\Vec{v}}}{m}
    \middle|
    \RV{L}_{\Vec{u},\Vec{v}} = \Variable{\ell}_{\Vec{u},\Vec{v}}
  \right)
  \leq
  2 e^{-\frac{1}{2} \Variable{\ell}_{\Vec{u},\Vec{v}} \Variable{t}_{\Vec{u},\Vec{v}}^{2}}
\\
\label{pf:lemma:technical:raic:large-scale:eqn:3:2}
  \Pr
  \left(
    \left|
      \left\langle
        \frac{\Vec{u} + \Vec{v}}{\left\| \Vec{u} + \Vec{v} \right\|_{2}},
        \frac{1}{\Eta}
        \hA[\Coords{J}]( \Vec{u}, \Vec{v} )
      \right\rangle
    \right|
    >
    \frac{\Variable{\ell}_{\Vec{u},\Vec{v}} \Variable{t}_{\Vec{u},\Vec{v}}}{m}
    \middle|
    \RV{L}_{\Vec{u},\Vec{v}} = \Variable{\ell}_{\Vec{u},\Vec{v}}
  \right)
  \leq
  2 e^{-\frac{1}{2} \Variable{\ell}_{\Vec{u},\Vec{v}} \Variable{t}_{\Vec{u},\Vec{v}}^{2}}
\\
\label{pf:lemma:technical:raic:large-scale:eqn:3:3}
  \Pr
  \left(
    \left\|
      \frac{1}{\Eta}
      \gA[\Coords{J}]( \Vec{u}, \Vec{v} )
    \right\|_{2}
    >
    \frac{2\sqrt{\KO \Variable{\ell}_{\Vec{u},\Vec{v}}}}{m}
    +
    \frac{\Variable{\ell}_{\Vec{u},\Vec{v}} \Variable{t}_{\Vec{u},\Vec{v}}}{m}
  \middle|
    \RV{L}_{\Vec{u},\Vec{v}} = \Variable{\ell}_{\Vec{u},\Vec{v}}
  \right)
  \leq
  2 e^{-\frac{1}{8} \Variable{\ell}_{\Vec{u},\Vec{v}} \Variable{t}_{\Vec{u},\Vec{v}}^{2}}
\end{gather}
\ORIG{where \( \Vec{R}_{\Vec{u},\Vec{v}} \) and \( \RV{L}_{\Vec{u},\Vec{v}} \) are random variables
defined as
\(
  \Vec{R}_{\Vec{u},\Vec{v}}
  =
  ( \Vec*{R}_{1;\Vec{u},\Vec{v}}, \dots, \Vec*{R}_{m;\Vec{u},\Vec{v}} )
  =
  \frac{1}{2}
  \left( \sgn( \MeasMat \Vec{u} ) - \sgn( \MeasMat \Vec{v} ) \right)
\)
and
\(
  \RV{L}_{\Vec{u},\Vec{v}} = \left\| \Vec{R}_{\Vec{u},\Vec{v}} \right\|_{0}
\),
and
\(
  \Vec{r} \in \{-1,0,1\}^{m}
\),
\(
  \Variable{\ell}_{\Vec{u},\Vec{v}} \in [m]
\).}%
\EDIT{where \(  \RV{L}_{\Vec{u},\Vec{v}} = \| \frac{1}{2} \left( \sgn( \MeasMat \Vec{u} ) - \sgn( \MeasMat \Vec{v} ) \right) \|_{0}  \) and \(  \Variable{\ell}_{\Vec{u},\Vec{v}} \in [m]  \).}
\EQN \eqref{pf:lemma:technical:raic:large-scale:eqn:3:1} further implies
\begin{gather}
\label{pf:lemma:technical:raic:large-scale:eqn:3b:1}
\nonumber
  \Pr
  \left(
    \left|
      \left(
        \left\| \Vec{u} - \Vec{v} \right\|_{2}
        -
        \left\langle
          \frac{\Vec{u} - \Vec{v}}{\left\| \Vec{u} - \Vec{v} \right\|_{2}},
          \hA[\Coords{J}]( \Vec{u}, \Vec{v} )
        \right\rangle
      \right)
      -
      \left(
        \left\| \Vec{u} - \Vec{v} \right\|_{2}
        -
        \sqrt{\frac{\pi}{2}}
        \frac{\Eta \Variable{\ell}_{\Vec{u},\Vec{v}}}{m}
        \frac{\DistS{\Vec{u}}{\Vec{v}}}{\theta_{\Vec{u},\Vec{v}}}
      \right)
    \right|
    >
    \frac{\Eta \Variable{\ell}_{\Vec{u},\Vec{v}} \Variable{t}_{\Vec{u},\Vec{v}}}{m}
    \middle|
    \RV{L}_{\Vec{u},\Vec{v}} = \Variable{\ell}_{\Vec{u},\Vec{v}}
  \right)
  \\
  \leq
  2 e^{-\frac{1}{2} \Variable{\ell}_{\Vec{u},\Vec{v}} \Variable{t}_{\Vec{u},\Vec{v}}^{2}}
\end{gather}
while \EQNS \eqref{pf:lemma:technical:raic:large-scale:eqn:3:2} and
\eqref{pf:lemma:technical:raic:large-scale:eqn:3:3} can be written
\begin{gather}
\label{pf:lemma:technical:raic:large-scale:eqn:3b:2}
  \Pr
  \left(
    \left|
      \left\langle
        \frac{\Vec{u} + \Vec{v}}{\left\| \Vec{u} + \Vec{v} \right\|_{2}},
        \hA[\Coords{J}]( \Vec{u}, \Vec{v} )
      \right\rangle
    \right|
    >
    \frac{\Eta \Variable{\ell}_{\Vec{u},\Vec{v}} \Variable{t}_{\Vec{u},\Vec{v}}}{m}
    \middle|
    \RV{L}_{\Vec{u},\Vec{v}} = \Variable{\ell}_{\Vec{u},\Vec{v}}
  \right)
  \leq
  2 e^{-\frac{1}{2} \Variable{\ell}_{\Vec{u},\Vec{v}} \Variable{t}_{\Vec{u},\Vec{v}}^{2}}
\\
\label{pf:lemma:technical:raic:large-scale:eqn:3b:3}
  \Pr
  \left(
    \left\|
      \gA[\Coords{J}]( \Vec{u}, \Vec{v} )
    \right\|_{2}
    >
    \frac{2\Eta \sqrt{\KO \Variable{\ell}_{\Vec{u},\Vec{v}}}}{m}
    +
    \frac{\Eta \Variable{\ell}_{\Vec{u},\Vec{v}} \Variable{t}_{\Vec{u},\Vec{v}}}{m}
  \middle|
    \RV{L}_{\Vec{u},\Vec{v}} = \Variable{\ell}_{\Vec{u},\Vec{v}}
  \right)
  \leq
  2 e^{-\frac{1}{8} \Variable{\ell}_{\Vec{u},\Vec{v}} \Variable{t}_{\Vec{u},\Vec{v}}^{2}}
\end{gather}
%
It follows that given
\(
  \RV{L}_{\Vec{u},\Vec{v}} = \Variable{\ell}_{\Vec{u},\Vec{v}}
\),
with probability at least
\(
  1 - 6 e^{-\frac{1}{8} \Variable{\ell}_{\Vec{u},\Vec{v}} \Variable{t}_{\Vec{u},\Vec{v}}^{2}}
\),
the following holds:
\begin{subequations}
\label{pf:lemma:technical:raic:large-scale:eqn:4:0}
\begin{align}
\label{pf:lemma:technical:raic:large-scale:eqn:4}
  &
  \left\| ( \Vec{u} - \Vec{v} ) - \hA[\Coords{J}]( \Vec{u}, \Vec{v} ) \right\|_{2}
  \\
  &\leq
  \left|
    \left\| \Vec{u} - \Vec{v} \right\|_{2}
    -
    \left\langle
    \frac{\Vec{u} - \Vec{v}}{\left\| \Vec{u} - \Vec{v} \right\|_{2}},
    \hA[\Coords{J}]( \Vec{u}, \Vec{v} )
    \right\rangle
  \right|
  +
  \left|
    \left\langle
    \frac{\Vec{u} + \Vec{v}}{\left\| \Vec{u} + \Vec{v} \right\|_{2}},
    \hA[\Coords{J}]( \Vec{u}, \Vec{v} )
    \right\rangle
  \right|
  +
  \left\|
    \gA[\Coords{J}]( \Vec{u}, \Vec{v} )
  \right\|_{2}
  \\
  &\leq
  \left|
    \left\| \Vec{u} - \Vec{v} \right\|_{2}
    -
    \sqrt{\frac{\pi}{2}}
    \frac{\Eta \Variable{\ell}_{\Vec{u},\Vec{v}}}{m}
    \frac{\DistS{\Vec{u}}{\Vec{v}}}{\theta_{\Vec{u},\Vec{v}}}
  \right|
  +
  \frac{\Eta \Variable{\ell}_{\Vec{u},\Vec{v}} \Variable{t}_{\Vec{u},\Vec{v}}}{m}
  +
  \frac{\Eta \Variable{\ell}_{\Vec{u},\Vec{v}} \Variable{t}_{\Vec{u},\Vec{v}}}{m}
  +
  \frac{2\Eta \sqrt{\KO \Variable{\ell}_{\Vec{u},\Vec{v}}}}{m}
  +
  \frac{\Eta \Variable{\ell}_{\Vec{u},\Vec{v}} \Variable{t}_{\Vec{u},\Vec{v}}}{m}
  \\
  &=
  \left|
    \DistS{\Vec{u}}{\Vec{v}}
    -
    \sqrt{\frac{\pi}{2}}
    \frac{\Eta \Variable{\ell}_{\Vec{u},\Vec{v}}}{m}
    \frac{\DistS{\Vec{u}}{\Vec{v}}}{\theta_{\Vec{u},\Vec{v}}}
  \right|
  +
  \frac{3 \Eta \Variable{\ell}_{\Vec{u},\Vec{v}} \Variable{t}_{\Vec{u},\Vec{v}}}{m}
  +
  \frac{2\Eta \sqrt{\KO \Variable{\ell}_{\Vec{u},\Vec{v}}}}{m}
  \\
  &=
  \left|
    1
    -
    \sqrt{\frac{\pi}{2}}
    \frac{\Eta \Variable{\ell}_{\Vec{u},\Vec{v}}}{m}
    \frac{1}{\theta_{\Vec{u},\Vec{v}}}
  \right|
  \DistS{\Vec{u}}{\Vec{v}}
  +
  \frac{3 \Eta \Variable{\ell}_{\Vec{u},\Vec{v}} \Variable{t}_{\Vec{u},\Vec{v}}}{m}
  +
  \frac{2 \Eta \sqrt{\KO \Variable{\ell}_{\Vec{u},\Vec{v}}}}{m}
\end{align}
\end{subequations}
%
\par
Let us next get a handle on the random variable \( \RV{L}_{\Vec{u},\Vec{v}} \), which tallies up
the number of sign differences between
\(
  \sgn( \MeasMat \Vec{u} )
\)
and
\(
  \sgn( \MeasMat \Vec{v} )
\).
%
%
By Lemma \ref{lemma:technical:concentration-ineq:counting}, the random variable
\( \RV{L}_{\Vec{u},\Vec{v}} \) can be characterized by its expectation:
\begin{gather}
  \E \left[ \RV{L}_{\Vec{u},\Vec{v}} \right]
  =
  \frac{\theta_{\Vec{u},\Vec{v}} m}{\pi}
\end{gather}
and the concentration inequality:
\begin{gather}
  \Pr
  \left(
    \RV{L}_{\Vec{u},\Vec{v}}
    \notin
    \left[
      ( 1 - \Variable{s}_{\Vec{u},\Vec{v}} ) \frac{\theta_{\Vec{u},\Vec{v}} m}{\pi}
      ,
      ( 1 + \Variable{s}_{\Vec{u},\Vec{v}} ) \frac{\theta_{\Vec{u},\Vec{v}} m}{\pi}
    \right]
  \right)
  \leq
  2 e^{-\frac{1}{3\pi} \theta_{\Vec{u},\Vec{v}} m\Variable{s}_{\Vec{u},\Vec{v}}^{2}}
.\end{gather}
%
\par
Thus far, it has been shown that for a given pair
\(
  ( \Vec{u}, \Vec{v} ) \in \Net{\Tau} \times \Net{\Tau}
\),
where
\(
  \DistS{\Vec{u}}{\Vec{v}} \geq \Tau
\),
with probability at least
\(
  1
  - 6 e^{-\frac{1}{8} \Variable{\ell}_{\Vec{u},\Vec{v}} \Variable{t}_{\Vec{u},\Vec{v}}^{2}}
  - 2 e^{-\frac{1}{3\pi} \theta_{\Vec{u},\Vec{v}} m\Variable{s}_{\Vec{u},\Vec{v}}^{2}}
\),
\begin{gather}
\label{pf:lemma:technical:raic:large-scale:eqn:5}
  \left\| ( \Vec{u} - \Vec{v} ) - \hA[\Coords{J}]( \Vec{u}, \Vec{v} ) \right\|_{2}
  \leq
  \left|
    1
    -
    \sqrt{\frac{\pi}{2}}
    \frac{\Eta \Variable{\ell}_{\Vec{u},\Vec{v}}}{m}
    \frac{1}{\theta_{\Vec{u},\Vec{v}}}
  \right|
  \DistS{\Vec{u}}{\Vec{v}}
  +
  \frac{3 \Eta \Variable{\ell}_{\Vec{u},\Vec{v}} \Variable{t}_{\Vec{u},\Vec{v}}}{m}
  +
  \frac{2\Eta \sqrt{\KO \Variable{\ell}_{\Vec{u},\Vec{v}}}}{m}
\end{gather}
where
\(
  \Variable{\ell}_{\Vec{u},\Vec{v}}
  \in
  [
    ( 1 - \Variable{s}_{\Vec{u},\Vec{v}} ) \frac{\theta_{\Vec{u},\Vec{v}} m}{\pi}
    ,
    ( 1 + \Variable{s}_{\Vec{u},\Vec{v}} ) \frac{\theta_{\Vec{u},\Vec{v}} m}{\pi}
  ]
\).
Next, this result will be extended---via union bounding---to hold uniformly for over all pairs
\(
  ( \Vec{u}, \Vec{v} )
  \in
  \Net{\Tau} \times \Net{\Tau}
\)
with
\(
  \DistS{\Vec{u}}{\Vec{v}} \geq \Tau
\)
and each
\(
  \Coords{J} \subseteq [n]
\),
\(
  | \Coords{J} | \leq \KO
\).
Let
\(
  \Rho'_{1}, \Rho''_{1} \in (0,1)
\)
such that
\(
  \Rho'_{1} + \Rho''_{1} = \Rho_{1}
\).
For each pair
\(
  \Vec{u}, \Vec{v} \in \Net{\Tau}
\)
and every
\(
  \Coords{J} \subseteq [n]
\),
\(
  | \Coords{J} | = \KO
\),
the parameters
\(
  \Variable{s}_{\Vec{u},\Vec{v}}
\)
and
\(
  \Variable{t}_{\Vec{u},\Vec{v}}
\)
should ensure
\begin{gather}
\label{pf:lemma:technical:raic:large-scale:eqn:6}
  \Pr
  \left(
    \exists \Vec{u}, \Vec{v} \in \Net{\Tau},
    \SPACE
    \DistS{\Vec{u}}{\Vec{v}} \geq \Tau,
    \TAB
    \RV{L}_{\Vec{u},\Vec{v}}
    \notin
    \left[
      ( 1 - \Variable{s}_{\Vec{u},\Vec{v}} ) \frac{\theta_{\Vec{u},\Vec{v}} m}{\pi}
      ,
      ( 1 + \Variable{s}_{\Vec{u},\Vec{v}} ) \frac{\theta_{\Vec{u},\Vec{v}} m}{\pi}
    \right]
  \right)
  \leq
  \Rho'_{1}
\end{gather}
and
\begin{gather}
\label{pf:lemma:technical:raic:large-scale:eqn:7}
  \Pr
  \left(
  \begin{array}{l}
    \exists ( \Vec{u}, \Vec{v} )
    \in
    \Net{\Tau} \times \Net{\Tau},
    \SPACE
    \DistS{\Vec{u}}{\Vec{v}} \geq \Tau,
    \\
    \exists \Coords{J} \subseteq [n],
    \SPACE
    | \Coords{J} | \leq \KO,
    \\
    \left\| ( \Vec{u} - \Vec{v} ) - \hA[\Coords{J}]( \Vec{u}, \Vec{v} ) \right\|_{2}
    \\
    \TAB
    >
    \left|
      1
      -
      \sqrt{\frac{\pi}{2}}
      \frac{\Eta \Variable{\ell}_{\Vec{u},\Vec{v}}}{m}
      \frac{1}{\theta_{\Vec{u},\Vec{v}}}
    \right|
    \DistS{\Vec{u}}{\Vec{v}}
    +
    \frac{3 \Eta \Variable{\ell}_{\Vec{u},\Vec{v}} \Variable{t}_{\Vec{u},\Vec{v}}}{m}
    +
    \frac{2\Eta \sqrt{\KO \Variable{\ell}_{\Vec{u},\Vec{v}}}}{m}
  \end{array}
  \middle|
    \RV{L}_{\Vec{u},\Vec{v}} = \Variable{\ell}_{\Vec{u},\Vec{v}}
    \in
    \left[
      ( 1 \pm \Variable{s}_{\Vec{u},\Vec{v}} ) \frac{\theta_{\Vec{u},\Vec{v}} m}{\pi}
    \right]
  \right)
  \leq
  \Rho''_{1}
\end{gather}
%
For the former, \eqref{pf:lemma:technical:raic:large-scale:eqn:6}, observe,
\begin{subequations}
\begin{align}
\label{pf:lemma:technical:raic:large-scale:eqn:6b}
  &
  \Pr
  \left(
    \exists \Vec{u}, \Vec{v} \in \Net{\Tau},
    \SPACE
    \DistS{\Vec{u}}{\Vec{v}} \geq \Tau,
    \TAB
    \RV{L}_{\Vec{u},\Vec{v}}
    \notin
    \left[
      ( 1 - \Variable{s}_{\Vec{u},\Vec{v}} ) \frac{\theta_{\Vec{u},\Vec{v}} m}{\pi}
      ,
      ( 1 + \Variable{s}_{\Vec{u},\Vec{v}} ) \frac{\theta_{\Vec{u},\Vec{v}} m}{\pi}
    \right]
  \right)
  \leq
  \Rho'_{1}
  \\
  &\dLn
  \binom{n}{k}^{2}
  \left( \frac{6}{\Tau} \right)^{\KO}
  \Pr
  \left(
    \RV{L}_{\Vec{u},\Vec{v}}
    \notin
    \left[
      ( 1 - \Variable{s}_{\Vec{u},\Vec{v}} ) \frac{\theta_{\Vec{u},\Vec{v}} m}{\pi}
      ,
      ( 1 + \Variable{s}_{\Vec{u},\Vec{v}} ) \frac{\theta_{\Vec{u},\Vec{v}} m}{\pi}
    \right]
  \right)
  \leq
  \Rho'_{1}
  \\
  &\dLn
  \binom{n}{k}^{2}
  \left( \frac{6}{\Tau} \right)^{\KO}
  2 e^{-\frac{1}{3\pi} \theta_{\Vec{u},\Vec{v}} m\Variable{s}_{\Vec{u},\Vec{v}}^{2}}
  \leq
  \Rho'_{1}
  \\
  &\dLn
  \Variable{s}_{\Vec{u},\Vec{v}}
  \geq
  \sqrt{
    \frac{3\pi}{\theta_{\Vec{u},\Vec{v}} m}
    \log
    \left(
      \binom{n}{k}^{2}
      \left( \frac{6}{\Tau} \right)^{\KO}
      \left( \frac{\NUMLSas}{\Rho'_{1}} \right)
    \right)
  }
\end{align}
\end{subequations}
%
Hence, the parameter is set as
\begin{gather}
\label{pf:lemma:technical:raic:large-scale:eqn:8}
  \Variable{s}_{\Vec{u},\Vec{v}}
  =
  \sqrt{
    \frac{3\pi}{\theta_{\Vec{u},\Vec{v}} m}
    \log
    \left(
      \binom{n}{k}^{2}
      \left( \frac{6}{\Tau} \right)^{\KO}
      \left( \frac{\NUMLSas}{\Rho'_{1}} \right)
    \right)
  }
  \in (0,1)
\end{gather}
%
Then,
\begin{gather}
\label{pf:lemma:technical:raic:large-scale:eqn:9}
  \Variable{\ell}_{\Vec{u},\Vec{v}}
  \leq
  ( 1 + \Variable{s}_{\Vec{u},\Vec{v}} )
  \frac{\theta_{\Vec{u},\Vec{v}} m}{\pi}
  \leq
  \left(
    1
    +
    \sqrt{
      \frac{3\pi}{\theta_{\Vec{u},\Vec{v}} m}
      \log
      \left(
        \binom{n}{k}^{2}
        \left( \frac{6}{\Tau} \right)^{\KO}
        \left( \frac{\NUMLSas}{\Rho'_{1}} \right)
      \right)
    }
  \right)
  \frac{\theta_{\Vec{u},\Vec{v}} m}{\pi}
  \leq
  \frac{2}{\pi} \theta_{\Vec{u},\Vec{v}} m
.\end{gather}
%
On the other hand, using \eqref{pf:lemma:technical:raic:large-scale:eqn:7},
\( \Variable{t}_{\Vec{u},\Vec{v}} \) is determined as follows.
Note that the number subsets
\(
  \Coords{J} \subseteq [n]
\),
\(
  | \Coords{J} | \leq \KO
\),
is at most
\(
  \binom{n}{\KO} 2^{\KO}
\)
(which will be used momentarily in a union bound),
and then observe,
\begin{subequations}
\begin{align}
\label{pf:lemma:technical:raic:large-scale:eqn:7b}
  &
  \Pr
  \left(
  \begin{array}{l}
    \exists
    ( \Vec{u}, \Vec{v} )
    \in
    \Net{\Tau} \times \Net{\Tau},
    \SPACE
    \DistS{\Vec{u}}{\Vec{v}} \geq \Tau,
    \\
    \exists \Coords{J} \subseteq [n],
    \SPACE
    | \Coords{J} | \leq \KO,
    \\
    \left\| ( \Vec{u} - \Vec{v} ) - \hA[\Coords{J}]( \Vec{u}, \Vec{v} ) \right\|_{2}
    \\
    \TAB
    >
    \left|
      1
      -
      \sqrt{\frac{\pi}{2}}
      \frac{\Eta \Variable{\ell}_{\Vec{u},\Vec{v}}}{m}
      \frac{1}{\theta_{\Vec{u},\Vec{v}}}
    \right|
    \DistS{\Vec{u}}{\Vec{v}}
    +
    \frac{3 \Eta \Variable{\ell}_{\Vec{u},\Vec{v}} \Variable{t}_{\Vec{u},\Vec{v}}}{m}
    +
    \frac{2\Eta \sqrt{\KO \Variable{\ell}_{\Vec{u},\Vec{v}}}}{m}
  \end{array}
  \middle|
    \RV{L}_{\Vec{u},\Vec{v}} = \Variable{\ell}_{\Vec{u},\Vec{v}}
    \in
    \left[
      ( 1 \pm \Variable{s}_{\Vec{u},\Vec{v}} ) \frac{\theta_{\Vec{u},\Vec{v}} m}{\pi}
    \right]
  \right)
  \leq
  \Rho''_{1}
  \\
  &\dLn
  \binom{n}{k}^{2}
  \left( \frac{6}{\Tau} \right)^{\KO}
  2^{\KO}
  \binom{n}{\KO}
  6 e^{-\frac{1}{8} \Variable{\ell}_{\Vec{u},\Vec{v}} \Variable{t}_{\Vec{u},\Vec{v}}^{2}}
  \leq
  \Rho''_{1}
  \\
  &\dLn
  \binom{n}{k}^{2} \binom{n}{\KO}
  \left( \frac{12}{\Tau} \right)^{\KO}
  6 e^{-\frac{1}{8} \Variable{\ell}_{\Vec{u},\Vec{v}} \Variable{t}_{\Vec{u},\Vec{v}}^{2}}
  \leq
  \Rho''_{1}
  \\
  &\dLn
  \Variable{t}_{\Vec{u},\Vec{v}}
  \geq
  \sqrt{
    \frac{8}{\Variable{\ell}_{\Vec{u},\Vec{v}}}
    \log
    \left(
      \binom{n}{k}^{2} \binom{n}{\KO}
      \left( \frac{12}{\Tau} \right)^{\KO}
      \left( \frac{\NUMLSat}{\Rho''_{1}} \right)
    \right)
  }
\end{align}
\end{subequations}
%
Thus, the parameter can be set as
\begin{gather}
\label{pf:lemma:technical:raic:large-scale:eqn:10}
  \Variable{t}_{\Vec{u},\Vec{v}}
  =
  \sqrt{
    \frac{8}{\Variable{\ell}_{\Vec{u},\Vec{v}}}
    \log
    \left(
      \binom{n}{k}^{2} \binom{n}{\KO}
      \left( \frac{12}{\Tau} \right)^{\KO}
      \left( \frac{\NUMLSat}{\Rho''_{1}} \right)
    \right)
  }
.\end{gather}
%
Note that
\begin{gather}
\label{pf:lemma:technical:raic:large-scale:eqn:11}
  \frac{\Variable{\ell}_{\Vec{u},\Vec{v}}}{m}
  \leq
  ( 1 + \Variable{s}_{\Vec{u},\Vec{v}} )
  \frac{\theta_{\Vec{u},\Vec{v}} m}{\pi}
  \cdot
  \frac{1}{m}
  =
  \frac{( 1 + \Variable{s}_{\Vec{u},\Vec{v}} )}{\pi}
  \theta_{\Vec{u},\Vec{v}}
  \leq
  \frac{2}{\pi}
  \theta_{\Vec{u},\Vec{v}}
\end{gather}
and
\begin{subequations}
\label{pf:lemma:technical:raic:large-scale:eqn:12:0}
\begin{align}
\label{pf:lemma:technical:raic:large-scale:eqn:12}
  \frac{\Variable{\ell}_{\Vec{u},\Vec{v}} \Variable{t}_{\Vec{u},\Vec{v}}}{m}
  &\leq
  \frac{\Variable{\ell}_{\Vec{u},\Vec{v}}}{m}
  \sqrt{
    \frac{8}{\Variable{\ell}_{\Vec{u},\Vec{v}}}
    \log
    \left(
      \binom{n}{k}^{2} \binom{n}{\KO}
      \left( \frac{12}{\Tau} \right)^{\KO}
      \left( \frac{\NUMLSat}{\Rho''_{1}} \right)
    \right)
  }
  =
  \frac{1}{m}
  \sqrt{
    8 \Variable{\ell}_{\Vec{u},\Vec{v}}
    \log
    \left(
      \binom{n}{k}^{2} \binom{n}{\KO}
      \left( \frac{12}{\Tau} \right)^{\KO}
      \left( \frac{\NUMLSat}{\Rho''_{1}} \right)
    \right)
  }
  \\
  &\leq
  \frac{1}{m}
  \sqrt{
    8
    \cdot
    \frac{2}{\pi}
    \theta_{\Vec{u},\Vec{v}} m
    \log
    \left(
      \binom{n}{k}^{2} \binom{n}{\KO}
      \left( \frac{12}{\Tau} \right)^{\KO}
      \left( \frac{\NUMLSat}{\Rho''_{1}} \right)
    \right)
  }
  \\
  &=
  \sqrt{
    8
    \cdot
    \frac{2}{\pi}
    \frac{\theta_{\Vec{u},\Vec{v}}}{m}
    \log
    \left(
      \binom{n}{k}^{2} \binom{n}{\KO}
      \left( \frac{12}{\Tau} \right)^{\KO}
      \left( \frac{\NUMLSat}{\Rho''_{1}} \right)
    \right)
  }
  \\
  &\leq
  \sqrt{
    \frac{\EDIT{8} \DistS{\Vec{u}}{\Vec{v}}}{m}
    \log
    \left(
      \binom{n}{k}^{2} \binom{n}{\KO}
      \left( \frac{12}{\Tau} \right)^{\KO}
      \left( \frac{\NUMLSat}{\Rho''_{1}} \right)
    \right)
  }
  \\
  &\leq
  \EDIT{\sqrt{8}}
  \cdot
  \sqrt{
    \frac{\DistS{\Vec{u}}{\Vec{v}}}{m}
    \log
    \left(
      \binom{n}{k}^{2} \binom{n}{\KO}
      \left( \frac{12}{\Tau} \right)^{\KO}
      \left( \frac{\NUMLSat}{\Rho''_{1}} \right)
    \right)
  }
\end{align}
\end{subequations}
%
In regard to the parameter \( \Variable{s}_{\Vec{u},\Vec{v}} \), observe
\begin{subequations}
\label{pf:lemma:technical:raic:large-scale:eqn:revision:1}
\begin{align}
  \Variable{s}_{\Vec{u},\Vec{v}}
  \DistS{\Vec{u}}{\Vec{v}}
  &=
  \DistS{\Vec{u}}{\Vec{v}}
  \sqrt{
    \frac{3\pi}{\theta_{\Vec{u},\Vec{v}} m}
    \log
    \left(
      \binom{n}{k}^{2}
      \left( \frac{6}{\Tau} \right)^{\KO}
      \left( \frac{\NUMLSas}{\Rho'_{1}} \right)
    \right)
  }
  \\
  &\leq
  \sqrt{
    \frac{3\pi \DistS{\Vec{u}}{\Vec{v}}}{m}
    \log
    \left(
      \binom{n}{k}^{2}
      \left( \frac{6}{\Tau} \right)^{\KO}
      \left( \frac{\NUMLSas}{\Rho'_{1}} \right)
    \right)
  }
\end{align}
\end{subequations}
%
Then, from the above discussion, with high probability,
\(
  \left\| ( \Vec{u} - \Vec{v} ) - \hA[\Coords{J}]( \Vec{u}, \Vec{v} ) \right\|_{2}
\)
is upper bounded as follows.
\begin{subequations}
\begin{align*}
  &
  \left\| ( \Vec{u} - \Vec{v} ) - \hA[\Coords{J}]( \Vec{u}, \Vec{v} ) \right\|_{2}
  \\
  &\leq
  \left|
    1
    -
    \sqrt{\frac{\pi}{2}}
    \frac{\Eta \Variable{\ell}_{\Vec{u},\Vec{v}}}{m}
    \frac{1}{\theta_{\Vec{u},\Vec{v}}}
  \right|
  \DistS{\Vec{u}}{\Vec{v}}
  +
  \frac{3 \Eta \Variable{\ell}_{\Vec{u},\Vec{v}} \Variable{t}_{\Vec{u},\Vec{v}}}{m}
  +
  \frac{2\Eta \sqrt{\KO \Variable{\ell}_{\Vec{u},\Vec{v}}}}{m}
  \\
  &\EDIT{\dCmt \text{by \EQN \eqref{pf:lemma:technical:raic:large-scale:eqn:4:0}}}
  \\
  &\leq
  \left|
    1
    -
    \sqrt{\frac{\pi}{2}}
    \Eta
    \frac{( 1 + \Variable{s}_{\Vec{u},\Vec{v}} ) \theta_{\Vec{u},\Vec{v}}}{\pi}
    \frac{1}{\theta_{\Vec{u},\Vec{v}}}
  \right|
  \DistS{\Vec{u}}{\Vec{v}}
  +
  \frac{3 \Eta \Variable{\ell}_{\Vec{u},\Vec{v}} \Variable{t}_{\Vec{u},\Vec{v}}}{m}
  +
  \frac{2\Eta \sqrt{\KO \Variable{\ell}_{\Vec{u},\Vec{v}}}}{m}
  \\
  &\EDIT{\dCmt \text{by \EQN \eqref{pf:lemma:technical:raic:large-scale:eqn:11}}}
  \\
  &=
  \Variable{s}_{\Vec{u},\Vec{v}}
  \DistS{\Vec{u}}{\Vec{v}}
  +
  \frac{3 \Eta \Variable{\ell}_{\Vec{u},\Vec{v}} \Variable{t}_{\Vec{u},\Vec{v}}}{m}
  +
  \frac{2\Eta \sqrt{\KO \Variable{\ell}_{\Vec{u},\Vec{v}}}}{m}
  \\
  &\EDIT{\dCmt \text{by canceling terms and using \(  \Eta = \sqrt{2\pi}  \)}}
  \\
  &\leq
  \sqrt{
    \frac{3\pi \DistS{\Vec{u}}{\Vec{v}}}{m}
    \log
    \left(
      \binom{n}{k}^{2}
      \left( \frac{6}{\Tau} \right)^{\KO}
      \left( \frac{\NUMLSas}{\Rho'_{1}} \right)
    \right)
  }
  +
  \frac{3 \Eta \Variable{\ell}_{\Vec{u},\Vec{v}} \Variable{t}_{\Vec{u},\Vec{v}}}{m}
  +
  \frac{2\Eta \sqrt{\KO \Variable{\ell}_{\Vec{u},\Vec{v}}}}{m}
  \\
  &\EDIT{\dCmt \text{by \EQN \eqref{pf:lemma:technical:raic:large-scale:eqn:revision:1}}}
  \\
  &\leq
  \sqrt{
    \frac{3\pi \DistS{\Vec{u}}{\Vec{v}}}{m}
    \log
    \left(
      \binom{n}{k}^{2}
      \left( \frac{6}{\Tau} \right)^{\KO}
      \left( \frac{\NUMLSas}{\Rho'_{1}} \right)
    \right)
  }
  +
  \EDIT{3\sqrt{8} \Eta}
  \cdot
  \sqrt{
    \frac{\DistS{\Vec{u}}{\Vec{v}}}{m}
    \log
    \left(
      \binom{n}{k}^{2} \binom{n}{\KO}
      \left( \frac{12}{\Tau} \right)^{\KO}
      \left( \frac{\NUMLSat}{\Rho''_{1}} \right)
    \right)
  }
  \\ \nonumber
  &\TAB
  +
  \EDIT{\sqrt{8} \Eta \cdot \sqrt{\frac{k}{m} \cdot \frac{2}{\pi} \theta_{\Vec{u},\Vec{v}}}}
  \\
  &\EDIT{\dCmt \text{by \EQNS \eqref{pf:lemma:technical:raic:large-scale:eqn:11} and \eqref{pf:lemma:technical:raic:large-scale:eqn:12:0}}}
  \\
  &\leq
  \sqrt{
    \frac{3\pi \DistS{\Vec{u}}{\Vec{v}}}{m}
    \log
    \left(
      \binom{n}{k}^{2}
      \left( \frac{6}{\Tau} \right)^{\KO}
      \left( \frac{\NUMLSas}{\Rho'_{1}} \right)
    \right)
  }
  +
  \EDIT{3\sqrt{8} \Eta}
  \cdot
  \sqrt{
    \frac{\DistS{\Vec{u}}{\Vec{v}}}{m}
    \log
    \left(
      \binom{n}{k}^{2} \binom{n}{\KO}
      \left( \frac{12}{\Tau} \right)^{\KO}
      \left( \frac{\NUMLSat}{\Rho''_{1}} \right)
    \right)
  }
  \\ \nonumber
  &\TAB
  +
  \EDIT{\sqrt{8} \Eta} \cdot \sqrt{\frac{k \DistS{\Vec{u}}{\Vec{v}}}{m}}
  \\
  &\EDIT{\dCmt \text{by Fact~\ref{fact:technical:theta/d}}}
  \\
  &=
  \sqrt{3\pi}
  \cdot
  \sqrt{
    \frac{\DistS{\Vec{u}}{\Vec{v}}}{m}
    \log
    \left(
      \binom{n}{k}^{2}
      \left( \frac{6}{\Tau} \right)^{\KO}
      \left( \frac{\NUMLSas}{\Rho'_{1}} \right)
    \right)
  }
  +
  \EDIT{3\sqrt{8} \Eta}
  \cdot
  \sqrt{
    \frac{\DistS{\Vec{u}}{\Vec{v}}}{m}
    \log
    \left(
      \binom{n}{k}^{2} \binom{n}{\KO}
      \left( \frac{12}{\Tau} \right)^{\KO}
      \left( \frac{\NUMLSat}{\Rho''_{1}} \right)
    \right)
  }
  \\ \nonumber
  &\TAB
  +
  \EDIT{\sqrt{8} \Eta}
  \cdot
  \sqrt{\frac{k \DistS{\Vec{u}}{\Vec{v}}}{m}}
  \\
  &\EDIT{\dCmt \text{by rearranging terms}}
  \\
  &\leq
  \sqrt{3\pi}
  \cdot
  \sqrt{
    \frac{\DistS{\Vec{u}}{\Vec{v}}}{m}
    \log
    \left(
      \binom{n}{k}^{2} \binom{n}{\KO}
      \left( \frac{12}{\Tau} \right)^{\KO}
      \left( \frac{\NUMLSas}{\Rho'_{1}} \right)
    \right)
  }
  +
  \EDIT{3\sqrt{8} \Eta}
  \cdot
  \sqrt{
    \frac{\DistS{\Vec{u}}{\Vec{v}}}{m}
    \log
    \left(
      \binom{n}{k}^{2} \binom{n}{\KO}
      \left( \frac{12}{\Tau} \right)^{\KO}
      \left( \frac{\NUMLSat}{\Rho''_{1}} \right)
    \right)
  }
  \\ \nonumber
  &\TAB
  +
  \EDIT{\sqrt{8} \Eta}
  \cdot
  \sqrt{\frac{k \DistS{\Vec{u}}{\Vec{v}}}{m}}
  \\
  &\EDIT{\dCmt \because     \log \left( \binom{n}{k}^{2} \left( \frac{6}{\Tau} \right)^{\KO} \left( \frac{\NUMLSas}{\Rho'_{1}} \right) \right) \leq \log \left( \binom{n}{k}^{2} \binom{n}{\KO} \left( \frac{12}{\Tau} \right)^{\KO} \left( \frac{\NUMLSas}{\Rho'_{1}} \right) \right)}
  \\
  &=
  \sqrt{3\pi}
  \cdot
  \sqrt{
    \frac{\DistS{\Vec{u}}{\Vec{v}}}{m}
    \log
    \left(
      \binom{n}{k}^{2} \binom{n}{\KO}
      \left( \frac{12}{\Tau} \right)^{\KO}
      \left( \frac{\NUMLSa}{\Rho_{1}} \right)
    \right)
  }
  +
  \EDIT{3\sqrt{8} \Eta}
  \cdot
  \sqrt{
    \frac{\DistS{\Vec{u}}{\Vec{v}}}{m}
    \log
    \left(
      \binom{n}{k}^{2} \binom{n}{\KO}
      \left( \frac{12}{\Tau} \right)^{\KO}
      \left( \frac{\NUMLSa}{\Rho_{1}} \right)
    \right)
  }
  \\ \nonumber
  &\TAB
  +
  \EDIT{\sqrt{8} \Eta}
  \cdot
  \sqrt{\frac{k \DistS{\Vec{u}}{\Vec{v}}}{m}}
  \\
  &\dCmt \text{set \(  \Rho'_{1} = \frac{1}{4} \Rho_{1},\ \Rho''_{1} = \frac{3}{4} \Rho_{1}  \) such that \(  \Rho'_{1} + \Rho'_{2} = \Rho_{1}  \)}
  \\
  &\leq
  \sqrt{3\pi}
  \cdot
  \sqrt{
    \frac{\DistS{\Vec{u}}{\Vec{v}}}{m}
    \log
    \left(
      \binom{n}{k}^{2} \binom{n}{\KO}
      \left( \frac{12}{\Tau} \right)^{\KO}
      \left( \frac{\NUMLSa}{\Rho_{1}} \right)
    \right)
  }
  +
  \EDIT{3\sqrt{8} \Eta}
  \cdot
  \sqrt{
    \frac{\DistS{\Vec{u}}{\Vec{v}}}{m}
    \log
    \left(
      \binom{n}{k}^{2} \binom{n}{\KO}
      \left( \frac{12}{\Tau} \right)^{\KO}
      \left( \frac{\NUMLSa}{\Rho_{1}} \right)
    \right)
  }
  \\ \nonumber
  &\TAB
  +
  \EDIT{\sqrt{8} \Eta}
  \cdot
  \sqrt{
    \frac{\DistS{\Vec{u}}{\Vec{v}}}{m}
    \log
    \left(
      \binom{n}{k}^{2} \binom{n}{\KO}
      \left( \frac{12}{\Tau} \right)^{\KO}
      \left( \frac{\NUMLSa}{\Rho_{1}} \right)
    \right)
  }
  \\
  &\EDIT{\dCmt \because k \leq k \log \left( \frac{n}{k} \right) \leq \log \binom{n}{k} \leq \log  \left( \binom{n}{k}^{2} \binom{n}{\KO} \left( \frac{12}{\Tau} \right)^{\KO} \left( \frac{\NUMLSa}{\Rho_{1}} \right) \right)}
  \\
  &=
  \left( \sqrt{3\pi} + \EDIT{3\sqrt{8} \Eta} + \EDIT{\sqrt{8} \Eta} \right)
  \sqrt{
    \frac{\DistS{\Vec{u}}{\Vec{v}}}{m}
    \log
    \left(
      \binom{n}{k}^{2} \binom{n}{\KO}
      \left( \frac{12}{\Tau} \right)^{\KO}
      \left( \frac{\NUMLSa}{\Rho_{1}} \right)
    \right)
  }
  \\
  &\EDIT{\dCmt \text{by distributivity}}
  \\
  &=
  \left( \sqrt{3\pi} + \EDIT{12\sqrt{\pi}} + \EDIT{4\sqrt{\pi}} \right)
  \sqrt{
    \frac{\DistS{\Vec{u}}{\Vec{v}}}{m}
    \log
    \left(
      \binom{n}{k}^{2} \binom{n}{\KO}
      \left( \frac{12}{\Tau} \right)^{\KO}
      \left( \frac{\NUMLSa}{\Rho_{1}} \right)
    \right)
  }
  \\
  &\EDIT{\dCmt \text{using \(  \Eta = \sqrt{2\pi}  \)}}
  \\
  &=
  \EDIT{\sqrt{\pi} ( \sqrt{3} + 16 )}
  \sqrt{
    \frac{\DistS{\Vec{u}}{\Vec{v}}}{m}
    \log
    \left(
      \binom{n}{k}^{2} \binom{n}{\KO}
      \left( \frac{12}{\Tau} \right)^{\KO}
      \left( \frac{\NUMLSa}{\Rho_{1}} \right)
    \right)
  }
  \\
  &\EDIT{\dCmt \text{by simplifying terms}}
  \\
  &\leq
  \EDIT{\sqrt{\pi} ( \sqrt{3} + 16 )}
  \sqrt{
    \frac{\DistS{\Vec{u}}{\Vec{v}}}{m}
    \log
    \left(
      \binom{n}{k}^{2} \binom{n}{\KO}
      \left( \frac{12}{\Tau} \right)^{\KO}
      \left( \frac{\UnivConstA}{\Rho} \right)
    \right)
  }
  \\
  &\EDIT{\dCmt \because \Rho_{1} = \frac{\Rho}{2}, \UnivConstA = 16}
\end{align*}
\end{subequations}
%
Recall that
\begin{align*}
  \Tau = \TauValue
\end{align*}
and
\begin{align*}
  m
  &\geq
  \CMPLXRAICOne
  \\
  &=
  \CMPLXRAICOneX
  \\
  &=
  \frac{\UnivConstB \UnivConstD}{\DDelta}
  \log
  \left(
    \binom{n}{k}^{2} \binom{n}{\KO}
    \left( \frac{12}{\Tau} \right)^{\KO}
    \left( \frac{\UnivConstA}{\Rho} \right)
  \right)
.\end{align*}
%
Thus, continuing the above derivation,
\begin{align*}
  &\left\| ( \Vec{u} - \Vec{v} ) - \hA[\Coords{J}]( \Vec{u}, \Vec{v} ) \right\|_{2}
  \\
  &\leq
  \EDIT{\sqrt{\pi}( \sqrt{3} + 16 )}
  \sqrt{
    \frac{\DistS{\Vec{u}}{\Vec{v}}}{m}
    \log
    \left(
      \binom{n}{k}^{2} \binom{n}{\KO}
      \left( \frac{12}{\Tau} \right)^{\KO}
      \left( \frac{\UnivConstA}{\Rho} \right)
    \right)
  }
  \\
  &\leq
  \EDIT{\sqrt{\pi}( \sqrt{3} + 16 )}
  \sqrt{\frac{\DDelta \DistS{\Vec{u}}{\Vec{v}}}{\UnivConstB \UnivConstD}}
  \\
  &=
  \EDIT{\sqrt{\frac{\pi}{\UnivConstB \UnivConstD}}( \sqrt{3} + 16 )}
  \sqrt{\DDelta \DistS{\Vec{u}}{\Vec{v}}}
\end{align*}
%
In short, the above step yields
\begin{gather}
  \left\| ( \Vec{u} - \Vec{v} ) - \hA[\Coords{J}]( \Vec{u}, \Vec{v} ) \right\|_{2}
  \leq
  \UnivConstb_{1}
  \sqrt{\DDelta \DistS{\Vec{u}}{\Vec{v}}}
\end{gather}
where the universal constant is set as
\begin{gather}
  \UnivConstb_{1}
  =
  \EDIT{\sqrt{\frac{\pi}{\UnivConstB \UnivConstD}}( \sqrt{3} + 16 )}
.\end{gather}
%
Then, the lemma's universal result follows---with probability at least
\(
  1 - \Rho_{1}
\),
\begin{gather}
  \left\| ( \Vec{u} - \Vec{v} ) - \hA[\Coords{J}]( \Vec{u}, \Vec{v} ) \right\|_{2}
  \leq
  \UnivConstb_{1}
  \sqrt{\DDelta \DistS{\Vec{u}}{\Vec{v}}}
\end{gather}
uniformly for all
\(
  ( \Vec{u}, \Vec{v} )
  \in
  \Net{\Tau} \times \Net{\Tau}
\),
\(
  \DistS{\Vec{u}}{\Vec{v}} \geq \Tau
\),
and all
\(
  \Coords{J} \subseteq [n]
\),
\(
  | \Coords{J} | \leq \KO
\).
\end{proof}


\subsection{``Small distances'' regime}
\label{outline:technical:pf|>small-scale}

In contrast to the regime in Section \ref{outline:technical:pf|>large-scale}, the regime
under consideration in this section looks at points in the \( \Tau \)-ball around every
\(  k  \)-sparse unit vector, \(  \Vec{u} \in \SparseSphereSubspace{k}{n}  \).
\EDIT{Note that here, the argument will not use the constructed  \(  \Tau  \)-net, \(  \Net{\Tau}  \), but rather provide a uniform result for all of \(  \SparseSphereSubspace{k}{n}  \).}
Lemma \ref{lemma:technical:raic:small-scale} states the formal result.

\begin{lemma}
\label{lemma:technical:raic:small-scale}
Let
\(
  \UnivConstb_{2} > 0
\)
be a universal constant.
\EDIT{Define \(  \GammaX, m > 0  \) as in \THEOREM \ref{thm:technical:raic:modified}.}
Fix
\(
  \DDelta, \Rho_{2} \in (0,1)
\),
where
\(  \Rho_{2} \defeq \frac{\Rho}{2}  \),
and let
\(
  \Tau \defeq \TauValue
\).
%
%
Uniformly with probability at least
\(
  1 - \Rho_{2}
\),
\begin{gather}
\label{eqn:technical:raic:small-scale}
  \left\| ( \Vec{x} - \Vec{u} ) - \hA[\Coords{J}]( \Vec{x}, \Vec{u} ) \right\|_{2}
  \leq
  \UnivConstb_{2} \DDelta
\end{gather}
for all
\(
  \Vec{u} \in \SparseSphereSubspace{k}{n}
\),
for all
\(
  \Vec{x} \in \BallSparseSphere{\Tau}( \Vec{u} )
\),
and for all
\(
  \Coords{J} \subseteq [n]
\),
\(
  | \Coords{J} | \leq \KO
\).
\end{lemma}
%
\BEGINEDIT
\begin{proof}
{\LEMMA \ref{lemma:technical:raic:small-scale}}
\label{pf:lemma:technical:raic:small-scale}
Fix
\(
  \Vec{u} \in \SparseSphereSubspace{k}{n}
\),
\(
  \Vec{x} \in \BallSparseSphere{\Tau}( \Vec{u} )
\),
and
\(
  \Coords{J} \subseteq [n]
\),
\(
  | \Coords{J} | \leq \KO
\),
arbitrarily.
The definition of \(  \BallSparseSphere{\Tau}( \Vec{u} )  \) directly implies that
\(  \| \Vec{x} - \Vec{u} \|_{2} \leq \Tau  \).
%
Combining this observation with the triangle inequality yields the following:
\begin{gather}
\label{pf:lemma:technical:raic:small-scale:eqn:2}
  \left\| ( \Vec{x} - \Vec{u} ) - \hA[\Coords{J}]( \Vec{x}, \Vec{u} ) \right\|_{2}
  \leq
  \left\| \Vec{x} - \Vec{u} \right\|_{2}
  +
  \left\| \hA[\Coords{J}]( \Vec{x}, \Vec{u} ) \right\|_{2}
  \leq
  \Tau + \left\| \hA[\Coords{J}]( \Vec{x}, \Vec{u} ) \right\|_{2}
.\end{gather}
%
Hence, the primary task in proving the lemma is controlling the rightmost term in \eqref{pf:lemma:technical:raic:small-scale:eqn:2}, \(  \| \hA[\Coords{J}]( \Vec{x}, \Vec{u} ) \|_{2}  \), uniformly with high probability for all
\(
  \Vec{u} \in \SparseSphereSubspace{k}{n}
\),
\(
  \Vec{x} \in \BallSparseSphere{\Tau}( \Vec{u} )
\),
and
\(
  \Coords{J} \subseteq [n]
\),
\(
  | \Coords{J} | \leq \KO
\).
The overall approach is similar to that taken in the proof of \LEMMA \ref{lemma:technical:raic:large-scale} for the ``large distances'' regime with some key differences in counting the number of sign mismatches and applying an appropriate union bound to obtain a uniform result.
Specifically, in place of \LEMMA \ref{lemma:technical:concentration-ineq:counting}, we will use \LEMMA \ref{lemma:technical:concentration-ineq:lbe-local-deviations:union} to upper bound the number of Gaussian measurements, \(  \MeasVec^{(i)}  \), on which \(  \sgn( \langle \MeasVec^{(i)}, \Vec{u} \rangle )  \) and \(  \sgn( \langle \MeasVec^{(i)}, \Vec{v} \rangle )  \) differ.
Moreover, rather than taking a union bound over a \(  \Tau  \)-net, we will consider the cardinality of the image of \(  \hA[\Coords{J}]  \) for each \(  \Coords{J} \subseteq [n]  \), \(  | \Coords{J} | \leq \KO  \), and union bound accordingly.
%
%
%
Let us begin by introducing some notations and discussing this more formally.
%
\par 
%
Let \(  \Set{I} \subseteq 2^{[m]}  \) be the (random) set of all possible subsets of \(  [m]  \) indexing mismatches for at least one vector in \(  \SparseSphereSubspace{k}{n}  \) and one vector at most \(  \Tau  \)-far away from it---or more formally,
\begin{gather}
\label{pf:lemma:technical:raic:small-scale:eqn:I:def}
  \Set{I} \defeq
  \left\{
    \{ i \in [m] : \Sgn{}( \langle \MeasVec^{(i)}, \Vec{x} \rangle ) \neq \Sgn{}( \langle \MeasVec^{(i)}, \Vec{u} \rangle )
    \} :
    \Vec{u} \in \SparseSphereSubspace{k}{n}, \Vec{x} \in \BallSparseSphere{\Tau}( \Vec{u} )
  \right\}
.\end{gather}
%
Additionally, define the random variable
\begin{gather}
\label{pf:lemma:technical:raic:small-scale:eqn:q:def}
  \RV{Q} \defeq \max_{\Coords{I} \in \Set{I}} |I|
.\end{gather}
%
As discussed earlier, the random variable \(  \RV{L}_{\Vec{x},\Vec{u}} \defeq \| \frac{1}{2} ( \sgn( \MeasMat \Vec{x} ) - \sgn( \MeasMat \Vec{u} ) ) \|_{0}  \) is equivalently given by \(  \RV{L}_{\Vec{x},\Vec{u}} = \| \I{ \sgn( \MeasMat \Vec{x} ) \neq \sgn( \MeasMat \Vec{u} ) } \|_{0}  \).
Therefore,
\begin{gather}
\label{pf:lemma:technical:raic:small-scale:eqn:q:def:b}
  \RV{Q} = \sup_{\substack{\Vec{u} \in \SparseSphereSubspace{k}{n}, \\ \Vec{x} \in \BallSparseSphere{\Tau}( \Vec{u} )}} \RV{L}_{\Vec{x},\Vec{u}}
.\end{gather}
%
Recall that as in \EQN \eqref{eqn:intro:techniques:hA:rewritten}, \(  \hA  \) can be written as follows:
\begin{align}
\label{pf:lemma:technical:raic:small-scale:eqn:hA:rewritten}
  \hA( \Vec{x}, \Vec{u} )
  &=
  \frac{\sqrt{2\pi}}{m}
  \sum_{i=1}^{m}
  \MeasVec^{(i)}
  \cdot
  \Sgn{}( \langle \MeasVec^{(i)}, \Vec{x} \rangle )
  \cdot
  \I{}(\Sgn{}( \langle \MeasVec^{(i)}, \Vec{x} \rangle )
     \neq \Sgn{}( \langle \MeasVec^{(i)}, \Vec{u} \rangle ))
\end{align}
for
\(
  \Vec{u}, \Vec{x} \in \SparseSphereSubspace{k}{n}
\)
and
\(
  \Coords{J} \subseteq [n]
\),
and hence,
\begin{align*}
  \hA[\Coords{J}]( \Vec{x}, \Vec{u} )
  &=
  \ThresholdSet{\Supp( \Vec{x} ) \cup \Supp( \Vec{u} ) \cup \Coords{J}}(
  \frac{\sqrt{2\pi}}{m}
  \sum_{i=1}^{m}
  \MeasVec^{(i)}
  \cdot
  \Sgn{}( \langle \MeasVec^{(i)}, \Vec{x} \rangle )
  \cdot
  \I{}(\Sgn{}( \langle \MeasVec^{(i)}, \Vec{x} \rangle )
     \neq \Sgn{}( \langle \MeasVec^{(i)}, \Vec{u} \rangle ))
  )
  \\
  &=
  \frac{\sqrt{2\pi}}{m}
  \sum_{i=1}^{m}
  \ThresholdSet{\Supp( \Vec{x} ) \cup \Supp( \Vec{u} ) \cup \Coords{J}}{}( \MeasVec^{(i)} )
  \cdot
  \Sgn{}( \langle \MeasVec^{(i)}, \Vec{x} \rangle )
  \cdot
  \I{}(\Sgn{}( \langle \MeasVec^{(i)}, \Vec{x} \rangle )
     \neq \Sgn{}( \langle \MeasVec^{(i)}, \Vec{u} \rangle ))
\TagEqn\label{pf:lemma:technical:raic:small-scale:eqn:hAJ:rewritten}
.\end{align*}
%
From \EQN \eqref{pf:lemma:technical:raic:small-scale:eqn:hAJ:rewritten}, it is clear that upon fixing the Gaussian vectors, \(  \MeasVec^{(i)}  \), \(  i \in [m]  \), the image of \(  \hA[\Coords{J}]  \) can only take finitely many values for each of the (finitely many) choices of \(  \Coords{J}  \).
As such, writing
\(  \Set{Y} \defeq \bigcup_{\Coords{J} \subseteq [m]: |\Coords{J}| \leq \kO} \hA[\Coords{J}] [ \Set{S} ]  \),
where
\(  \Set{S} \defeq \{ ( \Vec{x}, \Vec{u} ) : \Vec{u} \in \SparseSphereSubspace{k}{n}, \Vec{x} \in \BallSparseSphere{\Tau}( \Vec{u} ) \}  \),
the following claim bounds \(  | \Set{Y} |  \).
%
%
\begin{claim}
\label{pf:lemma:technical:raic:small-scale:claim:1}
Fix \(  \MeasVec^{(i)} \in \R^{n}  \), \(  i \in [m]  \).
%
%
Suppose \(  \RV{Q} = \Variable{q}  \).
Then,
\begin{gather}
\label{pf:lemma:technical:raic:small-scale:claim:1:eqn:1}
  | \Set{Y} |
  \leq
  \EDITX{\left( \frac{2e m}{\Variable{q}} \right)^{\Variable{q}} \left( \frac{en}{\kOX} \right)^{\kOX}}
.\end{gather}
\end{claim}
%
\begin{subproof}
{\CLAIM \ref{pf:lemma:technical:raic:small-scale:claim:1}}
Looking at \EQN \eqref{pf:lemma:technical:raic:small-scale:eqn:hAJ:rewritten}, the cardinality of \(  \Set{Y}  \) can be upper bounded by considering the set of all possible subsets that can comprise \(  \Supp( \Vec{x} ) \cup \Supp( \Vec{u} ) \cup \Coords{J}  \) and the set of all vectors that can be taken by \(  \sgn( \MeasMat \Vec{x} ) \odot \I{}( \sgn( \MeasMat \Vec{x} ) \neq \sgn( \MeasMat \Vec{u} ) )  \), and then multiplying the sizes of these two sets.
More concretely, let
\(  \Set{Y}_{1}, \Set{Y}'_{1} \subseteq 2^{[n]}  \) and
\(  \Set{Y}_{2}, \Set{Y}'_{2} \subseteq \{ -1, 0, 1 \}^{m} \)
be the sets given by
\begin{gather}
  \Set{Y}_{1} \defeq \{ \Supp(\Vec{u}) \cup \Supp( \Vec{v} ) \cup \Coords{J} : \Vec{u} \in \SparseSphereSubspace{k}{n}, \Vec{x} \in \BallSparseSphere{\Tau}( \Vec{u} ), \Coords{J} \subseteq [n], | \Coords{J} | \leq \kO \}
  ,\\
  \Set{Y}'_{1} \defeq \{ \Coords{J'} \subseteq [n] : 1 \leq | \Coords{J'} | \leq \kOX \}
  ,\\
  \Set{Y}_{2} \defeq \{ \sgn( \MeasMat \Vec{x} ) \odot \I{}( \sgn( \MeasMat \Vec{x} ) \neq \sgn( \MeasMat \Vec{u} ) ) : \Vec{u} \in \SparseSphereSubspace{k}{n}, \Vec{x} \in \BallSparseSphere{\Tau}( \Vec{u} ) \}
  ,\\
  \Set{Y}'_{2} \defeq \{ \Vec{w} \in \{ -1, 0, 1 \}^{m} : \| \Vec{w} \|_{0} \leq \Variable{q} \}
.\end{gather}
%
Note that
\(  \Set{Y}_{1} \subseteq \Set{Y}'_{1}  \) and
\(  \Set{Y}_{2} \subseteq \Set{Y}'_{2}  \),
where the latter holds due to the claim's assumption that \(  \RV{Q} = \Variable{q}  \).
These properties imply that
\(  | \Set{Y}_{1} | \leq | \Set{Y}'_{1} |  \) and
\(  | \Set{Y}_{2} | \leq | \Set{Y}'_{2} |  \),
where
\begin{gather}
  | \Set{Y}'_{1} |
  = \sum_{\ell=1}^{\kOX} \binom{n}{\ell}
  \leq \EDITX{\left( \frac{en}{\kOX} \right)^{\kOX}}
  ,\\
  | \Set{Y}'_{2} |
  = \sum_{\ell=0}^{\Variable{q}} 2^{\ell} \binom{m}{\ell}
  \leq 2^{\Variable{q}} \left( \frac{e m}{\Variable{q}} \right)^{\Variable{q}}
  = \left( \frac{2e m}{\Variable{q}} \right)^{\Variable{q}}
.\end{gather}
%
From the discussion earlier, the claim's bound on \(  | \Set{Y} |  \) now follows:
\begin{gather}
  | \Set{Y} |
  \leq | \Set{Y}_{1} | | \Set{Y}_{2} |
  \leq | \Set{Y}'_{1} | | \Set{Y}'_{2} |
  \EDITX{\leq \left( \frac{2e m}{\Variable{q}} \right)^{\Variable{q}} \left( \frac{en}{\kOX} \right)^{\kOX}}
.\end{gather}
\end{subproof}
%
Per \CLAIM \ref{pf:lemma:technical:raic:small-scale:claim:1}, it is possible to bound \(  \| \hA[\Coords{J}]( \Vec{x}, \Vec{u} ) \|_{2}  \) for an arbitrary choice of \(  \Vec{u} \in \SparseSphereSubspace{k}{n}  \), \(  \Vec{x} \in \BallSparseSphere{\Tau}( \Vec{u} )  \), and \(  \Coords{J} \subseteq [n]  \), \(  | \Coords{J} | \leq \KO  \), and subsequently union bound over \(  \Set{Y}  \).
Given the above discussion, the remainder of the proof of \LEMMA \ref{lemma:technical:raic:small-scale} will be carried out as follows.
%
\Enum[\label{pf:lemma:technical:raic:small-scale:enum:1:i}]{i}
First, arbitrarily fixing \(  \Vec{u} \in \SparseSphereSubspace{k}{n}  \), \(  \Vec{x} \in \BallSparseSphere{\Tau}( \Vec{u} )  \), and \(  \Coords{J} \subseteq [n]  \), \(  | \Coords{J} | \leq \KO  \), a probabilistic upper bound on \(  \| \hA[\Coords{J}]( \Vec{x}, \Vec{u} ) \|_{2}  \) will be derived.
This will subsequently be extended to all \(  \Vec{u} \in \SparseSphereSubspace{k}{n}  \), \(  \Vec{x} \in \BallSparseSphere{\Tau}( \Vec{u} )  \), and \(  \Coords{J} \subseteq [n]  \), \(  | \Coords{J} | \leq \KO  \) via a union bound over \(  \Set{Y}  \).
The result obtained in this step will be established in terms of the random variable \(  \RV{Q}  \), defined in \EQN \eqref{pf:lemma:technical:raic:small-scale:eqn:q:def}.
\Enum[\label{pf:lemma:technical:raic:small-scale:enum:1:ii}]{ii}
Then, the random variable \(  \RV{Q}  \) will be upper bounded with high probability via \LEMMA \ref{lemma:technical:concentration-ineq:lbe-local-deviations:union}, leading to further bounds on \(  | \Set{Y} |  \) in \EQN \eqref{pf:lemma:technical:raic:small-scale:claim:1:eqn:1} and on \(  \| \hA[\Coords{J}]( \Vec{x}, \Vec{u} ) \|_{2}  \) from \STEP \ref{pf:lemma:technical:raic:small-scale:enum:1:i}.
\Enum[\label{pf:lemma:technical:raic:small-scale:enum:1:iii}]{iii}
Finally, via appropriate union bounds together with the results obtained in \STEPS \ref{pf:lemma:technical:raic:small-scale:enum:1:i} and \ref{pf:lemma:technical:raic:small-scale:enum:1:ii}, the lemma will follow.
%
\paragraph{\STEP \ref{pf:lemma:technical:raic:small-scale:enum:1:i}.} 
%
Fix \(  \Vec{u} \in \SparseSphereSubspace{k}{n}  \), \(  \Vec{x} \in \BallSparseSphere{\Tau}( \Vec{u} )  \), and \(  \Coords{J} \subseteq [n]  \), \(  | \Coords{J} | \leq \KO  \).
As in the proof of Lemma \ref{lemma:technical:raic:large-scale},
the function \( \hA[\Coords{J}] \) can be expressed using orthogonal projections as
\begin{gather}
\label{pf:lemma:technical:raic:small-scale:eqn:3}
  \hA[\Coords{J}]( \Vec{x}, \Vec{u} )
  =
  \left\langle
  \frac{\Vec{x} - \Vec{u}}{\left\| \Vec{x} - \Vec{u} \right\|_{2}},
  \hA[\Coords{J}]( \Vec{u}, \Vec{v} )
  \right\rangle
  \frac{\Vec{x} - \Vec{u}}{\left\| \Vec{x} - \Vec{u} \right\|_{2}}
  +
  \left\langle
  \frac{\Vec{x} + \Vec{u}}{\left\| \Vec{x} + \Vec{u} \right\|_{2}},
  \hA[\Coords{J}]( \Vec{x}, \Vec{u} )
  \right\rangle
  \frac{\Vec{x} + \Vec{v}}{\left\| \Vec{x} + \Vec{u} \right\|_{2}}
  +
  \gA[\Coords{J}]( \Vec{x}, \Vec{u} )
,\end{gather}
and by the triangle inequality
\begin{subequations}
\label{pf:lemma:technical:raic:small-scale:eqn:4}
\begin{align}
  &
  \left\| \hA[\Coords{J}]( \Vec{x}, \Vec{u} ) \right\|_{2}
  \\
  &=
  \left\|
    \left\langle
    \frac{\Vec{x} - \Vec{u}}{\left\| \Vec{x} - \Vec{u} \right\|_{2}},
    \hA[\Coords{J}]( \Vec{u}, \Vec{v} )
    \right\rangle
    \frac{\Vec{x} - \Vec{u}}{\left\| \Vec{x} - \Vec{u} \right\|_{2}}
    +
    \left\langle
    \frac{\Vec{x} + \Vec{u}}{\left\| \Vec{x} + \Vec{u} \right\|_{2}},
    \hA[\Coords{J}]( \Vec{x}, \Vec{u} )
    \right\rangle
    \frac{\Vec{x} + \Vec{v}}{\left\| \Vec{x} + \Vec{u} \right\|_{2}}
    +
    \gA[\Coords{J}]( \Vec{x}, \Vec{u} )
  \right\|_{2}
  \\
  &\leq
  \left\|
    \left\langle
    \frac{\Vec{x} - \Vec{u}}{\left\| \Vec{x} - \Vec{u} \right\|_{2}},
    \hA[\Coords{J}]( \Vec{u}, \Vec{v} )
    \right\rangle
    \frac{\Vec{x} - \Vec{u}}{\left\| \Vec{x} - \Vec{u} \right\|_{2}}
  \right\|_{2}
  +
  \left\|
    \left\langle
    \frac{\Vec{x} + \Vec{u}}{\left\| \Vec{x} + \Vec{u} \right\|_{2}},
    \hA[\Coords{J}]( \Vec{x}, \Vec{u} )
    \right\rangle
    \frac{\Vec{x} + \Vec{v}}{\left\| \Vec{x} + \Vec{u} \right\|_{2}}
   \right\|_{2}
   +
   \left\|
    \gA[\Coords{J}]( \Vec{x}, \Vec{u} )
  \right\|_{2}
  \\
  &=
  \left|
    \left\langle
    \frac{\Vec{x} - \Vec{u}}{\left\| \Vec{x} - \Vec{u} \right\|_{2}},
    \hA[\Coords{J}]( \Vec{u}, \Vec{v} )
    \right\rangle
  \right|
  +
  \left|
    \left\langle
    \frac{\Vec{x} + \Vec{u}}{\left\| \Vec{x} + \Vec{u} \right\|_{2}},
    \hA[\Coords{J}]( \Vec{x}, \Vec{u} )
    \right\rangle
   \right|
   +
   \left\|
    \gA[\Coords{J}]( \Vec{x}, \Vec{u} )
  \right\|_{2}
.\end{align}
\end{subequations}
%
\par
Recall the concentration inequalities provided by
Lemma \ref{lemma:technical:concentration-ineq:(u,v)}:
\begin{gather}
  \label{pf:lemma:technical:raic:small-scale:eqn:5:1}
  \Pr
  \left(
      \left\langle
        \frac{\Vec{x} - \Vec{u}}{\left\| \Vec{x} - \Vec{u} \right\|_{2}},
        \frac{1}{\Eta}
        \hA[\Coords{J}]( \Vec{x}, \Vec{u} )
      \right\rangle
      -
      \sqrt{\frac{\pi}{2}}
      \frac{\Variable{\ell}_{\Vec{x},\Vec{u}}}{m}
      \frac{\DistS{\Vec{x}}{\Vec{u}}}{\theta_{\Vec{x},\Vec{u}}}
    \geq
    \frac{\Variable{\ell}_{\Vec{x},\Vec{u}} \Variable{t}_{\Vec{x},\Vec{u}}}{m}
    \middle|
    \RV{L}_{\Vec{x},\Vec{u}} = \Variable{\ell}_{\Vec{x},\Vec{u}}
  \right)
  \leq
  e^{-\frac{1}{2} \Variable{\ell}_{\Vec{x},\Vec{u}} \Variable{t}_{\Vec{x},\Vec{u}}^{2}}
  ,\\
  \label{pf:lemma:technical:raic:small-scale:eqn:5:2}
  \Pr
  \left(
    \left|
      \left\langle
        \frac{\Vec{x} + \Vec{u}}{\left\| \Vec{x} + \Vec{u} \right\|_{2}},
        \frac{1}{\Eta}
        \hA[\Coords{J}]( \Vec{x}, \Vec{u} )
      \right\rangle
    \right|
    \geq
    \frac{\Variable{\ell}_{\Vec{x},\Vec{u}} \Variable{t}_{\Vec{x},\Vec{u}}}{m}
    \middle|
    \RV{L}_{\Vec{x},\Vec{u}} = \Variable{\ell}_{\Vec{x},\Vec{u}}
  \right)
  \leq
  2 e^{-\frac{1}{2} \Variable{\ell}_{\Vec{x},\Vec{u}} \Variable{t}_{\Vec{x},\Vec{u}}^{2}}
  ,\\
  \label{pf:lemma:technical:raic:small-scale:eqn:5:3}
  \Pr
  \left(
    \left\|
      \frac{1}{\Eta}
      \gA[\Coords{J}]( \Vec{x}, \Vec{u} )
    \right\|_{2}
    \geq
    \frac{2 \sqrt{\KO \Variable{\ell}_{\Vec{x},\Vec{u}}}}{m}
    +
    \frac{\Variable{\ell}_{\Vec{x},\Vec{u}} \Variable{t}_{\Vec{x},\Vec{u}}}{m}
    \middle|
    \RV{L}_{\Vec{x},\Vec{u}} = \Variable{\ell}_{\Vec{x},\Vec{u}}
  \right)
  \leq
  2 e^{-\frac{1}{8} \Variable{\ell}_{\Vec{x},\Vec{u}} \Variable{t}_{\Vec{x},\Vec{u}}^{2}}
,\end{gather}
\ORIG{where \( \Vec{R}_{\Vec{x},\Vec{u}} \) and \( \RV{L}_{\Vec{x},\Vec{u}} \) are random variables
defined as
\(
  \Vec{R}_{\Vec{x},\Vec{u}}
  =
  ( \Vec*{R}_{1;\Vec{x},\Vec{u}}, \dots, \Vec*{R}_{m;\Vec{x},\Vec{u}} )
  =
  \frac{1}{2}
  \left( \sgn( \MeasMat \Vec{x} ) - \sgn( \MeasMat \Vec{u} ) \right)
\)
and
\(
  \RV{L}_{\Vec{x},\Vec{u}} = \left\| \Vec{R}_{\Vec{x},\Vec{u}} \right\|_{0}
\),
and
\(
  \Vec{r} \in \{-1,0,1\}^{m}
\),
\(
  \Variable{\ell}_{\Vec{x},\Vec{u}} \in [m]
\).}%
\EDIT{where \(  \RV{L}_{\Vec{x},\Vec{u}} = \| \frac{1}{2} \left( \sgn( \MeasMat \Vec{x} ) - \sgn( \MeasMat \Vec{u} ) \right) \|_{0}  \) and \(  \Variable{\ell}_{\Vec{x},\Vec{u}} \in [m]  \).}
Note that \EQN \eqref{pf:lemma:technical:raic:small-scale:eqn:5:1} uses the one-sided version of \EQN \eqref{eqn:technical:concentration-ineq:(u,v):u-v}, which can be seen in the proof of \LEMMA \ref{lemma:technical:concentration-ineq:(u,v)}.
Additionally, \EQN \eqref{pf:lemma:technical:raic:small-scale:eqn:5:1} can be replaced by
\begin{align}
  \nonumber
  &
  \Pr
  \left(
    \left\langle
      \frac{\Vec{x} - \Vec{u}}{\left\| \Vec{x} - \Vec{u} \right\|_{2}},
      \frac{1}{\Eta}
      \hA[\Coords{J}]( \Vec{x}, \Vec{u} )
    \right\rangle
    \geq
    \sqrt{\frac{\pi}{2}}
    \frac{\Variable{\ell}_{\Vec{x},\Vec{u}}}{m}
    \frac{\DistS{\Vec{x}}{\Vec{u}}}{\theta_{\Vec{x},\Vec{u}}}
    +
    \frac{\Variable{\ell}_{\Vec{x},\Vec{u}} \Variable{t}_{\Vec{x},\Vec{u}}}{m}
    \middle|
    \RV{L}_{\Vec{x},\Vec{u}} = \Variable{\ell}_{\Vec{x},\Vec{u}}
  \right)
  \leq
  e^{-\frac{1}{2} \Variable{\ell}_{\Vec{x},\Vec{u}} \Variable{t}_{\Vec{x},\Vec{u}}^{2}}
  \\ \nonumber 
  &\dLn
  \Pr
  \left(
    \left\langle
      \frac{\Vec{x} - \Vec{u}}{\left\| \Vec{x} - \Vec{u} \right\|_{2}},
      \frac{1}{\Eta}
      \hA[\Coords{J}]( \Vec{x}, \Vec{u} )
    \right\rangle
    \geq
    \sqrt{\frac{\pi}{2}}
    \frac{\Variable{\ell}_{\Vec{x},\Vec{u}}}{m}
    +
    \frac{\Variable{\ell}_{\Vec{x},\Vec{u}} \Variable{t}_{\Vec{x},\Vec{u}}}{m}
    \middle|
    \RV{L}_{\Vec{x},\Vec{u}} = \Variable{\ell}_{\Vec{x},\Vec{u}}
  \right)
  \leq
  e^{-\frac{1}{2} \Variable{\ell}_{\Vec{x},\Vec{u}} \Variable{t}_{\Vec{x},\Vec{u}}^{2}}
  \\ \label{pf:lemma:technical:raic:small-scale:eqn:5b:1:3}
  &\dLn
  \Pr
  \left(
    \left\langle
      \frac{\Vec{x} - \Vec{u}}{\left\| \Vec{x} - \Vec{u} \right\|_{2}},
      \frac{1}{\Eta}
      \hA[\Coords{J}]( \Vec{x}, \Vec{u} )
    \right\rangle
    \geq
    \left( \sqrt{\frac{\pi}{2}} + \Variable{t}_{\Vec{x},\Vec{u}} \right)
    \frac{\Variable{\ell}_{\Vec{x},\Vec{u}}}{m}
    \middle|
    \RV{L}_{\Vec{x},\Vec{u}} = \Variable{\ell}_{\Vec{x},\Vec{u}}
  \right)
  \leq
  e^{-\frac{1}{2} \Variable{\ell}_{\Vec{x},\Vec{u}} \Variable{t}_{\Vec{x},\Vec{u}}^{2}}
.\end{align}
%
Therefore, due to \EQN \eqref{pf:lemma:technical:raic:small-scale:eqn:4} in combination with \EQNS \eqref{pf:lemma:technical:raic:small-scale:eqn:5:2}, \eqref{pf:lemma:technical:raic:small-scale:eqn:5:3}, and \eqref{pf:lemma:technical:raic:small-scale:eqn:5b:1:3}, given \(  \RV{L}_{\Vec{x},\Vec{u}} = \Variable{\ell}_{\Vec{x},\Vec{u}}  \), the norm of \(  \hA[\Coords{J}]( \Vec{x}, \Vec{u} )  \) is bounded from above by
\begin{align}
  \nonumber
  \left\| \hA[\Coords{J}]( \Vec{x}, \Vec{u} ) \right\|_{2}
  &\leq
  \left( \sqrt{\frac{\pi}{2}} + \Variable{t}_{\Vec{x},\Vec{u}} \right)
  \frac{\Variable{\ell}_{\Vec{x},\Vec{u}}}{m}
  +
  \frac{\Variable{\ell}_{\Vec{x},\Vec{u}} \Variable{t}_{\Vec{x},\Vec{u}}}{m}
  +
  \frac{2 \sqrt{\KO \Variable{\ell}_{\Vec{x},\Vec{u}}}}{m}
  +
  \frac{\Variable{\ell}_{\Vec{x},\Vec{u}} \Variable{t}_{\Vec{x},\Vec{u}}}{m}
  \\ \nonumber
  &=
  \left( \sqrt{\frac{\pi}{2}} + 3 \Variable{t}_{\Vec{x},\Vec{u}} \right)
  \frac{\Variable{\ell}_{\Vec{x},\Vec{u}}}{m}
  +
  \frac{2 \sqrt{\KO \Variable{\ell}_{\Vec{x},\Vec{u}}}}{m}
  \\
  &=
  \left( 3 \Variable{t}_{\Vec{x},\Vec{u}} + \sqrt{\frac{\pi}{2}} \right)
  \frac{\Variable{\ell}_{\Vec{x},\Vec{u}}}{m}
  +
  \frac{2 \sqrt{\KO \Variable{\ell}_{\Vec{x},\Vec{u}}}}{m}
\label{pf:lemma:technical:raic:small-scale:eqn:4b}
\end{align}
with probability at least
\begin{gather}
\label{pf:lemma:technical:raic:small-scale:eqn:5b}
  1
  - e^{-\frac{1}{2} \Variable{\ell}_{\Vec{x},\Vec{u}} \Variable{t}_{\Vec{x},\Vec{u}}^{2}}
  - 2 e^{-\frac{1}{2} \Variable{\ell}_{\Vec{x},\Vec{u}} \Variable{t}_{\Vec{x},\Vec{u}}^{2}}
  - 2 e^{-\frac{1}{8} \Variable{\ell}_{\Vec{x},\Vec{u}} \Variable{t}_{\Vec{x},\Vec{u}}^{2}}
  \geq
  1 - 5 e^{-\frac{1}{8} \Variable{\ell}_{\Vec{x},\Vec{u}} \Variable{t}_{\Vec{x},\Vec{u}}^{2}}
.\end{gather}
%
Due to the conditioning in the above concentration bounds, we will need to have a handle on the
random variable
\(
  \RV{L}_{\Vec{x},\Vec{u}}
\).
%
By \EQN \eqref{pf:lemma:technical:raic:small-scale:eqn:q:def:b}, it is always the case that \(  \RV{L}_{\Vec{x},\Vec{u}} \leq \RV{Q}  \) for any \(  \Vec{u} \in \SparseSphereSubspace{k}{n}  \) and \(  \Vec{x} \in \BallSparseSphere{\Tau}( \Vec{u} )  \).
Notice that the \RHS of \EQN \eqref{pf:lemma:technical:raic:small-scale:eqn:4b} increases with \(  \RV{L}_{\Vec{x},\Vec{u}} = \Variable{\ell}_{\Vec{x},\Vec{u}}  \).
Therefore, we can consider the bound on \(  \| \hA[\Coords{J}]( \Vec{x}, \Vec{u} ) \|_{2}  \) when \(  \RV{L}_{\Vec{x},\Vec{u}}  \) is at the boundary, \(  \RV{L}_{\Vec{x},\Vec{u}} = \RV{Q}  \).
Specifically, given \(  \RV{Q} = \Variable{q}  \),
\begin{gather}
\label{pf:lemma:technical:raic:small-scale:eqn:4c}
  \left\| \hA[\Coords{J}]( \Vec{x}, \Vec{u} ) \right\|_{2}
  \leq
  \left( 3 \Variable{t}_{\Vec{x},\Vec{u}} + \sqrt{\frac{\pi}{2}} \right)
  \frac{\Variable{q}}{m}
  +
  \frac{2 \sqrt{\KO \Variable{q}}}{m}
\end{gather}
with probability at least
\begin{gather}
\label{pf:lemma:technical:raic:small-scale:eqn:5c}
  1 - 5 e^{-\frac{1}{8} \Variable{q} \Variable{t}_{\Vec{x},\Vec{u}}^{2}}
.\end{gather}
%
As the specific choice of \(  \Vec{u} \in \SparseSphereSubspace{k}{n}  \) and \(  \Vec{x} \in \BallSparseSphere{\Tau}( \Vec{u} )  \) will not be important, and since, indeed, we ultimately want to obtain a uniform result, let us rewrite this result with simpler notation:
if \(  \RV{Q} = \Variable{q}  \), then
\begin{gather}
\label{pf:lemma:technical:raic:small-scale:eqn:4d}
  \left\| \hA[\Coords{J}]( \Vec{x}, \Vec{u} ) \right\|_{2}
  \leq
  \left( 3 \Variable{t} + \sqrt{\frac{\pi}{2}} \right)
  \frac{\Variable{q}}{m}
  +
  \frac{2 \sqrt{\KO \Variable{q}}}{m}
\end{gather}
with probability at least
\begin{gather}
\label{pf:lemma:technical:raic:small-scale:eqn:5d}
  1 - 5 e^{-\frac{1}{8} \Variable{q} \Variable{t}^{2}}
,\end{gather}
where \(  \Variable{t} > 0  \) is determined later.
Applying \CLAIM \ref{pf:lemma:technical:raic:small-scale:claim:1} and union bounding over \(  \Set{Y}  \), it follows that for all \(  \Vec{u} \in \SparseSphereSubspace{k}{n}  \), \(  \Vec{x} \in \BallSparseSphere{\Tau}( \Vec{u} )  \), and \(  \Coords{J} \subseteq [n]  \), \(  | \Coords{J} | \leq \KO  \), \EQN \eqref{pf:lemma:technical:raic:small-scale:eqn:4d} holds uniformly with probability at least
\begin{gather}
\label{pf:lemma:technical:raic:small-scale:eqn:5e}
  1 - 5 | \Set{Y} | e^{-\frac{1}{8} \Variable{q} \Variable{t}^{2}}
  \geq
  1 -
  \EDITX{5 \left( \frac{2e m}{\Variable{q}} \right)^{\Variable{q}} \left( \frac{en}{\kOX} \right)^{\kOX}}
  e^{-\frac{1}{8} \Variable{q} \Variable{t}^{2}}
\end{gather}
when \(  \RV{Q} = \Variable{q}  \).
Finally, for \(  \Rho'_{2} \in (0,1)  \), specified later, setting
\begin{align}
  \Variable{t}
  &=
  \sqrt{
    \frac{8}{\Variable{q}}
    \Log(
      \EDITX{5 \left( \frac{2e m}{\Variable{q}} \right)^{\Variable{q}} \left( \frac{en}{\kOX} \right)^{\kOX}}
      \left( \frac{1}{\Rho'_{2}} \right)
    )
  }
  =
  \sqrt{
    \frac{8}{\Variable{q}}
    \Log(
      \EDITX{\left( \frac{2e m}{\Variable{q}} \right)^{\Variable{q}} \left( \frac{en}{\kOX} \right)^{\kOX} \left( \frac{5}{\Rho'_{2}} \right)}
    )
  }
  \\ \nonumber
  &=
  \BigO(
  \sqrt{
    \frac{k}{\Variable{q}}
    \Log( \frac{n}{k} )
    +
    \Log( \frac{m}{\Variable{q}} )
    +
    \frac{1}{\Variable{q}}
    \Log( \frac{\EDITX{1}}{\Rho} )
  }
  )
,\end{align}
uniformly with probability at least \(  1 - \Rho'_{2}  \), for all \(  \Vec{u} \in \SparseSphereSubspace{k}{n}  \), \(  \Vec{x} \in \BallSparseSphere{\Tau}( \Vec{u} )  \), and \(  \Coords{J} \subseteq [n]  \), \(  | \Coords{J} | \leq \KO  \),
\begin{align*}
  \left\| \hA[\Coords{J}]( \Vec{x}, \Vec{u} ) \right\|_{2}
  &\leq
  \left(
    3
    \sqrt{
      \frac{8}{\Variable{q}}
      \Log(
        \EDITX{\left( \frac{2e m}{\Variable{q}} \right)^{\Variable{q}} \left( \frac{en}{\kOX} \right)^{\kOX} \left( \frac{5}{\Rho'_{2}} \right)}
      )
    }
    +
    \sqrt{\frac{\pi}{2}}
  \right)
  \frac{\Variable{q}}{m}
  +
  \frac{2 \sqrt{\KO \Variable{q}}}{m}
  \\
  &=
  \frac{3}{m}
  \sqrt{
    8 \Variable{q}
    \Log(
      \EDITX{\left( \frac{2e m}{\Variable{q}} \right)^{\Variable{q}} \left( \frac{en}{\kOX} \right)^{\kOX} \left( \frac{5}{\Rho'_{2}} \right)}
    )
  }
  +
  \sqrt{\frac{\pi}{2}} \frac{\Variable{q}}{m}
  +
  \frac{\sqrt{8k \Variable{q}}}{m}
  \\
  &=
  \frac{\sqrt{72 \Variable{q}}}{m}
  \sqrt{
    \Log(
      \EDITX{\left( \frac{2e m}{\Variable{q}} \right)^{\Variable{q}} \left( \frac{en}{\kOX} \right)^{\kOX} \left( \frac{5}{\Rho'_{2}} \right)}
    )
  }
  +
  \sqrt{\frac{\pi}{2}} \frac{\Variable{q}}{m}
  +
  \frac{\sqrt{8k \Variable{q}}}{m}
  \\
  &=
  \frac{\sqrt{72 \Variable{q}}}{m}
  \sqrt{
  \Log( \left( \frac{2e m}{\Variable{q}} \right)^{\Variable{q}} )
  +
  \Log( \EDITX{\left( \frac{en}{\kOX} \right)^{\kOX}} )
  +
  \Log( \frac{\EDITX{5}}{\Rho'_{2}} )
  }
  +
  \sqrt{\frac{\pi}{2}} \frac{\Variable{q}}{m}
  +
  \frac{\sqrt{8k \Variable{q}}}{m}
  \\
  &\leq
  \frac{\sqrt{72 \Variable{q}}}{m}
  \sqrt{\Log( \left( \frac{2e m}{\Variable{q}} \right)^{\Variable{q}} )}
  +
  \frac{\sqrt{72 \Variable{q}}}{m}
  \sqrt{\Log( \EDITX{\left( \frac{en}{\kOX} \right)^{\kOX}} )}
  +
  \frac{\sqrt{72 \Variable{q}}}{m}
  \sqrt{\Log( \frac{\EDITX{5}}{\Rho'_{2}} )}
  +
  \sqrt{\frac{\pi}{2}} \frac{\Variable{q}}{m}
  +
  \frac{\sqrt{8k \Variable{q}}}{m}
  \\
  &=
  \frac{\sqrt{72 \Variable{q}}}{m}
  \sqrt{\Variable{q} \Log( \frac{2e m}{\Variable{q}} )}
  +
  \frac{\sqrt{72 \Variable{q}}}{m}
  \sqrt{\EDITX{\kOX \Log( \frac{en}{\kOX} )}}
  +
  \frac{\sqrt{72 \Variable{q}}}{m}
  \sqrt{\Log( \frac{\EDITX{5}}{\Rho'_{2}} )}
  +
  \sqrt{\frac{\pi}{2}} \frac{\Variable{q}}{m}
  +
  \frac{\sqrt{8k \Variable{q}}}{m}
  \\
  &=
  \frac{\sqrt{72} \Variable{q}}{m} \sqrt{\Log( \frac{2e m}{\Variable{q}} )}
  +
  \frac{\sqrt{72 \kOX \Variable{q}}}{m} \sqrt{\EDITX{\Log( \frac{en}{\kOX} )}}
  +
  \frac{\sqrt{72 \Variable{q}}}{m} \sqrt{\Log( \frac{\EDITX{5}}{\Rho'_{2}} )}
  +
  \sqrt{\frac{\pi}{2}} \frac{\Variable{q}}{m}
  +
  \frac{\sqrt{8k \Variable{q}}}{m}
.\end{align*}
%
Lastly, set
\(  \Rho'_{2} = \frac{\Rho_{2}}{2} = \frac{\Rho}{4}  \).
%
Then, with probability at least \(  1 - \frac{\Rho_{2}}{2}  \), for all \(  \Vec{u} \in \SparseSphereSubspace{k}{n}  \), \(  \Vec{x} \in \BallSparseSphere{\Tau}( \Vec{u} )  \), and \(  \Coords{J} \subseteq [n]  \), \(  | \Coords{J} | \leq \KO  \),
\begin{align*}
  \left\| \hA[\Coords{J}]( \Vec{x}, \Vec{u} ) \right\|_{2}
  &\leq
  \frac{\sqrt{72} \Variable{q}}{m} \sqrt{\Log( \frac{2e m}{\Variable{q}} )}
  +
  \frac{\sqrt{72 \kOX \Variable{q}}}{m} \sqrt{\EDITX{\Log( \frac{en}{\kOX} )}}
  +
  \frac{\sqrt{72 \Variable{q}}}{m} \sqrt{\Log( \frac{\EDITX{5}}{\Rho'_{2}} )}
  +
  \sqrt{\frac{\pi}{2}} \frac{\Variable{q}}{m}
  +
  \frac{\sqrt{8k \Variable{q}}}{m}
  \\
  &=
  \frac{\sqrt{72} \Variable{q}}{m} \sqrt{\Log( \frac{2e m}{\Variable{q}} )}
  +
  \frac{\sqrt{72 \kOX \Variable{q}}}{m} \sqrt{\EDITX{\Log( \frac{en}{\kOX} )}}
  +
  \frac{\sqrt{72 \Variable{q}}}{m} \sqrt{\Log( \frac{\EDITX{20}}{\Rho} )}
  +
  \sqrt{\frac{\pi}{2}} \frac{\Variable{q}}{m}
  +
  \frac{\sqrt{8k \Variable{q}}}{m}
  \\
  &=
  \frac{\sqrt{72} \Variable{q}}{m} \sqrt{\Log( \frac{2e m}{\Variable{q}} )}
  +
  \frac{\sqrt{72 \kOX \Variable{q}}}{m} \sqrt{\EDITX{\Log( \frac{en}{\kOX} )}}
  +
  \frac{\sqrt{72 \Variable{q}}}{m} \sqrt{\Log( \frac{\UnivConstAX}{\Rho} )}
  +
  \sqrt{\frac{\pi}{2}} \frac{\Variable{q}}{m}
  +
  \frac{\sqrt{8k \Variable{q}}}{m}
\end{align*}
%
This completes \STEP \ref{pf:lemma:technical:raic:small-scale:enum:1:i}.
%
\paragraph{\STEP \ref{pf:lemma:technical:raic:small-scale:enum:1:ii}.} 
%
Proceeding to the next step, the goal now is to upper bound the random variable \(  \RV{Q}  \) with high probability.
Here, we will leverage \LEMMA \ref{lemma:technical:concentration-ineq:lbe-local-deviations:union} to establish a uniform result over all \(  \Vec{u} \in \SparseSphereSubspace{k}{n}  \) and \(  \Vec{x} \in \BallSparseSphere{\Tau}( \Vec{u} )  \).
By \LEMMA \ref{lemma:technical:concentration-ineq:lbe-local-deviations:union}, the random variable \(  \RV{L}_{\Vec{x},\Vec{u}}  \) is upper bounded by \(  \RV{L}_{\Vec{x},\Vec{u}} \leq \GammaX m  \) uniformly with probability at least \(  1 - \EDITX{2} \binom{n}{\KO} e^{-\frac{1}{64} \GammaX m}  \) for all choices of \(  \Vec{u} \in \SparseSphereSubspace{k}{n}  \) and \(  \Vec{x} \in \BallSparseSphere{\Tau}( \Vec{u} )  \).
Due to the relationship between \(  \RV{L}_{\Vec{x},\Vec{u}}  \) and \(  \RV{Q}  \) stated in \EQN \eqref{pf:lemma:technical:raic:small-scale:eqn:q:def:b}, it follows that
\begin{gather}
\label{pf:lemma:technical:raic:small-scale:eqn:Q-bound}
  \RV{Q} \leq \QValue
\end{gather}
with probability at least
\begin{gather}
\label{pf:lemma:technical:raic:small-scale:eqn:6}
  1 - \EDITX{2} \binom{n}{\KO} e^{-\frac{1}{64} \GammaX m}
.\end{gather}
%
Recall that
\begin{align*}
  m
  &\geq
  \frac{64 \UnivConstB}{\DDelta} \Log( \binom{n}{\KO} \EDITX{\frac{\UnivConstAXX}{\Rho}} )
  \sqrt{\Log( \frac{2e}{\GammaX} )}
  \\
  &=
  \frac{64}{\GammaX} \Log( \EDITX{2} \binom{n}{\KO} \frac{2}{\Rho_{2}} )
,\end{align*}
and therefore \EQN \eqref{pf:lemma:technical:raic:small-scale:eqn:6} is bounded from below by
\begin{gather}
\label{pf:lemma:technical:raic:small-scale:eqn:6b}
  1 - \EDITX{2} \binom{n}{\KO} e^{-\frac{1}{64} \GammaX m}
  \geq
  1 - \frac{\Rho_{2}}{2}
.\end{gather}
%
\paragraph{\STEP \ref{pf:lemma:technical:raic:small-scale:enum:1:iii}.} 
%
The final step will put together the analysis from \STEPS \ref{pf:lemma:technical:raic:small-scale:enum:1:i} and \ref{pf:lemma:technical:raic:small-scale:enum:1:ii} to complete the proof of the lemma.
By a union bound combining \EQN \eqref{pf:lemma:technical:raic:small-scale:eqn:5e} and \eqref{pf:lemma:technical:raic:small-scale:eqn:6}, together with \EQNS \eqref{pf:lemma:technical:raic:small-scale:eqn:4d} and \eqref{pf:lemma:technical:raic:small-scale:eqn:Q-bound}, with probability at least
\begin{gather}
\label{pf:lemma:technical:raic:small-scale:eqn:5f}
  1
  - \EDITX{5 \left( \frac{2e m}{\Variable{q}} \right)^{\Variable{q}} \left( \frac{en}{\kOX} \right)^{\kOX}}
    e^{-\frac{1}{8} \Variable{q} \Variable{t}^{2}}
  - \EDITX{2} \binom{n}{\KO} e^{-\frac{1}{64} \GammaX m}
  \geq
  1 - \frac{\Rho_{2}}{2} - \frac{\Rho_{2}}{2}
  =
  1 - \Rho_{2}
,\end{gather}
uniformly for all \(  \Vec{u} \in \SparseSphereSubspace{k}{n}  \), \(  \Vec{x} \in \BallSparseSphere{\Tau}( \Vec{u} )  \), and \(  \Coords{J} \subseteq [n]  \), \(  | \Coords{J} | \leq \KO  \), the following holds
\begin{align}
  \nonumber
  \left\| \hA[\Coords{J}]( \Vec{x}, \Vec{u} ) \right\|_{2}
  &\leq
  \frac{\sqrt{72} \GammaX m}{m} \sqrt{\Log( \frac{2e m}{\GammaX m} )}
  +
  \frac{\sqrt{72 \kOX \GammaX m}}{m} \sqrt{\EDITX{\Log( \frac{en}{\kOX} )}}
  +
  \frac{\sqrt{72 \GammaX m}}{m} \sqrt{\Log( \EDITX{\frac{\UnivConstAX}{\Rho}} )}
  +
  \frac{\GammaX m}{m} \sqrt{\frac{\pi}{2}}
  +
  \frac{\sqrt{8k \GammaX m}}{m}
  \\ \nonumber
  &\leq
  \GammaX \sqrt{72 \Log( \frac{2e}{\GammaX} )}
  +
  \sqrt{\frac{72 \kOX \GammaX}{m} \EDITX{\Log( \frac{en}{\kOX} )}}
  +
  \sqrt{\frac{72 \GammaX}{m} \Log( \EDITX{\frac{\UnivConstAX}{\Rho}} )}
  +
  \GammaX \sqrt{\frac{\pi}{2}}
  +
  \sqrt{\frac{8k \GammaX}{m}}
  \\ \nonumber
  &\leq
  \GammaX \sqrt{72 \Log( \frac{2e}{\GammaX} )}
  +
  \sqrt{\frac{72 \kOX \GammaX}{m} \EDITX{\Log( \frac{en}{\kOX} )}}
  +
  \sqrt{\frac{72 \GammaX}{m} \Log( \EDITX{\frac{\UnivConstAX}{\Rho}} )}
  +
  \GammaX \sqrt{\frac{\pi}{2}}
  +
  \sqrt{\frac{8k \GammaX}{m}}
  \\
  &=
  \GammaX \sqrt{72 \Log( \frac{2e}{\GammaX} )}
  +
  \GammaX \sqrt{\frac{\pi}{2}}
  +
  \sqrt{\frac{72 \kOX \GammaX}{m} \EDITX{\Log( \frac{en}{\kOX} )}}
  +
  \sqrt{\frac{72 \GammaX}{m} \Log( \EDITX{\frac{\UnivConstAX}{\Rho}} )}
  +
  \sqrt{\frac{8k \GammaX}{m}}
\label{pf:lemma:technical:raic:small-scale:eqn:4f}
.\end{align}
%
Observe:
\begin{gather*}
  \GammaX = \GammaXValue = \GammaXValueX < \frac{\UnivConstb_{2} \DDelta}{5 \sqrt{72}}
,\end{gather*}
where
\begin{gather*}
  \UnivConstb_{2} = \frac{30 \sqrt{2}}{\UnivConstBX} = \frac{5 \sqrt{72}}{\UnivConstBX}
,\end{gather*}
and recall that
\begin{align*}
  m
  \geq
  \frac{\UnivConstB \kOX}{\DDelta}
  \EDITX{\Log( \frac{en}{\kOX} )}
  +
  \frac{\UnivConstB}{\DDelta}
  \Log( \EDITX{\frac{\UnivConstAX}{\Rho}} )
  =
  \frac{5 \sqrt{72} \kOX}{\UnivConstb_{2} \DDelta}
  \EDITX{\Log( \frac{en}{\kOX} )}
  +
  \frac{5 \sqrt{72}}{\UnivConstb_{2} \DDelta}
  \Log( \EDITX{\frac{\UnivConstAX}{\Rho}} )
.\end{align*}
%
Thus, \EQN \eqref{pf:lemma:technical:raic:small-scale:eqn:4f} can be bounded from above as follows:
\begin{align*}
  \left\| \hA[\Coords{J}]( \Vec{x}, \Vec{u} ) \right\|_{2}
  &\leq
  \GammaX \sqrt{72 \Log( \frac{2e}{\GammaX} )}
  +
  \GammaX \sqrt{\frac{\pi}{2}}
  +
  \sqrt{\frac{72 \kOX \GammaX}{m} \EDITX{\Log( \frac{en}{\kOX} )}}
  +
  \sqrt{\frac{72 \GammaX}{m} \Log( \EDITX{\frac{\UnivConstAX}{\Rho}} )}
  +
  \sqrt{\frac{8k \GammaX}{m}}
  \\
  &\leq
  \GammaXValueX \sqrt{72 \Log( \frac{2e}{\GammaX} )}
  +
  \frac{\UnivConstb_{2} \DDelta}{5 \sqrt{72}} \sqrt{\frac{\pi}{2}}
  +
  \sqrt{72 \kOX \cdot \frac{\UnivConstb_{2} \DDelta}{5 \sqrt{72}} \cdot \frac{\UnivConstb_{2} \DDelta}{5 \sqrt{72} \kOX \EDITX{\Log( \frac{en}{\kOX} )}} \EDITX{\Log( \frac{en}{\kOX} )}}
  \\ &\AlignTab+
  \sqrt{72 \cdot \frac{\UnivConstb_{2} \DDelta}{5 \sqrt{72}} \cdot \frac{\UnivConstb_{2} \DDelta}{5 \sqrt{72} \Log( \EDITX{\frac{\UnivConstAX}{\Rho}} )} \Log( \EDITX{\frac{\UnivConstAX}{\Rho}} )}
  +
  \sqrt{8k \cdot \frac{\UnivConstb_{2} \DDelta}{5 \sqrt{72}} \cdot \frac{\UnivConstb_{2} \DDelta}{5 \sqrt{72} \kOX \EDITX{\Log( \frac{en}{\kOX} )}}}
  \\
  &\leq
  \GammaXValueX \sqrt{72 \Log( \frac{2e}{\GammaX} )}
  +
  \frac{\UnivConstb_{2} \DDelta}{5 \sqrt{72}} \sqrt{\frac{\pi}{2}}
  +
  \sqrt{72 \kOX \cdot \frac{\UnivConstb_{2} \DDelta}{5 \sqrt{72}} \cdot \frac{\UnivConstb_{2} \DDelta}{5 \sqrt{72} \kOX \EDITX{\Log( \frac{en}{\kOX} )}} \EDITX{\Log( \frac{en}{\kOX} )}}
  \\ &\AlignTab+
  \sqrt{72 \cdot \frac{\UnivConstb_{2} \DDelta}{5 \sqrt{72}} \cdot \frac{\UnivConstb_{2} \DDelta}{5 \sqrt{72} \Log( \EDITX{\frac{\UnivConstAX}{\Rho}} )} \Log( \EDITX{\frac{\UnivConstAX}{\Rho}} )}
  +
  \sqrt{8k \cdot \frac{\UnivConstb_{2} \DDelta}{5 \sqrt{72}} \cdot \frac{\UnivConstb_{2} \DDelta}{5 \sqrt{72} k}}
  \\
  &\leq
  5 \cdot \frac{\UnivConstb_{2} \DDelta}{5}
  \\
  &=
  \UnivConstb_{2} \DDelta
.\end{align*}
%
To summarize, uniformly with probability at least \(  1 - \Rho_{2}  \), for all \(  \Vec{u} \in \SparseSphereSubspace{k}{n}  \), \(  \Vec{x} \in \BallSparseSphere{\Tau}( \Vec{u} )  \), and \(  \Coords{J} \subseteq [n]  \), \(  | \Coords{J} | \leq \KO  \),
\begin{gather}
  \left\| \hA[\Coords{J}]( \Vec{x}, \Vec{u} ) \right\|_{2}
  \leq
  \UnivConstb_{2} \DDelta
,\end{gather}
as desired.
\end{proof}
\ENDEDIT


\subsection{Combining the regimes to prove Theorem \ref{thm:technical:raic:modified}}
\label{outline:technical:pf|>combine}

Using Lemmas \ref{lemma:technical:raic:large-scale} and \ref{lemma:technical:raic:small-scale},
Theorem \ref{thm:technical:raic:modified} can now be established with a direct argument.

%
\begin{proof}
{Theorem \ref{thm:technical:raic:modified}}
\label{pf:thm:technical:raic:modified}
Fix
\(  \Rho_{1} = \Rho_{2} = \frac{\Rho}{2}  \),
such that
\(
  \Rho_{1} + \Rho_{2} = \Rho
\).
Let
\(
  \Vec{x}, \Vec{y} \in \SparseSphereSubspace{k}{n}
\)
be an arbitrary pair of \( k \)-sparse unit vectors.
Suppose
\(
  \Vec{u}, \Vec{v} \in \Net{\Tau}
\)
are the closest points to
\(
  \Vec{x}, \Vec{y}
\),
respectively, subject to
\(
  \supp( \Vec{u} ) = \supp( \Vec{x} )
\)
and
\(
  \supp( \Vec{v} ) = \supp( \Vec{y} )
\),
where it is possible that \( \Vec{u} = \Vec{x} \) when \( \Vec{x} \) is in the net,
and similarly for \( \Vec{v} \) when \( \Vec{y} \) is in the net.
Formally,
\begin{gather}
\label{pf:thm:technical:raic:modified:eqn:1}
  \Vec{u}
  =
  \arg \min_{\substack{\Vec{u'} \in \Net{\Tau}: \\ \supp( \Vec{u'} ) = \supp( \Vec{x} )}}
  \left\| \Vec{x} - \Vec{u'} \right\|_{2}
  \\
  \Vec{v}
  =
  \arg \min_{\substack{\Vec{v'} \in \Net{\Tau}: \\ \supp( \Vec{v'} ) = \supp( \Vec{y} )}}
  \left\| \Vec{y} - \Vec{v'} \right\|_{2}
\end{gather}
%
Note that the requirement
\(
  \supp( \Vec{u} ) = \supp( \Vec{x} )
\)
and
\(
  \supp( \Vec{v} ) = \supp( \Vec{y} )
\)
is possible due to the design of the \( \Tau \)-net \( \Net{\Tau} \) as specified at the beginning
of Section \ref{outline:technical:pf}.
Observe
\begin{subequations}
\begin{align}
\label{pf:thm:technical:raic:modified:eqn:2}
  &
  ( \Vec{x} - \Vec{y} ) - \hA( \Vec{x}, \Vec{y} )
  \\
  &=
  ( \Vec{x} - \Vec{y} )
  -
  \sqrt{2\pi} \frac{1}{m}
  \MeasMat^{\T}
  \cdot
  \frac{1}{2}
  \left(
    \Sgn( \MeasMat \Vec{x} )
    -
    \Sgn( \MeasMat \Vec{y} )
  \right)
  \\
  &=
  ( \Vec{u} - \Vec{v} )
  +
  ( \Vec{x} - \Vec{u} )
  +
  ( \Vec{v} - \Vec{y} )
  -
  \sqrt{2\pi} \frac{1}{m}
  \MeasMat^{\T}
  \cdot
  \frac{1}{2}
  \left(
    \Sgn( \MeasMat \Vec{u} )
    -
    \Sgn( \MeasMat \Vec{v} )
  \right)
  \\ \nonumber
  &\TAB
  -
  \sqrt{2\pi} \frac{1}{m}
  \MeasMat^{\T}
  \cdot
  \frac{1}{2}
  \left(
    \Sgn( \MeasMat \Vec{x} )
    -
    \Sgn( \MeasMat \Vec{u} )
  \right)
  -
  \sqrt{2\pi} \frac{1}{m}
  \MeasMat^{\T}
  \cdot
  \frac{1}{2}
  \left(
    \Sgn( \MeasMat \Vec{v} )
    -
    \Sgn( \MeasMat \Vec{y} )
  \right)
  \\
  &=
  ( \Vec{u} - \Vec{v} )
  -
  \sqrt{2\pi} \frac{1}{m}
  \MeasMat^{\T}
  \cdot
  \frac{1}{2}
  \left(
    \Sgn( \MeasMat \Vec{x} )
    -
    \Sgn( \MeasMat \Vec{u} )
  \right)
  \\ \nonumber
  &\TAB
  +
  ( \Vec{x} - \Vec{u} )
  -
  \sqrt{2\pi} \frac{1}{m}
  \MeasMat^{\T}
  \cdot
  \frac{1}{2}
  \left(
    \Sgn( \MeasMat \Vec{u} )
    -
    \Sgn( \MeasMat \Vec{v} )
  \right)
  \\ \nonumber
  &\TAB
  +
  ( \Vec{v} - \Vec{y} )
  -
  \sqrt{2\pi} \frac{1}{m}
  \MeasMat^{\T}
  \cdot
  \frac{1}{2}
  \left(
    \Sgn( \MeasMat \Vec{v} )
    -
    \Sgn( \MeasMat \Vec{y} )
  \right)
  \\
  &=
  ( \Vec{u} - \Vec{v} )
  -
  \hA( \Vec{u}, \Vec{v} )
  +
  ( \Vec{x} - \Vec{u} )
  -
  \hA( \Vec{x}, \Vec{u} )
  +
  ( \Vec{v} - \Vec{y} )
  -
  \hA( \Vec{v}, \Vec{y} )
\end{align}
\end{subequations}
%
Write
\(
  \Coords{J}_{\Vec{x}}
  =
  \Coords{J} \cup \supp( \Vec{x} )
\)
and
\(
  \Coords{J}_{\Vec{y}}
  =
  \Coords{J} \cup \supp( \Vec{y} )
\),
where
\(
  | \Coords{J}_{\Vec{x}} |, | \Coords{J}_{\Vec{y}} | \leq  \KO
\).
Then,
\begin{subequations}
\label{pf:thm:technical:raic:modified:eqn:3}
\begin{align}
\label{pf:thm:technical:raic:modified:eqn:3:1}
  ( \Vec{x} - \Vec{y} ) - \hA[\Coords{J}]( \Vec{x}, \Vec{y} )
  &=
  ( \Vec{x} - \Vec{y} )
  -
  \Threshold{\supp( \Vec{x} ) \cup \supp( \Vec{y} ) \cup \Coords{J}}
  ( \hA( \Vec{x}, \Vec{y} ) )
  \\
  &=
  ( \Vec{u} - \Vec{v} )
  -
  \Threshold{\supp( \Vec{x} ) \cup \supp( \Vec{y} ) \cup \Coords{J}}
  ( \hA( \Vec{u}, \Vec{v} ) )
  \\ \nonumber
  &+
  ( \Vec{x} - \Vec{u} )
  -
  \Threshold{\supp( \Vec{x} ) \cup \supp( \Vec{y} ) \cup \Coords{J}}
  ( \hA( \Vec{x}, \Vec{u} ) )
  \\ \nonumber
  &+
  ( \Vec{v} - \Vec{y} )
  -
  \Threshold{\supp( \Vec{x} ) \cup \supp( \Vec{y} ) \cup \Coords{J}}
  ( \hA( \Vec{v}, \Vec{y} ) )
  \\
  &=
  ( \Vec{u} - \Vec{v} )
  -
  \Threshold{\supp( \Vec{u} ) \cup \supp( \Vec{v} ) \cup \Coords{J}}
  ( \hA( \Vec{u}, \Vec{v} ) )
  \\ \nonumber
  &+
  ( \Vec{x} - \Vec{u} )
  -
  \Threshold{\supp( \Vec{x} ) \cup \supp( \Vec{u} ) \cup \Coords{J}_{\Vec{y}}}
  ( \hA( \Vec{x}, \Vec{u} ) )
  \\ \nonumber
  &+
  ( \Vec{v} - \Vec{y} )
  -
  \Threshold{\supp( \Vec{v} ) \cup \supp( \Vec{y} ) \cup \Coords{J}_{\Vec{x}}}
  ( \hA( \Vec{v}, \Vec{y} ) )
  \\
  &=
  ( \Vec{u} - \Vec{v} )
  -
  \hA[\Coords{J}]( \Vec{u}, \Vec{v} )
  +
  ( \Vec{x} - \Vec{u} )
  -
  \hA[\Coords{J}_{\Vec{y}}]( \Vec{x}, \Vec{u} )
  +
  ( \Vec{v} - \Vec{y} )
  -
  \hA[\Coords{J}_{\Vec{x}}]( \Vec{v}, \Vec{y} )
\end{align}
\end{subequations}
%
The norm of \eqref{pf:thm:technical:raic:modified:eqn:3} is then bounded by the triangle inequality.
\begin{subequations}
\begin{align}
\label{pf:thm:technical:raic:modified:eqn:4}
  &
  \left\| ( \Vec{x} - \Vec{y} ) - \hA[\Coords{J}]( \Vec{x}, \Vec{y} ) \right\|_{2}
  \\
  &=
  \left\|
    ( \Vec{u} - \Vec{v} )
    -
    \hA[\Coords{J}]( \Vec{u}, \Vec{v} )
    +
    ( \Vec{x} - \Vec{u} )
    -
    \hA[\Coords{J}_{\Vec{y}}]( \Vec{x}, \Vec{u} )
    +
    ( \Vec{v} - \Vec{y} )
    -
    \hA[\Coords{J}_{\Vec{x}}]( \Vec{v}, \Vec{y} )
  \right\|_{2}
  \\
  &\leq
  \left\|
    ( \Vec{u} - \Vec{v} )
    -
    \hA[\Coords{J}]( \Vec{u}, \Vec{v} )
  \right\|_{2}
  +
  \left\|
    ( \Vec{x} - \Vec{u} )
    -
    \hA[\Coords{J}_{\Vec{y}}]( \Vec{x}, \Vec{u} )
  \right\|_{2}
  +
  \left\|
    ( \Vec{v} - \Vec{y} )
    -
    \hA[\Coords{J}_{\Vec{x}}]( \Vec{v}, \Vec{y} )
  \right\|_{2}
\end{align}
\end{subequations}
Suppose
\(
  \DistS{\Vec{u}}{\Vec{v}} < \Tau
\).
Then, by Lemma \ref{lemma:technical:raic:small-scale},
\begin{subequations}
\label{pf:thm:technical:raic:modified:eqn:6}
\begin{align}
  &
  \left\| ( \Vec{x} - \Vec{y} ) - \hA[\Coords{J}]( \Vec{x}, \Vec{y} ) \right\|_{2}
  \\
  &\leq
  \left\|
    ( \Vec{u} - \Vec{v} )
    -
    \hA[\Coords{J}]( \Vec{u}, \Vec{v} )
  \right\|_{2}
  +
  \left\|
    ( \Vec{x} - \Vec{u} )
    -
    \hA[\Coords{J}_{\Vec{y}}]( \Vec{x}, \Vec{u} )
  \right\|_{2}
  +
  \left\|
    ( \Vec{v} - \Vec{y} )
    -
    \hA[\Coords{J}_{\Vec{x}}]( \Vec{v}, \Vec{y} )
  \right\|_{2}
  \\
  &\leq
  3 \UnivConstb_{2} \DDelta
  \\
  &\leq
  \UnivConstb_{1} \sqrt{\DDelta \DistS{\Vec{u}}{\Vec{v}}}
  +
  3 \UnivConstb_{2} \DDelta
\end{align}
\end{subequations}
uniformly with probability at least
\(
  1 - \Rho_{2} = 1 - \frac{\Rho}{2}
\).
On the other hand, if
\(
  \DistS{\Vec{u}}{\Vec{v}} \geq \Tau
\),
then by Lemmas \ref{lemma:technical:raic:large-scale} and \ref{lemma:technical:raic:small-scale},
\begin{subequations}
\label{pf:thm:technical:raic:modified:eqn:7}
\begin{align}
  &
  \left\| ( \Vec{x} - \Vec{y} ) - \hA[\Coords{J}]( \Vec{x}, \Vec{y} ) \right\|_{2}
  \\
  &\leq
  \left\|
    ( \Vec{u} - \Vec{v} )
    -
    \hA[\Coords{J}]( \Vec{u}, \Vec{v} )
  \right\|_{2}
  +
  \left\|
    ( \Vec{x} - \Vec{u} )
    -
    \hA[\Coords{J}_{\Vec{y}}]( \Vec{x}, \Vec{u} )
  \right\|_{2}
  +
  \left\|
    ( \Vec{v} - \Vec{y} )
    -
    \hA[\Coords{J}_{\Vec{x}}]( \Vec{v}, \Vec{y} )
  \right\|_{2}
  \\
  &\leq
  \UnivConstb_{1} \sqrt{\DDelta \DistS{\Vec{u}}{\Vec{v}}}
  +
  \UnivConstb_{2} \DDelta
  +
  \UnivConstb_{2} \DDelta
  \\
  &=
  \UnivConstb_{1} \sqrt{\DDelta \DistS{\Vec{u}}{\Vec{v}}}
  +
  2 \UnivConstb_{2} \DDelta
  \\
  &\leq
  \UnivConstb_{1} \sqrt{\DDelta \DistS{\Vec{u}}{\Vec{v}}}
  +
  3 \UnivConstb_{2} \DDelta
\end{align}
\end{subequations}
uniformly with probability at least
\(
  1 - \Rho_{1} - \Rho_{2} = 1 - \Rho
\).
Therefore, with probability at least
\(
  1 - \Rho
\),
for all
\(
  \Vec{x}, \Vec{y} \in \SparseSphereSubspace{k}{n}
\)
and all
\(
  \Coords{J} \subseteq [n]
\),
\(
  | \Coords{J} | \leq k
\),
\begin{gather}
\label{pf:thm:technical:raic:modified:eqn:5}
  \left\| ( \Vec{x} - \Vec{y} ) - \hA[\Coords{J}]( \Vec{x}, \Vec{y} ) \right\|_{2}
  \leq
  \UnivConstc_{1} \sqrt{\DDelta \DistS{\Vec{u}}{\Vec{v}}}
  +
  \UnivConstc_{2} \DDelta
\end{gather}
where
\(
  \UnivConstc_{1}
  = \UnivConstb_{1}
  = \UnivConstcOneValue
\),
\(
  \UnivConstc_{2}
  = 3\UnivConstb_{2}
  = 3 \cdot \frac{5 \sqrt{72}}{\UnivConstBX}
  = \UnivConstcTwoValue
\),
\(
  \UnivConstB \gtrsim \UnivConstBValue
\),
and
\(
  \UnivConstD = \UnivConstDValue
\),
as specified in \EQN \eqref{eqn:univConstants}.
Succinctly, the measurement matrix \( \MeasMat \) satisfies the \( (k, n, \DDelta, \UnivConstc_{1}, \UnivConstc_{2}) \)-RAIC with probability at least \( 1 - \Rho \).
\end{proof}


\section{Proofs of the concentration inequalities,
Lemmas \ref{lemma:technical:concentration-ineq:(u,v)}-\ref{lemma:technical:concentration-ineq:lbe-local-deviations:union}}
\label{outline:normal|>concentration-ineq-pfs}

\subsection{Orthogonal projections: proof of Lemma \ref{lemma:technical:concentration-ineq:(u,v)}}
\label{outline:normal|>concentration-ineq-pfs|>orthogonal-projections}

%
%
%
%
This appendix proves \LEMMA \ref{lemma:technical:concentration-ineq:(u,v)}.
\APPENDIX \ref{outline:normal|>concentration-ineq-pfs|>orthogonal-projections|>intermediate} presents three intermediate concentration inequalities which, in comparison to \LEMMA \ref{lemma:technical:concentration-ineq:(u,v)}, have an additional condition.
Subsequently, \APPENDIX \ref{outline:normal|>concentration-ineq-pfs|>orthogonal-projections|>pf} proves \LEMMA \ref{lemma:technical:concentration-ineq:(u,v)}, while the proofs of the intermediate lemmas in \APPENDIX \ref{outline:normal|>concentration-ineq-pfs|>orthogonal-projections|>intermediate} are deferred to \APPENDIX \ref{outline:normal|>concentration-ineq-pfs|>orthogonal-projections|>intermediate-pfs}.
The analysis here, as well as in \APPENDIX \ref{outline:normal|>concentration-ineq-pfs|>orthogonal-projections|>intermediate-pfs}, will use the notation of \(  m  \) \iid Gaussian vectors,
\(  \Vec{Z}^{(1)}, \dots, \Vec{Z}^{(m)} \sim \N( \Vec{0}, \Mat{I}_{n \times n} )  \),
as well as (in separate contexts) a single Gaussian vector,
\(  \Vec{Z} \sim \N( \Vec{0}, \Mat{I}_{n \times n} )  \).
Additionally, for \(  \Vec{u}, \Vec{v} \in \R^{n}  \), define the random variables
\(  \RV{R_{i;\Vec{u},\Vec{v}}} \defeq \frac{1}{2} ( \sgn( \langle \Vec{Z}^{(i)}, \Vec{u} \rangle ) - \sgn( \langle \Vec{Z}^{(i)}, \Vec{v} \rangle ) )  \), \(  i \in [m]  \),
and the random vectors
\(  \Vec{R}_{\Vec{u},\Vec{v}} \defeq ( \RV{R_{1;\Vec{u},\Vec{v}}}, \dots, \RV{R_{m;\Vec{u},\Vec{v}}} )  \) and
\(  \Rhatuv \defeq \I{ \Vec{R}_{\Vec{u},\Vec{v}} \neq \Vec{0} }  \).
Then, write
\(  \RV{L}_{\Vec{u},\Vec{v}} \defeq \| \Rhatuv \|_{0}  \).


\subsubsection{Intermediate Lemmas}
\label{outline:normal|>concentration-ineq-pfs|>orthogonal-projections|>intermediate}

\begin{lemma}
\label{lemma:normal:concentration-ineq:proj_u-v}
Let
\(
  \Variable{\ell}, t > 0
\)
and
\(
  \Vec{r} \in \rSet
\)
such that
\(
  \left\| \Vec{r} \right\|_{0} = \Variable{\ell}
\).
Fix an ordered pair of real-valued unit vectors,
\(
  ( \Vec{u}, \Vec{v} ) \in \Sphere{n} \times \Sphere{n}
\).
\ORIG{Define the random variable
\(
  \RV{L}_{\Vec{u},\Vec{v}} = \left\| \Rhatuv \right\|_{0}
\),
and suppose
\(
  \Rhatuv = \Vec{r}
\)
and
\(
  \RV{L}_{\Vec{u},\Vec{v}} = \Variable{\ell}
\).
Then, the random variable}%
\EDIT{The random variable}
\(
  \RV{X}_{\Vec{u},\Vec{v}}
  =
  \left\langle
    \frac{\Vec{u} - \Vec{v}}{\left\| \Vec{u} - \Vec{v} \right\|_{2}},
    \sum_{i=1}^{m}
    \Vec{Z}^{(i)} \RV{R_{i;\Vec{u},\Vec{v}}}
  \right\rangle
\)
conditioned on
\(
  \Rhatuv = \Vec{r},
  \RV{L}_{\Vec{u},\Vec{v}} = \Variable{\ell}
\)
is concentrated around its mean such that
\begin{gather}
\label{eqn:normal:distribution:proj_u-v:inequality}
  \Pr
  \left(
    \left|
      \RV{X}_{\Vec{u},\Vec{v}}
      -
      \E \left[ \RV{X}_{\Vec{u},\Vec{v}} \middle| \Rhatuv = \Vec{r}, \RV{L}_{\Vec{u},\Vec{v}} = \Variable{\ell} \right]
    \right|
    \geq
    \Variable{\ell} \Variable{t}
    \phantom{\Big|}\middle|
    \Rhatuv = \Vec{r},
    \RV{L}_{\Vec{u},\Vec{v}} = \Variable{\ell}
  \right)
  \leq
  2 e^{-\frac{1}{2} \Variable{\ell} t^{2}}
,\end{gather}
\BEGINEDIT
where
\begin{gather}
  \E \left[ \RV{X}_{\Vec{u},\Vec{v}} \middle| \Rhatuv = \Vec{r}, \RV{L}_{\Vec{u},\Vec{v}} = \Variable{\ell} \right]
  =
  \sqrt{\frac{\pi}{2}}
  \frac{\Variable{\ell} \DistS{\Vec{u}}{\Vec{v}}}{\theta_{\Vec{u},\Vec{v}}}
.\end{gather}
\ENDEDIT
\end{lemma}

\begin{lemma}
\label{lemma:normal:concentration-ineq:proj_u+v}
Let
\(
  \Variable{\ell}, t > 0
\)
and
\(
  \Vec{r} \in \rSet
\)
such that
\(
  \left\| \Vec{r} \right\|_{0} = \Variable{\ell}
\).
Fix an ordered pair of real-valued unit vectors,
\(
  ( \Vec{u}, \Vec{v} ) \in \Sphere{n} \times \Sphere{n}
\).
\ORIG{Define the random variable
\(
  \RV{L}_{\Vec{u},\Vec{v}} = \left\| \Rhatuv \right\|_{0}
\),
and suppose
\(
  \Rhatuv = \Vec{r}
\)
and
\(
  \RV{L}_{\Vec{u},\Vec{v}} = \Variable{\ell}
\).}%
\ORIG{Then, the random variable}%
\EDIT{The random variable}
\(
  \RV{X}_{\Vec{u},\Vec{v}}
  =
  \left\langle
    \frac{\Vec{u} + \Vec{v}}{\left\| \Vec{u} + \Vec{v} \right\|_{2}},
    \sum_{i=1}^{m}
    \Vec{Z}^{(i)} \RV{R_{i;\Vec{u},\Vec{v}}}
  \right\rangle
\)
conditioned on
\(
  \Rhatuv = \Vec{r},
  \RV{L}_{\Vec{u},\Vec{v}} = \Variable{\ell}
\)
is concentrated around zero such that
\begin{gather}
\label{eqn:normal:distribution:proj_u+v:inequality}
  \Pr
  \left(
    \left| \RV{X}_{\Vec{u},\Vec{v}} \right|
    \geq
    \Variable{\ell} \Variable{t}
    \phantom{\Big|}\middle|
    \Rhatuv = \Vec{r},
    \RV{L}_{\Vec{u},\Vec{v}} = \Variable{\ell}
  \right)
  \leq
  2 e^{-\frac{1}{2} \Variable{\ell} t^{2}}
.\end{gather}
\end{lemma}

\begin{lemma}
\label{lemma:normal:concentration-ineq:third}
Let
\(
  \Variable{d}, \Variable{\ell}, \Variable{t} > 0
\).
\EDITX{Write \(  \kO \defeq \min \{ 2k, n \}  \).}
Fix an ordered pair of \( k \)-sparse, real-valued unit vectors,
\(
  ( \Vec{u}, \Vec{v} ) \in (\SparseSphereSubspace{k}{n}) \times (\SparseSphereSubspace{k}{n})
\),
and let
\(
  \Coords{J} \subseteq [n]
\)
with
\(
  | \Coords{J} | \leq \Variable{d}
\).
Define the random variables
\(
  \Vec{Y}_{\Vec{u},\Vec{v}}^{(i)}
  =
  \Vec{Z}^{(i)}
  -
  \left\langle
    \frac{\Vec{u} - \Vec{v}}{\left\| \Vec{u} - \Vec{v} \right\|_{2}},
    \Vec{Z}^{(i)}
  \right\rangle
  \frac{\Vec{u} - \Vec{v}}{\left\| \Vec{u} - \Vec{v} \right\|_{2}}
  -
  \left\langle
    \frac{\Vec{u} + \Vec{v}}{\left\| \Vec{u} + \Vec{v} \right\|_{2}},
    \Vec{Z}^{(i)}
  \right\rangle
  \frac{\Vec{u} + \Vec{v}}{\left\| \Vec{u} + \Vec{v} \right\|_{2}}
\)
\EDIT{and}
\(
  \RV{X}_{\Vec{u},\Vec{v}}
  =
  \left\|
    \ThresholdSet{\EDIT{\Supp( \Vec{u} ) \cup \Supp( \Vec{v} ) \cup} \Coords{J}}
    (
      \sum_{i=1}^{m} \Vec{Y}_{\Vec{u},\Vec{v}}^{(i)} \RV{R}_{i;\Vec{u},\Vec{v}}
    )
  \right\|_{2}
\).
\ORIG{
and
\(
  \RV{L}_{\Vec{u},\Vec{v}} = \left\| \Rhatuv \right\|_{0}
\),
and suppose
\(
  \Rhatuv = \Vec{r}
\)
and
\(
  \RV{L}_{\Vec{u},\Vec{v}} = \Variable{\ell}
\).}%
Then,
\begin{gather}
\label{eqn:normal:concentration-ineq:third}
  \Pr
  \left(
    \RV{X}_{\Vec{u},\Vec{v}}
    \geq
    \left( \sqrt{\KXO} + \sqrt{\Variable{d}} \right) \sqrt{\Variable{\ell}}
    +
    \Variable{\ell} \Variable{t}
    \phantom{\Big|}\middle|
    \Rhatuv = \Vec{r},
    \RV{L}_{\Vec{u},\Vec{v}} = \Variable{\ell}
  \right)
  \leq
  2 e^{-\frac{1}{8} \Variable{\ell} \Variable{t}^{2}}
\end{gather}
\end{lemma}

\BEGINEDIT
\subsubsection{Proof of \LEMMA \ref{lemma:technical:concentration-ineq:(u,v)}}
\label{outline:normal|>concentration-ineq-pfs|>orthogonal-projections|>pf}

\begin{proof}
{\LEMMA \ref{lemma:technical:concentration-ineq:(u,v)}}
Fix \(  \Variable{t} > 0  \), \(  \Variable{\ell} \in \ZeroTo{m}  \), and \(  \Vec{u}, \Vec{v} \in \R^{n}  \) arbitrarily.
Towards proving the lemma, the following claim about the random vector \(  \Rhatuv  \) is stated and verified.
\begin{claim}
\label{pf:lemma:technical:concentration-ineq:(u,v):claim:1}
Write
\(  \Set{R}_{\Variable{\ell}} \defeq \{ \Vec{r} \in \rSet : \| \Vec{r} \|_{0} = \Variable{\ell} \}  \).
For \(  \Vec{r} \in \rSet  \),
\begin{gather}
  \Pr \left( \Rhatuv = \Vec{r} \middle| \RV{L}_{\Vec{u},\Vec{v}} = \Variable{\ell} \right)
  =
  \begin{cases}
  0,                                       &\cIf \Vec{r} \notin \Set{R}_{\Variable{\ell}}, \\
  \frac{1}{| \Set{R}_{\Variable{\ell}} |}, &\cIf \Vec{r} \in \Set{R}_{\Variable{\ell}}.
  \end{cases}
\end{gather}
\end{claim}
\begin{subproof}
{\CLAIM \ref{pf:lemma:technical:concentration-ineq:(u,v):claim:1}}
Notice that \(  \Rhatuv \in \Set{R}_{\Variable{\ell}}  \) if and only if \(  \RV{L}_{\Vec{u},\Vec{v}} = \Variable{\ell}  \).
Thus, for \(  \Vec{r} \in \Set{R}_{\Variable{\ell}}  \), there is an equality:
\(  \Pr \left( \Rhatuv = \Vec{r} \middle| \RV{L}_{\Vec{u},\Vec{v}} = \Variable{\ell} \right) = \Pr \left( \Rhatuv = \Vec{r} \middle| \Rhatuv \in \Set{R}_{\Variable{\ell}} \right)  \).
Hence,
\begin{align*}
  \sum_{\Vec{r} \in \Set{R}_{\Variable{\ell}}}
  \Pr \left( \Rhatuv = \Vec{r} \middle| \RV{L}_{\Vec{u},\Vec{v}} = \Variable{\ell} \right)
  =
  \sum_{\Vec{r} \in \Set{R}_{\Variable{\ell}}}
  \Pr \left( \Rhatuv = \Vec{r} \middle| \Rhatuv \in \Set{R}_{\Variable{\ell}} \right)
  = 1
,\end{align*}
and by complementation
\begin{align*}
  \sum_{\Vec{r} \in \rSet \setminus \Set{R}_{\Variable{\ell}}}
  \Pr \left( \Rhatuv = \Vec{r} \middle| \RV{L}_{\Vec{u},\Vec{v}} = \Variable{\ell} \right)
  =
  \sum_{\Vec{r} \in \rSet \setminus \Set{R}_{\Variable{\ell}}}
  \Pr \left( \Rhatuv = \Vec{r} \middle| \Rhatuv \in \Set{R}_{\Variable{\ell}} \right)
  = 0
.\end{align*}
By the latter equation and the first probability axiom (nonnegativity), if \(  \Vec{r} \in \rSet \setminus \Set{R}_{\Variable{\ell}}  \), then
\(  \Pr \left( \Rhatuv = \Vec{r} \middle| \RV{L}_{\Vec{u},\Vec{v}} = \Variable{\ell} \right) = 0  \).
On the other hand, because the random variables
\(  \RV{R}_{i;\Vec{u},\Vec{v}}  \), \(  i \in [m]  \),
are \iid and \(  \RV{L}_{\Vec{u},\Vec{v}}  \) (as an unweighted sum of \iid random variables) is invariant to permutation of the indexes, \(  i \in [m]  \), it follows that
\(  \Pr \left( \Rhatuv = \Vec{r} \middle| \RV{L}_{\Vec{u},\Vec{v}} = \Variable{\ell} \right) = \Pr \left( \Rhatuv = \Vec{r'} \middle| \RV{L}_{\Vec{u},\Vec{v}} = \Variable{\ell} \right)  \)
for all pairs, \(  \Vec{r}, \Vec{r'} \in \Set{R}_{\Variable{\ell}}  \).
Writing \(  p = \Pr \left( \Rhatuv = \Vec{r} \middle| \RV{L}_{\Vec{u},\Vec{v}} = \Variable{\ell} \right)  \) for an arbitrary choice of \(  \Vec{r} \in \Set{R}_{\Variable{\ell}}  \), observe:
\begin{align*}
  &
  1
  =
  \sum_{\Vec{r} \in \Set{R}_{\Variable{\ell}}}
  \Pr \left( \Rhatuv = \Vec{r} \middle| \RV{L}_{\Vec{u},\Vec{v}} = \Variable{\ell} \right)
  =
  \sum_{\Vec{r} \in \Set{R}_{\Variable{\ell}}} p
  =
  | \Set{R}_{\Variable{\ell}} | p
  \\
  &\longrightarrow
  p = \frac{1}{| \Set{R}_{\Variable{\ell}} |}
.\end{align*}
Thus,
\(  \Pr \left( \Rhatuv = \Vec{r} \middle| \RV{L}_{\Vec{u},\Vec{v}} = \Variable{\ell} \right) = \frac{1}{| \Set{R}_{\Variable{\ell}} |}  \)
for \(  \Vec{r} \in \Set{R}_{\Variable{\ell}}  \).
This completes the proof of the claim.
\end{subproof}
Throughout the remainder of the proof of \LEMMA \ref{lemma:technical:concentration-ineq:(u,v)}, the notation of the set \(  \Set{R}_{\Variable{\ell}}  \) will persist.
With \CLAIM \ref{pf:lemma:technical:concentration-ineq:(u,v):claim:1} established, \EQN \ref{eqn:technical:concentration-ineq:(u,v):u-v} will be the concentration inequality proved first.
It is recalled below:
\begin{gather*}
  \Pr
  \left(
    \left|
      \RV{X}_{\Vec{u},\Vec{v}}
      -
      \sqrt{\frac{\pi}{2}}
      \frac{\Variable{\ell}}{m}
      \frac{\DistS{\Vec{u}}{\Vec{v}}}{\theta_{\Vec{u},\Vec{v}}}
    \right|
    \geq
    \frac{\Variable{\ell} \Variable{t}}{m}
    \middle|
    \RV{L}_{\Vec{u},\Vec{v}} = \Variable{\ell}
  \right)
  \leq
  2 e^{-\frac{1}{2} \Variable{\ell} \Variable{t}^{2}}
,\end{gather*}
where
\begin{gather*}
  \RV{X}_{\Vec{u},\Vec{v}}
  \defeq
  \left\langle
    \frac{\Vec{u} - \Vec{v}}{\left\| \Vec{u} - \Vec{v} \right\|_{2}},
    \frac{1}{m}
    \sum_{i=1}^{m}
    \MeasVec^{(i)}
    \cdot
    \frac{1}{2} \left( \sgn( \langle \MeasVec^{(i)}, \Vec{u} \rangle ) - \sgn( \langle \MeasVec^{(i)}, \Vec{v} \rangle ) \right)
  \right\rangle
.\end{gather*}
This concentration inequality is derived as follows, using the law of total probability, \LEMMA \ref{lemma:normal:concentration-ineq:proj_u-v}, and \CLAIM \ref{pf:lemma:technical:concentration-ineq:(u,v):claim:1}:
\begin{align*}
  &
  \Pr
  \left(
    \left|
      \RV{X}_{\Vec{u},\Vec{v}}
      -
      \sqrt{\frac{\pi}{2}}
      \frac{\Variable{\ell}}{m}
      \frac{\DistS{\Vec{u}}{\Vec{v}}}{\theta_{\Vec{u},\Vec{v}}}
    \right|
    \geq
    \frac{\Variable{\ell} \Variable{t}}{m}
    \middle|
    \RV{L}_{\Vec{u},\Vec{v}} = \Variable{\ell}
  \right)
  \\
  &=
  \sum_{\Vec{r} \in \rSet}
  \Pr
  \left(
    \left|
      \RV{X}_{\Vec{u},\Vec{v}}
      -
      \sqrt{\frac{\pi}{2}}
      \frac{\Variable{\ell}}{m}
      \frac{\DistS{\Vec{u}}{\Vec{v}}}{\theta_{\Vec{u},\Vec{v}}}
    \right|
    \geq
    \frac{\Variable{\ell} \Variable{t}}{m}
    \middle|
    \Rhatuv = \Vec{r},
    \RV{L}_{\Vec{u},\Vec{v}} = \Variable{\ell}
  \right)
  \Pr \left( \Rhatuv = \Vec{r} \middle| \RV{L}_{\Vec{u},\Vec{v}} = \Variable{\ell} \right)
  \\
  &\dCmt \text{by the law of total probability}
  \\
  &=
  \sum_{\Vec{r} \in \Set{R}_{\Variable{\ell}}}
  \Pr
  \left(
    \left|
      \RV{X}_{\Vec{u},\Vec{v}}
      -
      \sqrt{\frac{\pi}{2}}
      \frac{\Variable{\ell}}{m}
      \frac{\DistS{\Vec{u}}{\Vec{v}}}{\theta_{\Vec{u},\Vec{v}}}
    \right|
    \geq
    \frac{\Variable{\ell} \Variable{t}}{m}
    \middle|
    \Rhatuv = \Vec{r},
    \RV{L}_{\Vec{u},\Vec{v}} = \Variable{\ell}
  \right)
  \Pr \left( \Rhatuv = \Vec{r} \middle| \RV{L}_{\Vec{u},\Vec{v}} = \Variable{\ell} \right)
  \\
  &\AlignTab
  +
  \sum_{\Vec{r} \in \rSet \setminus \Set{R}_{\Variable{\ell}}}
  \Pr
  \left(
    \left|
      \RV{X}_{\Vec{u},\Vec{v}}
      -
      \sqrt{\frac{\pi}{2}}
      \frac{\Variable{\ell}}{m}
      \frac{\DistS{\Vec{u}}{\Vec{v}}}{\theta_{\Vec{u},\Vec{v}}}
    \right|
    \geq
    \frac{\Variable{\ell} \Variable{t}}{m}
    \middle|
    \Rhatuv = \Vec{r},
    \RV{L}_{\Vec{u},\Vec{v}} = \Variable{\ell}
  \right)
  \Pr \left( \Rhatuv = \Vec{r} \middle| \RV{L}_{\Vec{u},\Vec{v}} = \Variable{\ell} \right)
  \\
  &\dCmt \text{by partitioning the image of \(  \Rhatuv  \)}
  \\
  &=
  \sum_{\Vec{r} \in \Set{R}_{\Variable{\ell}}}
  \Pr
  \left(
    \left|
      \RV{X}_{\Vec{u},\Vec{v}}
      -
      \sqrt{\frac{\pi}{2}}
      \frac{\Variable{\ell}}{m}
      \frac{\DistS{\Vec{u}}{\Vec{v}}}{\theta_{\Vec{u},\Vec{v}}}
    \right|
    \geq
    \frac{\Variable{\ell} \Variable{t}}{m}
    \middle|
    \Rhatuv = \Vec{r},
    \RV{L}_{\Vec{u},\Vec{v}} = \Variable{\ell}
  \right)
  \Pr \left( \Rhatuv = \Vec{r} \middle| \RV{L}_{\Vec{u},\Vec{v}} = \Variable{\ell} \right)
  \\
  &\AlignTab
  +
  \sum_{\Vec{r} \in \rSet \setminus \Set{R}_{\Variable{\ell}}}
  \Pr
  \left(
    \left|
      \RV{X}_{\Vec{u},\Vec{v}}
      -
      \sqrt{\frac{\pi}{2}}
      \frac{\Variable{\ell}}{m}
      \frac{\DistS{\Vec{u}}{\Vec{v}}}{\theta_{\Vec{u},\Vec{v}}}
    \right|
    \geq
    \frac{\Variable{\ell} \Variable{t}}{m}
    \middle|
    \Rhatuv = \Vec{r},
    \RV{L}_{\Vec{u},\Vec{v}} = \Variable{\ell}
  \right)
  \cdot 0
  \\
  &\dCmt \text{by \CLAIM \ref{pf:lemma:technical:concentration-ineq:(u,v):claim:1}}
  \\
  &=
  \sum_{\Vec{r} \in \Set{R}_{\Variable{\ell}}}
  \Pr
  \left(
    \left|
      \RV{X}_{\Vec{u},\Vec{v}}
      -
      \sqrt{\frac{\pi}{2}}
      \frac{\Variable{\ell}}{m}
      \frac{\DistS{\Vec{u}}{\Vec{v}}}{\theta_{\Vec{u},\Vec{v}}}
    \right|
    \geq
    \frac{\Variable{\ell} \Variable{t}}{m}
    \middle|
    \Rhatuv = \Vec{r},
    \RV{L}_{\Vec{u},\Vec{v}} = \Variable{\ell}
  \right)
  \Pr \left( \Rhatuv = \Vec{r} \middle| \RV{L}_{\Vec{u},\Vec{v}} = \Variable{\ell} \right)
  \\
  &=
  \frac{1}{| \Set{R}_{\Variable{\ell}} |}
  \sum_{\Vec{r} \in \Set{R}_{\Variable{\ell}}}
  \Pr
  \left(
    \left|
      \RV{X}_{\Vec{u},\Vec{v}}
      -
      \sqrt{\frac{\pi}{2}}
      \frac{\Variable{\ell}}{m}
      \frac{\DistS{\Vec{u}}{\Vec{v}}}{\theta_{\Vec{u},\Vec{v}}}
    \right|
    \geq
    \frac{\Variable{\ell} \Variable{t}}{m}
    \middle|
    \Rhatuv = \Vec{r},
    \RV{L}_{\Vec{u},\Vec{v}} = \Variable{\ell}
  \right)
  \\
  &\dCmt \text{by \CLAIM \ref{pf:lemma:technical:concentration-ineq:(u,v):claim:1}}
  \\
  &=
  \frac{1}{| \Set{R}_{\Variable{\ell}} |}
  \sum_{\Vec{r} \in \Set{R}_{\Variable{\ell}}}
  \Pr
  \left(
    \left|
      \RV{X}_{\Vec{u},\Vec{v}}
      -
      \E \left[ \RV{X}_{\Vec{u},\Vec{v}} \middle| \Rhatuv = \Vec{r}, \RV{L}_{\Vec{u},\Vec{v}} = \Variable{\ell} \right]
    \right|
    \geq
    \frac{\Variable{\ell} \Variable{t}}{m}
    \middle|
    \Rhatuv = \Vec{r},
    \RV{L}_{\Vec{u},\Vec{v}} = \Variable{\ell}
  \right)
  \\
  &\dCmt \text{by \LEMMA \ref{lemma:normal:concentration-ineq:proj_u-v}}
  \\
  &\leq
  \frac{1}{| \Set{R}_{\Variable{\ell}} |}
  \sum_{\Vec{r} \in \Set{R}_{\Variable{\ell}}}
  2 e^{-\frac{1}{2} \Variable{\ell} \Variable{t}^{2}}
  \\
  &\dCmt \text{by \LEMMA \ref{lemma:normal:concentration-ineq:proj_u-v}}
  \\
  &=
  \frac{1}{| \Set{R}_{\Variable{\ell}} |}
  \cdot | \Set{R}_{\Variable{\ell}} |
  \cdot
  2 e^{-\frac{1}{2} \Variable{\ell} \Variable{t}^{2}}
  \\
  &=
  2 e^{-\frac{1}{2} \Variable{\ell} \Variable{t}^{2}}
.\end{align*}
This completes the derivation of \EQN \eqref{eqn:technical:concentration-ineq:(u,v):u-v}.
%
\par 
%
The derivations of \EQNS \eqref{eqn:technical:concentration-ineq:(u,v):u+v} and \eqref{eqn:technical:concentration-ineq:(u,v):g_Z} will follow an analogous approach.
Towards verifying \EQN \eqref{eqn:technical:concentration-ineq:(u,v):u+v}, define the random variable
\begin{gather*}
  \RV{X'}_{\Vec{u},\Vec{v}}
  \defeq
  \left\langle
    \frac{\Vec{u} + \Vec{v}}{\left\| \Vec{u} + \Vec{v} \right\|_{2}},
    \frac{1}{m}
    \sum_{i=1}^{m}
    \MeasVec^{(i)}
    \cdot
    \frac{1}{2} \left( \sgn( \langle \MeasVec^{(i)}, \Vec{u} \rangle ) - \sgn( \langle \MeasVec^{(i)}, \Vec{v} \rangle ) \right)
  \right\rangle
\end{gather*}
so that with this notation, \EQN \eqref{eqn:technical:concentration-ineq:(u,v):u+v} is stated as follows:
\begin{gather*}
  \Pr
  \left(
    \left|
      \RV{X'}_{\Vec{u},\Vec{v}}
    \right|
    \geq
    \frac{\Variable{\ell} \Variable{t}}{m}
    \middle|
    \RV{L}_{\Vec{u},\Vec{v}} = \Variable{\ell}
  \right)
  \leq
  2 e^{-\frac{1}{2} \Variable{\ell} \Variable{t}^{2}}
.\end{gather*}
As similarly seen in the derivation of \EQN \eqref{eqn:technical:concentration-ineq:(u,v):u-v}, \EQN \eqref{eqn:technical:concentration-ineq:(u,v):u+v} is obtained as follows:
\begin{align*}
  &
  \Pr
  \left(
    \left|
      \RV{X'}_{\Vec{u},\Vec{v}}
    \right|
    \geq
    \frac{\Variable{\ell} \Variable{t}}{m}
    \middle|
    \RV{L}_{\Vec{u},\Vec{v}} = \Variable{\ell}
  \right)
  \\
  &=
  \sum_{\Vec{r} \in \rSet}
  \Pr
  \left(
    \left|
      \RV{X'}_{\Vec{u},\Vec{v}}
    \right|
    \geq
    \frac{\Variable{\ell} \Variable{t}}{m}
    \middle|
    \Rhatuv = \Vec{r},
    \RV{L}_{\Vec{u},\Vec{v}} = \Variable{\ell}
  \right)
  \Pr \left( \Rhatuv = \Vec{r} \middle| \RV{L}_{\Vec{u},\Vec{v}} = \Variable{\ell} \right)
  \\
  &\dCmt \text{by the law of total probability}
  \\
  &=
  \sum_{\Vec{r} \in \Set{R}_{\Variable{\ell}}}
  \Pr
  \left(
    \left|
      \RV{X'}_{\Vec{u},\Vec{v}}
    \right|
    \geq
    \frac{\Variable{\ell} \Variable{t}}{m}
    \middle|
    \Rhatuv = \Vec{r},
    \RV{L}_{\Vec{u},\Vec{v}} = \Variable{\ell}
  \right)
  \Pr \left( \Rhatuv = \Vec{r} \middle| \RV{L}_{\Vec{u},\Vec{v}} = \Variable{\ell} \right)
  \\
  &\AlignTab
  +
  \sum_{\Vec{r} \in \rSet \setminus \Set{R}_{\Variable{\ell}}}
  \Pr
  \left(
    \left|
      \RV{X'}_{\Vec{u},\Vec{v}}
    \right|
    \geq
    \frac{\Variable{\ell} \Variable{t}}{m}
    \middle|
    \Rhatuv = \Vec{r},
    \RV{L}_{\Vec{u},\Vec{v}} = \Variable{\ell}
  \right)
  \Pr \left( \Rhatuv = \Vec{r} \middle| \RV{L}_{\Vec{u},\Vec{v}} = \Variable{\ell} \right)
  \\
  &\dCmt \text{by partitioning the image of \(  \Rhatuv  \)}
  \\
  &=
  \sum_{\Vec{r} \in \Set{R}_{\Variable{\ell}}}
  \Pr
  \left(
    \left|
      \RV{X'}_{\Vec{u},\Vec{v}}
    \right|
    \geq
    \frac{\Variable{\ell} \Variable{t}}{m}
    \middle|
    \Rhatuv = \Vec{r},
    \RV{L}_{\Vec{u},\Vec{v}} = \Variable{\ell}
  \right)
  \Pr \left( \Rhatuv = \Vec{r} \middle| \RV{L}_{\Vec{u},\Vec{v}} = \Variable{\ell} \right)
  \\
  &\AlignTab
  +
  \sum_{\Vec{r} \in \rSet \setminus \Set{R}_{\Variable{\ell}}}
  \Pr
  \left(
    \left|
      \RV{X'}_{\Vec{u},\Vec{v}}
    \right|
    \geq
    \frac{\Variable{\ell} \Variable{t}}{m}
    \middle|
    \Rhatuv = \Vec{r},
    \RV{L}_{\Vec{u},\Vec{v}} = \Variable{\ell}
  \right)
  \cdot 0
  \\
  &\dCmt \text{by \CLAIM \ref{pf:lemma:technical:concentration-ineq:(u,v):claim:1}}
  \\
  &=
  \sum_{\Vec{r} \in \Set{R}_{\Variable{\ell}}}
  \Pr
  \left(
    \left|
      \RV{X'}_{\Vec{u},\Vec{v}}
    \right|
    \geq
    \frac{\Variable{\ell} \Variable{t}}{m}
    \middle|
    \Rhatuv = \Vec{r},
    \RV{L}_{\Vec{u},\Vec{v}} = \Variable{\ell}
  \right)
  \Pr \left( \Rhatuv = \Vec{r} \middle| \RV{L}_{\Vec{u},\Vec{v}} = \Variable{\ell} \right)
  \\
  &=
  \frac{1}{| \Set{R}_{\Variable{\ell}} |}
  \sum_{\Vec{r} \in \Set{R}_{\Variable{\ell}}}
  \Pr
  \left(
    \left|
      \RV{X'}_{\Vec{u},\Vec{v}}
    \right|
    \geq
    \frac{\Variable{\ell} \Variable{t}}{m}
    \middle|
    \Rhatuv = \Vec{r},
    \RV{L}_{\Vec{u},\Vec{v}} = \Variable{\ell}
  \right)
  \\
  &\dCmt \text{by \CLAIM \ref{pf:lemma:technical:concentration-ineq:(u,v):claim:1}}
  \\
  &\leq
  \frac{1}{| \Set{R}_{\Variable{\ell}} |}
  \sum_{\Vec{r} \in \Set{R}_{\Variable{\ell}}}
  2 e^{-\frac{1}{2} \Variable{\ell} \Variable{t}^{2}}
  \\
  &\dCmt \text{by \LEMMA \ref{lemma:normal:concentration-ineq:proj_u+v}}
  \\
  &=
  \frac{1}{| \Set{R}_{\Variable{\ell}} |}
  \cdot | \Set{R}_{\Variable{\ell}} |
  \cdot
  2 e^{-\frac{1}{2} \Variable{\ell} \Variable{t}^{2}}
  \\
  &=
  2 e^{-\frac{1}{2} \Variable{\ell} \Variable{t}^{2}}
.\end{align*}
%
\par 
%
Lastly, recall \EQN \eqref{eqn:technical:concentration-ineq:(u,v):g_Z}:
\begin{gather*}
  \Pr
  \left(
    \RV{X''}_{\Vec{u},\Vec{v}}
    \geq
    \frac{2\sqrt{\KXO \Variable{\ell}}}{m}
    +
    \frac{\Variable{\ell} \Variable{t}}{m}
    \middle|
    \RV{L}_{\Vec{u},\Vec{v}} = \Variable{\ell}
  \right)
  \leq
  2 e^{-\frac{1}{8} \Variable{\ell} \Variable{t}^{2}}
,\end{gather*}
where
\begin{align*}
  \RV{X''}_{\Vec{u},\Vec{v}}
  \defeq
  \| \ThresholdSet{\Supp( \Vec{u} ) \cup \Supp( \Vec{v} ) \cup \Coords{J}}( \Vec{Y}_{\Vec{u},\Vec{v}} ) \|_{2}
\end{align*}
and
\begin{align*}
  \Vec{Y}_{\Vec{u},\Vec{v}}
  &\defeq
  \frac{1}{m}
  \sum_{i=1}^{m}
  \MeasVec^{(i)}
  \cdot
  \frac{1}{2} \left( \sgn( \langle \MeasVec^{(i)}, \Vec{u} \rangle ) - \sgn( \langle \MeasVec^{(i)}, \Vec{v} \rangle ) \right)
  \\ &\AlignTab
  -
  \left\langle
    \frac{\Vec{u} - \Vec{v}}{\left\| \Vec{u} - \Vec{v} \right\|_{2}},
    \frac{1}{m}
    \sum_{i=1}^{m}
    \MeasVec^{(i)}
    \cdot
    \frac{1}{2} \left( \sgn( \langle \MeasVec^{(i)}, \Vec{u} \rangle ) - \sgn( \langle \MeasVec^{(i)}, \Vec{v} \rangle ) \right)
  \right\rangle
  \frac{\Vec{u} - \Vec{v}}{\left\| \Vec{u} - \Vec{v} \right\|_{2}}
  \\ &\AlignTab
  -
  \left\langle
    \frac{\Vec{u} + \Vec{v}}{\left\| \Vec{u} + \Vec{v} \right\|_{2}},
    \frac{1}{m}
    \sum_{i=1}^{m}
    \MeasVec^{(i)}
    \cdot
    \frac{1}{2} \left( \sgn( \langle \MeasVec^{(i)}, \Vec{u} \rangle ) - \sgn( \langle \MeasVec^{(i)}, \Vec{v} \rangle ) \right)
  \right\rangle
  \frac{\Vec{u} + \Vec{v}}{\left\| \Vec{u} + \Vec{v} \right\|_{2}}
\end{align*}
This result can again be verified with the same techniques as presented above:
\begin{align*}
  &
  \Pr
  \left(
    \RV{X''}_{\Vec{u},\Vec{v}}
    \geq
    \frac{2\sqrt{\KXO \Variable{\ell}}}{m}
    +
    \frac{\Variable{\ell} \Variable{t}}{m}
    \middle|
    \RV{L}_{\Vec{u},\Vec{v}} = \Variable{\ell}
  \right)
  \\
  &=
  \sum_{\Vec{r} \in \rSet}
  \Pr
  \left(
    \RV{X''}_{\Vec{u},\Vec{v}}
    \geq
    \frac{2\sqrt{\KXO \Variable{\ell}}}{m}
    +
    \frac{\Variable{\ell} \Variable{t}}{m}
    \middle|
    \Rhatuv = \Vec{r},
    \RV{L}_{\Vec{u},\Vec{v}} = \Variable{\ell}
  \right)
  \Pr \left( \Rhatuv = \Vec{r} \middle| \RV{L}_{\Vec{u},\Vec{v}} = \Variable{\ell} \right)
  \\
  &\dCmt \text{by the law of total probability}
  \\
  &=
  \sum_{\Vec{r} \in \Set{R}_{\Variable{\ell}}}
  \Pr
  \left(
    \RV{X''}_{\Vec{u},\Vec{v}}
    \geq
    \frac{2\sqrt{\KXO \Variable{\ell}}}{m}
    +
    \frac{\Variable{\ell} \Variable{t}}{m}
    \middle|
    \Rhatuv = \Vec{r},
    \RV{L}_{\Vec{u},\Vec{v}} = \Variable{\ell}
  \right)
  \Pr \left( \Rhatuv = \Vec{r} \middle| \RV{L}_{\Vec{u},\Vec{v}} = \Variable{\ell} \right)
  \\
  &\AlignTab
  +
  \sum_{\Vec{r} \in \rSet \setminus \Set{R}_{\Variable{\ell}}}
  \Pr
  \left(
    \RV{X''}_{\Vec{u},\Vec{v}}
    \geq
    \frac{2\sqrt{\KXO \Variable{\ell}}}{m}
    +
    \frac{\Variable{\ell} \Variable{t}}{m}
    \middle|
    \Rhatuv = \Vec{r},
    \RV{L}_{\Vec{u},\Vec{v}} = \Variable{\ell}
  \right)
  \Pr \left( \Rhatuv = \Vec{r} \middle| \RV{L}_{\Vec{u},\Vec{v}} = \Variable{\ell} \right)
  \\
  &\dCmt \text{by partitioning the image of \(  \Rhatuv  \)}
  \\
  &=
  \sum_{\Vec{r} \in \Set{R}_{\Variable{\ell}}}
  \Pr
  \left(
    \RV{X''}_{\Vec{u},\Vec{v}}
    \geq
    \frac{2\sqrt{\KXO \Variable{\ell}}}{m}
    +
    \frac{\Variable{\ell} \Variable{t}}{m}
    \middle|
    \Rhatuv = \Vec{r},
    \RV{L}_{\Vec{u},\Vec{v}} = \Variable{\ell}
  \right)
  \Pr \left( \Rhatuv = \Vec{r} \middle| \RV{L}_{\Vec{u},\Vec{v}} = \Variable{\ell} \right)
  \\
  &\AlignTab
  +
  \sum_{\Vec{r} \in \rSet \setminus \Set{R}_{\Variable{\ell}}}
  \Pr
  \left(
    \RV{X''}_{\Vec{u},\Vec{v}}
    \geq
    \frac{2\sqrt{\KXO \Variable{\ell}}}{m}
    +
    \frac{\Variable{\ell} \Variable{t}}{m}
    \middle|
    \Rhatuv = \Vec{r},
    \RV{L}_{\Vec{u},\Vec{v}} = \Variable{\ell}
  \right)
  \cdot 0
  \\
  &\dCmt \text{by \CLAIM \ref{pf:lemma:technical:concentration-ineq:(u,v):claim:1}}
  \\
  &=
  \sum_{\Vec{r} \in \Set{R}_{\Variable{\ell}}}
  \Pr
  \left(
    \RV{X''}_{\Vec{u},\Vec{v}}
    \geq
    \frac{2\sqrt{\KXO \Variable{\ell}}}{m}
    +
    \frac{\Variable{\ell} \Variable{t}}{m}
    \middle|
    \Rhatuv = \Vec{r},
    \RV{L}_{\Vec{u},\Vec{v}} = \Variable{\ell}
  \right)
  \Pr \left( \Rhatuv = \Vec{r} \middle| \RV{L}_{\Vec{u},\Vec{v}} = \Variable{\ell} \right)
  \\
  &=
  \frac{1}{| \Set{R}_{\Variable{\ell}} |}
  \sum_{\Vec{r} \in \Set{R}_{\Variable{\ell}}}
  \Pr
  \left(
    \RV{X''}_{\Vec{u},\Vec{v}}
    \geq
    \frac{2\sqrt{\KXO \Variable{\ell}}}{m}
    +
    \frac{\Variable{\ell} \Variable{t}}{m}
    \middle|
    \Rhatuv = \Vec{r},
    \RV{L}_{\Vec{u},\Vec{v}} = \Variable{\ell}
  \right)
  \\
  &\dCmt \text{by \CLAIM \ref{pf:lemma:technical:concentration-ineq:(u,v):claim:1}}
  \\
  &=
  \frac{1}{| \Set{R}_{\Variable{\ell}} |}
  \sum_{\Vec{r} \in \Set{R}_{\Variable{\ell}}}
  \Pr
  \left(
    \RV{X''}_{\Vec{u},\Vec{v}}
    \geq
    \frac{2\sqrt{\KXO \Variable{\ell}}}{m}
    +
    \frac{\Variable{\ell} \Variable{t}}{m}
    \middle|
    \Rhatuv = \Vec{r},
    \RV{L}_{\Vec{u},\Vec{v}} = \Variable{\ell}
  \right)
  \\
  &\dCmt \text{by \LEMMA \ref{lemma:normal:concentration-ineq:proj_u-v}}
  \\
  &\leq
  \frac{1}{| \Set{R}_{\Variable{\ell}} |}
  \sum_{\Vec{r} \in \Set{R}_{\Variable{\ell}}}
  2 e^{-\frac{1}{8} \Variable{\ell} \Variable{t}^{2}}
  \\
  &\dCmt \text{by \LEMMA \ref{lemma:normal:concentration-ineq:third}, setting \(  d = | \Coords{J} | \leq \KXO  \)}
  \\
  &=
  \frac{1}{| \Set{R}_{\Variable{\ell}} |}
  \cdot | \Set{R}_{\Variable{\ell}} |
  \cdot
  2 e^{-\frac{1}{8} \Variable{\ell} \Variable{t}^{2}}
  \\
  &=
  2 e^{-\frac{1}{8} \Variable{\ell} \Variable{t}^{2}}
.\end{align*}
\end{proof}
\ENDEDIT


\subsection{Proof of \LEMMAS \ref{lemma:normal:concentration-ineq:proj_u-v}-\ref{lemma:normal:concentration-ineq:third}}
\label{outline:normal|>concentration-ineq-pfs|>orthogonal-projections|>intermediate-pfs}

Before proving the lemmas
(see Appendix \ref{outline:normal|>concentration-ineq-pfs|>orthogonal-projections|>main-lemma-pf}),
several intermediate results are stated and proved in
Appendix \ref{outline:normal|>concentration-ineq-pfs|>orthogonal-projections|>distributions}
to facilitate the proofs.


\subsubsection{The distributions of orthogonal projections of i.i.d. standard normal vectors}
\label{outline:normal|>concentration-ineq-pfs|>orthogonal-projections|>distributions}

\begin{lemma}
\label{lemma:normal:distribution:proj_u-v}
Fix an ordered pair of real-valued vectors,
\(
  ( \Vec{u}, \Vec{v} ) \in \Sphere{n} \times \Sphere{n}
\),
of unit norm.
Let
\(
  \Vec{Z} \sim \N(\Vec{0},\Mat{I}_{n \times n})
\)
be a standard normal random vector,
and let
\(
  \RV{R}
\)
be the (discrete) random variable taking values in
\(
  \{ -1, 0, 1 \}
\)
and given by
\(
  \RV{R}_{\Vec{u},\Vec{v}}
  =
  \frac{1}{2}
  \left(
    \sgn( \langle \Vec{u}, \Vec{Z} \rangle )
    -
    \sgn( \langle \Vec{v}, \Vec{Z} \rangle )
  \right)
\).
Define the map
\(
  \Alpha : \R \to \R
\)
by
\(
  \Alpha(x)
  = x \Tan( \frac{\theta_{\Vec{u},\Vec{v}}}{2} )
  = x \sqrt{\frac{\DistS[^{2}]{\Vec{u}}{\Vec{v}}}{4 - \DistS[^{2}]{\Vec{u}}{\Vec{v}}}}
\).
Then, the density function
\(
  f_{\RV{X} \Mid \RV{R}} : \R \to \R_{\geq 0}
\)
for the random variable
\(
  \RV{X}_{\Vec{u},\Vec{v}}
  =
  \left\langle
    \frac{\Vec{u}-\Vec{v}}{\left\| \Vec{u}-\Vec{v} \right\|_{2}},
    \Vec{Z}
  \right\rangle
  \RV{R}_{\Vec{u},\Vec{v}}
\)
conditioned on
\(
  \RV{R} \neq 0
\)
is given by
\begin{gather}
\label{eqn:normal:distribution:proj_u-v:density-function}
  f_{\RV{X}_{\Vec{u},\Vec{v}} \Mid \RV{R}_{\Vec{u},\Vec{v}}}(x \Mid \Value{r} \neq 0)
  =
  \begin{cases}
    \frac{\pi}{\theta_{\Vec{u},\Vec{v}}}
    \sqrt{\frac{2}{\pi}}
    e^{-\frac{x^{2}}{2}}
    \cdot
    \frac{1}{\sqrt{2\pi}}
    \int_{y=-\Alpha(x)}^{y=\Alpha(x)}
    e^{-\frac{y^{2}}{2}}
    dy
    , &\cIf x \geq 0,
    \\
    0, &\cIf x < 0.
  \end{cases}
\end{gather}
Moreover, in expectation,
\begin{gather}
\label{eqn:normal:distribution:proj_u-v:expectation}
  \E(\RV{X}_{\Vec{u},\Vec{v}} \Mid \RV{R}_{\Vec{u},\Vec{v}} \neq 0)
  =
  \sqrt{\frac{\pi}{2}}
  \frac{\DistS{\Vec{u}}{\Vec{v}}}{\theta_{\Vec{u},\Vec{v}}}
.\end{gather}
\end{lemma}

\begin{proof}
{Lemma \ref{lemma:normal:distribution:proj_u-v}}
\label{pf:lemma:normal:distribution:proj_u-v}
Before deriving the density function of \( \RV{X}_{\Vec{u},\Vec{v}} \),
\(
  \Vec{u}, \Vec{v} \in \Sphere{n}
\),
\EDIT{let us introduce some helpful observations.
First, notice that \(  \RV{R}_{\Vec{u},\Vec{v}} \neq 0  \) implies that \(  \sgn( \langle \Vec{u}, \Vec{Z} \rangle ) = -\sgn( \langle \Vec{v}, \Vec{Z} \rangle )  \).}
Second, let us show that for
\(
  \Vec{u}, \Vec{v}, \Vec{u'}, \Vec{v'} \in \Sphere{n}
\),
such that
\(
  \theta_{\Vec{u},\Vec{v}} = \theta_{\Vec{u'},\Vec{v'}}
\),
the pair of random variables
\( ( \RV{X}_{\Vec{u},\Vec{v}} \Mid \RV{R}_{\Vec{u},\Vec{v}} = 0 ) \)
and
\( ( \RV{X}_{\Vec{u'},\Vec{v'}} \Mid \RV{R}_{\Vec{u'},\Vec{v'}} = 0 ) \)
follow the same distribution, as do the pair
\( ( \RV{X}_{\Vec{u},\Vec{v}} \Mid \RV{R}_{\Vec{u},\Vec{v}} \neq 0 ) \)
and
\( ( \RV{X}_{\Vec{u'},\Vec{v'}} \Mid \RV{R}_{\Vec{u'},\Vec{v'}} \neq 0 ) \).
This will simplify the characterization of the distribution of \( \RV{X}_{\Vec{u},\Vec{v}} \) by
allowing \( \Vec{u}, \Vec{v} \) to be chosen non-arbitrarily.
Conditioned on
\(
  \RV{R}_{\Vec{u},\Vec{v}} = \RV{R}_{\Vec{u'},\Vec{v'}} = 0
\),
\( \RV{X}_{\Vec{u},\Vec{v}} = \RV{X}_{\Vec{u'},\Vec{v'}} = 0 \)
with probability 1.
Otherwise, when
\(
  \RV{R}_{\Vec{u},\Vec{v}}, \RV{R}_{\Vec{u'},\Vec{v'}} \neq 0
\),
write
\(
  q = \| \Vec{u} - \Vec{v} \|_{2} = \| \Vec{u'} - \Vec{v'} \|_{2}
\),
and observe
\begin{subequations}
\begin{align}
  \RV{X}_{\Vec{u},\Vec{v}}
  &=
  \left\langle
    \frac{\Vec{u}-\Vec{v}}{\left\| \Vec{u}-\Vec{v} \right\|_{2}},
    \Vec{Z}
  \right\rangle
  \RV{R}_{\Vec{u},\Vec{v}}
  \\
  &=
  \frac{1}{q}
  \left(
    \langle \Vec{u}, \Vec{Z} \rangle
    \RV{R}_{\Vec{u},\Vec{v}}
    -
    \langle \Vec{v}, \Vec{Z} \rangle
    \RV{R}_{\Vec{u},\Vec{v}}
  \right)
  \\
  &=
  \frac{1}{q}
  \left(
    \langle \Vec{u}, \Vec{Z} \rangle
    \sgn( \langle \Vec{u}, \Vec{Z} \rangle )
    -
    \langle \Vec{v}, \Vec{Z} \rangle
    (-\sgn( \langle \Vec{v}, \Vec{Z} \rangle ))
  \right)
  \\
  &=
  \frac{1}{q}
  \left(
    \langle \Vec{u}, \Vec{Z} \rangle
    \sgn( \langle \Vec{u}, \Vec{Z} \rangle )
    +
    \langle \Vec{v}, \Vec{Z} \rangle
    \sgn( \langle \Vec{v}, \Vec{Z} \rangle )
  \right)
  \\
  &=
  \frac{1}{q}
  \left(
    \left| \langle \Vec{u}, \Vec{Z} \rangle \right|
    +
    \left| \langle \Vec{v}, \Vec{Z} \rangle \right|
  \right)
\end{align}
\end{subequations}
%
Likewise,
\begin{align}
  \RV{X}_{\Vec{u'},\Vec{v'}}
  &=
  \frac{1}{q}
  \left(
    \langle \Vec{u'}, \Vec{Z} \rangle
    \sgn( \langle \Vec{u'}, \Vec{Z} \rangle )
    +
    \langle \Vec{v'}, \Vec{Z} \rangle
    \sgn( \langle \Vec{v'}, \Vec{Z} \rangle )
  \right)
  =
  \frac{1}{q}
  \left(
    \left| \langle \Vec{u'}, \Vec{Z} \rangle \right|
    +
    \left| \langle \Vec{v'}, \Vec{Z} \rangle \right|
  \right)
\end{align}
%
Then, letting
\begin{gather}
  ( \RV{Y}, \RV{Y'} )
  \sim
  \N
  \left(
    \begin{pmatrix}
      0 \\
      0
    \end{pmatrix},
    \begin{pmatrix}
      1 & \cos( \theta_{\Vec{u},\Vec{v}} ) \\
      \cos( \theta_{\Vec{u},\Vec{v}} ) & 1
    \end{pmatrix}
  \right)
  \equiv
  \N
  \left(
    \begin{pmatrix}
      0 \\
      0
    \end{pmatrix},
    \begin{pmatrix}
      1 & \cos( \theta_{\Vec{u'},\Vec{v'}} ) \\
      \cos( \theta_{\Vec{u'},\Vec{v'}} ) & 1
    \end{pmatrix}
  \right)
,\end{gather}
notice that \( \RV{X}_{\Vec{u},\Vec{v}} \) and \( \RV{X}_{\Vec{u'},\Vec{v'}} \), conditioned on
\(
  \RV{R}_{\Vec{u},\Vec{v}}, \RV{R}_{\Vec{u'},\Vec{v'}} \neq 0
\),
both follow the same distribution as
\(
  \frac{1}{q}
  \left(
    | \RV{Y} |
    +
    | \RV{Y'} |
  \right)
\).
Hence, the claim is proved.
\par
We are ready to derive Lemma \ref{lemma:normal:distribution:proj_u-v}.
To simplify notation, we will drop the subscript of \( \Vec{u},\Vec{v} \) on the random variables,
writing
\(
  \RV{X} = \RV{X}_{\Vec{u},\Vec{v}},
  \RV{R} = \RV{R}_{\Vec{u},\Vec{v}}
\).
Let
\(
  \Vec{Z} = ( \Vec*{Z}_{1}, \dots, \Vec*{Z}_{n} ) \sim \N(\Vec{0},\Mat{I}_{n \times n})
\).
%
%
%
For an arbitrary choice of
\(
  \theta \in [0,2\pi)
\),
fix
\(
  \Vec{u}, \Vec{v} \in \Sphere{n}
\)
such that
\(
  \theta_{\Vec{u},\Vec{v}} = \theta
\)
and
\(
  \Vec{u} = ( \Vec*{u}_{1}, \Vec*{u}_{2}, \dots, \Vec*{u}_{n} )
\),
\(
  \Vec{v} = ( -\Vec*{u}_{1}, \Vec*{u}_{2}, \dots, \Vec*{u}_{n} )
\)
with
\(
  \Vec*{u}_{1} > 0
\),
which is made possible by the claim argued above.
This choice will now be shown to induce the distribution of
\(
  ( | \Vec*{Z}_{1} | \Mid \RV{R} \neq 0 )
\)
on the random variable
\(
  ( \RV{X} \Mid \RV{R} \neq 0 )
\).
First, observe that
\begin{gather}
\label{deriv:lemma:normal:distribution:proj_u-v:1}
  \frac
  {\Vec{u} - \Vec{v}}
  {\left\| \Vec{u} - \Vec{v} \right\|_{2}}
  =
  ( 1, 0, \dots, 0 )
\end{gather}
and thus
\begin{gather}
\label{deriv:lemma:normal:distribution:proj_u-v:2}
  \RV{X}
  =
  \left\langle
    \frac{\Vec{u}-\Vec{v}}{\left\| \Vec{u}-\Vec{v} \right\|_{2}},
    \Vec{Z}
  \right\rangle
  \RV{R}
  =
  \Vec*{Z}_{1} \RV{R}
.\end{gather}
%
\BEGINEDIT
Moreover, observe that the event
\(
  \RV{R} \neq 0
\)
implies that
\(
  \sgn( \langle \Vec{u}, \Vec{Z} \rangle )
  =
 - \sgn( \langle \Vec{v}, \Vec{Z} \rangle )
\).
%
Then,
\begin{subequations}
\label{deriv:lemma:normal:distribution:proj_u+v:4}
\begin{align}
  \RV{R}
  &=
  \frac{1}{2}
  \left(
    \Sgn( \langle \Vec{u}, \Vec{Z} \rangle )
    -
    \Sgn( \langle \Vec{v}, \Vec{Z} \rangle )
  \right)
  \\
  &=
  \Sgn
  (
    \Sgn( \langle \Vec{u}, \Vec{Z} \rangle )
    -
    \Sgn( \langle \Vec{v}, \Vec{Z} \rangle )
  )
  \\
  &=
  \Sgn
  (
    \Sgn( \langle \Vec{u}, \Vec{Z} \rangle )
    +
    \Sgn( \langle -\Vec{v}, \Vec{Z} \rangle )
  )
  \\
  &=
  \Sgn( \langle \Vec{u} - \Vec{v}, \Vec{Z} \rangle ).
\end{align}
\end{subequations}
%
Therefore, conditioned on
\(
  \RV{R} \neq 0
\),
by the above observation, \( \RV{R} \) takes the value
\begin{gather}
\label{deriv:lemma:normal:distribution:proj_u-v:3}
  \RV{R}
  =
  \Sgn
  (
    \left\langle
      \frac{\Vec{u}-\Vec{v}}{\left\| \Vec{u}-\Vec{v} \right\|_{2}},
      \Vec{Z}
    \right\rangle
  )
  =
  \Sgn( \Vec*{Z}_{1} )
.\end{gather}
\ENDEDIT
It follows that
\begin{gather}
\label{deriv:lemma:normal:distribution:proj_u-v:4}
  \left( \RV{X} \Mid \RV{R} \neq 0 \right)
  =
  \left(
  \left.
    \left\langle
      \frac{\Vec{u}-\Vec{v}}{\left\| \Vec{u}-\Vec{v} \right\|_{2}},
      \Vec{Z}
    \right\rangle
    \RV{R}
  \right|
    \RV{R} \neq 0
  \right)
  =
  \left( \Vec*{Z}_{1} \RV{R} \Mid \RV{R} \neq 0 \right)
  =
  \left( \Vec*{Z}_{1} \Sgn( \Vec*{Z}_{1} ) \Mid \RV{R} \neq 0 \right)
  =
  \left( | \Vec*{Z}_{1} | \Mid \RV{R} \neq 0 \right)
,\end{gather}
as claimed.
\par
Next, the density function
\(
  f_{\RV{X} \Mid \RV{R} \neq 0} : \R \to \R_{\geq 0}
\)
of the conditioned random variable
\(
  ( \RV{X} \Mid \RV{R} \neq 0 )
\)
is found by deriving the equivalent density function
\(
  f_{| \Vec*{Z}_{1} | \Mid \RV{R} \neq 0} : \R \to \R_{\geq 0}
\).
By Bayes' rule, this density function can be written as
\begin{gather}
\label{deriv:lemma:normal:distribution:proj_u-v:5}
  f_{| \Vec*{Z}_{1} | \Mid \RV{R}}(x \Mid \Value{r} \neq 0)
  =
  \frac
  {f_{| \Vec*{Z}_{1} |}(x) p_{\RV{R} \Mid | \Vec*{Z}_{1} |}(\Value{r} \neq 0 \Mid x)}
  {p_{\RV{R}}(\Value{r} \neq 0)}
,\end{gather}
which expresses
\(
  f_{| \Vec*{Z}_{1} | \Mid \RV{R} \neq 0}
\)
using three more manageable density (mass) functions.
Beginning with \( p_{\RV{R}}(\Value{r} \neq 0) \), let the random variable \( \RV{I} \) be the indicator
of the event
\(
  \RV{R} \neq 0
\),
formally,
\(
  \RV{I} = \I{\RV{R} \neq 0}
\).
Observing the following biconditionals
\begin{align}
\label{deriv:lemma:normal:distribution:proj_u-v:6}
  \RV{R} \neq 0
  \iff
  \frac{1}{2}
  \left(
    \Sgn( \langle \Vec{u}, \Vec{Z} \rangle ) - \Sgn( \langle \Vec{v}, \Vec{Z} \rangle )
  \right)
  \neq 0
  \iff
  \left(
    \Sgn( \langle \Vec{u}, \Vec{Z} \rangle ) - \Sgn( \langle \Vec{v}, \Vec{Z} \rangle )
  \right)
  \neq 0
,\end{align}
it follows that
\begin{subequations}
\label{deriv:lemma:normal:distribution:proj_u-v:7}
\begin{gather}
  \label{deriv:lemma:normal:distribution:proj_u-v:7:1}
  \RV{I} = \I{\RV{R} \neq 0}
  \\ \label{deriv:lemma:normal:distribution:proj_u-v:7:2}
  \RV{I} = \I{\frac{1}{2} \Sgn( \langle \Vec{u}, \Vec{Z} \rangle )
       - \Sgn( \langle \Vec{v}, \Vec{Z} \rangle ) \neq 0}
  \\ \label{deriv:lemma:normal:distribution:proj_u-v:7:3}
  \RV{I} = \I{\Sgn( \langle \Vec{u}, \Vec{Z} \rangle ) - \Sgn( \langle \Vec{v}, \Vec{Z} \rangle ) \neq 0}
\end{gather}
\end{subequations}
are equivalent definitions for the random variable \( \RV{I} \).
Then, the mass associated with
\(
  \RV{R} \neq 0
\)
is
\(
  p_{\RV{R}}(\Value{r} \neq 0)
  = \Pr ( \RV{I} = 1 )
  = \frac{\theta_{\Vec{u},\Vec{v}}}{\pi}
\),
where the last equality follows from Lemma \ref{lemma:prob-normal-vector-mismatch}, stated below.
%
\begin{lemma}[\cite{charikar2002similarity}]
\label{lemma:prob-normal-vector-mismatch}
Fix any pair of real-valued vectors
\(
  \Vec{u},\Vec{v} \in \R^{n}
\),
and suppose
\(
  \Vec{Z} \sim \N(\Vec{0},\Mat{I}_{n \times n})
\)
is a standard normal vector with i.i.d. entries.
Define the indicator random variable
\(
  \RV{I} = \I{\sgn(\langle \Vec{u},\Vec{Z} \rangle) - \sgn(\langle \Vec{v},\Vec{Z} \rangle) \neq 0}
\).
Then,
\begin{gather}
\label{eqn:prob-normal-vector-mismatch}
  \Pr( \RV{I} = 1 ) = \frac{\theta_{\Vec{u},\Vec{v}}}{\pi}
.\end{gather}
\end{lemma}
%

%
\par
In short, the above argument yields
\(
  p_{\RV{R}}(\Value{r} \neq 0) = \Pr ( \RV{I} = 1 ) = \frac{\theta_{\Vec{u},\Vec{v}}}{\pi}
\).
\par
Next, the density function for the random variable
\(
  | \Vec*{Z}_{1} |
\),
which is the absolute value of the standard normal random variable
\(
  \Vec*{Z}_{1}
\),
is the well-known folded standard normal distribution and takes the form
\begin{align}
\label{deriv:lemma:normal:distribution:proj_u-v:8}
  f_{| \Vec*{Z}_{1} |}(x)
  &=
  \begin{cases}
    f_{\Vec*{Z}_{1}}(-x) + f_{\Vec*{Z}_{1}}(x),
    &\cIf x \geq 0,
    \\
    0, &\cIf x < 0.
  \end{cases}
  \\
  &=
  \begin{cases}
    \frac{1}{\sqrt{2\pi}} e^{-\frac{(-x)^{2}}{2}} + \frac{1}{\sqrt{2\pi}} e^{-\frac{x^{2}}{2}},
    &\cIf x \geq 0,
    \\
    0, &\cIf x < 0.
  \end{cases}
  \\
\end{align}
In summary,
\begin{gather}
\label{deriv:lemma:normal:distribution:proj_u-v:8:b}
  f_{| \Vec*{Z}_{1} |}(x)
  =
  \begin{cases}
    \sqrt{\frac{2}{\pi}} e^{-\frac{x^{2}}{2}}, &\cIf x \geq 0,
    \\
    0, &\cIf x < 0.
  \end{cases}
\end{gather}
%
\par
Lastly, consider the mass function of
\(
  ( \RV{R} \Mid | \Vec*{Z}_{1} | )
\),
which need only be evaluated when
\(
  \RV{R} \neq 0
\).
The next argument will show that
\begin{gather}
\label{deriv:lemma:normal:distribution:proj_u-v:9}
  p_{\RV{R} \Mid | \Vec*{Z}_{1} |}(\Value{r} \neq 0 \Mid x)
  =
  \frac{1}{\sqrt{2\pi}} \int_{y=-\Alpha(x)}^{y=\Alpha(x)} e^{-\frac{y^{2}}{2}} dy
\end{gather}
where
\(
  \Alpha : \R \to \R
\)
is as defined in the lemma (and repeated here for convenience):
\begin{gather}
\label{deriv:lemma:normal:distribution:proj_u-v:10}
  \Alpha(x)
  = x \Tan( \frac{\theta_{\Vec{u},\Vec{v}}}{2} )
  = x \sqrt{\frac{\DistS[^{2}]{\Vec{u}}{\Vec{v}}}{4 - \DistS[^{2}]{\Vec{u}}{\Vec{v}}}}
.\end{gather}
%
Notice that by basic geometry, given
\(
  | \Vec*{Z}_{1} | = x
\),
\(
  x \geq 0
\),
the event
\(
  \RV{R} \neq 0
\)
occurs precisely when
\begin{gather}
\label{deriv:lemma:normal:distribution:proj_u-v:11}
  \left\langle
    \frac{\Vec{u}+\Vec{v}}{\left\| \Vec{u}+\Vec{v} \right\|_{2}},
    \Vec{Z}
  \right\rangle
  \in
  \left[
    -x \Tan( \frac{\theta_{\Vec{u},\Vec{v}}}{2} ),
    x \Tan( \frac{\theta_{\Vec{u},\Vec{v}}}{2} )
  \right]
\end{gather}
where
\(
  \Tan( \frac{\theta_{\Vec{u},\Vec{v}}}{2} )
\)
can be expressed as follows by using the half-angle trigonometric formula
(applied in \eqref{deriv:lemma:normal:distribution:proj_u-v:12:1}):
\begin{subequations}
\begin{align}
\label{deriv:lemma:normal:distribution:proj_u-v:12:1}
  \Tan( \frac{\theta_{\Vec{u},\Vec{v}}}{2} )
  &=
  \sqrt
  {
    \frac
    {1 - \Cos( \theta_{\Vec{u},\Vec{v}} )}
    {1 + \Cos( \theta_{\Vec{u},\Vec{v}} )}
  }
  \\
  &=
  \sqrt
  {
    \frac
    {1 - \Cos( \Arccos( 1 - \frac{\DistS[^{2}]{\Vec{u}}{\Vec{v}}}{2} ) )}
    {1 + \Cos( \Arccos( 1 - \frac{\DistS[^{2}]{\Vec{u}}{\Vec{v}}}{2} ) )}
  }
  \\
  &=
  \sqrt
  {
    \frac
    {\frac{\DistS[^{2}]{\Vec{u}}{\Vec{v}}}{2}}
    {2 - \frac{\DistS[^{2}]{\Vec{u}}{\Vec{v}}}{2}}
  }
  \\
  &=
  \sqrt
  {
    \frac
    {\DistS[^{2}]{\Vec{u}}{\Vec{v}}}
    {4 - \DistS[^{2}]{\Vec{u}}{\Vec{v}}}
  }
  \\
  &=
  \frac{\Alpha(x)}{x}
\end{align}
\end{subequations}
Thus,
\begin{subequations}
\begin{align}
\label{deriv:lemma:normal:distribution:proj_u-v:13}
  p_{\RV{R}}( \Value{r} \neq 0 )
  &=
  \Pr
  \left(
    \left\langle
      \frac{\Vec{u}+\Vec{v}}{\left\| \Vec{u}+\Vec{v} \right\|_{2}},
      \Vec{Z}
    \right\rangle
    \in
    \left[
      -x \Tan( \frac{\theta_{\Vec{u},\Vec{v}}}{2} ),
      x \Tan( \frac{\theta_{\Vec{u},\Vec{v}}}{2} )
    \right]
  \right)
  \\
  &=
  \Pr
  \left(
    \left\langle
      \frac{\Vec{u}+\Vec{v}}{\left\| \Vec{u}+\Vec{v} \right\|_{2}},
      \Vec{Z}
    \right\rangle
    \in
    \left[
      -x \frac{\Alpha(x)}{x},
      x \frac{\Alpha(x)}{x}
    \right]
  \right)
  \\
  &=
  \Pr
  \left(
    \left\langle
      \frac{\Vec{u}+\Vec{v}}{\left\| \Vec{u}+\Vec{v} \right\|_{2}},
      \Vec{Z}
    \right\rangle
    \in
    \left[ -\Alpha(x), \Alpha(x) \right]
  \right)
\end{align}
\end{subequations}
But \( \Vec{Z} \) is invariant under inner products with unit vectors,
and hence, the distribution of
\(
  \left\langle
    \frac{\Vec{u}+\Vec{v}}{\left\| \Vec{u}+\Vec{v} \right\|_{2}},
    \Vec{Z}
  \right\rangle
\)
follows that of
\(
  \left\langle
    \frac{\Vec{u}+\Vec{v}}{\left\| \Vec{u}+\Vec{v} \right\|_{2}},
    \Vec{Z}
  \right\rangle
  \sim
  \N(0,1)
\).
Therefore,
\begin{gather}
  p_{\RV{R}}( \Value{r} \neq 0 )
  =
  \Pr_{Y \sim \N(0,1)}( Y \in \left[ -\Alpha(x), \Alpha(x) \right] )
  =
  \frac{1}{\sqrt{2\pi}} \int_{y=-\Alpha(x)}^{y=\Alpha(x)} e^{-\frac{y^{2}}{2}} dy
,\end{gather}
as claimed.
\par
Combining the above derivations, the density function of
\(
  | \Vec*{Z}_{1} | \Mid \RV{R} \neq 0
\)
is obtained via \eqref{deriv:lemma:normal:distribution:proj_u-v:5}:
\begin{subequations}
\begin{align}
\label{deriv:lemma:normal:distribution:proj_u-v:14}
  f_{| \Vec*{Z}_{1} | \Mid \RV{R}}(x \Mid \Value{r} \neq 0)
  &=
  \frac
  {f_{| \Vec*{Z}_{1} |}(x) p_{\RV{R} \Mid | \Vec*{Z}_{1} |}(\Value{r} \neq 0 \Mid x)}
  {p_{\RV{R}}(\Value{r} \neq 0)}
  =
  \frac
  {\sqrt{\frac{2}{\pi}} e^{-\frac{x^{2}}{2}}
   \cdot
   \frac{1}{\sqrt{2\pi}} \int_{y=-\Alpha(x)}^{y=\Alpha(x)} e^{-\frac{y^{2}}{2}} dy}
  {\frac{\theta_{\Vec{u},\Vec{v}}}{\pi}}
  \\
  &=
  \frac{\pi}{\theta_{\Vec{u},\Vec{v}}}
  \sqrt{\frac{2}{\pi}} e^{-\frac{x^{2}}{2}}
  \cdot
  \frac{1}{\sqrt{2\pi}} \int_{y=-\Alpha(x)}^{y=\Alpha(x)} e^{-\frac{y^{2}}{2}} dy
\end{align}
\end{subequations}
if
\(
  x \geq 0
\),
and
\(
  f_{| \Vec*{Z}_{1} | \Mid \RV{R}}( x \Mid \Value{r} \neq 0 ) = 0
\)
if
\(
  x < 0
\),
where the support of \( f_{| \Vec*{Z}_{1} | \Mid \RV{R}} \) is restricted to the interval
\( [0,\infty) \) due the the latter case in
\eqref{deriv:lemma:normal:distribution:proj_u-v:8}.
\par
The remaining task is finding the expectation of
\(
  ( \RV{X} \Mid \RV{R} \neq 0 )
\)
to verify \eqref{eqn:normal:distribution:proj_u-v:expectation}, which is done by a direct
calculation using the density function,
\eqref{eqn:normal:distribution:proj_u-v:density-function}, that was just proved:
\begin{subequations}
\begin{align}
  \E( \RV{X} \Mid \RV{R} \neq 0 )
  &=
  \int_{-\infty}^{\infty}
  x f_{| \Vec*{Z}_{1} | \Mid \RV{R}}( x \Mid \Value{r} \neq 0 ) dx
  \\
  &=
  \lim_{t \to \infty}
  \int_{x=0}^{x=t}
  \frac{\pi}{\theta_{\Vec{u},\Vec{v}}}
  \sqrt{\frac{2}{\pi}}
  x e^{-\frac{x^{2}}{2}}
  \cdot
  \frac{1}{\sqrt{2\pi}} \int_{y=-\Alpha(x)}^{y=\Alpha(x)} e^{-\frac{y^{2}}{2}}
  dy dx
  \\
  &=
  \frac{\pi}{\theta_{\Vec{u},\Vec{v}}}
  \sqrt{\frac{2}{\pi}}
  \frac{\DistS{\Vec{u}}{\Vec{v}}}{2}
  \\
  &=
  \sqrt{\frac{\pi}{2}}
  \frac{\DistS{\Vec{u}}{\Vec{v}}}{\theta_{\Vec{u},\Vec{v}}}
\end{align}
\end{subequations}
as claimed.
\end{proof}

\begin{lemma}
\label{lemma:normal:distribution:proj_u+v}
Fix an ordered pair of real-valued vectors,
\(
  ( \Vec{u}, \Vec{v} ) \in \Sphere{n} \times \Sphere{n}
\),
of unit norm.
Let
\(
  \Vec{Z} \sim \N(\Vec{0},\Mat{I}_{n \times n})
\)
be a standard normal random vector,
and let
\(
  \RV{R}_{\Vec{u},\Vec{v}}
\)
be a discrete random variable given by
\(
  \RV{R}_{\Vec{u},\Vec{v}}
  =
  \frac{1}{2}
  \left(
    \sgn( \langle \Vec{u}, \Vec{Z} \rangle )
    -
    \sgn( \langle \Vec{v}, \Vec{Z} \rangle )
  \right)
\),
which takes values in
\(
  \{ -1, 0, 1 \}
\).
Then, the distribution of the random variable
\(
  \RV{Y}_{\Vec{u},\Vec{v}}
  =
  \left\langle
    \frac{\Vec{u}+\Vec{v}}{\left\| \Vec{u}+\Vec{v} \right\|_{2}},
    \Vec{Z}
  \right\rangle
  \RV{R}_{\Vec{u},\Vec{v}}
\)
conditioned on
\(
  \RV{R}_{\Vec{u},\Vec{v}} \neq 0
\)
is standard normal, i.e.,
\(
  ( \RV{Y}_{\Vec{u},\Vec{v}} \Mid \RV{R}_{\Vec{u},\Vec{v}} \neq 0 ) \sim \N(0,1)
\).
\end{lemma}

\begin{proof}
{Lemma \ref{lemma:normal:distribution:proj_u+v}}
\label{pf:lemma:normal:distribution:proj_u+v}
Analogously to the claim in the proof of Lemma \ref{lemma:normal:distribution:proj_u-v},
it can be shown that for
\(
  \Vec{u}, \Vec{v}, \Vec{u'}, \Vec{v'} \in \Sphere{n}
\),
such that
\(
  \theta_{\Vec{u},\Vec{v}} = \theta_{\Vec{u'},\Vec{v'}}
\),
the random variables
\( ( \RV{Y}_{\Vec{u},\Vec{v}} \Mid \RV{R}_{\Vec{u},\Vec{v}} = 0 ) \)
and
\( ( \RV{Y}_{\Vec{u'},\Vec{v'}} \Mid \RV{R}_{\Vec{u'},\Vec{v'}} = 0 ) \)
follow the same distribution, as do
\( ( \RV{Y}_{\Vec{u},\Vec{v}} \Mid \RV{R}_{\Vec{u},\Vec{v}} \neq 0 ) \)
and
\( ( \RV{Y}_{\Vec{u'},\Vec{v'}} \Mid \RV{R}_{\Vec{u'},\Vec{v'}} \neq 0 ) \).
We will omit the formal argument since it is nearly identical to that provided in
the proof of Lemma \ref{lemma:normal:distribution:proj_u-v}.
\par
Fix any
\(
  \theta \in [0,2\pi)
\),
and let
\(
  \Vec{u}
  = ( \Vec*{u}_{1}, \dots, \Vec*{u}_{n} )
  \in \Sphere{n}
\)
and take
\(
  \Vec{v} = ( \Vec*{u}_{1}, -\Vec*{u}_{2} \dots, -\Vec*{u}_{n} )
\)
such that
\(
  \Vec*{u}_{1} > 0
\)
and
\(
  \theta_{\Vec{u},\Vec{v}} = \theta
\).
This construction yields
\begin{gather}
\label{deriv:lemma:normal:distribution:proj_u+v:1}
  \frac
  {\Vec{u} + \Vec{v}}
  {\left\| \Vec{u} + \Vec{v} \right\|_{2}}
  =
  ( 1, 0, \dots, 0 )
\end{gather}
as well as
\begin{gather}
\label{deriv:lemma:normal:distribution:proj_u+v:2}
  \Vec{u} - \Vec{v} \propto ( 0, \Vec*{u}_{2}, \dots \Vec*{u}_{n} )
\end{gather}
%
We will again drop the subscript \( \Vec{u},\Vec{v} \) from the random variables for simplicity
and denote
\(
  \RV{Y} = \RV{Y}_{\Vec{u},\Vec{v}},
  \RV{R} = \RV{R}_{\Vec{u},\Vec{v}}
\).
From \eqref{deriv:lemma:normal:distribution:proj_u+v:1}, it follows that
\begin{gather}
\label{deriv:lemma:normal:distribution:proj_u+v:3}
  \RV{X}
  =
  \left\langle
    \frac{\Vec{u}+\Vec{v}}{\left\| \Vec{u}+\Vec{v} \right\|_{2}},
    \Vec{Z}
  \right\rangle
  =
  \Vec*{Z}_{1}
\end{gather}
%
\BEGINEDIT
On the other hand, recall from the proof of \LEMMA \ref{lemma:normal:distribution:proj_u-v} that the event
\(
  \RV{R} \neq 0
\)
implies that \(  \RV{R} = \Sgn( \langle \Vec{u} - \Vec{v}, \Vec{Z} \rangle )  \).
\ENDEDIT
But recall from \eqref{deriv:lemma:normal:distribution:proj_u+v:2} that
\(
  \Vec{u} - \Vec{v} \propto ( 0, \Vec*{u}_{2}, \dots \Vec*{u}_{n} )
\),
and thus, given
\(
  \RV{R} \neq 0
\),
\begin{gather}
\label{deriv:lemma:normal:distribution:proj_u+v:5}
  \RV{R}
  =
  \Sgn( \langle \Vec{u} - \Vec{v}, \Vec{Z} \rangle )
  =
  \Sgn( \langle ( 0, \Vec*{u}_{2}, \dots, \Vec*{u}_{n} ), \Vec{Z} \rangle )
\end{gather}
which implies conditional independence of
\(
  ( \RV{R} \Mid \RV{R} \neq 0 )
\)
and
\(
  ( \Vec*{Z}_{1} \Mid \RV{R} \neq 0 )
  =
  ( \RV{X} \Mid \RV{R} \neq 0 )
\).
Then,
\(
  ( \RV{Y} \Mid \RV{R} \neq 0 )
  =
  ( \RV{X} \RV{R} \Mid \RV{R} \neq 0 )
  =
  ( \Vec*{Z}_{1} \RV{R} \Mid \RV{R} \neq 0 )
\),
and so
\(
  ( \RV{Y} \Mid \RV{R} \neq 0 )
\)
follows the same distribution as either the random variable
\(
  \RV{Z'}
\)
or
\(
  -\RV{Z'}
\),
where
\(
  \RV{Z'} \sim \N(0,1)
\).
But it is well-known that the standard normal random variable
\(
  \RV{Z'}
\)
and its negation
\(
  -\RV{Z'}
\)
have the same distribution, implying that
\(
  ( \RV{Y} \Mid \RV{R} \neq 0 ) \sim \N(0,1)
\),
as claimed.
\end{proof}

\begin{lemma}
\label{lemma:normal:distribution:third}
Fix an ordered pair of real-valued unit vectors,
\(
  ( \Vec{u}, \Vec{v} ) \in \Sphere{n} \times \Sphere{n}
\),
and let
\(
  \Vec{w} \in \Sphere{n} \cap \Span( \{ \Vec{u} - \Vec{v}, \Vec{u} + \Vec{v} \} )^{\perp}
\)
be any real-valued unit vector in the orthogonal complement of
\( \Span( \{ \Vec{u} - \Vec{v}, \Vec{u} + \Vec{v} \} ) \).
Let
\(
  \Vec{Z} \sim \N(\Vec{0},\Mat{I}_{n \times n})
\)
be a standard normal random vector, let \( \Vec{Y} \) be the random vector given by
\begin{gather}
  \Vec{Y}
  =
  \Vec{Z}
  -
  \left\langle
    \frac{\Vec{u} - \Vec{v}}{\left\| \Vec{u} - \Vec{v} \right\|_{2}},
    \Vec{Z}
  \right\rangle
  \frac{\Vec{u} - \Vec{v}}{\left\| \Vec{u} - \Vec{v} \right\|_{2}}
  -
  \left\langle
    \frac{\Vec{u} + \Vec{v}}{\left\| \Vec{u} + \Vec{v} \right\|_{2}},
    \Vec{Z}
  \right\rangle
  \frac{\Vec{u} + \Vec{v}}{\left\| \Vec{u} + \Vec{v} \right\|_{2}}
,\end{gather}
and let
\(
  \RV{R}
\)
be the (discrete) random variable taking values in
\(
  \{ -1, 0, 1 \}
\)
and given by
$$
  \RV{R}
  =
  \frac{1}{2}
  \left(
    \sgn( \langle \Vec{u}, \Vec{Z} \rangle )
    -
    \sgn( \langle \Vec{v}, \Vec{Z} \rangle )
  \right).
$$
Then, the random vector
\(
  \RV{X} = \langle \Vec{w}, \Vec{Y} \rangle \RV{R}
\)
conditioned on
\(
  \RV{R} \neq 0
\)
is standard normal, i.e.,
\(
  ( \RV{X} \Mid \RV{R} \neq 0 ) \sim \N(0,1)
\).
\end{lemma}

\begin{proof}
{Lemma \ref{lemma:normal:distribution:third}}
\label{pf:lemma:normal:distribution:third}
As in the previous two lemmas, the ordered pair of unit vectors
\(
  ( \Vec{u}, \Vec{v} ) \in \Sphere{n} \times \Sphere{n}
\)
can be chosen nonarbitrarily due to the rotational invariance of the standard normal distribution
and the argument laid out in the proof of Lemma \ref{lemma:normal:distribution:proj_u-v}.
For the purposes of this proof, we will select \( \Vec{u} \) and \( \Vec{v} \) as follows.
For any pair of constants
\(
  \Const{p}, \Const{q}
\),
subject to
\(
  \Const{p}^{2} + \Const{q}^{2} = 1
\),
set
\(
  \Vec{u} = ( \Const{p}, \Const{q}, 0, \dots, 0 )
\)
and
\(
  \Vec{v} = ( -\Const{p}, \Const{q}, 0, \dots, 0 )
\).
Note that
\begin{gather}
  \left\| \Vec{u} \right\|_{2} = \left\| \Vec{v} \right\|_{2} = 1
  \\
  \Vec{u} - \Vec{v} = ( 2 \Const{p}, 0, \dots, 0 )
  ,\quad
  \frac{\Vec{u} - \Vec{v}}{\left\| \Vec{u} - \Vec{v} \right\|_{2}}
  = ( 1, 0, \dots, 0 )
  = \Vec{e}_{1}
  \\
  \Vec{u} + \Vec{v} = ( 0, 2 \Const{q}, \dots, 0 )
  ,\quad
  \frac{\Vec{u} + \Vec{v}}{\left\| \Vec{u} + \Vec{v} \right\|_{2}}
  = ( 0, 1, \dots, 0 )
  = \Vec{e}_{2}
\end{gather}
where
\(
  \Vec{e}_{1} = (1,0,\dots,0),
  \Vec{e}_{1} = (0,1,\dots,0)
  \in \R^{n}
\)
are the first and second standard basis vectors or \( \R^{n} \).
Fix any
\(
  \Vec{w} \in \Sphere{n} \cap \Span( \{ \Vec{u} - \Vec{v}, \Vec{u} + \Vec{v} \} )^{\perp}
\).
Then,
\begin{align}
  \Vec{Y}
  &=
  \Vec{Z}
  -
  \left\langle
    \frac{\Vec{u} - \Vec{v}}{\left\| \Vec{u} - \Vec{v} \right\|_{2}},
    \Vec{Z}
  \right\rangle
  \frac{\Vec{u} - \Vec{v}}{\left\| \Vec{u} - \Vec{v} \right\|_{2}}
  -
  \left\langle
    \frac{\Vec{u} + \Vec{v}}{\left\| \Vec{u} + \Vec{v} \right\|_{2}},
    \Vec{Z}
  \right\rangle
  \frac{\Vec{u} + \Vec{v}}{\left\| \Vec{u} + \Vec{v} \right\|_{2}}
  \\
  &=
  \Vec{Z}
  -
  \Vec*{Z}_{1} \Vec{e}_{1}
  -
  \Vec*{Z}_{2} \Vec{e}_{2}
  \\
  &=
  ( 0, 0, \Vec*{Z}_{3}, \dots, \Vec*{Z}_{n} )
\end{align}
%
Notice that
\(
  \Span( \{ \Vec{u} - \Vec{v}, \Vec{u} + \Vec{v} \} )
  =
  \Span( \{ \Vec{e}_{1}, \Vec{e}_{2} \} )
\)
and
\(
  \Span( \{ \Vec{u} - \Vec{v}, \Vec{u} + \Vec{v} \} )^{\perp}
  =
  \Span( \{ \Vec{e}_{1}, \Vec{e}_{2} \} )^{\perp}
  =
  \{ \Vec{x} \in \R^{n} : \Vec*{x}_{1} = \Vec*{x}_{2} = 0 \}
\).
Then, writing
\(
  \Vec{\tilde{Z}} = ( \Vec*{Z}_{3}, \dots, \Vec*{Z}_{n} )
\)
and
\(
  \Vec{\tilde{w}} = ( \Vec*{w}_{3}, \dots, \Vec*{w}_{n} )
\),
the random variable
\(
  \langle \Vec{w}, \Vec{Y} \rangle
\)
follows the same distribution as
\(
  \langle \Vec{\tilde{w}}, \Vec{\tilde{Z}} \rangle
  =
  \langle \frac{\Vec{\tilde{w}}}{\| \Vec{\tilde{w}} \|_{2}}, \Vec{\tilde{Z}} \rangle
\)
with
\(
  \| \Vec{\tilde{w}} \|_{2} = 1
\).
But it is well-known that
\(
  \langle \Vec{\tilde{w}}, \Vec{\tilde{Z}} \rangle
  \sim
  \N(0,1)
\).
\par
Recall the definition of the random variable
\(
  \RV{R}
  =
  \frac{1}{2}
  ( \sgn( \langle \Vec{u}, \Vec{Z} \rangle ) - \sgn( \langle \Vec{v}, \Vec{Z} \rangle ) )
\).
Because
\(
  \Vec{u}, \Vec{v} \in \Span( \{ \Vec{u} - \Vec{v}, \Vec{u} + \Vec{v} \} )
\),
the random variable \( \RV{R} \) is entirely dependent on the projection of \( \Vec{Z} \) onto
\( \Span( \{ \Vec{u} - \Vec{v}, \Vec{u} + \Vec{v} \} ) \) and hence independent of its projection
onto \( \Span( \{ \Vec{u} - \Vec{v}, \Vec{u} + \Vec{v} \} )^{\perp} \).
More formally,
\begin{subequations}
\begin{align}
  \RV{R}
  &=
  \frac{1}{2}
  \left( \sgn( \langle \Vec{u}, \Vec{Z} \rangle ) - \sgn( \langle \Vec{v}, \Vec{Z} \rangle ) \right)
  \\
  &=
  \frac{1}{2}
  \left(
    \sgn( \Const{p} \Vec*{Z}_{1} + \Const{q} \Vec*{Z}_{2} )
    -
    \sgn( -\Const{p} \Vec*{Z}_{1} + \Const{q} \Vec*{Z}_{2} )
  \right)
\end{align}
\end{subequations}
and thus, \( \RV{R} \) depends only on the random variables \( \Vec*{Z}_{1} \) and
\( \Vec*{Z}_{2} \).
However, it was already noted that
\(
  \Span( \{ \Vec{u} - \Vec{v}, \Vec{u} + \Vec{v} \} )^{\perp}
  =
  \{ \Vec{x} \in \R^{n} : \Vec*{x}_{1} = \Vec*{x}_{2} = 0 \}
\),
which implies that the projection \( \RV{Y} \) depend only on a (possibly improper) subset of
\(
  \{ \Vec*{Z}_{j} \}_{j \in [n] \setminus \{1,2\}}
\).
The independence of \( \RV{Y} \) and \( \RV{R} \) follows.
Then, the conditioned random variable
\(
  ( \RV{X} \Mid \RV{R} \neq 0 ) = ( \langle \Vec{w}, \Vec{Y} \rangle \RV{R} \Mid \RV{R} \neq 0 )
\)
is equivalent to either
\(
  \langle \Vec{w}, \Vec{Y} \rangle \RV{R}
\)
or
\(
  -\langle \Vec{w}, \Vec{Y} \rangle \RV{R}
\),
both of which follow the standard normal distribution.
Hence,
\(
  ( \RV{X} \Mid \RV{R} \neq 0 ) \sim \N(0,1)
\).
\end{proof}


\subsubsection{Concentration inequalities for orthogonal projections of normal vectors}
\label{outline:normal|>concentration-ineq-pfs|>orthogonal-projections|>main-lemma-pf}

We are ready to prove
Lemmas \ref{lemma:normal:concentration-ineq:proj_u-v}-%
\ref{lemma:normal:concentration-ineq:third}.
Note that the subscripts \( \Vec{u},\Vec{v} \) are dropped from some random variables for ease of
notation.

\begin{proof}
{Lemma \ref{lemma:normal:concentration-ineq:proj_u-v}}
\label{pf:lemma:normal:concentration-ineq:proj_u-v}
Using the linearity of inner products, the random variable \( \RV{X} \) can be written as
\begin{gather}
  \RV{X}
  =
  \left\langle
    \frac{\Vec{u} - \Vec{v}}{\left\| \Vec{u} - \Vec{v} \right\|_{2}},
    \sum_{i=1}^{m}
    \Vec{Z}^{(i)} \RV{R_{i;\Vec{u},\Vec{v}}}
  \right\rangle
  =
  \sum_{i=1}^{m}
  \left\langle
    \frac{\Vec{u} - \Vec{v}}{\left\| \Vec{u} - \Vec{v} \right\|_{2}},
    \Vec{Z}^{(i)} \RV{R_{i;\Vec{u},\Vec{v}}}
  \right\rangle
  =
  \sum_{i=1}^{m} \RV{X}_{i}
,\end{gather}
where the random variables
\(
  \RV{X}_{i}
  =
  \left\langle
    \frac{\Vec{u} - \Vec{v}}{\left\| \Vec{u} - \Vec{v} \right\|_{2}},
    \Vec{Z}^{(i)} \RV{R_{i;\Vec{u},\Vec{v}}}
  \right\rangle
\),
\(
  i \in [m]
\),
are i.i.d. and have (conditional) distributions formally defined in
Lemma \ref{lemma:normal:distribution:proj_u-v}.
The concentration inequality will follow from
\begin{EnumerateInline}
\item \label{enum:lemma:normal:distribution:proj_u-v:1}
controlling the MGF,
\(
  \psi_{\RV{X}_{i} - \mu \Mid \RV{R_{i;\Vec{u},\Vec{v}}} \neq 0}
\),
of each zero-mean i.i.d. random variable
\(
  ( \RV{X}_{i} - \mu \Mid \RV{R_{i;\Vec{u},\Vec{v}}} \neq 0 )
\),
such that
\(
  \psi_{\RV{X}_{i} - \mu \Mid \RV{R_{i;\Vec{u},\Vec{v}}} \neq 0}(s) \leq e^{\frac{s^{2}}{2}}
\).
The negation of this random variable,
\(
  ( -\RV{X}_{i} + \mu \Mid \RV{R_{i;\Vec{u},\Vec{v}}} \neq 0 )
\),
is handled likewise.
\item \label{enum:lemma:normal:distribution:proj_u-v:2}
Then, \EDIT{writing \(  \RV{L}_{\Vec{u},\Vec{v}} \defeq \| \Rhatuv \|_{0}  \),} the MGFs of
\(
  ( \RV{X} - \E{}[ \RV{X} ] \Mid \Rhatuv, \RV{L}_{\Vec{u},\Vec{v}} )
\)
and
\(
  ( -\RV{X} + \E{}[ \RV{X} ] \Mid \Rhatuv, \RV{L}_{\Vec{u},\Vec{v}} )
\)
follow from step \ref{enum:lemma:normal:distribution:proj_u-v:1} and the i.i.d. property of
\(
  \{ \RV{X}_{i} \}_{i \in [m]}
\).
\item \label{enum:lemma:normal:distribution:proj_u-v:3}
Lastly, two Chernoff bounds using the MGFs found in
step \ref{enum:lemma:normal:distribution:proj_u-v:2} will yield the lemma's two-sided bound.
in \eqref{eqn:normal:distribution:proj_u-v:inequality}.
\end{EnumerateInline}
%
\par
Beginning with the derivation of the MGF of the i.i.d. random variables, as outlined in
step \ref{enum:lemma:normal:distribution:proj_u-v:1}, fix any
\(
  i \in [m]
\)
such that
\(
  \RV{R_{i;\Vec{u},\Vec{v}}} \neq 0
\).
Then, the density function of
\(
  ( \RV{X}_{i} \Mid \RV{R_{i;\Vec{u},\Vec{v}}} \neq 0 )
\)
is given in \EQN \eqref{eqn:normal:distribution:proj_u-v:density-function} of
Lemma \ref{lemma:normal:distribution:proj_u-v}:
\begin{gather}
\label{deriv:lemma:normal:concentration-ineq:proj_u-v:1}
  f_{\RV{X}_{i} \Mid \RV{R_{i;\Vec{u},\Vec{v}}}}( x \Mid \Value{r} \neq 0 )
  =
  \begin{cases}
    \frac{\pi}{\theta_{\Vec{u},\Vec{v}}}
    \sqrt{\frac{2}{\pi}}
    e^{-\frac{x^{2}}{2}}
    \cdot
    \frac{1}{\sqrt{2\pi}}
    \int_{y=-\Alpha(x)}^{y=\Alpha(x)}
    e^{-\frac{y^{2}}{2}}
    dy
    , &\cIf x \geq 0,
    \\
    0, &\cIf x < 0.
  \end{cases}
\end{gather}
with
\begin{gather}
\label{deriv:lemma:normal:concentration-ineq:proj_u-v:2}
  \mu \defeq \E( \RV{X}_{i} \Mid \RV{R_{i;\Vec{u},\Vec{v}}} \neq 0 )
  =
  \sqrt{\frac{\pi}{2}}
  \frac{\DistS{\Vec{u}}{\Vec{v}}}{\theta_{\Vec{u},\Vec{v}}}
,\end{gather}
as specified in \eqref{eqn:normal:distribution:proj_u-v:expectation} of
Lemma \ref{lemma:normal:distribution:proj_u-v}.
The MGF of \( ( \RV{X}_{i} \Mid \RV{R_{i;\Vec{u},\Vec{v}}} \neq 0 ) \) at
\(
  s \geq 0
\)
is then bounded from above by
\begin{gather}
\label{deriv:lemma:normal:concentration-ineq:proj_u-v:3}
  \psi_{\RV{X}_{i} - \mu \Mid \RV{R_{i;\Vec{u},\Vec{v}}} \neq 0}(s)
  \leq
  e^{\frac{s^{2}}{2}}
\end{gather}
as derived next in \eqref{deriv:lemma:normal:concentration-ineq:proj_u-v:4}.
%
\begin{subequations}
\label{deriv:lemma:normal:concentration-ineq:proj_u-v:4}
\begin{align}
  \label{deriv:lemma:normal:concentration-ineq:proj_u-v:4:1}
  \psi_{\RV{X}_{i} - \mu \Mid \RV{R_{i;\Vec{u},\Vec{v}}} \neq 0}(s)
  &=
  \E
  \left[
    e^{s ( \RV{X}_{i} - \E( \RV{X}_{i} \Mid \RV{R_{i;\Vec{u},\Vec{v}}} \neq 0 ) )}
    \middle| \RV{R_{i;\Vec{u},\Vec{v}}} \neq 0
  \right]
  \\
  &=
  \E
  \left[
    e^{s ( \RV{X}_{i} - \mu )}
    \middle| \RV{R_{i;\Vec{u},\Vec{v}}} \neq 0
  \right]
  \\
  &=
  e^{-s \mu}
  \E
  \left[
    e^{s \RV{X}_{i}}
    \middle| \RV{R_{i;\Vec{u},\Vec{v}}} \neq 0
  \right]
  \\
  &=
  e^{-s \mu}
  \int_{x=-\infty}^{x=\infty}
  e^{sx}
  f_{\RV{X}_{i} \Mid \RV{R_{i;\Vec{u},\Vec{v}}}}( x \Mid \Value{r} \neq 0 )
  dx
  \\
  &=
  e^{-s \mu}
  \int_{x=0}^{x=\infty}
  e^{sx}
  \cdot
  \frac{\pi}{\theta_{\Vec{u},\Vec{v}}}
  \sqrt{\frac{2}{\pi}}
  e^{-\frac{x^{2}}{2}}
  \cdot
  \frac{1}{\sqrt{2\pi}}
  \int_{y=-\Alpha(x)}^{y=\Alpha(x)}
  e^{-\frac{y^{2}}{2}}
  dy dx
  \\
  &=
  e^{-s \mu}
  \int_{x=0}^{x=\infty}
  e^{sx} e^{-\frac{x^{2}}{2}}
  \cdot
  \frac{\pi}{\theta_{\Vec{u},\Vec{v}}}
  \sqrt{\frac{2}{\pi}}
  \cdot
  \frac{1}{\sqrt{2\pi}}
  \int_{y=-\Alpha(x)}^{y=\Alpha(x)}
  e^{-\frac{y^{2}}{2}}
  dy dx
  \\
  &=
  e^{-s \mu}
  \int_{x=0}^{x=\infty}
  e^{-\left( \frac{x^{2}}{2} - sx \right)}
  \cdot
  \frac{\pi}{\theta_{\Vec{u},\Vec{v}}}
  \sqrt{\frac{2}{\pi}}
  \cdot
  \frac{1}{\sqrt{2\pi}}
  \int_{y=-\Alpha(x)}^{y=\Alpha(x)}
  e^{-\frac{y^{2}}{2}}
  dy dx
  \\
  &=
  e^{-s \mu}
  \int_{x=0}^{x=\infty}
  e^{-\frac{x^{2} - 2sx}{2}}
  \cdot
  \frac{\pi}{\theta_{\Vec{u},\Vec{v}}}
  \sqrt{\frac{2}{\pi}}
  \cdot
  \frac{1}{\sqrt{2\pi}}
  \int_{y=-\Alpha(x)}^{y=\Alpha(x)}
  e^{-\frac{y^{2}}{2}}
  dy dx
  \\
  &=
  e^{-s \mu}
  \int_{x=0}^{x=\infty}
  e^{-\frac{x^{2} - 2sx + s^{2} - s^{2}}{2}}
  \cdot
  \frac{\pi}{\theta_{\Vec{u},\Vec{v}}}
  \sqrt{\frac{2}{\pi}}
  \cdot
  \frac{1}{\sqrt{2\pi}}
  \int_{y=-\Alpha(x)}^{y=\Alpha(x)}
  e^{-\frac{y^{2}}{2}}
  dy dx
  \\
  &=
  e^{-s \mu}
  \int_{x=0}^{x=\infty}
  e^{\frac{s^{2}}{2}}
  e^{-\frac{x^{2} - 2sx + s^{2}}{2}}
  \cdot
  \frac{\pi}{\theta_{\Vec{u},\Vec{v}}}
  \sqrt{\frac{2}{\pi}}
  \cdot
  \frac{1}{\sqrt{2\pi}}
  \int_{y=-\Alpha(x)}^{y=\Alpha(x)}
  e^{-\frac{y^{2}}{2}}
  dy dx
  \\
  &=
  e^{-s \mu}
  \int_{x=0}^{x=\infty}
  e^{\frac{s^{2}}{2}}
  e^{-\frac{(x-s)^{2}}{2}}
  \cdot
  \frac{\pi}{\theta_{\Vec{u},\Vec{v}}}
  \sqrt{\frac{2}{\pi}}
  \cdot
  \frac{1}{\sqrt{2\pi}}
  \int_{y=-\Alpha(x)}^{y=\Alpha(x)}
  e^{-\frac{y^{2}}{2}}
  dy dx
  \\ \label{deriv:lemma:normal:concentration-ineq:proj_u-v:4:2}
  &=
  e^{\frac{s^{2}}{2}}
  e^{-s \mu}
  \int_{x=0}^{x=\infty}
  e^{-\frac{(x-s)^{2}}{2}}
  \cdot
  \frac{\pi}{\theta_{\Vec{u},\Vec{v}}}
  \sqrt{\frac{2}{\pi}}
  \cdot
  \frac{1}{\sqrt{2\pi}}
  \int_{y=-\Alpha(x)}^{y=\Alpha(x)}
  e^{-\frac{y^{2}}{2}}
  dy dx
\end{align}
%
Note that the function
\begin{gather}
\label{deriv:lemma:normal:concentration-ineq:proj_u-v:5:q}
  q(s)
  =
  e^{-s \mu}
  \int_{x=0}^{x=\infty}
  e^{-\frac{(x-s)^{2}}{2}}
  \cdot
  \frac{\pi}{\theta_{\Vec{u},\Vec{v}}}
  \sqrt{\frac{2}{\pi}}
  \cdot
  \frac{1}{\sqrt{2\pi}}
  \int_{y=-\Alpha(x)}^{y=\Alpha(x)}
  e^{-\frac{y^{2}}{2}}
  dy dx
  =
  \E \left[ e^{s (\RV{X} - \mu)} e^{-\frac{s^{2}}{2}} \right]
\end{gather}
decreases monotonically w.r.t. \( s \) over the interval
\(
  s \in [0,\infty)
\)
(see Lemma \ref{lemma:appendix:monotonicity-q(s)}).
Formally, this implies
\begin{gather}
\label{deriv:lemma:normal:concentration-ineq:proj_u-v:6}
  \max_{s \in [0,\infty)} q(s) = q(0) = 1
\end{gather}
where the last equality follows from the fact that \( q(0) \) reduces to the evaluation of
the density function \( f_{\RV{X}_{i} \Mid \RV{R_{i;\Vec{u},\Vec{v}}}} \) over its entire support.
Then, continuing
\eqref{deriv:lemma:normal:concentration-ineq:proj_u-v:4:1}-%
\eqref{deriv:lemma:normal:concentration-ineq:proj_u-v:4:2} arrives at the desired bound,
\eqref{deriv:lemma:normal:concentration-ineq:proj_u-v:3}:
\begin{align}
\label{deriv:lemma:normal:concentration-ineq:proj_u-v:4b}
  \psi_{\RV{X}_{i} - \mu \Mid \RV{R_{i;\Vec{u},\Vec{v}}} \neq 0}(s)
  &=
  e^{\frac{s^{2}}{2}}
  e^{-s \mu}
  \int_{x=0}^{x=\infty}
  e^{-\frac{(x-s)^{2}}{2}}
  \cdot
  \frac{\pi}{\theta_{\Vec{u},\Vec{v}}}
  \sqrt{\frac{2}{\pi}}
  \cdot
  \frac{1}{\sqrt{2\pi}}
  \int_{y=-\Alpha(x)}^{y=\Alpha(x)}
  e^{-\frac{y^{2}}{2}}
  dy dx
  \\
  &\leq
  e^{\frac{s^{2}}{2}} \cdot 1
  \\
  &=
  e^{\frac{s^{2}}{2}}
\end{align}
\end{subequations}
%
\par
Next, the MGF of the negated random variable,
\( ( -\RV{X}_{i} + \mu \Mid \RV{R_{i;\Vec{u},\Vec{v}}} \neq 0 ) \) is upper bounded by
\begin{gather}
\label{deriv:lemma:normal:concentration-ineq:proj_u-v:7}
  \psi_{-\RV{X}_{i} + \mu \Mid \RV{R_{i;\Vec{u},\Vec{v}}} \neq 0}(s)
  \leq
  e^{\frac{s^{2}}{2}}
.\end{gather}
%
The derivation of \eqref{deriv:lemma:normal:concentration-ineq:proj_u-v:7} is similar to that
above.
\begin{subequations}
\label{deriv:lemma:normal:concentration-ineq:proj_u-v:8}
\begin{align}
  \label{deriv:lemma:normal:concentration-ineq:proj_u-v:8:1}
  \psi_{-\RV{X}_{i} + \mu \Mid \RV{R_{i;\Vec{u},\Vec{v}}} \neq 0}(s)
  &=
  \E
  \left[
    e^{s ( -\RV{X}_{i} + \E( \RV{X}_{i} \Mid \RV{R_{i;\Vec{u},\Vec{v}}} \neq 0 ) )}
    \middle| \RV{R_{i;\Vec{u},\Vec{v}}} \neq 0
  \right]
  \\
  &=
  \E
  \left[
    e^{-s ( \RV{X}_{i} - \mu )}
    \middle| \RV{R_{i;\Vec{u},\Vec{v}}} \neq 0
  \right]
  \\
  &=
  e^{s \mu}
  \E
  \left[
    e^{-s \RV{X}_{i}}
    \middle| \RV{R_{i;\Vec{u},\Vec{v}}} \neq 0
  \right]
  \\
  &=
  e^{s \mu}
  \int_{x=-\infty}^{x=\infty}
  e^{-sx}
  f_{\RV{X}_{i} \Mid \RV{R_{i;\Vec{u},\Vec{v}}}}( x \Mid \Value{r} \neq 0 )
  dx
  \\
  &=
  e^{s \mu}
  \int_{x=0}^{x=\infty}
  e^{-sx}
  \cdot
  \frac{\pi}{\theta_{\Vec{u},\Vec{v}}}
  \sqrt{\frac{2}{\pi}}
  e^{-\frac{x^{2}}{2}}
  \cdot
  \frac{1}{\sqrt{2\pi}}
  \int_{y=-\Alpha(x)}^{y=\Alpha(x)}
  e^{-\frac{y^{2}}{2}}
  dy dx
  \\
  &=
  e^{s \mu}
  \int_{x=0}^{x=\infty}
  e^{-sx} e^{-\frac{x^{2}}{2}}
  \cdot
  \frac{\pi}{\theta_{\Vec{u},\Vec{v}}}
  \sqrt{\frac{2}{\pi}}
  \cdot
  \frac{1}{\sqrt{2\pi}}
  \int_{y=-\Alpha(x)}^{y=\Alpha(x)}
  e^{-\frac{y^{2}}{2}}
  dy dx
  \\
  &=
  e^{s \mu}
  \int_{x=0}^{x=\infty}
  e^{-\left( \frac{x^{2}}{2} + sx \right)}
  \cdot
  \frac{\pi}{\theta_{\Vec{u},\Vec{v}}}
  \sqrt{\frac{2}{\pi}}
  \cdot
  \frac{1}{\sqrt{2\pi}}
  \int_{y=-\Alpha(x)}^{y=\Alpha(x)}
  e^{-\frac{y^{2}}{2}}
  dy dx
  \\
  &=
  e^{s \mu}
  \int_{x=0}^{x=\infty}
  e^{-\frac{x^{2} + 2sx}{2}}
  \cdot
  \frac{\pi}{\theta_{\Vec{u},\Vec{v}}}
  \sqrt{\frac{2}{\pi}}
  \cdot
  \frac{1}{\sqrt{2\pi}}
  \int_{y=-\Alpha(x)}^{y=\Alpha(x)}
  e^{-\frac{y^{2}}{2}}
  dy dx
  \\
  &=
  e^{s \mu}
  \int_{x=0}^{x=\infty}
  e^{-\frac{x^{2} + 2sx + s^{2} - s^{2}}{2}}
  \cdot
  \frac{\pi}{\theta_{\Vec{u},\Vec{v}}}
  \sqrt{\frac{2}{\pi}}
  \cdot
  \frac{1}{\sqrt{2\pi}}
  \int_{y=-\Alpha(x)}^{y=\Alpha(x)}
  e^{-\frac{y^{2}}{2}}
  dy dx
  \\
  &=
  e^{s \mu}
  \int_{x=0}^{x=\infty}
  e^{\frac{s^{2}}{2}}
  e^{-\frac{x^{2} - 2sx + s^{2}}{2}}
  \cdot
  \frac{\pi}{\theta_{\Vec{u},\Vec{v}}}
  \sqrt{\frac{2}{\pi}}
  \cdot
  \frac{1}{\sqrt{2\pi}}
  \int_{y=-\Alpha(x)}^{y=\Alpha(x)}
  e^{-\frac{y^{2}}{2}}
  dy dx
  \\
  &=
  e^{s \mu}
  \int_{x=0}^{x=\infty}
  e^{\frac{s^{2}}{2}}
  e^{-\frac{(x+s)^{2}}{2}}
  \cdot
  \frac{\pi}{\theta_{\Vec{u},\Vec{v}}}
  \sqrt{\frac{2}{\pi}}
  \cdot
  \frac{1}{\sqrt{2\pi}}
  \int_{y=-\Alpha(x)}^{y=\Alpha(x)}
  e^{-\frac{y^{2}}{2}}
  dy dx
  \\ \label{deriv:lemma:normal:concentration-ineq:proj_u-v:8:2}
  &=
  e^{\frac{s^{2}}{2}}
  e^{-s \mu}
  \int_{x=0}^{x=\infty}
  e^{-\frac{(x+s)^{2}}{2}}
  \cdot
  \frac{\pi}{\theta_{\Vec{u},\Vec{v}}}
  \sqrt{\frac{2}{\pi}}
  \cdot
  \frac{1}{\sqrt{2\pi}}
  \int_{y=-\Alpha(x)}^{y=\Alpha(x)}
  e^{-\frac{y^{2}}{2}}
  dy dx
\end{align}
%
Again, the function
\begin{gather}
\label{deriv:lemma:normal:concentration-ineq:proj_u-v:5:r}
  r(s)
  =
  e^{s \mu}
  \int_{x=0}^{x=\infty}
  e^{-\frac{(x+s)^{2}}{2}}
  \cdot
  \frac{\pi}{\theta_{\Vec{u},\Vec{v}}}
  \sqrt{\frac{2}{\pi}}
  \cdot
  \frac{1}{\sqrt{2\pi}}
  \int_{y=-\Alpha(x)}^{y=\Alpha(x)}
  e^{-\frac{y^{2}}{2}}
  dy dx
  =
  \E \left[ e^{-s(X-\mu)} e^{-\frac{s^{2}}{2}} \right]
\end{gather}
decreases monotonically w.r.t.
\(
  s \in [0,\infty)
\)
(see, again, Lemma \ref{lemma:appendix:monotonicity-q(s)}), and thus
\begin{gather}
\label{deriv:lemma:normal:concentration-ineq:proj_u-v:9}
  \max_{s \in [0,\infty)} r(s) = r(0) = 1
\end{gather}
where, as before, the last equality holds because \( r(0) \) simply evaluates
the density function \( f_{\RV{X}_{i} \Mid \RV{R_{i;\Vec{u},\Vec{v}}}} \) over its entire support.
Then, the desired bound in \eqref{deriv:lemma:normal:concentration-ineq:proj_u-v:7} can now
be established by continuing from
\eqref{deriv:lemma:normal:concentration-ineq:proj_u-v:8:1}-%
\eqref{deriv:lemma:normal:concentration-ineq:proj_u-v:8:2} as follows.
\begin{align}
\label{deriv:lemma:normal:concentration-ineq:proj_u-v:10}
  \psi_{-\RV{X}_{i} + \mu \Mid \RV{R_{i;\Vec{u},\Vec{v}}} \neq 0}(s)
  &=
  e^{\frac{s^{2}}{2}}
  e^{s \mu}
  \int_{x=0}^{x=\infty}
  e^{-\frac{(x+s)^{2}}{2}}
  \cdot
  \frac{\pi}{\theta_{\Vec{u},\Vec{v}}}
  \sqrt{\frac{2}{\pi}}
  \cdot
  \frac{1}{\sqrt{2\pi}}
  \int_{y=-\Alpha(x)}^{y=\Alpha(x)}
  e^{-\frac{y^{2}}{2}}
  dy dx
  \\
  &\leq
  e^{\frac{s^{2}}{2}} \cdot 1
  \\
  &=
  e^{\frac{s^{2}}{2}}
\end{align}
\end{subequations}
%
Note that \eqref{deriv:lemma:normal:concentration-ineq:proj_u-v:3} and
\eqref{deriv:lemma:normal:concentration-ineq:proj_u-v:7} holds likewise for every
\(
  i \in [m]
\).
This completes the first outline step.
\par
The second task, outlined in \ref{enum:lemma:normal:distribution:proj_u-v:2}, is controlling
the MGFs of the sums of i.i.d. random variables,
\(
  ( \RV{X} - \E[ \RV{X} \Mid \RLCOND ] \Mid \RLCOND )
\)
and
\(
  ( -\RV{X} + \E[ \RV{X} \Mid \RLCOND ] \Mid \RLCOND )
\)
\EDIT{for an arbitrary choice of \(  \Vec{r} \in \{ 0,1 \}^{m}  \) and \(  \Variable{\ell} = \| \Vec{r} \|_{0}  \).}
\BEGINEDIT
Note that \(  \Rhatuv  \) completely determines \(  \RV{L}_{\Vec{u},\Vec{v}}  \).
Therefore,
\begin{gather}
  \label{deriv:lemma:normal:concentration-ineq:proj_u-v:R-R,L:1}
  ( \RV{X} - \E[ \RV{X} \Mid \RLCOND ] \Mid \RLCOND ) \sim ( \RV{X} - \E[ \RV{X} \Mid \RCOND ] \Mid \RCOND )
  ,\\ \label{deriv:lemma:normal:concentration-ineq:proj_u-v:R-R,L:2}
  ( -\RV{X} + \E[ \RV{X} \Mid \RLCOND ] \Mid \RLCOND ) \sim ( -\RV{X} + \E[ \RV{X} \Mid \RCOND ] \Mid \RCOND )
.\end{gather}
\ENDEDIT
Writing
\EDIT{\(  \mu_{\RV{X} \Mid \RLCOND} = \E[ \RV{X} \Mid \RLCOND ]  \)
and
\(  \mu_{\RV{X} \Mid \RCOND} = \E[ \RV{X} \Mid \RCOND ]  \),}
the MGF for the sum of i.i.d. random variables,
\(
  \RV{X} = \sum_{i=1}^{m} (\RV{X}_{i} - \mu)
\),
conditioned on
\EDIT{\(  \Rhatuv = \Vec{r}  \)}
\(  \Luv = \Variable{\ell}  \)
can then be bounded from above as follows:
\BEGINEDIT
\begin{subequations}
\label{deriv:lemma:normal:concentration-ineq:proj_u-v:11}
\begin{align}
  \psi_{\RV{X} - \mu_{\RV{X} \Mid \RLCOND} \Mid \RLCOND}(s)
  &=
  \psi_{\RV{X} - \mu_{\RV{X} \Mid \RCOND} \Mid \RCOND}(s)
  \\
  &\dCmt \text{by \EQN \eqref{deriv:lemma:normal:concentration-ineq:proj_u-v:R-R,L:1}}
  \\
  &=
  \E \left[ e^{s (\RV{X} - \mu_{\RV{X} \Mid \RCOND})} \middle| \RCOND \right]
  \\
  &=
  \E
  \left[
    e^{s \sum_{i=1}^{m} (\RV{X}_{i} - \mu)} \middle| \RCOND
  \right]
  \\
  &=
  \E
  \left[
    e^{s \sum_{i \in \supp( \Vec{r} )} (\RV{X}_{i} - \mu)}
  \middle| \RCOND
  \right]
  \\
  &=
  \prod_{i \in \supp( \Vec{r} )}
  \E \left[ e^{s (\RV{X}_{i} - \mu)} \middle| \RCOND \right]
  \\ \nonumber
  &\dCmt \because \text{the random variables \( \RV{X}_{i}, i \in [m], \) are independent}
  \\
  &=
  \prod_{i \in \supp( \Vec{r} )}
  \E \left[ e^{s (\RV{X}_{i} - \mu)} \middle| \Vec*{R}_{i;\Vec{u},\Vec{v}} \neq 0 \right]
  \\
  &=
  \E \left[ e^{s (\RV{X}_{i} - \mu)} \middle| \Vec*{R}_{i;\Vec{u},\Vec{v}} \neq 0 \right]^{\ell}
  \\ \nonumber
  &\dCmt \text{where \(  i \in \supp( \Vec{r} )  \)}
  \\ \nonumber
  &\dCmt \because \text{the random variables \( \RV{X}_{i'}, i' \in [m], \) are identically distributed}
  \\
  &\leq
  e^{\frac{1}{2} \Variable{\ell} s^{2}}
  \\ \nonumber
  &\dCmt \Text{by \eqref{deriv:lemma:normal:concentration-ineq:proj_u-v:3}}
.\end{align}
\end{subequations}
\ENDEDIT
Moreover, by an analogous argument, the MGF of the negated random variable
\(
  ( -\RV{X} - \E[ -\RV{X} ] \Mid \RLCOND )
  =
  ( -\RV{X} + \E[ \RV{X} ] \Mid \RLCOND )
\)
can be upper bounded.
%
%
Notice that
\(
  -\RV{X}
  = -\sum_{i=1}^{m} (\RV{X}_{i} - \mu)
  = \sum_{i=1}^{m} (-\RV{X}_{i} + \mu)
\),
which allows the MGF of
\(
  -\RV{X} + \E[ \RV{X} ]
\)
conditioned on
\EDIT{\(  \Rhatuv = \Vec{r}  \) and}
\(
  \Luv = \Variable{\ell}
\)
to be upper bounded by the following:
\BEGINEDIT
\begin{subequations}
\label{deriv:lemma:normal:concentration-ineq:proj_u-v:12}
\begin{align}
  \psi_{-\RV{X} + \mu_{-\RV{X} \Mid \RLCOND} \Mid \RLCOND}(s)
  &=
  \psi_{-\RV{X} + \mu_{-\RV{X} \Mid \RCOND} \Mid \RCOND}(s)
  \\
  &\dCmt \text{by \EQN \eqref{deriv:lemma:normal:concentration-ineq:proj_u-v:R-R,L:2}}
  \\
  &=
  \E \left[ e^{s (-\RV{X} + \mu_{-\RV{X} \Mid \RCOND})} \middle| \RCOND \right]
  \\
  &=
  \E
  \left[
    e^{s \sum_{i=1}^{m} (-\RV{X}_{i} + \mu)} \middle| \RCOND
  \right]
  \\
  &=
  \E
  \left[
    e^{s \sum_{i \in \supp( \Vec{r} )} (-\RV{X}_{i} + \mu)}
  \middle| \RCOND
  \right]
  \\
  &=
  \prod_{i \in \supp( \Vec{r} )}
  \E \left[ e^{s (-\RV{X}_{i} + \mu)} \middle| \RCOND \right]
  \\ \nonumber
  &\dCmt \because \text{the random variables \( \RV{X}_{i}, i \in [m], \) are independent}
  \\
  &=
  \prod_{i \in \supp( \Vec{r} )}
  \E \left[ e^{s (-\RV{X}_{i} + \mu)} \middle| \Vec*{R}_{i;\Vec{u},\Vec{v}} \neq 0 \right]
  \\
  &=
  \E \left[ e^{s (-\RV{X}_{i} + \mu)} \middle| \Vec*{R}_{i;\Vec{u},\Vec{v}} \neq 0 \right]^{\ell}
  \\ \nonumber
  &\dCmt \text{where \(  i \in \supp( \Vec{r} )  \)}
  \\ \nonumber
  &\dCmt \because \text{the random variables \( \RV{X}_{i'}, i' \in [m], \) are identically distributed}
  \\
  &\leq
  e^{\frac{1}{2} \Variable{\ell} s^{2}}
  \\ \nonumber
  &\dCmt \Text{by \eqref{deriv:lemma:normal:concentration-ineq:proj_u-v:7}}
\end{align}
\end{subequations}
\ENDEDIT
To summarize, this step, \ref{enum:lemma:normal:distribution:proj_u-v:2}, has shown
\begin{gather}
  \label{deriv:lemma:normal:concentration-ineq:proj_u-v:13:positive}
  \psi_{\RV{X} - \mu_{\RV{X} \Mid \RLCOND} \Mid \RLCOND}(s)
  \leq
  e^{\frac{1}{2} \Variable{\ell} s^{2}}
  \\ \label{deriv:lemma:normal:concentration-ineq:proj_u-v:13:negative}
  \psi_{-\RV{X} + \mu_{\RV{X} \Mid \RLCOND} \Mid \RLCOND}(s)
  \leq
  e^{\frac{1}{2} \Variable{\ell} s^{2}}
.\end{gather}
%
\par
The aim in the final outlined step, \ref{enum:lemma:normal:distribution:proj_u-v:3},
is bounding \( \RV{X} \) from each sides by a Chernoff bound and subsequently union bounding to
obtain the lemma's two-sided result.
The upper bound, derived first, will use the MGF of
\(
  ( \RV{X} - \E[ \RV{X} \Mid \RLCOND ] \Mid \RLCOND )
\),
while the lower bound will use the MGF of
\(
  ( -\RV{X} + \E[ \RV{X} \Mid \RLCOND ] \Mid \RLCOND )
\).
In both cases, the bounds will be shown to fail with probability not exceeding
\(
  e^{-\frac{1}{2} \Variable{\ell} t^{2}}
\).
For the upper bound,
\begin{subequations}
\label{deriv:lemma:normal:concentration-ineq:proj_u-v:14}
\begin{align}
  &
  \Pr
  \left(
    \RV{X}
    -
    \E \left[ \RV{X} \middle| \RLCOND \right]
    \geq
    \Variable{\ell} \Variable{t}
    \phantom{\Big|}\middle| \RLCOND
  \right)
  \\
  &=
  \Pr
  \left(
    \RV{X}
    -
    \mu_{\RV{X} \Mid \RLCOND}
    \geq
    \Variable{\ell} \Variable{t}
    \phantom{\Big|}\middle| \RLCOND
  \right)
  \\
  &=
  \Pr
  \left(
    e^{
      \RV{X}
      -
      \E \left[ \RV{X} \middle| \RLCOND \right]
    }
    \geq
    e^{
      \Variable{\ell} \Variable{t}
    }
    \phantom{\Big|}\middle|
    \RLCOND
  \right)
  \\
  &\leq
  \min_{s \geq 0}
  e^{-\Variable{\ell} st}
  \cdot
  \psi_{\RV{X} - \mu_{\RV{X} \Mid \RLCOND} \Mid \RLCOND}(s)
  \\ \nonumber
  &\dCmt \text{due to Bernstein (see, e.g., \cite{vershynin2018high})}
  \\
  &\leq
  \min_{s \geq 0}
  e^{-\Variable{\ell} st} e^{\frac{1}{2} \Variable{\ell} s^{2}}
  \\ \nonumber
  &\dCmt \text{by \EQN \eqref{deriv:lemma:normal:concentration-ineq:proj_u-v:13:positive}}
  \\ \label{deriv:lemma:normal:concentration-ineq:proj_u-v:14:1}
  &=
  \min_{s \geq 0}
  e^{-\Variable{\ell} \left( st - \frac{s^{2}}{2} \right)}
\end{align}
%
A maximizer of
\(
  st - \frac{s^{2}}{2}
\)
a minimizer of
\(
  e^{-\Variable{\ell} ( st - \frac{s^{2}}{2} )}
\).
The unique zero of
\(
  \frac{\partial}{\partial s} st - \frac{s^{2}}{2}
\)
is at
\(
  s = t
\)
(moreover,
\(
  \frac{\partial^{2}}{\partial s^{2}} st - \frac{s^{2}}{2} < 0
\)
and hence this is indeed a (global) maximum).
Note additionally that setting
\(
  s = t
\)
ensures that
\(
  s \in [0,1]
\),
which was assumed in step \ref{enum:lemma:normal:distribution:proj_u-v:1}.
Then, continuing from above,
\begin{align}
  \Pr
  \left(
    \RV{X}
    -
    \E \left[ \RV{X} \middle| \RLCOND \right]
    \geq
    \Variable{\ell} \Variable{t}
    \phantom{\Big|}\middle| \RLCOND
  \right)
  &\leq
  \min_{s \geq 0}
  e^{-\Variable{\ell} \left( st - \frac{s^{2}}{2} \right)}
  \\
  &=
  e^{-\Variable{\ell} \left( t^{2} - \frac{t^{2}}{2} \right)}
  \\ \nonumber
  &\dCmt \text{as argued above}
  \\
  &\leq
  e^{-\frac{1}{2} \Variable{\ell} t^{2}}
\label{deriv:lemma:normal:concentration-ineq:proj_u-v:14:end}
\end{align}
\end{subequations}
as desired.
The derivation of the lower bound is nearly identical, as seen next.
\begin{subequations}
\label{deriv:lemma:normal:concentration-ineq:proj_u-v:15}
\begin{align}
  &
  \Pr
  \left(
    \RV{X}
    -
    \E \left[ \RV{X} \middle| \RLCOND \right]
    \leq
    -\Variable{\ell} \Variable{t}
    \phantom{\Big|}\middle|
    \RLCOND
  \right)
  \\
  &=
  \Pr
  \left(
    -\RV{X}
    +
    \E \left[ \RV{X} \middle| \RLCOND \right]
    \geq
    \Variable{\ell} \Variable{t}
    \phantom{\Big|}\middle|
    \RLCOND
  \right)
  \\
  &=
  \Pr
  \left(
    -\RV{X}
    +
    \mu_{\RV{X} \Mid \RLCOND}
    \geq
    \Variable{\ell} \Variable{t}
    \phantom{\Big|}\middle|
    \RLCOND
  \right)
  \\
  &=
  \Pr
  \left(
    e^{
      -\RV{X}
      +
      \E \left[ \RV{X} \middle| \RLCOND \right]
    }
    \geq
    e^{
      \Variable{\ell} \Variable{t}
    }
    \phantom{\Big|}\middle|
    \RLCOND
  \right)
  \\
  &\leq
  \min_{s \geq 0}
  e^{-\Variable{\ell} st}
  \cdot
  \psi_{-\RV{X} + \mu_{\RV{X} \Mid \RLCOND} \Mid \RLCOND}(s)
  \\ \nonumber
  &\dCmt \text{due to Bernstein (see, e.g., \cite{vershynin2018high})}
  \\
  &\leq
  \min_{s \geq 0}
  e^{-\Variable{\ell} st} e^{\frac{1}{2} \Variable{\ell} s^{2}}
  \\ \nonumber
  &\dCmt \text{by \EQN \eqref{deriv:lemma:normal:concentration-ineq:proj_u-v:13:negative}}
  \\
  &=
  \min_{s \geq 0}
  e^{-\Variable{\ell} \left( st - \frac{s^{2}}{2} \right)}
  \\
  &=
  e^{-\Variable{\ell} \left( t^{2} - \frac{t^{2}}{2} \right)}
  \\ \nonumber
  &\dCmt
  \text{the same minimization problem as \eqref{deriv:lemma:normal:concentration-ineq:proj_u-v:14:1}, whose solution is at \(  s = t  \)}
  \\
  &=
  e^{-\frac{1}{2} \Variable{\ell} t^{2}}
\label{deriv:lemma:normal:concentration-ineq:proj_u-v:15:end}
\end{align}
\end{subequations}
%
Thus far, it has been shown that
\begin{gather}
  \label{deriv:lemma:normal:concentration-ineq:proj_u-v:16:upper}
  \Pr
  \left(
    \RV{X}
    -
    \E \left[ \RV{X} \middle| \RLCOND \right]
    \geq
    \Variable{\ell} \Variable{t}
    \phantom{\Big|}\middle|
    \RLCOND
  \right)
  \leq
  e^{-\frac{1}{2} \Variable{\ell} t^{2}},
  \\ \label{deriv:lemma:normal:concentration-ineq:proj_u-v:16:lower}
  \Pr
  \left(
    \RV{X}
    -
    \E \left[ \RV{X} \middle| \RLCOND \right]
    \leq
    -\Variable{\ell} \Variable{t}
    \phantom{\Big|}\middle|
    \RLCOND
  \right)
  \leq
  e^{-\frac{1}{2} \Variable{\ell} t^{2}}
.\end{gather}
%
To complete the proof, \eqref{deriv:lemma:normal:concentration-ineq:proj_u-v:16:upper} and
\eqref{deriv:lemma:normal:concentration-ineq:proj_u-v:16:lower} are combined by a union bound,
yielding the lemma's concentration inequality,
\begin{gather}
  \Pr
  \left(
    \left|
      \RV{X}
      -
      \E \left[ \RV{X} \middle| \RLCOND \right]
    \right|
    \geq
    \Variable{\ell} \Variable{t}
    \phantom{\Big|}\middle|
    \RLCOND
  \right)
  \leq
  2 e^{-\frac{1}{2} \Variable{\ell} t^{2}}
.\end{gather}
\end{proof}

\begin{proof}
{Lemma \ref{lemma:normal:concentration-ineq:proj_u+v}}
\label{pf:lemma:normal:concentration-ineq:proj_u+v}
As in the proof of \EDIT{Lemma \ref{lemma:normal:concentration-ineq:proj_u-v}},
let
\(
  \RV{X}_{i}
  =
  \left\langle
    \frac{\Vec{u} + \Vec{v}}{\left\| \Vec{u} + \Vec{v} \right\|_{2}},
    \Vec{Z}^{(i)} \RV{R_{i;\Vec{u},\Vec{v}}}
  \right\rangle
\)
for each \( i \in [m] \), which are \iid with (conditional) distributions described in
Lemma \ref{lemma:normal:distribution:proj_u+v}.
Then the random variable \( \RV{X} \) can be written as
\begin{gather}
\label{deriv:lemma:normal:concentration-ineq:proj_u+v:1}
  \RV{X}
  =
  \left\langle
    \frac{\Vec{u} + \Vec{v}}{\left\| \Vec{u} + \Vec{v} \right\|_{2}},
    \sum_{i=1}^{m}
    \Vec{Z}^{(i)} \RV{R_{i;\Vec{u},\Vec{v}}}
  \right\rangle
  =
  \sum_{i=1}^{m}
  \left\langle
    \frac{\Vec{u} + \Vec{v}}{\left\| \Vec{u} + \Vec{v} \right\|_{2}},
    \Vec{Z}^{(i)} \RV{R_{i;\Vec{u},\Vec{v}}}
  \right\rangle
  =
  \sum_{i=1}^{m} \RV{X}_{i}
.\end{gather}
%
\EDIT{Fix \(  \Vec{r} \in \rSet  \) arbitrarily, and suppose \(  \Rhatuv = \Vec{r}  \).
Let \(  \Variable{\ell} \defeq \| \Vec{r} \|_{0}  \).}
Recall from Lemma \ref{lemma:normal:distribution:proj_u+v} that
the random variables
\(
  ( \RV{X}_{i} \Mid \Vec*{R}_{i;\Vec{u},\Vec{v}} \neq 0 )
\),
\(
  i \in [m]
\),
are standard normal, and moreover, they are in fact \iid[ ]
It follows that
\BEGINEDIT
\begin{align}
\label{deriv:lemma:normal:concentration-ineq:proj_u+v:2}
  ( \RV{X} \Mid \RCOND )
  &=
  \left( \sum_{i=1}^{m} \RV{X}_{i} \middle| \RCOND \right)
  \\
  &=
  \sum_{i=1}^{m} ( \RV{X}_{i} \Mid \RCOND )
  \\
  &=
  \sum_{i=1}^{m}
  ( \RV{X}_{i} \Mid \Rhatuvi = \Vec*{r}_{i} )
  \\
  &=
  \sum_{i \in \supp( \Vec{r} )}
  ( \RV{X}_{i} \Mid \Vec*{R}_{i;\Vec{u},\Vec{v}} \neq 0 )
  \sim
  \N(0, \sigma^{2} = \Variable{\ell})
,\end{align}
\ENDEDIT
\EDIT{where the rightmost statement uses the fact that the variance of the sum of \(  \Variable{\ell}  \) independent random variables is the sum of the individual variances.}
\ORIG{where the distribution of \( \RV{X} \) depends only on the number \( \Variable{\ell} \) of
random variables summed up but not the exact subset
\(
  \supp(\Rhatuv) \subseteq [n]
\)
(since the random variables \( \RV{X}_{i} \), \( i \in [m] \), are identically distributed).
Therefore,}%
%
\EDIT{As noted in the proof of Lemma \ref{lemma:normal:concentration-ineq:proj_u-v}, \(  \Rhatuv  \) completely determines \(  \Luv  \), and therefore}
\begin{gather}
  ( \RV{X} \Mid \RLCOND ) \sim ( \RV{X} \Mid \RCOND ) \sim \N(0, \sigma^{2} = \Variable{\ell})
.\end{gather}
Therefore,
\begin{gather}
\label{deriv:lemma:normal:concentration-ineq:proj_u+v:3}
  \Pr
  \left(
    \left| \RV{X} \right|
    \geq
    \Variable{t'}
  \middle|
    \RLCOND
  \right)
  \leq
  2 e^{-\frac{\Variable{t'}^{2}}{2 \Variable{\ell}}}
.\end{gather}
%
Taking
\(
  \Variable{t'} = \Variable{\ell} \Variable{t}
\),
\eqref{deriv:lemma:normal:concentration-ineq:proj_u+v:3} implies
\begin{gather}
\label{deriv:lemma:normal:concentration-ineq:proj_u+v:4}
  \Pr
  \left(
    \left| \RV{X} \right|
    \geq
    \Variable{\ell} \Variable{t}
  \middle|
    \RLCOND
  \right)
  \leq
  2 e^{-\frac{\Variable{\ell}^{2} \Variable{t}^{2}}{2 \Variable{\ell}}}
  =
  2 e^{-\frac{1}{2} \Variable{\ell} \Variable{t}^{2}}
.\end{gather}
%
Thus proved.
\end{proof}

\begin{proof}
{Lemma \ref{lemma:normal:concentration-ineq:third}}
\label{pf:lemma:normal:concentration-ineq:third}
\ORIG{
By Lemma \ref{lemma:normal:distribution:third}, for each
\(
  i \in [m]
\),
\(
  \Vec{Z}^{(i)} \RV{R}_{i;\Vec{u},\Vec{v}}
\)}%
Write
\(
  \Coords{J'} \ORIG{= \Coords{J} \cap (\supp( \Vec{u} ) \cup \supp( \Vec{v} ))}\EDIT{= \supp( \Vec{u} ) \cup \supp( \Vec{v} )}
\)
and
\(
  \Coords{J''} = \Coords{J} \setminus (\supp( \Vec{u} ) \cup \supp( \Vec{v} )) \EDIT{= \Coords{J} \setminus \Coords{J'}}
\)
\EDIT{such that \(  \EDIT{\supp( \Vec{u} ) \cup \supp( \Vec{v} ) \cup} \Coords{J} = \Coords{J'} \sqcup \Coords{J''}  \).}
By the triangle inequality,
\begin{subequations}
\begin{align}
  \left\|
    \ThresholdSet{\EDIT{\supp( \Vec{u} ) \cup \supp( \Vec{v} ) \cup} \Coords{J}}( \sum_{i=1}^{m} \Vec{Y}^{(i)} \RV{R}_{i;\Vec{u},\Vec{v}} )
  \right\|_{2}
  &=
  \left\|
    \ThresholdSet{\EDIT{\Coords{J'}}}( \sum_{i=1}^{m} \Vec{Y}^{(i)} \RV{R}_{i;\Vec{u},\Vec{v}} )
    +
    \ThresholdSet{\Coords{J''}}( \sum_{i=1}^{m} \Vec{Y}^{(i)} \RV{R}_{i;\Vec{u},\Vec{v}} )
  \right\|_{2}
  \\
  &\leq
  \left\|
    \ThresholdSet{\Coords{J'}}( \sum_{i=1}^{m} \Vec{Y}^{(i)} \RV{R}_{i;\Vec{u},\Vec{v}} )
  \right\|_{2}
  +
  \left\|
    \ThresholdSet{\Coords{J''}}( \sum_{i=1}^{m} \Vec{Y}^{(i)} \RV{R}_{i;\Vec{u},\Vec{v}} )
  \right\|_{2}
  \\
  &=
  \left\|
    \ThresholdSet{\supp( \Vec{u} ) \cup \supp( \Vec{v} )}
    ( \sum_{i=1}^{m} \Vec{Y}^{(i)} \RV{R}_{i;\Vec{u},\Vec{v}} )
  \right\|_{2}
  +
  \left\|
    \ThresholdSet{\Coords{J''}}( \sum_{i=1}^{m} \Vec{Y}^{(i)} \RV{R}_{i;\Vec{u},\Vec{v}} )
  \right\|_{2}
.\end{align}
\end{subequations}
%
Let
\(
  \Variable{d'} = | \supp( \Vec{u} ) \cup \supp( \Vec{v} ) |
\)
and
\(
  \Vec{V}^{(i)}
  = \Vec*{V}_{1}^{(i)},\dots,\Vec*{V}_{\Variable{d'}-2}^{(i)}
  \sim \N(\Vec{0},\Mat{I}_{(\Variable{d'}-2) \times (\Variable{d'}-2)})
\),
\(
  i \in [m]
\),
and suppose
\(
  \{ \Vec{b}_{j} \in \R^{n} \}_{j \in [d'-2]}
\)
is an orthonormal basis over
\(
  \Span( \{ \Vec{u}-\Vec{v}, \Vec{u}+\Vec{v} \} )^{\perp}
  \cap
  \{ \Vec{x} \in \R^{n} : \supp( \Vec{x} ) \subseteq \supp( \Vec{u} ) \cup \supp( \Vec{v} ) \}
\)
with
\(
  \Vec{Y}^{(i)}
  =
  \sum_{j=1}^{\Variable{d'}-2}
  \langle \Vec{b}_{j}, \Vec{Y}^{(i)} \rangle
  \Vec{b}_{j}
\).
\EDITX{Note that \(  \Variable{d'} \leq \KXO  \).}
Due to Lemma \ref{lemma:normal:distribution:third},
\(
  \langle \Vec{b}_{j}, \Vec{Y}^{(i)} \rangle \sim \N(0,1)
\).
%
\begin{subequations}
\begin{align}
  \left\|
    \ThresholdSet{\supp( \Vec{u} ) \cup \supp( \Vec{v} )}
    ( \sum_{i=1}^{m} \Vec{Y}^{(i)} \RV{R}_{i;\Vec{u},\Vec{v}} )
  \right\|_{2}
  &=
  \left\|
    \sum_{i=1}^{m}
    \ThresholdSet{\supp( \Vec{u} ) \cup \supp( \Vec{v} )}
    ( \Vec{Y}^{(i)} \RV{R}_{i;\Vec{u},\Vec{v}} )
  \right\|_{2}
  \\
  &=
  \left\|
    \sum_{\substack{i=1 : \\ \RV{R}_{i;\Vec{u},\Vec{v}} \neq 0}}^{m}
    \sum_{j=1}^{\Variable{d'}-2}
    \langle \Vec{b}_{j}, \Vec{Y}^{(i)} \rangle \Vec{b}_{j}
  \right\|_{2}
  \\
  &=
  \left\|
    \sum_{j=1}^{\Variable{d'}-2}
    \Vec{b}_{j}
    \sum_{\substack{i=1 : \\ \RV{R}_{i;\Vec{u},\Vec{v}} \neq 0}}^{m}
    \langle \Vec{b}_{j}, \Vec{Y}^{(i)} \rangle
  \right\|_{2}
  \\
  &=
  \left(
    \sum_{j=1}^{\Variable{d'}-2}
    \sum_{j'=1}^{\Variable{d'}-2}
    \langle \Vec{b}_{j}, \Vec{b}_{j'} \rangle
    \left(
      \sum_{\substack{i=1 : \\ \RV{R}_{i;\Vec{u},\Vec{v}} \neq 0}}^{m}
      \langle \Vec{b}_{j}, \Vec{Y}^{(i)} \rangle
    \right)^{2}
  \right)^{\frac{1}{2}}
  \\
  &=
  \left(
    \sum_{j=1}^{\Variable{d'}-2}
    \left(
      \sum_{\substack{i=1 : \\ \RV{R}_{i;\Vec{u},\Vec{v}} \neq 0}}^{m}
      \langle \Vec{b}_{j}, \Vec{Y}^{(i)} \rangle
    \right)^{2}
  \right)^{\frac{1}{2}}
  \\
  &\sim
  \left(
    \sum_{j=1}^{\Variable{d'}-2}
    \left(
      \sum_{\substack{i=1 : \\ \RV{R}_{i;\Vec{u},\Vec{v}} \neq 0}}^{m}
      \RV{V}_{j}^{(i)}
    \right)^{2}
  \right)^{\frac{1}{2}}
  \\
  &=
  \left\|
    \sum_{\substack{i=1 : \\ \RV{R}_{i;\Vec{u},\Vec{v}} \neq 0}}^{m} \Vec{V}^{(i)}
  \right\|_{2}
  \sim
  \left\| \sum_{i}^{\Variable{\ell}} \Vec{V}^{(i)} \right\|_{2}
.\end{align}
\end{subequations}
%
Then, by a standard Chernoff bound for standard normal random vectors,
\begin{align}
  &
  \Pr
  \left(
    \left\|
      \ThresholdSet{\Coords{J'}}( \sum_{i=1}^{m} \Vec{Y}^{(i)} \RV{R}_{i;\Vec{u},\Vec{v}} )
    \right\|_{2}
    >
    \sqrt{\KXO \Variable{\ell}}
    +
    \frac{1}{2} \Variable{\ell} \Variable{t}
  \middle|
    \RLCOND
  \right)
  \\
  &=
  \Pr
  \left(
    \left\|
      \ThresholdSet{\supp( \Vec{u} ) \cup \supp( \Vec{v} )}
      ( \sum_{i=1}^{m} \Vec{Y}^{(i)} \RV{R}_{i;\Vec{u},\Vec{v}} )
    \right\|_{2}
    >
    \sqrt{\KXO \Variable{\ell}}
    +
    \frac{1}{2} \Variable{\ell} \Variable{t}
  \middle|
      \RLCOND
  \right)
  \\
  &=
  \Pr
  \left(
    \left\| \sum_{i=1}^{\Variable{\ell}} \Vec{V}^{(i)} \right\|_{2}
    >
    \sqrt{\KXO \Variable{\ell}}
    +
    \frac{1}{2} \Variable{\ell} \Variable{t}
  \middle|
    \RLCOND
  \right)
  \\
  &\leq
  \Pr
  \left(
    \left\| \sum_{i=1}^{\Variable{\ell}} \Vec{V}^{(i)} \right\|_{2}
    >
    \E
    \left[
      \left\| \sum_{i=1}^{\Variable{\ell}} \Vec{V}^{(i)} \right\|_{2}
    \right]
    +
    \frac{1}{2} \Variable{\ell} \Variable{t}
  \middle|
    \RLCOND
  \right)
  \\
  &\leq
  e^{-\frac{1}{8} \Variable{\ell} \Variable{t}^{2}}
\end{align}
%
On the other hand, observe,
%
\begin{subequations}
\begin{align}
  \left\|
    \ThresholdSet{\Coords{J''}}( \sum_{i=1}^{m} \Vec{Y}^{(i)} \RV{R}_{i;\Vec{u},\Vec{v}} )
  \right\|_{2}
  &=
  \left\|
    \sum_{i=1}^{m}
    \sum_{j \in \Coords{J''}}
    \left\langle
      \Vec{e}_{j},
      \Vec{Y}^{(i)} \RV{R}_{i;\Vec{u},\Vec{v}}
    \right\rangle
    \Vec{e}_{j}
  \right\|_{2}
  \\
  &=
  \left\|
    \sum_{i=1}^{m}
    \sum_{j \in \Coords{J''}}
    \Vec*{Y}_{j}^{(i)} \RV{R}_{i;\Vec{u},\Vec{v}}
    \Vec{e}_{j}
  \right\|_{2}
\end{align}
\end{subequations}
%
Let
\(
  \Variable{d''} = | \Coords{J''} |
\)
and
\(
  \Vec{W}^{(i)}
  = (\Vec*{W}_{1},\dots,\Vec*{W}_{\Variable{d''}})
  \sim \N(\Vec{0},\Mat{I}_{\Variable{d''} \times \Variable{d''}})
\),
\(
  i \in [m]
\).
Due to Lemma \ref{lemma:normal:distribution:third},
\(
  (
    \|
      \sum_{j \in \Coords{J''}}
      \Vec*{Y}_{j}^{(i)} \RV{R}_{i;\Vec{u},\Vec{v}}
      \Vec{e}_{j}
    \|_{2}
    \mid
    \RV{R}_{i;\Vec{u},\Vec{v}} \neq 0
  )
\)
and
\(
  \| \Vec{W}^{(i)} \|_{2}
\),
\(
  i \in [m]
\),
share the same distribution.
Then, by a standard Chernoff bound for standard normal random vectors,
\begin{align}
  &
  \Pr
  \left(
    \left\|
      \ThresholdSet{\Coords{J''}}( \sum_{i=1}^{m} \Vec{Y}^{(i)} \RV{R}_{i;\Vec{u},\Vec{v}} )
    \right\|_{2}
    >
    \sqrt{\Variable{d} \Variable{\ell}}
    +
    \frac{1}{2} \Variable{\ell} \Variable{t}
    \middle|
      \RLCOND
  \right)
  \\
  &\leq
  \Pr
  \left(
    \left\|
      \ThresholdSet{\Coords{J''}}( \sum_{i=1}^{m} \Vec{Y}^{(i)} \RV{R}_{i;\Vec{u},\Vec{v}} )
    \right\|_{2}
    >
    \sqrt{\Variable{d''} \Variable{\ell}}
    +
    \frac{1}{2} \Variable{\ell} \Variable{t}
    \middle|
      \RLCOND
  \right)
  \\
  &=
  \Pr
  \left(
    \left\| \sum_{\substack{i=1 : \\ \RV{R}_{i;\Vec{u},\Vec{v}} \neq 0}}^{m}
    \Vec{W}^{(i)} \right\|_{2}
    >
    \sqrt{\Variable{d''} \Variable{\ell}}
    +
    \frac{1}{2} \Variable{\ell} \Variable{t}
    \middle|
      \RLCOND
  \right)
  \\
  &\leq
  \Pr
  \left(
  \left\| \sum_{i=1}^{\Variable{\ell}} \Vec{W}^{(i)} \right\|_{2}
  >
  \E
  \left[
    \left\| \sum_{i=1}^{\Variable{\ell}} \Vec{W}^{(i)} \right\|_{2}
  \right]
  +
  \frac{1}{2} \Variable{\ell} \Variable{t}
  \middle|
    \RLCOND
  \right)
  \\
  &\leq
  e^{-\frac{1}{8} \Variable{\ell} \Variable{t}^{2}}
\end{align}
%
Then, since
\begin{gather}
  \sqrt{\KXO \Variable{\ell}}
  + \frac{1}{2} \Variable{\ell} \Variable{t}
  + \sqrt{\Variable{d} \Variable{\ell}}
  + \frac{1}{2} \Variable{\ell} \Variable{t}
  =
  \left( \sqrt{\KXO} + \sqrt{\Variable{d}} \right) \sqrt{\Variable{\ell}}
  +
  \Variable{\ell} \Variable{t}
\end{gather}
and
\begin{align}
  \left\|
    \ThresholdSet{\EDIT{\supp( \Vec{u} ) \cup \supp( \Vec{v} ) \cup} \Coords{J}}( \sum_{i=1}^{m} \Vec{Y}^{(i)} \RV{R}_{i;\Vec{u},\Vec{v}} )
  \right\|_{2}
  &\leq
  \left\|
    \ThresholdSet{\Coords{J'}}
    ( \sum_{i=1}^{m} \Vec{Y}^{(i)} \RV{R}_{i;\Vec{u},\Vec{v}} )
  \right\|_{2}
  +
  \left\|
    \ThresholdSet{\Coords{J''}}( \sum_{i=1}^{m} \Vec{Y}^{(i)} \RV{R}_{i;\Vec{u},\Vec{v}} )
  \right\|_{2}
,\end{align}
it follows from a union bound that
\begin{subequations}
\begin{align}
  &
  \Pr
  \left(
    \left\|
      \ThresholdSet{\EDIT{\supp( \Vec{u} ) \cup \supp( \Vec{v} ) \cup} \Coords{J}}( \sum_{i=1}^{m} \Vec{Y}^{(i)} \RV{R}_{i;\Vec{u},\Vec{v}} )
    \right\|_{2}
    \geq
    \left( \sqrt{\KXO} + \sqrt{d} \right)\sqrt{\Variable{\ell}} + \Variable{\ell} \Variable{t}
  \middle|
    \RLCOND
  \right)
  \\
  &\leq
  \Pr
  \left(
    \left\|
      \ThresholdSet{\Coords{J'}}( \sum_{i=1}^{m} \Vec{Y}^{(i)} \RV{R}_{i;\Vec{u},\Vec{v}} )
    \right\|_{2}
    \geq
    \sqrt{\KXO \Variable{\ell}}
    +
    \frac{1}{2} \Variable{\ell} \Variable{t}
  \middle|
    \RLCOND
  \right)
  \\
  &\Tab
  +
  \Pr
  \left(
    \left\|
      \ThresholdSet{\Coords{J''}}( \sum_{i=1}^{m} \Vec{Y}^{(i)} \RV{R}_{i;\Vec{u},\Vec{v}} )
    \right\|_{2}
    \geq
    \sqrt{\Variable{d} \Variable{\ell}} + \frac{1}{2} \Variable{\ell} \Variable{t}
  \middle|
    \RLCOND
  \right)
  \\
  &\leq
  2 e^{-\frac{1}{8} \Variable{\ell} \Variable{t}^{2}}
\end{align}
\end{subequations}
\end{proof}


\subsubsection{Proof of Lemma \ref{lemma:appendix:monotonicity-q(s)}}
\label{outline:normal|>concentration-ineq-pfs|>monotonicity}

\begin{lemma}
\label{lemma:appendix:monotonicity-q(s)}
Let \( \RV{X} \) be a random variable with a finite, positive mean
\(
  \mu = \E [ \RV{X} ]
\)
and a density function \( f \) of the form
\begin{gather}
  f(x)
  =
  \begin{cases}
    \sqrt{\frac{2}{\pi}}
    e^{-\frac{x^{2}}{2}} p(x), &\cIf x \geq 0, \\
    0, &\cIf x < 0,
  \end{cases}
\end{gather}
where the image of the function
\(
  p : \R \to \R
\)
is given by
\(
  p(x)
  =
  \frac{\pi}{\theta}
  \frac{1}{\sqrt{2\pi}}
  \int_{y=-x \tan( \frac{\theta}{2} )}^{y=x \tan( \frac{\theta}{2} )}
  e^{-\frac{y^{2}}{2}}
  dy
\)
for
\(
  x \in \R
\).
Define the functions
\(
  q,r : \R \to \R
\)
by
\begin{gather}
  \label{eqn:appendix:monotonicity-q(s):q(s)-def}
  q(s) = \E_{\RV{X} \sim f} \left[ e^{s (\RV{X} - \mu)} e^{-\frac{s^{2}}{2}} \right]
  \\
  \label{eqn:appendix:monotonicity-q(s):r(s)-def}
  r(s) = \E_{\RV{X} \sim f} \left[ e^{-s (\RV{X} - \mu)} e^{-\frac{s^{2}}{2}} \right]
\end{gather}
for
\(
  s \in \R
\).
Then, \( q(s) \) and \( r(s) \) monotonically decrease with \( s \) over the interval
\(
  s \in [0,\infty)
\).
\end{lemma}

\begin{proof}
{Lemma \ref{lemma:appendix:monotonicity-q(s)}}
\label{pf:lemma:appendix:monotonicity-q(s)}
Let
\(
  s \in \R
\),
\(
  f,p,q,r : \R \to \R
\)
be satisfy the conditions of the lemma.
Notice that \( q,r \) can be expressed as
\begin{gather}
\label{deriv:lemma:appendix:monotonicity-q(s):1}
  q(s)
  =
  \int_{x=-\infty}^{x=\infty}
  e^{s(x-\mu)}
  e^{-\frac{s^{2}}{2}}
  f(x) dx
  =
  \int_{x=0}^{x=\infty}
  \sqrt{\frac{2}{\pi}}
  e^{-s\mu}
  e^{-\frac{(x-s)^{2}}{2}}
  p(x)
  dx
  \\
\label{deriv:lemma:appendix:monotonicity-q(s):2}
  r(s)
  =
  \int_{x=-\infty}^{x=\infty}
  e^{-s(x - \mu)}
  e^{-\frac{s^{2}}{2}}
  f(x) dx
  =
  \int_{x=0}^{x=\infty}
  \sqrt{\frac{2}{\pi}}
  e^{s\mu}
  e^{-\frac{(x+s)^{2}}{2}}
  p(x)
  dx
\end{gather}
%
The functions \( q,r \) can be shown to (non-strictly) monotonically decrease with \( s \) over
the interval
\(
  s \in [0,\infty)
\)
by verifying that their partial derivatives w.r.t. \( s \) are non-positive on this interval,
which will be argued by contradiction.
First, suppose \( q(s) \) is not monotonically decreasing with \( s \) over all
\(
  s \geq 0
\),
such that there exists
\(
  s' \geq 0
\)
for which
\(
  \frac{\partial}{\partial s} q(s) \big|_{s=s'} > 0
\).
Write
\(
  p'(a,b)
  =
  \frac{\pi}{\theta}
  \frac{1}{\sqrt{2\pi}}
  \int_{a \tan( \frac{\theta}{2} )}^{b \tan( \frac{\theta}{2} )}
  e^{-\frac{y^{2}}{2}}
  dy
\),
\(
  a \leq b \in \R
\),
and notice that
\(
  p'(a,b) \leq p'(0,b-a)
\).
Then, observe
\begin{subequations}
\begin{align}
\label{deriv:lemma:appendix:monotonicity-q(s):3:1}
  &
  \left. \frac{\partial}{\partial s} q(s) \right|_{s=s'}
  \\
  &=
  \left.
  \frac{\partial}{\partial s}
  \int_{x=0}^{x=\infty}
  \sqrt{\frac{2}{\pi}}
  e^{-s\mu}
  e^{-\frac{(x-s)^{2}}{2}}
  p(x)
  dx
  \right|_{s=s'}
  \\
  &=
  \left.
  \int_{x=0}^{x=\infty}
  \frac{\partial}{\partial s}
  \sqrt{\frac{2}{\pi}}
  e^{-s\mu}
  e^{-\frac{(x-s)^{2}}{2}}
  p(x)
  dx
  \right|_{s=s'}
  \\
  &=
  \left.
  \int_{x=0}^{x=\infty}
  (x - s - \mu)
  \sqrt{\frac{2}{\pi}}
  e^{-s\mu}
  e^{-\frac{(x-s)^{2}}{2}}
  p(x)
  dx
  \right|_{s=s'}
  \\
  &=
  \int_{x=0}^{x=\infty}
  (x - s' - \mu)
  \sqrt{\frac{2}{\pi}}
  e^{-s\mu}
  e^{-\frac{(x-s')^{2}}{2}}
  p(x)
  dx
  \\
  &=
  e^{-s'\mu}
  \int_{x=0}^{x=\infty}
  (x - s' - \mu)
  \sqrt{\frac{2}{\pi}}
  e^{-\frac{(x-s')^{2}}{2}}
  p(x)
  dx
  \\
  &=
  e^{-s'\mu}
  \int_{u=-s'}^{u=\infty}
  (u - \mu)
  \sqrt{\frac{2}{\pi}}
  e^{-\frac{u^{2}}{2}}
  p(u+s')
  du
  ,\dCmt u=x-s'
  \\
  &=
  e^{-s'\mu}
  \int_{u=-s'}^{u=\infty}
  (u - \mu)
  \sqrt{\frac{2}{\pi}}
  e^{-\frac{u^{2}}{2}}
  \left( p(u) + 2p'(u,u+s') \right)
  du
  \\
  &=
  e^{-s'\mu}
  \left(
    \int_{u=-s'}^{u=\infty}
    u
    \sqrt{\frac{2}{\pi}}
    e^{-\frac{u^{2}}{2}}
    \left( p(u) + 2p'(u,u+s') \right)
    du
    -
    \mu
    \int_{u=-s'}^{u=\infty}
    \sqrt{\frac{2}{\pi}}
    e^{-\frac{u^{2}}{2}}
    \left( p(u) + 2p'(u,u+s') \right)
    du
  \right)
  \\ \label{deriv:lemma:appendix:monotonicity-q(s):3:2}
  &=
  e^{-s'\mu}
  \left(
    \int_{u=-s'}^{u=0}
    u
    \sqrt{\frac{2}{\pi}}
    e^{-\frac{u^{2}}{2}}
    \left( p(u) + 2p'(u,u+s') \right)
    du
    +
    \int_{u=0}^{u=\infty}
    u
    \sqrt{\frac{2}{\pi}}
    e^{-\frac{u^{2}}{2}}
    \left( p(u) + 2p'(u,u+s') \right)
    du
  \right.
  \\ \nonumber
  &\Tab
  \left.
    -
    \mu
    \int_{u=-s'}^{u=0}
    \sqrt{\frac{2}{\pi}}
    e^{-\frac{u^{2}}{2}}
    \left( p(u) + 2p'(u,u+s') \right)
    du
    -
    \mu
    \int_{u=0}^{u=\infty}
    \sqrt{\frac{2}{\pi}}
    e^{-\frac{u^{2}}{2}}
    \left( p(u) + 2p'(u,u+s') \right)
    du
  \right)
  \\
  &\leq
  e^{-s'\mu}
  \left(
    \int_{u=0}^{u=\infty}
    u
    \sqrt{\frac{2}{\pi}}
    e^{-\frac{u^{2}}{2}}
    \left( p(u) + 2p'(u,u+s') \right)
    du
    -
    \mu
    \int_{u=0}^{u=\infty}
    \sqrt{\frac{2}{\pi}}
    e^{-\frac{u^{2}}{2}}
    \left( p(u) + 2p'(u,u+s') \right)
    du
  \right)
  ,
  \\ \nonumber
  &\dCmt
  \text{the first integral in \eqref{deriv:lemma:appendix:monotonicity-q(s):3:2} is nonpositive;
  the third is nonnegative}
  \\
  &\leq
  e^{-s'\mu}
  \left(
    \int_{u=0}^{u=\infty}
    u
    \sqrt{\frac{2}{\pi}}
    e^{-\frac{u^{2}}{2}}
    \left( p(u) + 2p'(0,s') \right)
    du
    -
    \mu
    \int_{u=0}^{u=\infty}
    \sqrt{\frac{2}{\pi}}
    e^{-\frac{u^{2}}{2}}
    \left( p(u) + 2p'(0,s') \right)
    du
  \right)
  \\ \nonumber
  &\dCmt \text{at } s=s',\ \frac{\partial}{\partial s} q(s) > 0 \text{ by assumption}
  \\
  &=
  e^{-s'\mu}
  \left(
    \int_{u=0}^{u=\infty}
    u
    \sqrt{\frac{2}{\pi}}
    e^{-\frac{u^{2}}{2}}
    p(u)
    du
    +
    2 p'(0,s')
    \int_{u=0}^{u=\infty}
    u
    \sqrt{\frac{2}{\pi}}
    e^{-\frac{u^{2}}{2}}
    du
  \right.
  \\
  &\Tab[4]
  \left.
    -
    \mu
    \int_{u=0}^{u=\infty}
    \sqrt{\frac{2}{\pi}}
    e^{-\frac{u^{2}}{2}}
    p(u)
    du
    -
    2 \mu p'(0,s')
    \int_{u=0}^{u=\infty}
    \sqrt{\frac{2}{\pi}}
    e^{-\frac{u^{2}}{2}}
    du
  \right)
  \\
  &=
  e^{-s'\mu}
  \left(
    \int_{u=0}^{u=\infty}
    u f(u) du
    +
    2 p'(0,s')
    \int_{u=0}^{u=\infty}
    u f_{|Z|}(u) du
    -
    \mu
    \int_{u=0}^{u=\infty}
    f(u) du
    -
    2 \mu p'(0,s')
    \int_{u=0}^{u=\infty}
    f_{|Z|}(u) du
  \right)
  \\
  &=
  e^{-s'\mu}
  \left( \mu + 2 \sqrt{\frac{2}{\pi}} p'(0,s') - \mu - 2 \mu p'(0,s') \right)
  \\
  &=
  e^{-s'\mu}
  \left( (\mu - \mu) + 2 p'(0,s') (\sqrt{\frac{2}{\pi}} - \mu) \right)
  \\
  &\leq
  0
  ,\dCmt \text{equality only if \( \theta = \pi \)}
\label{deriv:lemma:appendix:monotonicity-q(s):3:end}
\end{align}
\end{subequations}
But this shows that
\(
  \frac{\partial}{\partial s} q(s) \big|_{s=s'} \leq 0
\)
which is a contradiction.
Hence, monotonicity of \( q \) holds.
\par
Now consider \( r(s) \), and again assume there exists
\(
  s' \geq 0
\)
such that
\(
  \frac{\partial}{\partial s} r(s) \big|_{s=s'} > 0
\).
The following will similarly arrive at a contradiction.
\begin{subequations}
\begin{align}
\label{deriv:lemma:appendix:monotonicity-q(s):4:1}
  &
  \left. \frac{\partial}{\partial s} r(s) \right|_{s=s'}
  \\
  &=
  \left.
  \frac{\partial}{\partial s}
  \int_{x=0}^{x=\infty}
  \sqrt{\frac{2}{\pi}}
  e^{s\mu}
  e^{-\frac{(x+s)^{2}}{2}}
  p(x)
  dx
  \right|_{s=s'}
  \\
  &=
  \left.
  \int_{x=0}^{x=\infty}
  \frac{\partial}{\partial s}
  \sqrt{\frac{2}{\pi}}
  e^{s\mu}
  e^{-\frac{(x+s)^{2}}{2}}
  p(x)
  dx
  \right|_{s=s'}
  \\
  &=
  \left.
  \int_{x=0}^{x=\infty}
  (\mu-s-x)
  \sqrt{\frac{2}{\pi}}
  e^{s\mu}
  e^{-\frac{(x+s)^{2}}{2}}
  p(x)
  dx
  \right|_{s=s'}
  \\
  &=
  \int_{x=0}^{x=\infty}
  (\mu-s-x)
  \sqrt{\frac{2}{\pi}}
  e^{s\mu}
  e^{-\frac{(x+s')^{2}}{2}}
  p(x)
  dx
  \\
  &\leq
  \int_{x=0}^{x=\infty}
  (\mu-s-x)
  \sqrt{\frac{2}{\pi}}
  e^{s'\mu}
  e^{-\frac{(x+s')^{2}}{2}}
  p(x)
  dx
  ,\dCmt \text{at } s=s',\ \frac{\partial}{\partial s} r(s) > 0 \text{ by assumption}
  \\
  &=
  e^{s'\mu}
  \int_{x=0}^{x=\infty}
  (\mu-s-x)
  \sqrt{\frac{2}{\pi}}
  e^{-\frac{(x+s')^{2}}{2}}
  p(x)
  dx
  \\
  &=
  e^{s'\mu}
  \int_{u=s'}^{u=\infty}
  (\mu-u)
  \sqrt{\frac{2}{\pi}}
  e^{-\frac{u^{2}}{2}}
  p(u-s')
  du
  ,\dCmt u=x+s'
  \\
  &\leq
  e^{s'\mu}
  \int_{u=s'}^{u=\infty}
  (\mu-u)
  \sqrt{\frac{2}{\pi}}
  e^{-\frac{u^{2}}{2}}
  p(u)
  du
  ,\dCmt \text{equality only if \( s'=0 \)}
  \\ \label{deriv:lemma:appendix:monotonicity-q(s):4:2}
  &=
  e^{s'\mu}
  \left(
    \int_{u=0}^{u=\infty}
    (\mu-u)
    \sqrt{\frac{2}{\pi}}
    e^{-\frac{u^{2}}{2}}
    p(u)
    du
    -
    \int_{u=0}^{u=s'}
    (\mu-u)
    \sqrt{\frac{2}{\pi}}
    e^{-\frac{u^{2}}{2}}
    p(u)
    du
  \right)
  \\
  &\leq
  e^{s'\mu}
  \int_{u=0}^{u=\infty}
  (\mu-u)
  \sqrt{\frac{2}{\pi}}
  e^{-\frac{u^{2}}{2}}
  p(u)
  du
  ,\dCmt
  \text{the right integral in \eqref{deriv:lemma:appendix:monotonicity-q(s):4:2} is nonnegative}
  \\
  &=
  e^{s'\mu}
  \left(
    \mu
    \int_{u=0}^{u=\infty}
    \sqrt{\frac{2}{\pi}}
    e^{-\frac{u^{2}}{2}}
    p(u)
    du
    -
    \int_{u=0}^{u=\infty}
    u
    \sqrt{\frac{2}{\pi}}
    e^{-\frac{u^{2}}{2}}
    p(u)
    du
  \right)
  \\
  &=
  e^{s'\mu}
  \left(
    \mu
    \int_{u=0}^{u=\infty}
    f(u) du
    -
    \int_{u=0}^{u=\infty}
    u f(u) du
  \right)
  \\
  &=
  e^{s'\mu} \left( \mu - \mu \right)
  =
  0.
\label{deriv:lemma:appendix:monotonicity-q(s):4:end}
\end{align}
\end{subequations}
%
Thus,
\(
  \frac{\partial}{\partial s} r(s) \big|_{s=s'} \leq 0
\)
implies
\(
  \frac{\partial}{\partial s} r(s) \big|_{s=s'} \leq 0
\),
a contradiction.
Therefore, the monotonicity of \( r \) also holds.
\end{proof}


\subsubsection{Proof of Lemma \ref{lemma:technical:concentration-ineq:counting}}
\label{outline:technical|>concentration-ineq-pfs|>counting}

\begin{lemma*}
[Lemma \ref{lemma:technical:concentration-ineq:counting}]
Fix
\(
  \Variable{t} \in (0,1)
\).
%
Let
\(
  \Vec{u}, \Vec{v} \in \R^{n}
\),
and define the random variable
\(  \RV{L}_{\Vec{u},\Vec{v}} = \left\|   \frac{1}{2} ( \sgn( \MeasMat \Vec{u} ) - \sgn( \MeasMat \Vec{v} ) ) \right\|_{0}  \),
as in \LEMMA \ref{lemma:technical:concentration-ineq:(u,v)}.
Then,
\begin{gather}
  \mu_{\RV{L}_{\Vec{u},\Vec{v}}}
  =
  \E \left[ \RV{L}_{\Vec{u},\Vec{v}} \right]
  =
  \frac{\theta_{\Vec{u},\Vec{v}} m}{\pi}
\end{gather}
and
\begin{gather}
  \Pr
  \left(
    \RV{L}_{\Vec{u},\Vec{v}}
    \notin
    \left[
      (1 - \Variable{t}) \mu_{\RV{L}_{\Vec{u},\Vec{v}}},
      (1 + \Variable{t}) \mu_{\RV{L}_{\Vec{u},\Vec{v}}}
    \right]
  \right)
  \leq
  2e^{-\frac{1}{3} \mu_{\RV{L}_{\Vec{u},\Vec{v}}} \Variable{t}^{2}}
.\end{gather}
\end{lemma*}

\begin{proof}
{Lemma \ref{lemma:technical:concentration-ineq:counting}}
\label{pf:lemma:technical:concentration-ineq:counting}
Denote the indicator random variables,
\(
  \RV{I}_{i}
  =
  \I{\Sgn{}( \langle \MeasVec^{(i)}, \Vec{u} \rangle ) \neq \Sgn{}( \langle \MeasVec^{(i)}, \Vec{v} \rangle )}
\),
\(
  i \in [m]
\).
%
By \LEMMA \ref{lemma:prob-normal-vector-mismatch}, each \( i\Th \) indicator random variable, \(  \RV{I}_{i}  \), \(  i \in [m]  \), has
\begin{gather}
  \Pr \left( \RV{I}_{i} = 1 \right) = \frac{\theta_{\Vec{u},\Vec{v}}}{\pi}
.\end{gather}
%
As seen earlier in the proof of \LEMMA \ref{lemma:normal:concentration-ineq:proj_u-v},
\(  \frac{1}{2} ( \Sgn{}( \langle \MeasVec^{(i)}, \Vec{u} \rangle ) - \Sgn{}( \langle \MeasVec^{(i)}, \Vec{u} \rangle ) ) \neq 0  \)
precisely when
\(  \Sgn{}( \langle \MeasVec^{(i)}, \Vec{u} \rangle ) \neq \Sgn{}( \langle \MeasVec^{(i)}, \Vec{v} \rangle )  \).
%
Hence,
\(  \RV{I}_{i} = \I{ \frac{1}{2} ( \Sgn{}( \langle \MeasVec^{(i)}, \Vec{u} \rangle ) - \Sgn{}( \langle \MeasVec^{(i)}, \Vec{u} \rangle ) ) \neq 0 }  \).
%
It follows that
\(
  \RV{L}_{\Vec{u},\Vec{v}} = \sum_{i=1}^{m} \RV{I}_{i}
\),
and by the linearity of expectation and the fact that the random variables
\(
  \{ \RV{I}_{i} \}_{i \in [m]}
\)
are i.i.d.,
\begin{gather}
  \mu_{\RV{L}_{\Vec{u},\Vec{v}}} = \E \left[ \RV{L}_{\Vec{u},\Vec{v}} \right] = \frac{\theta_{\Vec{u},\Vec{v}} m}{\pi}
\end{gather}
as desired.
Using standard Chernoff bounds, for any
\(
  \Variable{t} \in (0,1)
\),
\begin{gather}
\label{eqn:pf:lemma:technical:concentration-ineq:counting:1:a}
  \Pr \left( \RV{L}_{\Vec{u},\Vec{v}} < (1 - \Variable{t}) \mu_{\RV{L}_{\Vec{u},\Vec{v}}} \right)
  \leq
  e^{-\frac{1}{2} \mu_{\RV{L}_{\Vec{u},\Vec{v}}} \Variable{t}^{2}}
  ,\\
\label{eqn:pf:lemma:technical:concentration-ineq:counting:1:b}
  \Pr \left( \RV{L}_{\Vec{u},\Vec{v}} > (1 + \Variable{t}) \mu_{\RV{L}_{\Vec{u},\Vec{v}}} \right)
  \leq
  e^{-\frac{1}{3} \mu_{\RV{L}_{\Vec{u},\Vec{v}}} \Variable{t}^{2}}
,\end{gather}
and via a union bound over \EQNS \eqref{eqn:pf:lemma:technical:concentration-ineq:counting:1:a} and \eqref{eqn:pf:lemma:technical:concentration-ineq:counting:1:b}, above,
\begin{gather}
  \Pr
  \left(
    \RV{L}_{\Vec{u},\Vec{v}}
    \notin
    \left[
      (1 - \Variable{t}) \mu_{\RV{L}_{\Vec{u},\Vec{v}}},
      (1 + \Variable{t}) \mu_{\RV{L}_{\Vec{u},\Vec{v}}}
    \right]
  \right)
  \leq
  e^{-\frac{1}{2} \mu_{\RV{L}_{\Vec{u},\Vec{v}}} \Variable{t}^{2}}
  +
  e^{-\frac{1}{3} \mu_{\RV{L}_{\Vec{u},\Vec{v}}} \Variable{t}^{2}}
  \leq
  2 e^{-\frac{1}{3} \mu_{\RV{L}_{\Vec{u},\Vec{v}}} \Variable{t}^{2}}
,\end{gather}
as claimed.
\end{proof}


\subsubsection{Proof of Lemma \ref{lemma:technical:concentration-ineq:lbe-local-deviations:union}}
\label{outline:technical|>concentration-ineq-pfs|>lbe-local-deviations:union}

\begin{proof}
{\LEMMA \ref{lemma:technical:concentration-ineq:lbe-local-deviations:union}}
Let us begin by stating the result of \cite[\COR 3.3]{oymak2015near} to which \LEMMA \ref{lemma:technical:concentration-ineq:lbe-local-deviations:union} is a corollary.
%
\begin{lemma}[{equivalent to (part of) \cite[\COR 3.3]{oymak2015near}}]
\label{lemma:technical:concentration-ineq:lbe-local-deviations}
Fix \(  \ConstD = 256  \), and fix \(  \DDeltaX \in (0,1)  \).
For \(  \kX \in \Z_{+}  \), \(  \kX < n  \), let \(  \Set{W} \subseteq \Sphere{n}  \) be a set such that \(  \Set{\hat{W}} \defeq \{ \alpha \Vec{w} : \Vec{w} \in \Set{W}, \alpha \in \R \}  \) is a subspace with \(  \Dim\,\Set{\hat{W}} = \kX  \).
Let \(  \MeasMat \in \R^{m \times n}  \) be a standard Gaussian matrix with i.i.d. entries.
If
\(  m \geq \frac{\EDITX{\ConstD} \kX}{\DDeltaX} \Log( \frac{1}{\DDeltaX} )  \),
then, with probability at least \(  1 - 2 e^{-\frac{1}{64} \DDeltaX m}  \), uniformly for all \(  \Vec{u}, \Vec{v} \in \Set{W}  \) such that
\(  \| \Vec{u} - \Vec{v} \|_{2} \leq \frac{\DDeltaX}{\UnivConstD \sqrt{\Log( 1/\DDeltaX )}}  \),
the number of mismatches satisfies
\(  \| \I{\sgn( \MeasMat \Vec{u} ) \neq \sgn( \MeasMat \Vec{v} )} \|_{0} \leq \DDeltaX m  \).
\end{lemma}
%
Before proceeding with the argument, some notations are introduced.
\EDITX{Define
\(  \kO \defeq \min \{ \KXO, n \}  \).}
%
For \(  \Coords{J} \subseteq [n]  \), define
\EDITX{\(  \Set{W}^{\Coords{J}} \defeq {\Sphere{n} \cap \SparseRealSubspace{\Coords{J}}{n}}  \),}
and let
\EDITX{\(  \Set{\hat{W}}^{\Coords{J}} \defeq \SparseRealSubspace{\Coords{J}}{n}  \).}
%
\EDITX{Note that for any pair \(  \Vec{u}, \Vec{v} \in \SparseSphereSubspace{k}{n}  \), there exists a coordinate subset \(  \Coords{J} \subseteq [n]  \), \(  | \Coords{J} | = \kO  \), \(  \Supp( \Vec{u} ) \cup \Supp( \Vec{v} ) \subseteq \Coords{J}  \), which satisfies \(  \Vec{u}, \Vec{v} \in \Set{W}^{\Coords{J}}  \) since \(  | \Supp( \Vec{u} ) \cup \Supp( \Vec{v} ) | \leq \kO  \).}
Additionally, for any \(  \Vec{u}, \Vec{v} \in \R^{n}  \),
\(  \| \I{\sgn( \MeasMat \Vec{u} ) \neq \sgn( \MeasMat \Vec{v} )} \|_{0} = \| \frac{1}{2} ( \sgn( \MeasMat \Vec{u} ) - \sgn( \MeasMat \Vec{v} ) ) \|_{0} = \RV{L}_{\Vec{u},\Vec{v}}  \),
which will allow the result in \LEMMA \ref{lemma:technical:concentration-ineq:lbe-local-deviations} to be related to that which is sought in \LEMMA \ref{lemma:technical:concentration-ineq:lbe-local-deviations:union}.
The crucial idea for this proof is viewing the set \(  \SparseSphereSubspace{{\kO}}{n}  \) as a union of the sets \(  \Set{W}^{\Coords{J}}  \) for \(  \Coords{J} \subseteq [n]  \), \EDITX{\(  | \Coords{J} | = \kO  \)}, and applying \LEMMA \ref{lemma:technical:concentration-ineq:lbe-local-deviations} to each such \(  \Set{W}^{\Coords{J}}  \), where the corresponding subspace, \(  \Set{\hat{W}}^{\Coords{J}}  \), has dimension \(  \Dim\,\Set{\hat{W}}^{\Coords{J}} \EDITX{=} \kO  \).
With this in mind, fix \(  \Coords{J} \subseteq [n]  \), \EDITX{\(  | \Coords{J} | = \kO  \)}, arbitrarily.
Due to \LEMMA \ref{lemma:technical:concentration-ineq:lbe-local-deviations} and the fact that \(  \Dim\,\Set{\hat{W}}^{\Coords{J}} \EDITX{=} \kO  \), the following holds uniformly with probability at least
\(  1 - 2e^{-\frac{1}{64} \DDeltaX m}  \),
for all \(  \Vec{u}, \Vec{v} \in \Set{W}^{\Coords{J}}  \) such that
\(  \| \Vec{u} - \Vec{v} \|_{2} \leq \frac{\DDeltaX}{2\ConstD \sqrt{\Log( 1/\DDeltaX )}} = \frac{\DDeltaX}{\UnivConstD \sqrt{\Log( 1/\DDeltaX )}}  \):
\(  \RV{L}_{\Vec{u},\Vec{v}} = \| \I{\sgn( \MeasMat \Vec{u} ) \neq \sgn( \MeasMat \Vec{v} )} \|_{0} \leq \DDeltaX m  \).
%
All that remains is union bounding over the subsets
\(  \Set{U} \defeq \{ \Set{W}^{\Coords{J}} \}_{\Coords{J} \subseteq [n] : \EDITX{| \Coords{J} | = \kO}}  \)
and extending the argument to pairs of vectors in \(  \SparseSphereSubspace{k}{n}  \).
The number of these sets comprising \(  \Set{U}  \) is bounded from above by
\begin{align*}
  \left| \Set{U} \right|
  = | \{ \Coords{J} \subseteq [n] : \EDITX{| \Coords{J} | = \kO} \} |
  = \binom{n}{\kO}
.\end{align*}
%
Therefore, by a union bound over \(  \Set{U}  \), the earlier mentioned uniform bound on \(  \RV{L}_{\Vec{u},\Vec{v}}  \) holds for all \(  \Vec{u}, \Vec{v} \in \bigcup_{\Set{W}^{\Coords{J}} \in \Set{U}} \Set{W}^{\Coords{J}}  \) with probability at least
\begin{gather*}
  1 - \EDITX{2 \binom{n}{\kO} e^{-\frac{1}{64} \DDeltaX m}}
.\end{gather*}
%
Lastly, per the earlier discussion, for any \(  \Vec{u}, \Vec{v} \in \SparseSphereSubspace{k}{n}  \), there exists \(  \Set{W}^{\Coords{J}} \in \Set{U}  \) such that \(  \Vec{u}, \Vec{v} \in \Set{W}^{\Coords{J}}  \).
Thus, the same uniform result applies to all \(  \Vec{u}, \Vec{v} \in \SparseSphereSubspace{k}{n}  \), yielding the lemma's result.
\end{proof}


\section{Proof of Fact \ref{fact:misc:error-decay-recurrence}}
\label{outline:misc:error-decay-recurrence}

Recall Fact \ref{fact:misc:error-decay-recurrence} from
Section \ref{outline:biht:pf-main-thm|>intermediate-lemmas-pf|>error}.
%
\begin{fact*}
[Fact \ref{fact:misc:error-decay-recurrence}]
Let
\(
  \Variable{u}, \Variable{v}, \Variable{w}, \Variable{w}_{0} \in \R_{+}
\)
such that
\(
  \Variable{u} = \frac{1}{2} \left( 1 + \sqrt{1 + 4 \Variable{w}} \right)
\),
and
\(
  1 \leq \Variable{u} \leq \sqrt{\frac{2}{\Variable{v}}}
\).
Define the functions
\(
  \Function{f}_{1}, \Function{f}_{2} : \Z_{\geq 0} \to \R
\)
by
\begin{gather*}
  \Function{f}_{1}(0) = 2
  \\
  \Function{f}_{1}(\FunctionVariable{t})
  =
  \Variable{v} \Variable{w} + \sqrt{\Variable{v} \Function{f}_{1}(\FunctionVariable{t}-1)}
  ,\quad \FunctionVariable{t} \in \Z_{+}
  \\
  \Function{f}_{2}(\FunctionVariable{t})
  =
  2^{2^{-\FunctionVariable{t}}} (\Variable{u}^{2} \Variable{v})^{1 - 2^{-\FunctionVariable{t}}}
  ,\quad \FunctionVariable{t} \in \Z_{\geq 0}
.\end{gather*}
%
Then, \( \Function{f}_{1} \) and \( \Function{f}_{2} \) are strictly monotonically decreasing and
asymptotically converges to
\(
  \Variable{u}^{2} \Variable{v}
\).
Moreover, \( \Function{f}_{2} \) pointwise upper bounds \( \Function{f}_{1} \).
Formally,
\begin{gather*}
  \Function{f}_{1}(\FunctionVariable{t}) \leq \Function{f}_{2}(\FunctionVariable{t})
  ,\quad \forall \FunctionVariable{t} \in \Z_{\geq 0}
  \\
  \lim_{\FunctionVariable{t} \to \infty} \Function{f}_{2}(\FunctionVariable{t})
  =
  \lim_{\FunctionVariable{t} \to \infty} \Function{f}_{1}(\FunctionVariable{t})
  =
  \Variable{u}^{2} \Variable{v}
.\end{gather*}
\end{fact*}
%
The verification of the fact will use Fact \ref{fact:misc:nested-sqrt-recurrence}.
%
\begin{fact}
\label{fact:misc:nested-sqrt-recurrence}
Let
\(
  \Variable{u}, \Variable{w}, \Variable{w}_{0} \in \R_{+}
\)
\(
  \Variable{u} = \frac{1}{2} \left( 1 + \sqrt{1 + 4 \Variable{w}} \right)
\).
Define the function
\(
  \Function{f} : \Z_{\geq 0} \to \R
\)
by
\begin{gather}
\label{eqn:misc:nested-sqrt-recurrence:f}
  \Function{f}(0) = \Variable{w}_{0},
  \\
  \Function{f}(\FunctionVariable{t}) = \sqrt{\Variable{w} + \Function{f}(\FunctionVariable{t}-1)}
  ,\quad \FunctionVariable{t} \in \Z_{+}.
\end{gather}
%
Then,
\begin{gather}
  \lim_{\FunctionVariable{t} \to \infty} f(\FunctionVariable{t}) = u
\end{gather}
%
Moreover, when
\(
  \Variable{w}_{0} > u
\)
(%
\(
  \Variable{w}_{0} < u
\),
\(
  \Variable{w}_{0} = u
\)%
),
\( \Function{f} \) strictly monotonically decreases (respectively, strictly monotonically
increases, is constant) with respect to \( \FunctionVariable{t} \).
\end{fact}

\begin{proof}
{Fact \ref{fact:misc:nested-sqrt-recurrence}}
\label{pf:fact:misc:nested-sqrt-recurrence}
Let us first show that \( \Function{f} \) is monotone over
\(
  \FunctionVariable{t} \in \Z_{+}
\).
Write
\begin{gather}
\label{pf:fact:misc:nested-sqrt-recurrence:eqn:1}
  \sgnO(a)
  =
  \begin{cases}
     -1, & \cIf a < 0, \\
    \+0, & \cIf a = 0, \\
    \+1, & \cIf a > 0,
  \end{cases}
\end{gather}
and note that
\(
  \sgnO( \Function{f}^{2}(\FunctionVariable{t}) - \Function{f}^{2}(\FunctionVariable{t}') )
  =
  \sgnO( \Function{f}(\FunctionVariable{t}) - \Function{f}(\FunctionVariable{t}') )
\)
for any
\(
  t, t' \geq 0
\).
Moreover, notice that
\(
  \Function{f}^{2}(\FunctionVariable{t})
  =
  ( \sqrt{\Variable{w} + \Function{f}(\FunctionVariable{t}-1)} )^{2}
  =
  \Variable{w} + \Function{f}(\FunctionVariable{t}-1)
\),
\(
  \FunctionVariable{t} \in \Z_{\geq 0}
\).
The goal will be to show that for each
\(
  \FunctionVariable{t} \in \Z_{+}
\),
the sign of
\(
  \Function{f}(\FunctionVariable{t}) - \Function{f}(\FunctionVariable{t+1})
\)
and
\(
  \Function{f}(\FunctionVariable{t-1}) - \Function{f}(\FunctionVariable{t})
\)
match.
Fix
\(
  \FunctionVariable{t} \in \Z_{+}
\)
arbitrarily, and observe
\begin{align}
  \Function{f}^{2}(\FunctionVariable{t}) - \Function{f}^{2}(\FunctionVariable{t+1})
  &=
  \Variable{w} + \Function{f}(\FunctionVariable{t}-1)
  -
  ( \Variable{w} + \Function{f}(\FunctionVariable{t}) )
  \\
  &=
  \Function{f}(\FunctionVariable{t}-1) - \Function{f}(\FunctionVariable{t})
\end{align}
and thus
\begin{gather}
  \sgnO( \Function{f}(\FunctionVariable{t}) - \Function{f}(\FunctionVariable{t+1}) )
  =
  \sgnO( \Function{f}^{2}(\FunctionVariable{t}) - \Function{f}^{2}(\FunctionVariable{t+1}) )
  =
  \sgnO( \Function{f}(\FunctionVariable{t}-1) - \Function{f}(\FunctionVariable{t}) )
\end{gather}
as desired.
The monotonicity of \( \Function{f} \) over \( \Z_{\geq 0} \) follows.
\par
To find the direction of the monotonicity, it suffices to look at
\(
  \sgnO( \Function{f}(1) - \Function{f}(0) )
\)
since the monotonicity has already been argued.
This can be given by
\begin{gather}
  \sgnO( \Function{f}(1) - \Function{f}(0) )
  =
  \sgnO( \Function{f}^{2}(1) - \Function{f}^{2}(0) )
  =
  \sgnO( \Variable{w} + \Function{f}(0) - \Function{f}^{2}(0) )
  =
  \sgnO( \Variable{w} + \Variable{w}_{0} - \Variable{w}_{0}^{2} )
.\end{gather}
%
To determine from this the condition under which \( \Function{f} \) is constant, observe,
\begin{subequations}
\begin{align}
  &
  \sgnO( \Variable{w} + \Variable{w}_{0} - \Variable{w}_{0}^{2} ) = 0
  \\
  &\dLn
  \Variable{w} + \Variable{w}_{0} - \Variable{w}_{0}^{2} = 0
  \\
  &\dLn
  \Variable{w}_{0}
  \in
  \left\{ \frac{1}{2} ( 1 \pm \sqrt{1 + 4\Variable{w}} ) \right\}
  \\
  &\dLn
  \Variable{w}_{0}
  =
  \frac{1}{2} ( 1 + \sqrt{1 + 4\Variable{w}} )
  =
  \Variable{u}
\end{align}
\end{subequations}
%
\begin{gather}
  \Variable{w} + \Variable{w}_{0} - \Variable{w}_{0}^{2}
  \begin{cases}
    < 0, & \cIf \Variable{w}_{0} > \frac{1}{2} ( 1 + \sqrt{1 + 4\Variable{w}} ), \\
    = 0, & \cIf \Variable{w}_{0} = \frac{1}{2} ( 1 + \sqrt{1 + 4\Variable{w}} ), \\
    > 0, & \cIf \Variable{w}_{0} < \frac{1}{2} ( 1 + \sqrt{1 + 4\Variable{w}} ).
  \end{cases}
\end{gather}
%
Hence, \( \Function{f} \) is strictly monotonically decreasing when
\(
  \Variable{w}_{0} > \Variable{u}
\),
constant when
\(
  \Variable{w}_{0} = \Variable{u}
\),
and strictly monotonically increasing when
\(
  \Variable{w}_{0} > \Variable{u}
\),
as claimed.
\par
The final step is to determine the asymptotic behavior of \( \Function{f} \) as
\( \FunctionVariable{t} \to \infty \).
If
\(
  \Variable{w}_{0} = \Variable{u}
\),
then \( \Function{f} \) is constant, implying that
\(
  \lim_{\FunctionVariable{t} \to \infty} \Function{f}(\FunctionVariable{t})
  = \Function{f}(0)
  = \Variable{w}_{0}
  = \Variable{u}
\).
On the other hand, when
\(
  \Variable{w}_{0} \neq \Variable{u}
\)
we would like to characterize some behavior such as
\begin{gather}
  \lim_{\FunctionVariable{t} \to \infty}
  \Function{f}^{2}(\FunctionVariable{t}+1) - \Function{f}^{2}(\FunctionVariable{t})
  = 0
\end{gather}
%
Observe,
\begin{subequations}
\begin{align}
  &
  \Function{f}^{2}(\FunctionVariable{t}+1) - \Function{f}^{2}(\FunctionVariable{t}) = 0
  \\
  &\dLn
  \Variable{w} + \Function{f}(\FunctionVariable{t}) - \Function{f}^{2}(\FunctionVariable{t}) = 0
  \\
  &\dLn
  \Function{f}(\FunctionVariable{t})
  =
  \frac{1}{2} ( 1 + \sqrt{1 + 4 \Variable{w}} )
  =
  \Variable{u}
\end{align}
\end{subequations}
%
Hence, if
\(
  \Variable{w}_{0} > \Variable{u}
\),
the strictly monotonically decreasing function is lower bounded by
\(
  \inf_{\FunctionVariable{t} \in \Z_{\geq 0}} \Function{f}(\FunctionVariable{t}) = \Variable{u}
\),
while the strictly monotonically increasing function is upper bounded by
\(
  \sup_{\FunctionVariable{t} \in \Z_{\geq 0}} \Function{f}(\FunctionVariable{t}) = \Variable{u}
\)
when
\(
  \Variable{w}_{0} < \Variable{u}
\).
But in both cases, the function has strict monotonicity, and therefore it must happen that
\(
  \lim_{\FunctionVariable{t} \to \infty} \Function{f}(\FunctionVariable{t}) = \Variable{u}
\).
\end{proof}

\begin{proof}
{Fact \ref{fact:misc:error-decay-recurrence}}
In addition to defining \( \Function{f}_{1} \) and \( \Function{f}_{2} \) as in
Fact \ref{fact:misc:error-decay-recurrence}, let
\(
  \Function{f} : \Z_{\geq 0} \to \R
\)
be the function as defined in Fact \ref{fact:misc:nested-sqrt-recurrence}, which is given by the
recurrence relation
\begin{gather}
  \Function{f}(0) = \Variable{w}_{0}
  \\
  \Function{f}(\FunctionVariable{t}) = \sqrt{w + \Function{f}(\FunctionVariable{t}-1)}
\end{gather}
where for the purposes of this proof, \( \Variable{w}_{0} \) is fixed as
\(
  \Variable{w}_{0} = \sqrt{\frac{2}{\Variable{v}}}
\).
Notice that the function \( \Function{f}_{1} \) can be written as
\begin{gather}
  \Function{f}_{1}(\FunctionVariable{t})
  =
  \Variable{v} \Variable{w} + \sqrt{\Variable{v} \Function{f}_{1}(\FunctionVariable{t}-1)}
  =
  \Variable{v}
  \left( \Variable{w} + \sqrt{\frac{\Function{f}_{1}(\FunctionVariable{t}-1)}{\Variable{v}}} \right)
  =
  \Variable{v}
  \left( \Variable{w} + \Function{f}(\FunctionVariable{t}-1) \right)
  =
  \Variable{v} \Function{f}^{2}(\FunctionVariable{t})
\end{gather}
%
Then, the monotonicity and asymptotic behavior of the functions \( \Function{f}_{1} \) follow
directly from Fact \ref{fact:misc:nested-sqrt-recurrence}.
\begin{gather}
  \lim_{\FunctionVariable{t} \to \infty} \Function{f}_{1}(\FunctionVariable{t})
  =
  \lim_{\FunctionVariable{t} \to \infty}
  \Variable{v} \Function{f}^{2}(\FunctionVariable{t})
  =
  \Variable{u}^{2} \Variable{v}
\end{gather}
On the other hand, for \( \Function{f}_{2} \),
\begin{gather}
  \lim_{\FunctionVariable{t} \to \infty} \Function{f}_{2}(\FunctionVariable{t})
  =
  \lim_{\FunctionVariable{t} \to \infty}
  2^{2^{-\FunctionVariable{t}}} (\Variable{u}^{2} \Variable{v})^{1 - 2^{-\FunctionVariable{t}}}
  =
  1 \cdot \Variable{u}^{2} \Variable{v}
  =
  \Variable{u}^{2} \Variable{v}
\end{gather}
%
\par
The function \( \Function{f}_{2} \) can be shown inductively to pointwise upper bound
\( \Function{f}_{1} \).
The base case,
\(
  \FunctionVariable{t} = 0
\),
is trivial since
\(
  \Function{f}_{2}(0)
  = 2^{2^{0}} (\Variable{u}^{2} \Variable{v})^{1 - 2^{0}}
  = 2 \cdot 1
  = 2
  = \Function{f}_{1}(0)
\).
Letting
\(
  \FunctionVariable{t} \in \Z_{+}
\),
suppose that for each
\(
  \FunctionVariable{t'} \in \{ 2,\dots,\FunctionVariable{t}-1 \}
\),
the bound
\(
  \Function{f}_{1}(\FunctionVariable{t'}) \leq \Function{f}_{2}(\FunctionVariable{t'})
\)
holds.
Then, the desired result will follow from induction if it is shown that
\(
  \Function{f}_{1}(\FunctionVariable{t}) \leq \Function{f}_{2}(\FunctionVariable{t})
\).
To verify this, note that \( \Function{f}_{2} \) can be written as the following recurrence
relation
\begin{gather}
  \Function{f}_{2}(0) = 2
  \\
  \Function{f}_{2}(\FunctionVariable{t})
  = \sqrt{\Variable{u}^{2} \Variable{v} \Function{f}_{2}(\FunctionVariable{t}-1)}
\end{gather}
since it was already argued that
\(
  \Function{f}_{2}(0) = 2
\)
and otherwise for
\(
  \FunctionVariable{t} \in \Z_{+}
\),
\begin{subequations}
\begin{align}
  \sqrt{\Variable{u}^{2} \Variable{v} \Function{f}_{2}(\FunctionVariable{t}-1)}
  &=
  \left( \Variable{u}^{2} \Variable{v} \right)^{\frac{1}{2}}
  \left( \Function{f}_{2}(\FunctionVariable{t}-1) \right)^{\frac{1}{2}}
  \\
  &=
  \left( \Variable{u}^{2} \Variable{v} \right)^{\frac{1}{2}}
  \left( \Variable{u}^{2} \Variable{v} \right)^{\frac{1}{2^{2}}}
  \left( \Function{f}_{2}(\FunctionVariable{t}-2) \right)^{\frac{1}{2^{2}}}
  =
  \left( \Variable{u}^{2} \Variable{v} \right)^{\frac{1}{2} + \frac{1}{2^{2}}}
  \left( \Function{f}_{2}(\FunctionVariable{t}-2) \right)^{\frac{1}{2^{2}}}
  \\
  &=
  \left( \Variable{u}^{2} \Variable{v} \right)^{\frac{1}{2} + \frac{1}{2^{2}} + \frac{1}{2^{3}}}
  \left( \Function{f}_{2}(\FunctionVariable{t}-3) \right)^{\frac{1}{2^{3}}}
  \\
  &\vdots
  \\
  &=
  \left( \Variable{u}^{2} \Variable{v} \right)^{\sum_{s=1}^{\FunctionVariable{t'}} 2^{-s}}
  \left( \Function{f}_{2}(\FunctionVariable{t}-\FunctionVariable{t'}) \right)^{2^{-t'}}
  \\
  &\vdots
  \\
  &=
  \left( \Variable{u}^{2} \Variable{v} \right)^{\sum_{s=1}^{\FunctionVariable{t}} 2^{-s}}
  \left(
    \Function{f}_{2}(\FunctionVariable{t}-\FunctionVariable{t})
  \right)^{2^{-\FunctionVariable{t}}}
  =
  \left( \Variable{u}^{2} \Variable{v} \right)^{\sum_{s=1}^{\FunctionVariable{t}} 2^{-s}}
  \left(
    \Function{f}_{2}(0)
  \right)^{2^{-\FunctionVariable{t}}}
  =
  2^{2^{-\FunctionVariable{t}}}
  \left( \Variable{u}^{2} \Variable{v} \right)^{1-2^{-\FunctionVariable{t}}}
  \\
  &=
  \Function{f}_{2}(\FunctionVariable{t})
\end{align}
\end{subequations}
as desired.
With the above argument, it suffices to show that
\(
  \Function{f}_{1}(\FunctionVariable{t})
  \leq
  \sqrt{\Variable{u}^{2} \Variable{v} \Function{f}_{2}(\FunctionVariable{t}-1)}
\).
Note that
\begin{subequations}
\begin{align}
  &
  \Variable{u}^{2}
  =
  \frac{1}{4}
  \left( 1 + \sqrt{1 + \Variable{w}} \right)^{2}
  =
  \Variable{u} + \Variable{w}
  \\
  &\dLn
  \Variable{w} = \Variable{u}^{2} - \Variable{u}
\end{align}
\end{subequations}
%
Then, observe,
\begin{subequations}
\begin{align}
  \Function{f}_{1}(\FunctionVariable{t})
  -
  \sqrt{\Variable{u}^{2} \Variable{v} \Function{f}_{2}(\FunctionVariable{t}-1)}
  &=
  \Variable{v} \Variable{w} + \sqrt{\Variable{v} \Function{f}_{1}(\FunctionVariable{t}-1)}
  -
  \sqrt{\Variable{u}^{2} \Variable{v} \Function{f}_{2}(\FunctionVariable{t}-1)}
  \\
  &\leq
  \Variable{v} \Variable{w} + \sqrt{\Variable{v} \Function{f}_{2}(\FunctionVariable{t}-1)}
  -
  \sqrt{\Variable{u}^{2} \Variable{v} \Function{f}_{2}(\FunctionVariable{t}-1)}
  ,\dCmt \Text{by the inductive hypothesis}
  \\
  &=
  \Variable{v} \left( \Variable{u}^{2} - \Variable{u} \right)
  + \sqrt{\Variable{v} \Function{f}_{2}(\FunctionVariable{t}-1)}
  - \sqrt{\Variable{u}^{2} \Variable{v} \Function{f}_{2}(\FunctionVariable{t}-1)}
  \\
  &=
  \Variable{v} \Variable{u}^{2}
  - \Variable{v} \Variable{u}
  + \sqrt{\Variable{v} \Function{f}_{2}(\FunctionVariable{t}-1)}
  - \Variable{u} \sqrt{\Variable{v} \Function{f}_{2}(\FunctionVariable{t}-1)}
  \\
  &=
  (\Variable{u} - 1) \Variable{u} \Variable{v}
  -
  (\Variable{u} - 1) \sqrt{\Variable{v} \Function{f}_{2}(\FunctionVariable{t}-1)}
  \\
  &\leq
  (\Variable{u} - 1) \Variable{u} \Variable{v}
  -
  (\Variable{u} - 1) \sqrt{\Variable{v} ( \Variable{u}^{2} \Variable{v} )}
  =
  0.
\end{align}
\end{subequations}
Hence,
\begin{gather}
  \Function{f}_{1}(\FunctionVariable{t})
  -
  \sqrt{\Variable{u}^{2} \Variable{v} \Function{f}_{2}(\FunctionVariable{t}-1)}
  \leq 0
  \Longrightarrow
  \Function{f}_{1}(\FunctionVariable{t})
  \leq \sqrt{\Variable{u}^{2} \Variable{v} \Function{f}_{2}(\FunctionVariable{t}-1)}
  = \Function{f}_{2}(\FunctionVariable{t})
\end{gather}
%
By induction,
\(
  \Function{f}_{1}(\FunctionVariable{t}) \leq \Function{f}_{2}(\FunctionVariable{t})
\)
for every
\(
  \FunctionVariable{t} \in \Z_{\geq 0}
\).
\end{proof}


\section{A Different Invertibility Condition~\cite{friedlander2021nbiht}}
\label{outline:technical:raic:friedlander}

\begin{definition}
[{restricted approximate invertibility condition as defined in \cite[Def. 8]{friedlander2021nbiht}}]
\label{def:raic:friedlander}
Fix
\(
  \nu, \delta, \eta, r, r' > 0
\).
Let
\(
  \MeasMat \in \R^{m \times n}
\)
be a measurement matrix, and let
\(
  \Vec{x} \in \SparseSphereSubspace{k}{n}
\).
The \( (\nu, \delta, \eta, r, r') \)-RAIC holds for \( \MeasMat \) at \( \Vec{x} \) if for every
\(
  \Vec{y} \in \SparseSphereSubspace{k}{n}
\),
\(
  r \leq \DistS{\Vec{x}}{\Vec{y}} \leq r'
\),
\begin{gather}
  \left\|
    (\Vec{x} - \Vec{y})
    -
    \nu \MeasMat^{\T}
    \left( \Sgn( \MeasMat \Vec{x} ) - \Sgn( \MeasMat \Vec{y} ) \right)
  \right\|_{(\SparseSphereSubspace{k}{n})^{\circ}}
  \leq
  \delta
  \DistS{\Vec{x}}{\Vec{y}}
  +
  \eta
\end{gather}
where
\(
  \left\| \cdot \right\|_{(\SparseSphereSubspace{k}{n})^{\circ}}
\)
denotes the dual norm given by
\(
  \left\| \Vec{u} \right\|_{(\SparseSphereSubspace{k}{n})^{\circ}}
  =
  \sup_{\Vec{u}' \in \SparseSphereSubspace{k}{n}}
  \langle \Vec{u}, \Vec{u}' \rangle
\)
for
\(
  \Vec{u} \in \R^{n}
\).
\end{definition}

Instead of the $\ell_2$-norm as in our definition, this definition resorts to the dual norm. Furthermore, our definition of RAIC should hold for all pair of vectors uniformly; whereas in the above definition invertibitily condition is asked for vectors within distance $[r,r'].$ Both of these two differences make our definition simpler to state and handle, and also allow us to do a precise analysis in the ``small-distance'' regime.
\end{appendix}

\clearpage

\end{document}